\documentclass[a4paper,11pt,twoside]{article}
\usepackage{fontenc}
\usepackage{amsthm,amsmath,amsfonts,amssymb}
\RequirePackage[colorlinks]{hyperref}
\RequirePackage{epsfig, graphicx, color}
\RequirePackage{latexsym}

\usepackage{a4wide}
\usepackage{epsf,epsfig,subfigure}
\usepackage{natbib}
\usepackage{graphics,graphicx}
\usepackage{enumerate}
\usepackage{algorithmic,algorithm,float}

\usepackage{subfigure}
\usepackage{bbm}

\usepackage{multirow}

\newtheorem{theorem}{Theorem}

\newtheorem{lemma}[theorem]{Lemma}

\newtheorem{corollary}[theorem]{Corollary}

\theoremstyle{definition}

\newtheorem{assumption}{Assumption}

\newtheorem{prop}[theorem]{Proposition}

\theoremstyle{remark}

\algsetup{linenodelimiter=.}

\newcommand{\D}{\mathcal{D}}
\newcommand{\A}{\mathcal{A}}
\newcommand{\R}{\mathbb{R}}
\newcommand{\PP}{\mathbb{P}}
\newcommand{\E}{\mathbb{E}}
\newcommand{\var}{\textrm{Var}}
\newcommand{\ipw}{\textrm{IPW}}

\newcommand{\ave}{\textrm{AVE}}
\newcommand{\dipw}{\textrm{DIPW}}
\newcommand{\mdipw}{\textrm{mDIPW}}

\newcommand{\ora}{\mathrm{ORA}}
\newcommand{\T}{\mathcal{T}}

\newcommand{\abs}[1]{\left| #1\right|}
\newcommand{\norm}[1]{\left\lVert #1 \right\rVert}
\newcommand{\ind}{\mathbbm{1}}
\newcommand{\pr}{\mathbb{P}}

\newcommand{\floor}[1]{\left\lfloor #1 \right\rfloor}

\newcommand{\Var}{\mathrm{Var}}

\newcommand{\ci}{\mathrm{CI}}

\newcommand{\argmin}{\textrm{argmin}}

\newcommand{\kron}{\mathbbm{1}}

\def\diff{\mathrm{d}}
\def\bbP{\mathbb{P}}

\newcommand{\mb}{\mathbf}
\newcommand{\mbb}{\boldsymbol}

\newcommand{\vertiii}[1]{{\left\vert\kern-0.25ex\left\vert\kern-0.25ex\left\vert #1 
    \right\vert\kern-0.25ex\right\vert\kern-0.25ex\right\vert}}
    
\newcommand\independent{\protect\mathpalette{\protect\independenT}{\perp}}
    \def\independenT#1#2{\mathrel{\rlap{$#1#2$}\mkern2mu{#1#2}}}

\newcommand{\rev}[1]{{#1}}
\newcommand{\newrev}[1]{{  #1}}
    
\title{Debiased Inverse Propensity Score Weighting for Estimation of Average Treatment Effects with High-Dimensional Confounders}

\author{Yuhao Wang\thanks{Part of this work was done while YW was at the University of Cambridge, UK. 
	}\\
	Tsinghua University, China\\
	\url{yuhaow@tsinghua.edu.cn}
	\and
	Rajen D. Shah\\
	University of Cambridge, UK\\
	\url{r.shah@statslab.cam.ac.uk}}

\begin{document}

\maketitle
%

\begin{abstract}
 We consider estimation of average treatment effects given observational data with high-dimensional pretreatment variables. Existing methods for this problem typically assume some form of sparsity for the regression functions. In this work, we introduce a debiased inverse propensity score weighting (DIPW) scheme for average treatment effect estimation that delivers $\sqrt{n}$-consistent estimates
 when the propensity score follows a sparse logistic regression model; the outcome regression functions are permitted to be arbitrarily complex.
 \rev{We further demonstrate how confidence intervals centred on our estimates may be constructed.}
 Our theoretical results quantify the price to pay for permitting the regression functions to be unestimable, which shows up as an inflation of the variance of the estimator compared to the semiparametric efficient variance by a constant factor, under mild conditions. \rev{We also show that when outcome regressions can be estimated faster than a slow $1/\sqrt{ \log n}$ rate, our estimator achieves semiparametric efficiency.} 
 As our results accommodate arbitrary outcome regression functions, averages of transformed responses under each treatment may also be estimated at the $\sqrt{n}$ rate. Thus, for example, the variances of the potential outcomes may be estimated.
 We discuss extensions to estimating linear projections of the heterogeneous treatment effect function \rev{and explain how propensity score models with  more general link functions may be handled within our framework. An R package \texttt{dipw} implementing our methodology is available on CRAN}.
\end{abstract}

\section{Introduction} \label{sec:intro}
Estimating average treatment effects from observational data is a central topic in causal inference, and has received a great deal of attention in recent years.
The setup we consider involves data consisting of $n$ i.i.d.\ copies of the triple $(X, Y, T) \in \R^p \times \R \times \{0, 1\}$ where $X$ is a vector of covariates, $Y$ is the observed outcome, and $T$ is a binary treatment indicator. We work in the potential outcomes framework~\citep{Neyman23,Rubin74} and define the average treatment effect as $\tau := \E\{Y(1) - Y(0)\}$ where $Y(0)$ and $Y(1)$ are the potential outcomes such that $Y = Y(T)$.
Identification of $\tau$ relies on the treatment assignment being as good as random conditional on observed covariates, that is $\{Y(0), Y(1)\} \independent T \mid X$. This is an unverifiable assumption, but it may at least be more plausible when the number of covariates $p$ is large.
As a result, there has been growing interest in average treatment effect estimation with high-dimensional covariates where $p$ may potentially exceed the number of observations $n$.

There is a rich literature on doubly robust methods for average treatment effect estimation. Augmented inverse propensity weighting (AIPW)
\citep{RRZ94, RR95, SRR99}
is perhaps  the most prominent approach,
along with targeted maximum likelihood estimation \citep{VR06}. Important modifications of the basic AIPW scheme and associated theory applicable to high-dimensional regimes were proposed in \citet{belloni2014inference, Farrell15,BCFH17,CCD18}.
%
These involve first estimating the nuisance components, that is the propensity score $\pi(x) := \pr(T=1|X=x)$ and the regression functions $r_t(x):=\E\{Y(t) |X=x\}$, $t=0, 1$. In order that the resulting estimator achieves a parametric $1/\sqrt{n}$ rate, the nuisance components are typically required to be estimated at a relatively slower $o(1 / n^{1/4})$ rate. For example, if the propensity score and regression functions follow logistic and linear models respectively with the (maximum) sparsities of the corresponding regression coefficients being $s_\pi$ and $s_r$, this convergence rate requirement entails that $s_\pi, s_r = o(\sqrt{n} / \log(p))$. 

In recent years, there have been several advances that relax these conditions in various ways: for example \citet{NPI18} allow for the linear and logistic models to be approximately sparse, whilst \citet{Tan20b, smucler2019unifying, dukes2020inference} permit one of these models to be misspecified but ask for the corresponding parameter estimates to converge to parameters satisfying the sparsity requirements as above.
The work of \citet{BWZ19} works with a sparse logistic regression model for the propensity score and develops a method that allows for either (i) $s_r = o(\sqrt{n} / \log p)$ and $s_\pi = o(n / \log p)$, or (ii) $s_r = o(n^{3 / 4} / \log p)$ and $s_\pi = o(\sqrt{n} / \log p)$, whilst still retaining a $1/\sqrt{n}$ estimation error rate.
Remarkably \citet{AIW18} relaxes the sparsity requirement on the propensity score altogether by exploiting the structure of the bias in estimation of the sparse linear response functions (where the sparsity is assumed to satisfy $s_r = o(\sqrt{n} / \log p)$) and employing a tailored covariate balancing scheme to eliminate this bias asymptotically. 
\rev{Building on this, \citet{bradic2019minimax} introduces an estimator that additionally accommodates settings where we only have approximate sparsity 
	in both the outcome regression model and propensity score model, the requirement on the latter being that the projection of $(T, X) \mapsto T/\pi(X) -(1-T)/\{1-\pi(X)\}$ onto the set of linear combinations of $T$, $X$ and their interaction, is approximately sparse.
}

Whilst these advances are important developments, they all rely on either outcome regression functions or an approximation thereof being sparse, an assumption which may not always be plausible.
In this work, we introduce a debiased inverse propensity score weighting (DIPW) scheme that allows the regression functions to be completely misspecified: they can for instance be nonlinear and depend on all $p$ variables.
Our method does however rely on sparsity of the propensity score, that is $s_\pi = o(\sqrt{n} / \log p)$. The requirements thus precisely complement those of \citet{AIW18}. One advantage of only making assumptions on the propensity score is that, trivially, we can estimate any functional of the potential outcome $\E \{h(Y(t))\}$ for a given function $h \in \R \to \R$ and any $t \in \{0, 1\}$. To do this we simply transform the outcomes via $h$ and apply our procedure to the transformed data. For example, our method allows for estimation of $\Var\{Y(1)\}$ (see Section~\ref{sec:simvar}).


\subsection{Preliminaries and an overview of our main contributions} \label{sec:contrib}
Recall the basic setup outlined earlier: we have available i.i.d.\ copies $(X_1, Y_1, T_1), \ldots, (X_n, Y_n, T_n)$ of the triple $(X, Y, T) \in \R^p \times \R \times \{0, 1\}$ representing pretreatment covariates, the outcome and a binary treatment indicator. The first component of $X$ will typically be $1$ representing an intercept term. We collect these data into $\mb X \in \R^{n \times p}$, $\mb Y \in \R^n$ and $\mb T \in \{0,1\}^n$.
Throughout the paper we assume unconfoundedness as required for identification of the average treatment effect $\tau$.
\begin{assumption} \label{as:unconfounded}
	Conditioning on the observed covariate vector $X$, $Y(1)$ and  $Y(0)$ are independent from the treatment assignment $T$, that is $\{Y(1), Y(0)\} \independent T \mid X$.
\end{assumption}

\rev{
	We write
	\[
	r_t(x) = \E\{Y(t)| X=x\}, \qquad t=0,1
	\]
	for the outcome regression functions. Then $\tau = \E \{r_1(X) - r_0(X)\}$.}

\rev{Throughout the manuscript (with the exception of Section~\ref{sec:link}),} we will additionally assume a logistic regression model for the propensity score and that we have overlap, that is the propensity scores are bounded away from $0$ and $1$.
\begin{assumption} \label{as:logistic}
	We have
	\begin{align} \label{eq:logistic}
		\pr(T=1|X=x) = \pi(x) = \psi(x^\top  \gamma)
	\end{align}
	and for some $c_\pi > 0$, $c_\pi <\pi(X) < 1-c_\pi$  almost surely.
\end{assumption}
Here $\psi(u) := \{1 + \exp(-u)\}^{-1}$ denotes the standard logistic function. We have in mind the high-dimensional setting where $p$ is large and potentially $p \gg n$, but writing $s:= |\{ j : \gamma_j \neq 0\}|$, we have $s \ll p$. \rev{However, we aim to avoid such sparsity conditions or making any additional structural or smoothness assumptions on  $r_0$ and $r_1$.}

To motivate our approach, first consider the standard inverse propensity weighting (IPW) estimator $\hat{\tau}_{\ipw}$ given by
\begin{align} \label{eq:IPW}
	\hat{\tau}_{\ipw} := \frac{1}{n} \sum_{i=1}^n \frac{T_i Y_i}{\hat{\pi}(X_i)} - \frac{1}{n} \sum_{i=1}^n\frac{(1 - T_i) Y_i}{1 - \hat{\pi}(X_i)},
\end{align}
where $\hat{\pi}$ is an estimate of the propensity score function. The rationale of the IPW estimator is that were $\hat{\pi}$ to be replaced by the true propensity score $\pi$, the estimator would be an average of i.i.d.\ quantities each with expectation exactly equal to the target $\tau$. In high-dimensions, the difficulty in estimating $\pi$ results in the IPW having an unacceptably large bias.

Consider now the following modified IPW estimator
\begin{align} \label{eq:tau_hat}
	\hat{\tau} = \frac{1}{n}\sum_{i=1}^n \frac{T_i (Y_i-\mu_i)}{\hat{\pi}(X_i)} - \frac{1}{n} \sum_{i=1}^n\frac{(1 - T_i) (Y_i-\mu_i)}{1 - \hat{\pi}(X_i)},
\end{align}
where we have subtracted from the $i$th observed outcome $Y_i$ a quantity $\mu_i \in \R$.
Provided $\mu_i \independent T_i \mid \mb X$, $\hat{\tau}$ will retain the oracular unbiasedness of the IPW estimator, that is
\begin{align*}
	\E \left(\frac{T_i (Y_i-\mu_i)}{\pi(X_i)} - \frac{(1 - T_i) (Y_i-\mu_i)}{1 - \pi(X_i)}\right) - \tau = \E \left(\frac{(1-T_i)\mu_i}{1-\pi(X_i)} - \frac{T_i\mu_i}{\pi(X_i)}\right)=0. 
\end{align*}
Whilst this alone confers no advantage, we may attempt to choose $\mbb\mu := (\mu_i)_{i=1}^n$ so as to reduce the bias of the regular IPW estimator. \rev{Indeed, the popular augmented inverse propensity weighting estimator (AIPW) may be recovered from \eqref{eq:tau_hat} through taking $\mu_i = \{1-\hat{\pi}(X_i)\} \hat{r}_1(X_i) +\hat{\pi}(X_i)\hat{r}_0(X_i)$, where $\hat{r}_1$ and $\hat{r}_0$ are estimates of $r_1$ and $r_0$ respectively (see Section~\ref{sec:AIPW} in the appendix for a derivation).
	Thus AIPW chooses $\mbb\mu$ through estimating the function $\mu_{\ora} : \R^p \to \R$ given by
	\begin{equation} \label{eq:mu_ora_def}
		\mu_{\ora}(x) := \{1- \pi(x)\} r_1(x) + \pi(x) r_0(x).
	\end{equation}
	
	The estimator has the well-known double robustness property that if $\hat{r}_1=r_1$ and $\hat{r}_0=r_0$, so there is no estimation error for the outcome regressions, then an estimator of the form \eqref{eq:tau_hat} will remain unbiased regardless of the severity of the bias in $\hat{\pi}$. Moreover, the estimator achieves the semiparametric efficient variance bound for this problem when $\hat{\pi}$ is a sufficiently good estimator of the true propensity score $\pi$.
	
	As explained above though, here we are considering the setting where we may not have good estimates of $r_1$ and $r_0$, and so we cannot hope to emulate the all-encompassing debiasing effect of AIPW. The sparsity of the propensity score model does however provide us with an opportunity: roughly speaking, the bias of our propensity score estimates can be expected to lie `only in certain directions'---this is a vague statement which we make more precise in Section~\ref{sec:basic} (and specifically \eqref{eq:holder}) but does capture the underlying intuition.
}

\rev{Similarly to AIPW, our approach starts with an estimate $\tilde{\mu}$ of $\mu_{\ora}$; this may be constructed through applying machine learning methods to estimate each of $r_1$ and $r_0$, or a more traditional Lasso regression. Rather than using this directly in \eqref{eq:tau_hat} however, we then adjust this estimate by introducing an additional `orthogonalisation' step to compensate for the more specific biases in the propensity score estimate.} This involves solving a convex quadratic program depending on $\mb X$ and
an auxiliary dataset consisting of an i.i.d.\ copy of $(\mb X, \mb Y, \mb T)$; such auxiliary data may be obtained with a single dataset through sample-splitting. The optimisation problem is amenable to off-the-shelf convex optimisers and the approach is therefore computationally scalable to high-dimensional and large-scale settings.

\rev{In the case that $\tilde{\mu}$ is a sufficiently good estimate of $\mu_{\ora}$, as for example typically required for the theoretical guarantees of AIPW, our adjustment will likely result in little change. On the other hand, if $\tilde{\mu}$ is a poor estimate, the correction will ensure approximate unbiasedness of the resulting debiased IPW (DIPW) estimator. Importantly, this unbiasedness property makes essentially no requirements on the quality of the initial estimate $\tilde{\mu}$.  Instead, we see through our results that improved estimation here results in a smaller variance in the final estimator. In this way, our approach effectively separates out the debiasing and variance reduction effects of AIPW.}


In Section~\ref{sec:theory}, we show that the resulting debiased IPW (DIPW) estimator concentrates around the target $\tau$ at the parametric $1/\sqrt{n}$ rate and give a Berry--Esseen bound quantifying in finite samples, the deviation from Gaussianity  (Theorem~\ref{thm:split}). \rev{We show that for a cross-fitted version of our estimator, associated confidence intervals have finite sample coverage guarantees (Theorem~\ref{thm:ci_new}).}
\rev{Moreover, we show that $o(1 /\sqrt{\log n})$ rates of estimation for the outcome regressions result in semiparametric efficiency (see Theorem~\ref{thm:efficiency}). If sparse linear outcome regression models with sparsity $s_r$ are assumed, this translates to a requirement of $s_r = o(n / \{\log (n)\log (p)\})$, which is weaker than the corresponding $s_r=o(\sqrt{n} / \log p)$ of AIPW and $s_r = o(n^{3/4} / \log p)$ of \citet{BWZ19}; however no such sparse regression models are necessarily required in our theory.}

Whilst our basic DIPW has useful theoretical properties, a drawback of the estimator is that typically it has unwanted variance due to the randomness of the sample splitting required in its construction.
In Section~\ref{sec:multi} we present a version of the estimator that
averages over multiple splits to reduce this excess variance; it is this multi-sample splitting version of the estimator that we recommend using in practice. We outline extensions of our basic methodology for estimating linear projections of the heterogeneous treatment effect function, \rev{and handling propensity score models with other link functions other than the logistic link in Section~\ref{sec:extensions}}.
Numerical experiments are contained in Section~\ref{sec:numerical} and we conclude with a discussion in Section~\ref{sec:discuss}. The appendix of this paper contains the proofs of the theoretical results and further numerical results. An R package \texttt{dipw} implementing the methodology is available on CRAN.

\section{Debiased IPW: basic methodology} \label{sec:basic}
In this section, we describe a simple version of our debiased IPW estimator which relies on independent auxiliary datasets $\mathcal{D}_A := (\mb X_A, \mb Y_A, \mb T_A) \in \R^{n_A\times p} \times \R^{n_A} \times \{0, 1\}^{n_A}$ and $\mathcal{D}_B := (\mb X_B, \mb Y_B, \mb T_B) \in \R^{n_B\times p} \times \R^{n_B} \times \{0, 1\}^{n_B}$ that are independent  of the main dataset $(\mb X, \mb Y, \mb T)$ and consist of i.i.d.\ copies of the triple $(X, Y, T)$.
These datasets are used to ensure certain independencies, some of which are largely technical, which simplify our theoretical analysis in Section~\ref{sec:theory}.
Given only a single dataset, the separate datasets could be formed through splitting the original data into three parts. However in practice we recommend the multiple sample splitting approach described in Section~\ref{sec:multi}.

Using dataset $\mathcal{D}_B$ we first construct an estimate of the propensity score 
of the form $\hat{\pi}(x) = \psi(x^\top \hat{\gamma})$, where $\hat{\gamma}$ is an estimate of $\gamma$. Our default option, which we use in all our numerical experiments, is $\ell_1$-penalised logistic regression (see \eqref{eq:l_1_logistic}).

Recall that in order for the modified IPW estimator $\hat{\tau}$ \eqref{eq:tau_hat} to be useful, we must choose the correction terms $\mu_i$ such that the bias is negligible.
In order to derive the appropriate form of correction, it is helpful to introduce the oracle estimate
\begin{equation} \label{eq:tau_ora}
	\tau_{\ora} := \frac{1}{n}\sum_{i=1}^n \left(\frac{T_i (Y_i-\mu_i)}{\pi(X_i)} - \frac{(1 - T_i) (Y_i-\mu_i)}{1 - \pi(X_i)}\right),
\end{equation}
which, as discussed in the previous section, is unbiased given the conditional independence $\mbb\mu \independent \mb T \mid \mb X$.
We then have the following decomposition involving $\tau_{\ora}$ (see Section~\ref{sec:thmdipw} for a derivation):
\begin{equation} \label{eq:hat_tau_decomp}
	\begin{split}
		\hat{\tau} & = -\frac{1}{n}\sum_{i=1}^n \left(\frac{T_iY_i}{\hat{\pi}_i\pi_i} - \frac{(1-T_i)Y_i}{(1-\hat{\pi}_i)(1-\pi_i)} - \frac{\mu_i}{\hat{\pi}_i(1-\hat{\pi}_i)}\right)(\hat{\pi}_i - \pi_i)  \\
		&\qquad +  \frac{1}{n} \sum_{i=1}^n (T_i - \pi_i)\left(\frac{\mu_i}{\hat{\pi}_i \pi_i} +
		\frac{\mu_i}{(1-\hat{\pi}_i) (1-\pi_i)}\right)(\hat{\pi}_i - \pi_i)  + \tau_{\ora}. 
	\end{split}
\end{equation}
Here we have written $\hat{\pi}_i := \hat{\pi}(X_i) = \psi(X_i^\top \hat{\gamma})$ and $\pi_i = \pi(X_i)$ for simplicity. Now provided $\hat{\gamma}$ is independent of the main dataset, and given $\mbb\mu \independent \mb T \mid \mb X$, the penultimate term has mean zero due to the $T_i - \pi_i$ factor in each summand. The first term therefore constitutes the bias; 
%
its absolute value is approximately
\begin{align}
	& \abs{\frac{1}{n} \sum_{i=1}^n \left(\frac{T_i Y_i}{\hat{\pi}_i^2} + \frac{(1 - T_i) Y_i}{(1 - \hat{\pi}_i)^2} - \frac{\mu_i}{(1 - \hat{\pi}_i)\hat{\pi}_i}\right)(\hat{\pi}_i - \pi_i)} \notag \\
	\approx & \abs{\frac{1}{n} \sum_{i=1}^n \left(\frac{T_i Y_i}{\hat{\pi}_i^2} + \frac{(1 - T_i) Y_i}{(1 - \hat{\pi}_i)^2} - \frac{\mu_i}{(1-\hat{\pi}_i)\hat{\pi}_i}\right) \psi'(X_i^\top  \hat{\gamma}) X_i^\top  (\hat{\gamma} - \gamma)}, \label{eq:approx_bias}
\end{align}
the second approximation following from a Taylor series expansion. Now noting that the logistic function satisfies $\psi'(u) = \psi(u)\{1-\psi(u)\}$, and writing
\begin{align} \label{eq:tildeY}
	\tilde{Y}_i := \frac{T_iY_i(1-\hat{\pi}_i)}{\hat{\pi}_i} + \frac{(1-T_i)Y_i \hat{\pi}_i}{1 - \hat{\pi}_i},
\end{align}
and $\tilde{\mb Y} := (\tilde{Y}_i)_{i=1}^n$,
we have that \eqref{eq:approx_bias} is
\begin{equation}\label{eq:holder}
	\abs{\frac{1}{n} \sum_{i=1}^n (\tilde{Y}_i - \mu_i) X_i^\top (\hat{\gamma} - \gamma)} \leq \frac{1}{n} \| \mb X^\top  \tilde{\mb Y} - \mb X^\top  \mbb\mu \|_\infty \|\hat{\gamma} - \gamma\|_1,
\end{equation}
using H\"older's inequality. 
Under relatively weak conditions,  we have that with high probability $\|\hat{\gamma} - \gamma\|_1 \lesssim s \sqrt{\log(p)/n}$  \citep{van2008high}.
\rev{This inequality encapsulates the sense in which the bias in the propensity score estimate is `only in certain directions': we need only arrange that $\| \mb X^\top  (\tilde{\mb Y} - \mbb\mu) \|_\infty/n$ is small, that is $\tilde{\mb Y} - \mbb\mu$ is approximately orthogonal to the columns of $\mb X$, in order for the bias to be negligible.}

\rev{
	This, however, is not entirely straightforward to achieve since,} as explained in Section~\ref{sec:contrib}, in order to maintain unbiasedness of $\tau_{\ora}$, $\mbb\mu$ should be conditionally independent of $\mb T$ given $\mb X$. To ensure this, we construct $\mbb\mu$ as a function of $\mb X$  and the auxiliary data, as follows.
First observe that for any function $f:\R^p \to \R$,\footnote{We use the convention that a function $g: \R^p \to \R$ applied to a matrix $\mb M \in \R^{m \times p}$ is the vector formed from applying $g$ to each row of $\mb M$, that is, $g(\mb M):= (g(M_i))_{i=1}^m$ where $M_i \in \R^p$ is the $i$th row of $\mb M$.}
\begin{equation} \label{eq:l_inf_simple}
	\| \mb X^\top  \{\tilde{\mb Y} - f(\mb X)\} - \mb X^\top  \{\mbb\mu -f(\mb X)\}\|_\infty = \| \mb X^\top  \tilde{\mb Y} - \mb X^\top  \mbb\mu\|_\infty.
\end{equation}
Given $f$ (we discuss the choice of this below), we propose to approximate $\mb X^\top  \{\tilde{\mb Y} - f(\mb X)\} / n$ via $\mb X_A^\top  \{\tilde{\mb Y}_A - f(\mb X_A)\}/n_A$. Here $\tilde{\mb Y}_A$ is the equivalent of $\tilde{\mb Y}$ but with $Y_i$ and $T_i$ in definition \eqref{eq:tildeY} replaced by $Y_{A,i}$ and $T_{A,i}$ for $i=1,\ldots,n_A$; for later use, we similarly define $\mb{\tilde{Y}}_B$.
We may then choose $\mbb\mu$ to satisfy the constraint
\begin{equation} \label{eq:constraint}
	\Big\|\frac{1}{n_A}\mb X_A^\top  \{\tilde{\mb Y}_A - f(\mb X_A)\} - \frac{1}{n} \mb X^\top  \{\mbb\mu - f(\mb X)\}\Big\|_\infty \leq \eta.
\end{equation}
Under reasonable sub-Gaussianity conditions (see Assumptions \ref{as:subgY} and \ref{as:subgX}),
\begin{equation} \label{eq:l_inf_bd}
	\Big\|\frac{1}{n_A}\mb X_A^\top  \{\tilde{\mb Y}_A - f(\mb X_A)\} - \frac{1}{n} \mb X^\top  \{\tilde{\mb Y} - f(\mb X)\}\Big\|_\infty \leq c \sqrt{\frac{\log(p)}{\min(n, n_A)}}
\end{equation}
with high probability, for a constant $c > 0$. Thus assuming $n_A \gtrsim  n$, choosing $\eta \asymp \sqrt{\log(p)/n}$ should ensure that the bias is of order
\[
\sqrt{\frac{\log p}{n}} \times s \sqrt{\frac{\log p}{n}} = s \frac{\log p }{n} \ll \frac{1}{\sqrt{n}}
\]
under a sparsity condition of the form $s \ll \sqrt{n} / \log p$. Note that with high probability $\tilde{ \mb Y}$ will be feasible for \eqref{eq:constraint} so the constraint set will be non-empty.

\rev{This type of use of auxiliary data is related to the idea of cross-fitting (as popularised e.g.\ by \citet{chernozhukov2018double}) which involves estimating nuisance functions on auxiliary data to ensure the estimates are independent of the main dataset. The crucial difference here is that $\mbb\mu$ does depend on the main dataset, as it must in order to ensure the approximate orthogonality we require; however, this dependence is only through $\mb X$. 
	A related `partial' cross-fitting scheme also appears in the work of \citet{bradic2019minimax}, which is contemporaneous with ours here, but in the context of estimating regression coefficients.}

Constraint \eqref{eq:constraint}, even with $f$ fixed, does not uniquely identify $\mbb\mu$; we would thus like to pick an element of the constraint set so the variance of the estimator is small. Recall the relationship between $\hat{\tau}$ and $\tau_{\ora}$ in \eqref{eq:hat_tau_decomp}. The above analysis shows that we can expect $\hat{\tau} \approx \tau_{\ora}$, so $\Var(\hat{\tau}) \approx \Var(\tau_{\ora})$.
As indicated in Section~\ref{sec:contrib}, the variance minimising $\mu_i$ should be
\begin{equation} \label{eq:mu_ora}
	\mu_{\ora}(X_i) =: \mu_{\ora,i} = \E (\tilde{Y}_{\ora,i}|X_i),
\end{equation}
where $\tilde{Y}_{\ora,i}$ is like $\tilde{Y}_i$ \eqref{eq:tildeY} but with $\hat{\pi}_i$ replaced with the true propensity score $\pi_i$; see Section~\ref{sec:mu_ora} in the appendix for a derivation of the final equality.
Indeed, observe that the
$i$th summand of $\tau_{\ora}$ may be written as
\[
\tau_{\ora, i} := \frac{(T_i - \pi_i)(Y_i-\mu_i)}{\pi_i(1-\pi_i)}.
\]
We then have the following result.
\begin{lemma} \label{lem:var}
	The function $\mu_i$ of $X_1, \ldots, X_n$ that minimises $\Var(\tau_{\ora, i})$
	is $\mu_i(X_1,\ldots,X_n) = \mu_{\ora, i}$.
\end{lemma}
This suggests we should encourage $\mu_i$ to be close to $\mu_{\ora, i}$; however directly regressing $\tilde{\mb Y}$ on $\mb X$ would violate our requirement that $\mbb\mu \independent \mb T \mid \mb X$.
\rev{Instead, we propose to construct an estimate $\tilde{\mu}$ of the function $\mu_{\ora}$ using dataset $\mathcal{D}_B$; treating $\tilde{\mu}(X_i)$ as a proxy of $\mu_{\ora, i}$, we may pick the $\mbb\mu$ minimising $\|\mbb\mu - \tilde{\mu}(\mb X)\|_2^2$ subject to our constraint \eqref{eq:constraint}.} \rev{This $\tilde{\mu}$ may be formed either through
	\begin{enumerate}[(a)]
		\item regressing $\tilde{\mb Y}_B$ onto $\mb X_B$, or;
		\item estimating each of $r_1$ and $r_0$ by regressing treatment and control groups onto the confounders separately, and forming $\tilde{\mu}(x) = \{1-\hat{\pi}(X_i)\} \hat{r}_1(X_i) +\hat{\pi}(X_i)\hat{r}_0(X_i)$.
	\end{enumerate}
	See Section~\ref{sec:numerical} for illustrations of each of these approaches. }

\rev{
	It remains to fix function $f$, for which we can use $\tilde{\mu}$ once more: the constant $c$ in \eqref{eq:l_inf_bd} will be related to the variances of $X_{ij}\{\tilde{Y}_i - f(X_i)\}$ for $j=1,\ldots,p$, which should be small when $f$ is close to $\E(\tilde{Y}_i | X_i)$.

	We thus arrive at the following convex quadratic program for determining $\mbb\mu = \hat{\mbb\mu}$}
\begin{gather}
	\hat{\mbb\mu} = \argmin_{\mbb\mu \in \R^n} \frac{1}{n}\|\tilde{\mu}(\mb X) - \mbb\mu\|_2^2 \notag\\
	\text{subject to } \;\;  \Big\|\frac{1}{n_A}\mb X_A^\top  \{\tilde{\mb Y}_A - \tilde{\mu}(\mb X_A)\} - \frac{1}{n}\mb X^\top  \{\mbb\mu - \tilde{\mu}(\mb X)\} \Big\|_\infty \leq \eta. \label{eq:mu_hat}
\end{gather}
When no feasible $\hat{\mbb\mu}$ exists, we simply set $\hat{\mbb\mu} = \mb 0$; as explained earlier however, $\tilde{\mb Y}$ is a feasible solution with high probability.
With this we can define the basic version of our debiased inverse probability weighting estimator $\hat{\tau}_{\dipw}$:
\begin{equation} \label{eq:dipw}
	\hat{\tau}_{\dipw} := \frac{1}{n}\sum_{i=1}^n \left(\frac{T_i (Y_i-\hat{\mu}_i)}{\hat{\pi}_i} - \frac{(1 - T_i) (Y_i-\hat{\mu}_i)}{1 - \hat{\pi}_i}\right).
\end{equation}
The choice of regression method to produce $\tilde{\mu}$ will affect the variance of the final estimator and the constant multiplying the bias term, but not rate at which the bias decays. Importantly there is no need for $\tilde{\mu}(X_i)$ to genuinely approximate $\mu_{\ora, i}$
and even the simple choice $\tilde{\mu}(x) \equiv 0$ will result in $1/\sqrt{n}$ rates for estimating $\tau$. 

In the following section, we present some theoretical properties of $\hat{\tau}_{\dipw}$ before describing a multi-sample splitting version in Section~\ref{sec:multi} that does not require auxiliary datasets and is the default version we recommend using in practice.

\section{Theoretical properties} \label{sec:theory}
In this section,  we investigate some theoretical properties of our debiased IPW scheme. We first study the basic estimator $\hat{\tau}_{\dipw}$ \eqref{eq:dipw} which uses auxiliary datasets as described in Section~\ref{sec:basic}. We show  via a Berry--Esseen bound that $\hat{\tau}_{\dipw}$ is approximately Gaussian about the average treatment effect conditioned on the covariates (Theorem~\ref{thm:dipw}), with a variance that decays at the optimal $1/n$ rate. We further show that $\hat{\tau}_{\dipw}$ concentrates around the (unconditional) average treatment effect parameter $\tau$ at the optimal $1/n$ rate (Theorem~\ref{thm:split}).
A noteworthy feature of the results is that they require essentially no assumptions on the complexity of the regression functions that might allow them to be estimable, and quantify the price to pay for this which manifests itself as a  constant factor inflation of the variance of the estimator. \rev{Thus our theory here supports the idea that DIPW is particularly well-suited to settings where the outcome regression models may be very hard to estimate consistently, but the propensity score is expected to be estimable (as may be ascertained in practice via a goodness-of-fit test such as the generalised residual prediction test of \citet{jankova2020goodness}); this thinking is further evidenced numerically in Section~\ref{sec:numerical}.}

\rev{
	In Section~\ref{sec:crossfit} we consider a setting with only a single main sample, and use cross-fitting to construct a confidence interval for the conditional average treatment effect, for which we provide finite sample coverage guarantees. We also show that under conditions on how well the outcome regression functions $r_1$ and $r_0$ are estimated, our approach delivers efficient estimation of the average treatment effect $\tau$.}

\rev{For all of the results to follow,} in addition to the unconfoundedness, overlap and propensity score model assumptions (Assumptions \ref{as:unconfounded} and \ref{as:logistic}), it will be convenient to make the following assumptions about the distribution of the high-dimensional confounders $X$ and the potential outcomes $Y(0), Y(1)$.

\begin{assumption} \label{as:subgY}
	The potential outcomes $Y(0), Y(1)$ are sub-Gaussian random variables, so there exists $\sigma_Y > 0$ such that $\E\big(\exp[\alpha \{Y(t) - \E Y(t)\}]\big) \leq \exp(\alpha^2 \sigma_Y^2 / 2)$ for all $\alpha \in \R$ and $t=0, 1$. Furthermore we assume there exists a constant $m_Y > 0$ such that $\max_{t=0,1}|\E Y(t)| \leq m_Y$.
\end{assumption}

\begin{assumption} \label{as:subgX}
	The first component of $X$ is $1$, representing an intercept term. Denoting by $Z \in \R^{p-1}$ the remaining components of $X$, we assume $Z$ is 
	sub-Gaussian with $\E Z=0$, so there exists $\sigma_Z>0$ such that for each $u \in \R^{p-1}$ with $\|u\|_2=1$, and for all $\alpha \in \R$, $\E\{\exp(\alpha u^\top  Z)\} \leq \exp(\alpha^2 \sigma_Z^2/2)$.
\end{assumption}

\rev{The assumption that the potential outcomes are sub-Gaussian and of bounded mean will be satisfied for example if they are bounded random variables, or if $r_t(X)$ for $t=0,1$ are sub-Gaussian with bounded mean and the errors $Y(t) - r_t(X)$ are sub-Gaussian.}
Assumption~\ref{as:subgX} that $\E Z = 0$ above is here only made for simplicity and is not essential for our main arguments. The sub-Gaussianity of $Z$ then implies in particular that the largest eigenvalue of $\Var(Z)$ is bounded.
We also assume that $p$ is not too large compared to $n$, as is typical in the high-dimensional statistics literature, and also place a lower bound on $p$ and $s$ to simplify the statements of the results to follow.

\begin{assumption} \label{as:p_large}
	There exists a sequence $(a_n)_{n=1}^\infty$ with $\lim_{n \to \infty} a_n = 0$ such that $\log(p) / n = a_n$. Furthermore $p \geq 2$ and $s \geq 1$.
\end{assumption} 

\subsection{Results with auxiliary data} \label{sec:aux}
In this section, we investigate some theoretical properties of the basic debiased IPW estimator $\hat{\tau}_{\dipw}$ \eqref{eq:dipw}, assuming the existence of auxiliary datasets $\mathcal{D}_A$ and $\mathcal{D}_B$ as in Section~\ref{sec:basic}.
For simplicity of exposition, throughout this section we assume that $n = n_A$.



There are two potential targets we may be interested in estimating using $\hat{\tau}_{\dipw}$: the first is the average treatment effect $\tau = \E\{Y(1) - Y(0)\}$, and the second is a version conditional on the observed covariates
\begin{equation} \label{eq:tau_bar}
	\bar{\tau} := \frac{1}{n} \sum_{i=1}^n \E\{Y_i(1) - Y_i(0) | X_i\}.
\end{equation}
In order to present our results on the estimation of these quantities, we first introduce the following notation. Let us write
\[
\varepsilon(t) := Y(t) - r_t(X) \qquad \text{and} \qquad \varepsilon_i(t) = Y_i(t) - r_t(X_i), \qquad t=0,1.
\]
Let $\mathcal{D} := (\mb X, \mathcal{D}_A, \mathcal{D}_B)$.
Define the event $\Omega(c_\gamma, c_{\tilde{\mu}}, c_{\hat{\pi}})$ depending on constants $c_\gamma, c_{\tilde{\mu}} > 0$ and $c_{\hat{\pi}} \in (0, 1/2]$ to be such that
\begin{enumerate}[(i)]
	\item $\|\hat{\gamma} - \gamma\|_1\leq c_\gamma s \sqrt{\log (p) /n }$ and $\|\hat{\gamma} - \gamma\|_2 \leq c_\gamma \sqrt{s \log (p) / n}$;
	\item $c_{\hat{\pi}} \leq \hat{\pi}_i \leq 1 - c_{\hat{\pi}}$ and $c_{\hat{\pi}} \leq \hat{\pi}_{Ai} \leq 1 - c_{\hat{\pi}}$ for all $i=1,\ldots,n$;
	\item $|\E \{\tilde{\mu}(X) \,|\, \mathcal{D}_B\} | < c_{\tilde{\mu}}$ and for all $\alpha \in \R$, $\E \big(\exp [\alpha \{\tilde{\mu}(X) - \E(\tilde{\mu}(X) | \mathcal{D}_B)\} ] \,|\, \mathcal{D}_B\big) \leq \exp(\alpha^2 c_{\tilde{\mu}}^2/2)$.
\end{enumerate}
Here $\hat{\pi}_{Ai}$ is the equivalent of $\hat{\pi}_i$ but related to the dataset $\mathcal{D}_A$.
Our results will require that this event occurs with high probability.
\newrev{Note that if $\hat{\gamma}$ is the output of a Lasso logistic regression, we are guaranteed that both (i) and (ii) occur with high probability under the additional mild assumptions that $\Var(Z)$ (see Assumption~\ref{as:subgX}) has a minimum eigenvalue bounded away from zero, and $s \ll n / (\log (p) \log (n))  $; see Section~\ref{sec:omega} in the appendix.}

Condition (iii) requires $\tilde{\mu}(X)$ to be conditionally sub-Gaussian with bounded conditional mean. \rev{While this is a non-standard condition, we do not regard it as very strong, and there are several settings under which we can expect this to be satisfied for some constant $c_{\tilde{\mu}}$. Firstly, the results in this section, and their conclusions regarding the fast $1/\sqrt{n}$ rate of estimation of treatment effects, all hold true when $\tilde{\mu}$ is simply chosen to be the 0 function, or any constant function. More generally, this will hold whenever $\tilde{\mu}$ is bounded. For example, when the potential outcomes are bounded (see also Assumption~\ref{as:subgY} and the following discussion), using each of the approaches (a) or (b) in Section~\ref{sec:basic} to form $\tilde{\mu}$ and employing either regression trees \citep{breiman2017classification} or random forests \citep{breiman2001random} for the regressions will result in $\tilde{\mu}$ being bounded, and hence (iii) being satisfied. If $\tilde{\mu}$ were constructed through Lasso regressions, then the $\ell_2$-norm of the regression coefficient being bounded will ensure (iii) provided Assumption~\ref{as:subgX} holds.}

In order to state our first result concerning estimation of $\bar{\tau}$, it will be convenient to define the following quantities:
\begin{align*}
	\sigma_\mu^2 &:= \frac{1}{n} \sum_{i=1}^n \frac{(\hat{\mu}_i - \mu_{\ora, i})^2}{\pi_i (1 - \pi_i)} \\
	\bar{\sigma}^2 &:= \frac{1}{n} \sum_{i=1}^n \bigg(\frac{\E\{\varepsilon_i(1)^2 | X_i\}}{\pi_i} + \frac{\E\{\varepsilon_i(0)^2 | X_i\}}{1 - \pi_i}\bigg) \\
	\bar{\rho}^3 &:= \frac{1}{n} \sum_{i=1}^n [\E\{|\varepsilon_i(1)|^3 | X_i\} + \E\{|\varepsilon_i(0)|^3 | X_i\}].
\end{align*}
We will see that $\sigma_\mu^2$ is a term by which the asymptotic mean-squared error (MSE) of $\hat{\tau}_{\dipw}$ is increased beyond the
MSE $\bar{\sigma}^2$ attainable by AIPW when $\pi$, $r_0$ and $r_1$ are estimated sufficiently well.
The quantity $\bar{\rho}^3$ plays a role in bounding the rate of convergence of $\sqrt{n}(\hat{\tau}_{\dipw} - \bar{\tau})$ to a Gaussian distribution in a Berry--Esseen bound, as we see below. 
\begin{theorem} \label{thm:dipw}
	Let $\hat{\tau}_{\dipw}$ be the estimator \eqref{eq:dipw} where $\hat{\mu} \in \R^n$ is constructed via \eqref{eq:mu_hat} with tuning parameter $\eta = c_\eta\sqrt{\log(p) / n}$ and where the constant $c_\eta>0$ is sufficiently large. Suppose Assumptions \ref{as:unconfounded}--\ref{as:p_large} hold.
	Then we have the decomposition
	\[
	\sqrt{n}(\hat{\tau}_{\dipw} - \bar{\tau}) = \delta + \sqrt{ \sigma_{\mu}^2 + \bar{\sigma}^2} \zeta
	\]
	in which given constants $c_\gamma, c_{\tilde{\mu}} > 0$, $c_{\hat{\pi}} \in (0, \frac{1}{2}]$ and $m \in \mathbb{N}$, we have that $\delta$ and $\zeta$ satisfy the following properties:
	\begin{enumerate}[(i)]
		\item there exist constants\footnote{Here and below, the constants in the conclusions of our results may depend upon quantities introduced as
			constants in the relevant conditions for these results.} $c_\delta, c > 0$ 
		such that
		\begin{equation} \label{eq:delta}
			\pr\left(|\delta|  > c_\delta (s + \sqrt{s \log n}) \frac{\log p}{\sqrt{n}}\right) \leq \pr(\Omega^c(c_\gamma,c_{\tilde{\mu}},c_{\hat{\pi}})) + c(p^{-m} + n^{-m});
		\end{equation}
		\item there exists constant $c_\zeta >0$ such that
		\begin{align} \label{eq:berry-eseen}
			\sup_{t \in \R} |\pr(\zeta \leq t | \mathcal{D}) - \Phi(t)| \leq \frac{c_\zeta}{\sqrt{n}} \frac{\sigma_\mu^2\|\hat{\mbb\mu} - \mbb\mu_{\ora}\|_\infty + \bar{\rho}^3}{(\sigma_{\mu}^2 + \bar{\sigma}^2)^{3/2}}.
		\end{align}
	\end{enumerate}
\end{theorem}
Before discussing the specifics of the result, we note that informally, Theorem~\ref{thm:dipw} says that in an asymptotic regime where $n,p \to \infty$ and where $(s + \sqrt{s \log n}) \log p / \sqrt{n} \to 0$, we have approximately that
\[
\sqrt{n}(\hat{\tau}_{\dipw} - \bar{\tau}) \,|\, \mathcal{D} \sim \mathcal{N}(0, \sigma_\mu^2 + \bar{\sigma}^2).
\]
This interpretation relies in particular on $s \log (p) / \sqrt{n} \to 0$, which is a common sparsity requirement seen for example in the theory relating to the debiased Lasso \citep{zhang2014confidence, van2014asymptotically, javanmard2014confidence} and approximate residual balancing \citep{AIW18}; this condition comes from needing to control the first term in \eqref{eq:hat_tau_decomp}. 
Additionally, we require $\log p \sqrt{s \log (n) / n} \to 0$. 
This is a relatively weak sparsity requirement, and will be implied by the former more standard sparsity condition whenever $\log n \leq \text{const.} \times \sqrt{s}$. The need for this condition stems from requiring control of the second term in \eqref{eq:hat_tau_decomp}; an alternative would be to include the effect of this term in the approximately Gaussian quantitiy $\zeta$ as it is mean-zero.

Importantly, the variance term $\sigma_\mu^2$ that we incur in addition to the $\bar{\sigma}^2$ attainable by AIPW when all nuisance functions are known, is well-controlled.
We can show (see \eqref{eq:mu_bd} in the appendix) that there exist constants $c, c_\mu>0$ such that $\sigma^2_\mu \leq c_\mu$ with probability at least $1-\pr(\Omega^c(c_\gamma,c_{\tilde{\mu}},c_{\hat{\pi}})) + c(p^{-m} + n^{-m})$. Thus even if $\tilde{\mu}$ does a poor job at estimating $\mu_{\ora}$, the estimation error of $\hat{\tau}_{\dipw}$ decays at the parametric $1/\sqrt{n}$ rate.


\rev{We now turn to the quality of the Gaussian approximation as given by the right-hand side of \eqref{eq:berry-eseen}.}
First note that $\|\hat{\mbb\mu} - \mbb\mu_{\ora}\|_\infty \leq \|\hat{\mbb\mu}\|_\infty + \|\mbb\mu_{\ora}\|_\infty$. It is straightforward to show that $\|\mbb\mu_{\ora}\|_\infty = O_{\pr}(\sqrt{\log n})$ so the right-hand side of \eqref{eq:berry-eseen} is small whenever $\|\hat{\mbb\mu}\|_\infty / \sqrt{n}$ is small. We can easily include a constraint on $\|\hat{\mbb\mu}\|_\infty$ in our convex quadratic program \eqref{eq:mu_hat}; however our experience in practice is that typically $\|\hat{\mbb\mu}\|_\infty / \sqrt{n} \ll 1$ so the additional constraint appears unnecessary. 
Under reasonable conditions, we can expect that the ratio $\bar{\rho}^3 / \bar{\sigma}^3$ is bounded with high probability: for example if the errors $(\varepsilon(0), \varepsilon(1))$ are independent of $X$, the ratio is up to a constant factor bounded above by
\[
\frac{\E\{|\varepsilon(0)|^3\} + \E\{|\varepsilon(1)|^3\}}{[\E\{\varepsilon(0)^2\}]^{3/2} +[\E\{\varepsilon(1)^2\}]^{3/2}}.
\]

\newrev{Whilst the statement of Theorem~\ref{thm:dipw} above treats $c_\pi$ and $c_{\hat{\pi}}$ as constants bounded away from zero, in practice, these may be small. In Theorem~\ref{thm:dipwpf} of the appendix, we present a stronger version of Theorem~\ref{thm:dipw} that reveals in particular that the constants in the result have a relatively favourable at worst (low-degree) polynomial dependence on  $c_\pi^{-1}$ and $c_{\hat{\pi}}^{-1}$.}

If $\mb X$ is not considered to be deterministic, a more appropriate target parameter is $\tau$. Theorem~\ref{thm:split} below gives guarantees on the error in estimating this using $\hat{\tau}_{\dipw}$. In order to state our result, we introduce the following quantities:
\begin{align}
	\sigma^2 & := \Var\bigg(r_1(X) - r_0(X) - \tau + \frac{T \varepsilon(1)}{\pi(X)} - \frac{(1-T) \varepsilon(0)}{1 - \pi(X)}\bigg), \label{eq:sigma_def}\\
	\rho^3 & := \E \bigg|r_1(X) - r_0(X) - \tau + \frac{T \varepsilon(1)}{\pi(X)} - \frac{(1-T) \varepsilon(0)}{1 - \pi(X)}\bigg|^3. \label{eq:rho_def}
\end{align}
\begin{theorem}\label{thm:split}
	Consider the setup of Theorem~\ref{thm:dipw}. We have the decomposition
	\begin{equation} \label{eq:split_decomp}
		\sqrt{n}(\hat{\tau}_{\dipw} - \tau) = \delta + \sigma_{\mu} \zeta_1 + \sigma \zeta_2
	\end{equation}
	and given constants $c_\gamma, c_{\tilde{\mu}} >0$, and $c_{\hat{\pi}} \in (0, 1/2]$, we have that $\delta$, $\zeta_1$ and $\zeta_2$ satisfy the following properties:
	\begin{enumerate}[(i)]
		\item $\delta$ satisfies (i) of Theorem~\ref{thm:dipw};
		\item there exists constants $c_{\zeta, 1}, c_{\zeta, 2}$ such that
		\begin{align*}
			\sup_{t \in \R} |\pr(\zeta_1 \leq t \,|\, \mathcal{D}) - \Phi(t)| &\leq \frac{c_{\zeta, 1}}{\sqrt{n}} \frac{\|\hat{\mbb\mu} - \mbb\mu_{\ora}\|_\infty}{\sigma_{\mu}}, \\
			\sup_{t \in \R} |\pr(\zeta_2 \leq t ) - \Phi(t)| &\leq \frac{c_{\zeta, 2}}{\sqrt{n}} \frac{\rho^3}{\sigma^3};
		\end{align*}
		\item $\E (\zeta_1 \zeta_2 | \mathcal{D}) = 0$.
	\end{enumerate}
\end{theorem}
Theorem~\ref{thm:split} says that $\hat{\tau}_{\dipw}$ concentrates around $\tau$ at the optimal $\sqrt{n}$ rate. Moreover the uncorrelatedness property (iii) implies that its MSE satisfies
\begin{align*}
	n \E \{(\hat{\tau}_{\dipw} - \tau)^2\} &= n \E[\E \{(\hat{\tau}_{\dipw} - \tau)^2 | \mathcal{D}\}] \\
	&\approx \E \sigma_{\mu}^2 + \sigma^2.
\end{align*}
Note that $\sigma^2$ is the semiparametric efficient variance bound based on a dataset of size $n$ achieved for example by AIPW when $r_0, r_1$ and $\pi$ are all estimable at sufficiently fast rates. The term $\E \sigma_\mu^2$ may therefore be viewed as
a price to pay for not being able to estimate $r_0$ and $r_1$; \rev{in Section~\ref{sec:efficiency} we investigate conditions under which efficiency is achieved.}

We also remark that the decomposition \eqref{eq:split_decomp} is not quite an asymptotic normality result as whilst $\sigma \zeta_2$ is marginally Gaussian, $\sigma_\mu \zeta_1$ is only conditionally Gaussian given $\mathcal{D}$: marginally it is a mixture of mean-zero Gaussians with variances determined by the distribution of $\sigma_\mu^2$. This latter distribution is intractable without assumptions on $r_0$ and $r_1$, which we are avoiding here. \newrev{Nevertheless, it does suggest a confidence interval for $\tau$ employing a simple union bound. We give a construction for such an interval along with finite sample coverage guarantees in appendix~\ref{sec:CI_tau}. Although this the coverage is expected to be larger than the nominal coverage, the length of the interval contracts at the $1/\sqrt{n}$ rate.}

One unsatisfactory aspect of the setup here from a practical perspective, is that the results and indeed the method require the use of an auxiliary dataset. In the next section we introduce a cross-fitting scheme that avoids this issue, and later in Section~\ref{sec:multi} we develop a multiple sample spitting scheme that derandomises by aggregating over multiple sample splits that we recommend to use in practice.

\subsection{Cross-fitting, confidence intervals and efficiency} \label{sec:crossfit}
In this section we consider a single dataset of size $n$, which we assume for simplicity to be a multiple of $3$. We split the observation indices into three parts, $I_1:=\{1,\ldots,n/3\}$, $I_2 := \{n/3+1,\ldots,2n/3\}$ and $I_3 :=\{2n/3+1,\ldots,n\}$ to give corresponding datasets $\mathcal{D}_j := (X_i, Y_i, T_i)_{i \in I_j}$, $j=1,2,3$. 
Consider forming three corresponding estimates $\hat{\tau}_{\dipw,1}, \hat{\tau}_{\dipw,2}, \hat{\tau}_{\dipw,3}$ where $\hat{\tau}_{\dipw,j}$ is constructed as described in Section~\ref{sec:basic} but using $\mathcal{D}_j$ as the main dataset, and taking auxiliary datasets $\mathcal{D}_A$ and $\mathcal{D}_B$ as $\mathcal{D}_{j+1}$ and $\mathcal{D}_{j+2}$ respectively, with the additions $j+1$ and $j+2$ understood to be modulo $3$. Here we study properties of the aggregate estimator
\begin{equation} \label{eq:cross-fit}
	\hat{\tau}_{\ave}:= \frac{1}{3}( \hat{\tau}_{\dipw,1} + \hat{\tau}_{\dipw,2} + \hat{\tau}_{\dipw,3});
\end{equation}
because of the way auxiliary data and the main data are interchanged in $\hat{\tau}_{\dipw,j}$ as $j$ varies, the above is sometimes known as a cross-fit estimator.

In order to state our theoretical results, we introduce the following quantities, analogues of which appear in Sections~\ref{sec:basic} and \ref{sec:aux}.
For $j=1,2,3$, let $\hat{\mbb\mu}_j \in \R^{n/3}$ be the corresponding debiasing quantities constructed as in \eqref{eq:mu_hat} using estimates $\tilde{\mu}_1, \tilde{\mu}_2, \tilde{\mu}_3$ of $\mu_{\ora}$, and let $\Omega_j(c_\gamma, c_{\tilde{\mu}}, c_{\hat{\pi}})$ be defined as in the previous section, but in each case taking the main dataset as $\mathcal{D}_j$ and auxiliary datasets as above. Let $\hat{\mbb\mu}$ be the concatenation $\hat{\mbb\mu} := (\hat{\mbb\mu}_1, \hat{\mbb\mu}_2, \hat{\mbb\mu}_3) \in \R^n$. We also retain the definition of $\mbb\mu_{\ora} \in \R^n$ from Section~\ref{sec:basic}.


\subsubsection{Confidence intervals}\label{sec:ci}
In this section we consider constructing confidence intervals around the conditional average treatment effect $\bar{\tau}$ \eqref{eq:tau_bar} evaluated across the entire dataset.
Recall that by Theorem~\ref{thm:dipw}, under conditions, $\hat{\tau}_{\dipw,j}$, is an approximately unbiased and Gaussian estimator for $\bar{\tau}_j:=\sum_{i \in I_j} \E\{Y_i(1) - Y_i(0) \,|\, X_i\} / (n/3)$. Thus we can expect that their average $\hat{\tau}_{\ave}$
estimates $\bar{\tau}$.

One complication in constructing a confidence interval is that without making any further assumptions about how close $\hat{\mbb\mu}$ is to its oracular counterpart $\mbb\mu_{\ora}$, we cannot argue that the $\hat{\tau}_{\dipw,j}$ are independent, and thus $\hat{\tau}_{\ave}$ can have a complicated non-Gaussian distribution. However, we can construct a conservative confidence interval which still contracts at the parametric rate under reasonable conditions.

\newrev{
	To see how such a construction can work, note that by Theorem~\ref{thm:dipw}, we may prove that each $\sqrt{n / 3} (\hat{\tau}_{\dipw, j} - \bar{\tau}_j)$ is approximately normal conditional on $\{\{X_i\}_{i \in I_j}, \mathcal{D}_{j + 1}, \mathcal{D}_{j + 2}\}$, under some extra regularity conditions that will be specified below. This allows us to construct approximately valid confidence intervals for each $\bar{\tau}_j$ using  estimates of upper bounds on the conditional variances of $\sqrt{n / 3} (\hat{\tau}_{\dipw, j} - \bar{\tau}_j)$, which are for $j=1,2,3$ given by}
\begin{equation}\label{eq:sigma_est}
	\hat{\sigma}^2_j := \frac{3}{n} \sum_{i \in I_j} \left(\frac{T_i(Y_i - \hat{\mu}_i)}{\hat{\pi}_i} - \frac{(1 - T_i)(Y_i - \hat{\mu}_i)}{1 - \hat{\pi}_i} - \hat{\tau}_{\dipw,j}\right)^2.
\end{equation}
\newrev{
	With this, a valid confidence interval for $\bar{\tau}$ can be derived via applying a union bound. Specifically, we
	define $\tilde{\sigma} := (\hat{\sigma}_1 + \hat{\sigma}_2 + \hat{\sigma}_3)/\sqrt{3}$, and construct a confidence interval via}
\begin{equation} \label{eq:ci_new}
	\tilde{C}_\alpha := \Big[\hat{\tau}_{\ave} - \frac{\tilde{\sigma}}{\sqrt{n}} z_\alpha, \hat{\tau}_{\ave} + \frac{\tilde{\sigma}}{\sqrt{n}} z_\alpha \Big].
\end{equation}

\newrev{
	In order to guarantee that each $\sqrt{n / 3} (\hat{\tau}_{\dipw, j} - \bar{\tau}_j)$ is approximately conditionally normal, we require that the bias in $\hat{\tau}_{\dipw, j}$ is dominated by the variance.}
To this end, we require an explicit sparsity assumption on the propensity score model; note that we avoided making such an assumption in Theorem~\ref{thm:dipw} for example.

\begin{assumption} \label{as:sparsity}
	There exists a sequence $(b_n)_{n=1}^\infty$ with $\lim_{n \to \infty} b_n = 0$ such that $s = b_n \sqrt{n} / \log p$.
\end{assumption}
We also make the mild assumption that the variances of the errors $\varepsilon(0), \varepsilon(1)$ are bounded away from zero.
\begin{assumption} \label{as:subeps}
	There exists constant  $\sigma_{\varepsilon} > 0$ such that $\min_{t=0,1}\Var(\varepsilon(t)) \geq \sigma_{\varepsilon}^2$.
\end{assumption}
The result below shows that $\tilde{C}_{\alpha/3}$ has at least approximate $1-\alpha$ coverage. The confidence interval $\tilde{C}_{\alpha/3}$ is expected to be conservative, as the derivation of the guarantee below involves the use of a union bound involving each of the three estimators $\hat{\tau}_{\dipw,j}$, $j=1,2,3$; \newrev{see Section~\ref{sec:CI_proofs} of the appendix}. Nevertheless, the confidence interval does contract at the $1/\sqrt{n}$ rate.
\begin{theorem}\label{thm:ci_new}
	Suppose tuning parameter $\eta = c_\eta\sqrt{\log(p) / n}$ used to form each $\hat{\mbb\mu}_j$ involved in the construction of $\hat{\tau}_{\ave}$ is such that the constant $c_\eta>0$ is sufficiently large.
	Suppose Assumptions \ref{as:unconfounded}--\ref{as:subeps} hold and let the confidence interval $\tilde{C}_{\alpha}$ be as in \eqref{eq:ci_new}.
	Given constants $c_\gamma, c_{\tilde{\mu}}, c_\varepsilon > 0, c_{\hat{\pi}} \in (0, \frac{1}{2}]$ and $m \in \mathbb{N}$, there exist constants $c, c_\zeta >0$ such that with probability at least
	\begin{equation} \label{eq:full_prob}
		1 - \sum_{j=1}^3 \pr(\Omega_j^c(c_\gamma, c_{\tilde{\mu}}, c_{\hat{\pi}})) - c(n^{-m} + p^{-m}),
	\end{equation}
	we have for all $\alpha \in (0, 1]$ the finite sample coverage guarantee
	\begin{equation} \label{eq:coverage_new}
		\begin{split} 
			& \pr\Big( \bar{\tau} \in \tilde{C}_{\alpha/3} \,|\, \mb X \Big) \geq 1 -\alpha  \\
			& \qquad - c_\zeta \bigg(\E(\|\hat{\mbb\mu} - \mbb\mu_{\ora}\|_\infty \,|\, \mb X) \sqrt{\frac{\log n}{n}}  + \frac{\sqrt{ b_n \log p \log n}}{n^{1/4}} + b_n +  p^{-m}\bigg).
		\end{split}
	\end{equation}
\end{theorem}
Under reasonable conditions (see the discussion following Theorem~\ref{thm:dipw}) we can expect that $\E(\|\hat{\mbb\mu} - \mbb\mu_{\ora}\|_\infty \,|\, \mb X) = o_{\pr}(\sqrt{n/\log(n)})$. Note that this accommodates $\|\hat{\mbb\mu} - \mbb\mu_{\ora}\|_\infty$ growing rather fast with $n$; for comparison if each component of $\hat{\mbb\mu} - \mbb\mu_{\ora}$ was sub-Gaussian with mean of order at most $\sqrt{\log n}$ conditional on $\mb X$, we would have $\E(\|\hat{\mbb\mu} - \mbb\mu_{\ora}\|_\infty \,|\, \mb X) = O_\pr(\sqrt{\log n})$, which would certainly satisfy the requirement. 
Thus we may interpret the result above as guaranteeing that with high probability, the covariates $\mb X$ are such that the confidence interval has coverage approximately at least $1-\alpha$. We can also take expectations on both sides of \eqref{eq:coverage_new} to obtain a more conventional unconditional coverage guarantee, though without conditioning on $\mb X$, the target `parameter' $\bar{\tau}$ is not fixed.

\subsubsection{Efficiency} \label{sec:efficiency}
The previous results have placed no requirement on the quality of the estimate of $\mu_{\ora}$. Here we show that if we make such assumptions, $\hat{\tau}_\ave$ can attain efficiency.
Let us write $\tilde{\mu}(\mb X) \in \R^n$ for the vector with $i$th component $\tilde{\mu}_j(X_i)$ if $i \in I_j$. Recall that the efficient variance we hope to achieve is given by $\sigma^2$ \eqref{eq:sigma_def}. The following result shows that $\sqrt{n}(\hat{\tau}_{\ave} - \tau)$ is approximately Gaussian with variance $\sigma^2$, provided $\tilde{\mu}(\mb X)$ is a sufficiently good estimate of $\mbb\mu_{\ora}$, and the propensity score model is sufficiently sparse.
\begin{theorem}\label{thm:efficiency}
	Suppose we construct $\hat{\tau}_\ave$ as in Theorem~\ref{thm:ci_new} and suppose Assumptions \ref{as:unconfounded}--\ref{as:p_large} hold. We have the decomposition
	\[
	\sqrt{n} (\hat{\tau}_\ave - \tau) = \delta + \sigma \zeta,
	\]
	in which given constants $c_\gamma, c_{\tilde{\mu}} > 0, c_{\hat{\pi}} \in (0, 1/2]$ and $m \in \mathbb{N}$, we have that $\delta$ and $\zeta$ satisfy the following properties:
	\begin{enumerate}[(i)]
		\item \newrev{there exist constants $c_\delta, c > 0$ such that given any sequence $(e_n)_{n = 1}^\infty$, with probability at least
			\[
			1 - \sum_{j=1}^3 \pr(\Omega_j^c(c_\gamma, c_{\tilde{\mu}}, c_{\hat{\pi}})) - c(n^{-m} + p^{-m}) - 2 e^{-e_n^2},
			\]
			we have that
			\[
			|\delta| \le c_\delta (s + \sqrt{s \log n}) \frac{\log p}{\sqrt{n}} + c_\delta \frac{e_n}{\sqrt{n}} \|\tilde{\mu}({\mbb X}) - {\mbb \mu_\ora}\|_2;
			\]}
		\item there exists constant $c_\zeta > 0$ such that
		\[
		\sup_{t \in \R} |\pr(\zeta \leq t ) - \Phi(t)| \leq \frac{c_\zeta}{\sqrt{n}} \frac{\rho^3}{\sigma^3}.
		\]
	\end{enumerate}
\end{theorem}

Suppose that each $\tilde{\mu}_j$ is constructed by first estimating  $r_1$ and $r_0$ by $\hat{r}_{1,j}$ and $\hat{r}_{0,j}$ respectively, and then setting  $\tilde{\mu}_j(x) := \{1 - \hat{\pi}(x)\} \hat{r}_{1,j}(x) + \hat{\pi}(x) \hat{r}_{0,j}(x)$. For $t=0,1$, let $\hat{r}_t(\mb X) \in \R^n$ be the vector with $i$th component $\hat{r}_{tj}(X_i)$ if $i \in I_j$. The following corollary shows that provided the $\hat{r}_{1,j}$ and $\hat{r}_{0,j}$ are sufficiently good estimates for each $j=1,2,3$, under a sparsity assumption, $\delta$ above will be negligible.
\begin{corollary} \label{cor:efficiency}
	Consider the setup of Theorem~\ref{thm:efficiency} but specifically with $\tilde{\mu}$ constructed as indicated above. Assume Assumption~\ref{as:sparsity} and also that $\sqrt{n} \geq c \log p \log n$ for some constant $c>0$. Then we have that for some constant $c_{\delta}>0$,
	\[
	|\delta| \le c_\delta \sqrt{b_n} + c_\delta \frac{e_n}{\sqrt{n}} \max_{t=0,1} \|\hat{r}_t(\mb X) - r_t(\mb X)\|_2.
	\] 
\end{corollary}
Informally then Corollary~\ref{cor:efficiency} says that 
\newrev{provided there exists some $e_n \to \infty$ such that
	$\|\hat{r}_t({\mbb X}) - r_t({\mbb X})\|_2 / \sqrt{n} = o_\pr(e_n^{-1})$ for $t=0,1$,}
$\sqrt{n}(\hat{\tau}_\ave - \tau)$ is an asymptotically normally distributed random variable and  $\hat{\tau}_{\ave}$ attains the semiparametric efficiency bound. \newrev{Notice in order to arrive at this conclusion, $e_n$ is allowed to diverge to infinity at an arbitrarily slow rate.} This requirement on the outcome regression estimates may be contrasted with the stronger $n^{-1/4}$ rate required in theoretical guarantees for AIPW or the $n^{-1/8}$ rate required by the estimator of \citet{BWZ19}.
Moreover, no model assumptions on the functions $r_t$ are needed beyond the relatively weak requirement on the convergence rate. \newrev{In Section~\ref{sec:efficiency_appendix} in the appendix, we further discuss the efficiency of $\hat{\tau}_{\dipw}$ in estimating  $\bar{\tau}$, which follows a similar argument as that of the result presented here.}

\section{Debiased IPW: multiple sample splitting} \label{sec:multi}
In this section we describe a variant of the basic debiased inverse propensity weighting scheme introduced in Section~\ref{sec:basic} that in particular avoids the use of an auxiliary dataset through multiple sample splitting, and is the version we recommend using in practice.
The approach is summarised in Algorithm~\ref{alg:dipw} and we denote the resulting estimate by $\hat{\tau}_{\mdipw}$ for \textbf{m}ultiple sample-splitting DIPW.
Below we discuss the various steps taken in the algorithm. Next in Section~\ref{sec:confidence} we discuss the construction of confidence intervals centred on $\hat{\tau}_{\mdipw}$.
\begin{algorithm}[!h]
	\caption{DIPW with multiple sample splitting}
	\label{alg:dipw}
	\begin{algorithmic}
		\STATE {\bfseries Input:} Dataset $(\mb X, \mb Y, \mb T)$ with $n$ observations, where $\mb X$ has first column a vector of ones representing an intercept; number of splits $B$; tuning parameter $\kappa>0$ (default choice $\kappa = 0.5$).
		\STATE {\bfseries Output:} Estimate $\hat{\tau}_{\mdipw}$ for the average treatment effect.
		\vspace{0.1cm}
		\STATE \textbf{1.} Estimate the propensity score parameter $\gamma$ via a penalised logistic regression of $\mb T$ on $\mb X$, e.g.\
		\begin{equation} \label{eq:l_1_logistic}
			\hat{\gamma} := \argmin_{b \in \R^p} \left( \frac{1}{n} \sum_{i=1}^n [\log\{1 + \exp(X_i^\top  b)\}  - T_i X_i^\top  b]+ \lambda_\gamma \sum_{j=2}^p |b_j|\right)
		\end{equation}
		where $\lambda_\gamma$ is chosen by cross-validation. Write $\hat{\pi}(x) := 1/(1+\exp(-x^\top \hat{\gamma}))$ and $\hat{\pi}_i := \hat{\pi}(X_i)$.
		\STATE \rev{\textbf{2.} Form estimate $\tilde{\mu}$ of $\mu_{\ora}$ \eqref{eq:mu_ora} either
			\STATE \quad \textbf{2a.} through regressing $\tilde{\mb Y}$ given by  \eqref{eq:tildeY} onto $\mb X$,
			\STATE \quad \textbf{2b.} or first estimating $r_1$ and $r_0$ by regressing treatment and control groups onto the predictors separately to give $\hat{r}_1$ and $\hat{r}_0$, and then setting $\tilde{\mu}(x) := \{1-\hat{\pi}(x)\}\hat{r}_1(x) + \hat{\pi}(x)\hat{r}_0(x)$.}
		\STATE \textbf{3.}  Choose pairs of subsets $\{(I_{2b-1}, I_{2b}):b=1,\ldots,B\}$ of $\{1,\ldots,n\}$ where $|I_{2b}| = \floor{n/2}$ and $I_{2b-1} = I_{2b}^c$. Let $n_k:=|I_k|$ and $n_k^c := |I_k^c|$. Construct, for $k=1,\ldots,2B$, bias correction vectors $\check{\mbb\mu}_k \in \R^{n_k}$, which we shall index by elements of $I_k$, via the following steps:
		\STATE \quad \textbf{3a.} Writing $m_{k,j} := \sum_{i=I_k} X_{ij} / n_k$ and $m_{k,j}^c := \sum_{i \in I_k^c} X_{ij}/n_k^c$ for $j=1,\ldots,p$, form a mean-centred version $\mb X^m$ of $\mb X$ via $X_{ij}^m := X_{ij} - m_{k,j}$ for $i \in I_k$ and $X_{ij}^m := X_{ij} - m_{k,j}$ for $i \in I_k^c$.
		\STATE \quad \textbf{3b.} Let $\hat{\mbb\mu}_k$ be the minimiser over $\mbb\mu \in \R^{n_k}$ of
		\begin{align} \label{eq:mu_hat_opt}
			\frac{1-\kappa}{\kappa n_k^2}\sum_{i \in I_k} \{\tilde{\mu}(X_i) - \mu_i\}^2 +  \max_{j=1,\ldots,p} \Big | \frac{1}{n_k}\sum_{i \in I_k} X^m_{ij} (\mu_i - \tilde{\mu}(X_i))   - \frac{1}{n_k^c}\sum_{i \in I_k^c} X^m_{ij}(\tilde{Y}_i - \tilde{\mu}(X_i)) \Big |^2;
		\end{align}
		\STATE \quad \textbf{3c.} For all $i \in I_k$, set $\check{\mu}_{ki} := \hat{\mu}_{ki} + \frac{1}{n_k} \sum_{l \in I_k} (\tilde{Y}_l - \tilde{\mu}(X_l))$.
		\STATE \textbf{4.} For $b=1, \ldots, B$: define $\bar{\mbb\mu}_b \in \R^n$ by $\bar{\mu}_{bi} = \check{\mu}_{(2b-1)i}$ when $i \in I_{2b-1}$ and $\bar{\mu}_{bi} = \check{\mu}_{(2b)i}$ when $i \in I_{2b}$; set
		\begin{align}
			\hat{\tau}_b &:= \sum_{i=1}^n \frac{T_i(Y_i - \bar{\mu}_{bi})}{\hat{\pi}_i} \Big/ \sum_{l=1}^n \frac{T_l}{\hat{\pi}_l}
			\;-\; \sum_{i=1}^n \frac{(1-T_i)(Y_i - \bar{\mu}_{bi})}{1 - \hat{\pi}_i} \Big/ \sum_{l=1}^n \frac{1-T_l}{1-\hat{\pi}_l}. \label{eq:tau_b}\\
			\hat{\tau}_{\mdipw} &:= \frac{1}{B} \sum_{b=1}^{B} \hat{\tau}_b. \notag
		\end{align}
	\end{algorithmic}
\end{algorithm}
\subsection{Estimation} \label{sec:multi_est}
Steps 1 and 2 of Algorithm~\ref{alg:dipw} produce estimates of the propensity score and $\mbb\mu_{\ora}$ \eqref{eq:mu_ora}, though as indicated in
Theorems~\ref{thm:dipw} and \ref{thm:split}, the latter estimate can be poor without affecting convergence properties of the final estimator for the average treatment effect. \rev{Any regression methods of choice can be used at this stage; in our numerical experiments we looked at using Lasso \citep{tibshirani1996regression} regression for 2a and random Forest \citep{breiman2001random} in conjunction with approach 2b.}

Unlike the basic DIPW method presented in Section~\ref{sec:basic}, these are not constructed on auxiliary datasets but instead are computed using all available data. Whilst this introduces some dependence that makes theoretical analysis more problematic, our experience is that the empirical performance is improved.

Step 3 however does attempt to mimic the use of auxiliary datasets through sample splitting. This ensures that the conditional independence $\hat{\mbb\mu}_k \independent (T_i)_{i \in I_k} \mid (X_i)_{i \in I_k}$ guaranteeing unbiasedness given oracle propensity scores (see Section~\ref{sec:contrib}), holds at least approximately. The subsets $I_1, I_3, \ldots, I_{2B-1}, I_{2B}$ may be chosen uniformly at random, or in a deterministic fashion, for example according to the rows of Hadamard matrices, to keep them well-separated.

Analogously to \citet{AIW18}, we use a Lagrangian formulation \eqref{eq:mu_hat_opt} of our original convex program \eqref{eq:mu_hat} as this is more convenient to optimise; we use the interior point solver \texttt{mosek} \citep{mosek} in our implementation. The centring and re-centring steps performed in steps 3a and 3c respectively are almost equivalent to introducing the constraint that
\[
\frac{1}{n_k}\sum_{i \in I_k} \check{\mu}_{ki} = \frac{1}{n_k} \sum_{i \in I_k} \tilde{Y}_i,
\]
and so the first component of \eqref{eq:l_inf_simple} corresponding to the intercept term is exactly zero. This is beneficial as then the estimation error of the intercept term, which may be large, does not contribute to the LHS of \eqref{eq:holder}. The downside of this is that some additional dependence is introduced between $\check{\mbb\mu}_k$ and $(\tilde{ Y}_i)_{i \in I_k}$ which can introduce some further bias; however our experience is that in practice the dependence is sufficiently weak that it is not problematic.

Finally in step 4, we employ the standard practice of renormalising the propensity score weights \citep{imbens2004nonparametric,lunceford2004stratification}, and take as our final estimator a simple average over the estimates corresponding to the different sample-splits.

\subsection{Confidence intervals} \label{sec:confidence}
For $b=1,\ldots,B$, let $\check{\tau}_{2b-1}$ and $\check{\tau}_{2b}$ be the sums over $i \in I_{2b-1}$ and $i \in I_{2b}$ respectively in equation \eqref{eq:tau_b} defining $\hat{\tau}_b$. As discussed in Section~\ref{sec:ci},
Theorem~\ref{thm:dipw} indicates that under appropriate conditions, each $\check{\tau}_k$ should be approximately Gaussian about
\[
\bar{\tau}_k := \frac{1}{n} \sum_{i \in I_k} \E\{Y_i(1) - Y_i(0)|X_i\},
\]
but with potentially different variances. The variance of the aggregate $\hat{\tau}_{\mdipw}$ may be approximately bounded above as follows:
\begin{align*}
	\Var\bigg(\frac{1}{B}\sum_{b=1}^B\hat{\tau}_b  \,\Big|\, \mb X \bigg) &= \Var\bigg(\frac{1}{B}\sum_{k=1}^{2B}\check{\tau}_k \,\Big|\, \mb X \bigg) \\
	&\approx \E\bigg\{\bigg(\frac{2}{B}\sum_{k=1}^{2B}(\check{\tau}_k - \bar{\tau}_k)\bigg)^2 \, \Big|\, \mb X \bigg\} \\
	&\leq \E\bigg(\frac{2}{B}\sum_{k=1}^{2B} (\check{\tau}_k - \bar{\tau}_k)^2 \,\Big|\, \mb X \bigg) \\
	&\approx \frac{2}{B} \sum_{k=1}^{2B} \Var(\check{\tau}_k \,| \, \mb X),
\end{align*}
where the inequality in the penultimate line follows from Jensen's inequality. Now each $\Var(\check{\tau}_k| \mb X)$
may be estimated similarly to \eqref{eq:sigma_est}. Writing $(\check{\sigma}_k^2)_{b=1}^{2B}$ for these estimates, the argument above suggests approximating the variance of $\hat{\tau}_{\mdipw}$ by the average $\hat{\sigma}^2_{\text{m}} := 2\sum_{k=1}^{2B} \check{\sigma}_k^2/B$. Whilst each individual $\check{\tau}_k$ may be Gaussian, we cannot be sure that the average will be. Empirically however, we have found that confidence intervals constructed analogously to \eqref{eq:ci_new}
but centred on $\hat{\tau}_{\mdipw}$ and taking the variance as the approximate upper bound $\hat{\sigma}^2_{\text{m}}$ have reasonable coverage, though can sometimes be conservative. Section~\ref{sec:conf_res} gives the results of numerical experiments exploring the coverage properties of the confidence interval construction. \rev{We note that the approach of \citet{guo2023rank} gives an alternative scheme to address the challenges of constructing confidence intervals centred on estimators formed through multiple sample splitting, which may be less conservative; we leave exploring this further to future work.}

\section{Extensions} \label{sec:extensions}
Here we outline modifications of our framework that permit link functions other than the logistic link for the propensity sore model and allow for estimation of linear projections of the conditional average treatment effect function. 



\subsection{Other link functions} \label{sec:link}

In this section we sketch how to construct a version of the DIPW estimator when the propensity model follows satisfies $\pi(x) = \phi(x^\top \gamma)$ with some general link function $\phi$ where both $\phi'$ and $\phi''$ are uniformly bounded below by some constant; aside from this, we work in the setup of Section~\ref{sec:basic}. Following analogous derivations as in~\eqref{eq:approx_bias}, and recalling the definition of $\tilde{Y}_i$ in~\eqref{eq:tildeY}, the bias term in \eqref{eq:hat_tau_decomp} can instead be written as
\begin{align}
	& \abs{\frac{1}{n} \sum_{i=1}^n \left(\frac{T_i Y_i}{\hat{\pi}_i^2} + \frac{(1 - T_i) Y_i}{(1 - \hat{\pi}_i)^2} - \frac{\mu_i}{(1 - \hat{\pi}_i)\hat{\pi}_i}\right)(\hat{\pi}_i - \pi_i)} \notag \\
	\approx & \abs{\frac{1}{n} \sum_{i=1}^n \left(\frac{T_i Y_i}{\hat{\pi}_i^2} + \frac{(1 - T_i) Y_i}{(1 - \hat{\pi}_i)^2} - \frac{\mu_i}{(1-\hat{\pi}_i)\hat{\pi}_i}\right) \phi'(X_i^\top  \hat{\gamma}) X_i^\top  (\hat{\gamma} - \gamma)} \notag \\
	= & \left|\frac{1}{n} \sum_{i=1}^n (\tilde{Y}_i - \mu_i) \frac{\phi'(X_i^\top \hat{\gamma})}{\hat{\pi}_i (1 - \hat{\pi}_i)} X_i^\top(\hat{\gamma} - \gamma)\right| \le \frac{1}{n} \|{\mbb X}^\top {\mbb \Pi} (\tilde{\mbb Y} - {\mbb \mu})\|_\infty \|\hat{\gamma} - \gamma\|_1 \notag
\end{align}
where ${\mbb \Pi} \in \R^{n \times n}$ is a diagonal matrix with
\[
\Pi_{ii} = \frac{\phi'(X_i^\top \hat{\gamma})}{\hat{\pi}_i (1 - \hat{\pi}_i)}.
\]
In this case, instead of finding $\hat{\mbb \mu}$ such that $\|{\mbb X}^\top (\tilde{\mbb Y} - {\mbb \mu})\|_\infty / n = O\left(\sqrt{\log p /n}\right)$ with high probability, we should aim for $\|{\mbb X}^\top {\mbb \Pi} (\tilde{\mbb Y} - {\mbb \mu})\|_\infty = O\left(\sqrt{\log p / n}\right)$. This suggests adapting our convex program in~\eqref{eq:mu_hat} to
\begin{gather}
	\hat{\mbb\mu} = \argmin_{\mbb\mu \in \R^n} \frac{1}{n}\|\tilde{\mu}(\mb X) - \mbb\mu\|_2^2 \notag\\
	\text{subject to } \;\;  \Big\|\frac{1}{n_A}\mb X_A^\top {\mbb \Pi}_A \{\tilde{\mb Y}_A - \tilde{\mu}(\mb X_A)\} - \frac{1}{n}\mb X^\top {\mbb \Pi} \{\mbb\mu - \tilde{\mu}(\mb X)\} \Big\|_\infty \leq \eta, \notag
\end{gather}
where $\mbb\Pi_A$ is a version of $\mbb\Pi_A$ based on dataset $\mathcal{D}_A$.
Then based on an analysis analogous to the proof of Theorem~\ref{thm:dipw}, we may see that constructing $\hat{\mbb \mu}$ using the above convex program and taking $\eta \asymp \sqrt{\log p / n}$, the theoretical properties described in Theorem~\ref{thm:dipw} continue to hold in this new setting.

\newrev{In Appendix~\ref{sec:lower}, we additionally show that when the link function can be expressed as $\phi(u) := \frac{1 + \exp(u)}{1 + 2 \exp(u)}$, the requirement that $s \lesssim \sqrt{n} / \log p$ is also a necessary condition for $\sqrt{n}$-consistent estimation. It would be of interest to prove the necessity of such a sparsity constraint under a more common link function such as logistic link: we leave this for future work.}

\subsection{Heterogeneous treatment effects} \label{sec:hetero}
Our debiased inverse propensity score estimator estimates the average treatment effect $\E(\Delta(X))=\tau$, where $\Delta := r_1 - r_0$ is the conditional average treatment effect, which captures any heterogeneity in the treatment effect. In many settings it is helpful to estimate other properties of  $\Delta$.
Let $W = g(X) \in \R^d$ be a function of $X\in\R^p$, where $d \ll p$. We consider here estimating the linear projection of $\Delta(X)$ onto $W$; $W$ could for example be the first few principal components of $X$, a collection of best estimates of the heterogeneous treatment effect function, or a subset of variables of special interest. Estimation of such linear projections of the potentially highly complex function $\Delta$ was advocated in \citet{chernozhukov2018generic} in the context where the propensity score is known exactly. In this section we outline an approach for estimating $\beta := \{\E(W W^\top )\}^{-1} \E\{ \Delta(X)W \}$.
Note that the problem of average treatment effect estimation we have studied thus far, may be seen as a special case of this by  taking $W$ equal to the first component of $X$ (which is identically $1$) and thus considering the projection onto an intercept term.

Let $W_i := g(X_i)$ and $W_{A,i} := g(X_{A,i})$. Observe that
\[
\Delta(X) = \E \bigg( \frac{T Y}{\pi(X)} - \frac{(1-T)Y}{1-\pi(X)} \,\Big|\, X\bigg).
\]
This suggests estimating $\beta$ using least squares regression via
\[
\hat{\beta} = \argmin_{b \in \R^d} \frac{1}{n}\sum_{i=1}^n \bigg( \frac{T_i (Y_i - \mu_i)}{\hat{\pi}_i} - \frac{(1-T_i)(Y_i - \mu_i)}{1-\hat{\pi}_i} - W_i^\top  b\bigg)^2,
\]
with $\mbb\mu \in \R^n$ playing a bias correcting role as before. The main bias terms that we would need to control are analogous versions of equation \eqref{eq:holder} where the $i$th summand is weighted by $W_{ij}$ for each $j=1,\ldots,d$:
\[
\max_{j =1,\ldots,d} \abs{\frac{1}{n}\sum_{i=1}^n W_{ij}(\tilde{Y}_i - \mu_i)X_i^\top (\hat{\gamma}-\gamma)} \leq \max_{j =1,\ldots,d}  \norm{\frac{1}{n}\sum_{i=1}^n W_{ij}(\tilde{Y}_i - \mu_i)X_i}_\infty \|\hat{\gamma} - \gamma\|_1 ;
\]
the inequality follows from H\"older's inequality as before. We should thus seek to choose $\mbb\mu \in \R^n$ such that the first term on the RHS is small. In analogy with \eqref{eq:constraint}, we may pick $\mbb\mu = \hat{\mbb\mu}$ given by
\begin{gather}
	\hat{\mbb\mu} = \argmin_{\mbb\mu \in \R^n} \frac{1}{n}\|\tilde{\mu}(\mb X) - \mbb\mu\|_2^2 \notag\\
	\text{subject to } \;\;  \max_{j =1,\ldots,d}  \norm{\frac{1}{n_A}\sum_{i=1}^{n_A} W_{A,ij}\{\tilde{Y}_{A,i} - \tilde{\mu}(X_{A,i})\}X_{A,i} - \frac{1}{n}\sum_{i=1}^n W_{ij} \{\mu_i - \tilde{\mu}(X_i)\}X_i }_\infty \leq \eta. \label{eq:constr_ext}
\end{gather}
If we are willing to assume that $\Delta(x) = x^\top \beta$ for a sparse coefficient vector $\beta \in \R^p$, we can take $W_i = X_i$ and $W_{A,i}=X_{A,i}$ in the constraint \eqref{eq:constr_ext}, and estimate $\beta$ using Lasso regression \citep{tibshirani1996regression}:
\[
\hat{\beta} = \argmin_{b \in \R^p} \frac{1}{n}\sum_{i=1}^n \bigg( \frac{T_i (Y_i - \hat{\mu}_i)}{\hat{\pi}_i} - \frac{(1-T_i)(Y_i - \hat{\mu}_i)}{1-\hat{\pi}_i} - X_{i}^\top  b\bigg)^2 + \lambda \sum_{j=2}^p |b_j|.
\]

\section{Numerical experiments} \label{sec:numerical}
In this section we report the results of numerical experiments exploring the empirical performance of our DIPW estimator. We first consider the problem of average treatment effect estimation in Section~\ref{sec:simate}, and then in Section~\ref{sec:simvar} turn to the problem of estimating $\Var(Y(1))$, which we approach by applying our methodology to appropriately transformed responses. Section~\ref{sec:simate_hetero} in the appendix contains additional numerical results for the average treatment effect estimation problem where the errors are heteroscedastic.

\subsection{Average treatment effect}\label{sec:simate}
Below we describe our experimental setups for studying estimation of the average treatment effect.
\subsubsection{Experimental setups}
We generate data $(X_i, T_i, Y_i) \in \R^p \times \{0,1\} \times \R$ for $i=1,\ldots,n=100$ and $p=400$ according to the following logistic propensity score model for the treatment $T_i$ and regression model for $Y_i$:
\begin{align}
	\PP(T_i = 1 \mid X_i) &= \psi(X_i^\top  \gamma), \notag\\
	Y_i &= b(X_i) + T_i \Delta(X_i) + \varepsilon_i ; \label{eq:sim_Y}
\end{align}
here $\psi$ denotes the standard logistic function $\psi(u) := \{1+\exp(-u)\}^{-1}$ as before, the $\varepsilon_i$ are independent and standard Gaussian. To compare the performance of our method under different scenarios, we consider $3$ models for generating the covariates $X_i$, as well as $2$ constructions for the pair of functions $(b(\cdot), \Delta(\cdot))$ and $3$ settings the vector of coefficients $\gamma$, as detailed below. Each of the $3 \times 2 \times 3$ settings were simulated $250$ times, with design matrices, regression functions and propensity score coefficients generated anew in each run. The results are presented in Section~\ref{sec:ATE_results}.

\paragraph{Covariate designs}
\begin{itemize}
	\item \emph{Toeplitz:} $X_i \sim \mathcal{N}(0, \Sigma)$, where $\Sigma\in \R^{p \times p}$ is given by $\Sigma_{i,j} = 0.9^{|i - j|}$.
	\item \emph{Exponential decay:} As above, but with $\Sigma$ such that  $(\Sigma^{-1})_{i,j} = 0.9^{|i - j|}$, and then normalised such that all diagonal entries of $\Sigma$ equal $1$.
	\item \emph{Real:} We use gene expression data collated by \citet{gtex2017genetic}, and preprocessed as described in \citet[Sec.~5.2]{shah2020right}, available at \url{https://github.com/benjaminfrot/RSVP}. We selected $p$ genes with the largest empirical variance in the data and the first $n=100$ out of a total of $491$ observations.
\end{itemize}

\paragraph{Response function $b(\cdot)$} We consider the following two forms of the response function.
\begin{itemize}
	\item \emph{Dense linear function:} We set $b(x) = x^\top  \beta$, where we generate $\beta \in \R^p$ by first randomly selecting $50$ components and then independently assigning each selected entry a random value uniformly sampled from $[0,1]$. All other entries are set to $0$ and finally $\beta$ is normalised such that $\|\beta\|_2=2$.
	\item \emph{Nonlinear function:} We generate $\beta\in \R^p$ as above and then set $b(x) = 3 \tilde{x}^\top  \beta$, where $\tilde{x}_{j} = 2 \{1 + \exp(-x)\}^{-1}  - 1$.
\end{itemize}

\paragraph{Heterogeneous treatment effect function $\Delta(\cdot)$} We consider the following two forms of heterogeneous treatment effects.
\begin{itemize}
	\item \emph{Dense linear function:} We generate a vector $\delta \in \R^p$ in the same fashion as $\beta$ above, but normalise $\|\delta\|_2=1$; we then set $\Delta(x) = x^\top \delta$.
	\item \emph{Nonlinear function:} We first generate $\delta \in \R^p$ as above and then set $\Delta(x) = \{1 + \exp(x^\top \delta)\}^{-1} - 0.5$.
\end{itemize}

\paragraph{Propensity model parameter $\gamma$} Given the randomly generated $\beta$, we randomly select $s = 5, 20, 50$ indices corresponding to non-zero entries in $\beta$. These indices of $\gamma$ are then independently assigned a random value uniformly sampled from $[0,1]$, with the remaining components of $\gamma$ set to zero. We then normalise $\gamma$ such that $\|\gamma\|_2=1$.

\paragraph{Methods under comparison}\label{sec:methodsetup}
We apply our DIPW estimator as described in Algorithm~\ref{alg:dipw} with $B=3$ splits and the given default choices of regression methods with tuning parameters selected by 10-fold cross-validation and the $1$ standard error rule, as implemented in \citet{friedman2010regularization}.
To benchmark our results, we compare our method to inverse propensity weighting (IPW), augmented inverse propensity weighting (AIPW), approximate residual balancing (ARB) \citep{AIW18}, targetted maximum likelihood estimation (TMLE) \citep{VR06}, \rev{regularized
	calibrated estimation based treatment effect estimator (RCAL)~\citep{Tan20b} and high-dimensional covariate balancing propensity score based estimator (hdCBPS)~\citep{NPI18},} implemented as follows. For IPW, AIPW and TMLE, we use the same estimated propensity scores as for DIPW, and outcome regression models for AIPW and TMLE were estimated using the Lasso with tuning parameter selected as for the regressions involved in the DIPW estimator. For ARB, we use the default parameter settings in the associated \texttt{balanceHD} package. \rev{For RCAL, we use the cross-validated treatment effect estimator in the \texttt{RCAL} package, with number of folds equal to $10$ and number of tuning parameters equal to $100$ for the calibrated estimation of the two nuisance functions, mirroring the default options of the popular Lasso implementation \texttt{glmnet} \citep{friedman2010regularization}. For CBPS, we use the function \texttt{hdCBPS} in the \texttt{CBPS} package. As the optimisation algorithm used in \texttt{hdCBPS} is computationally demanding, we limit the number of iterations in the optimisation to $100$.}

\subsubsection{Results} \label{sec:ATE_results}

\begin{figure}[t!]
	\subfigure[Toeplitz design, $s = 5$]{\includegraphics[width=0.32\textwidth]{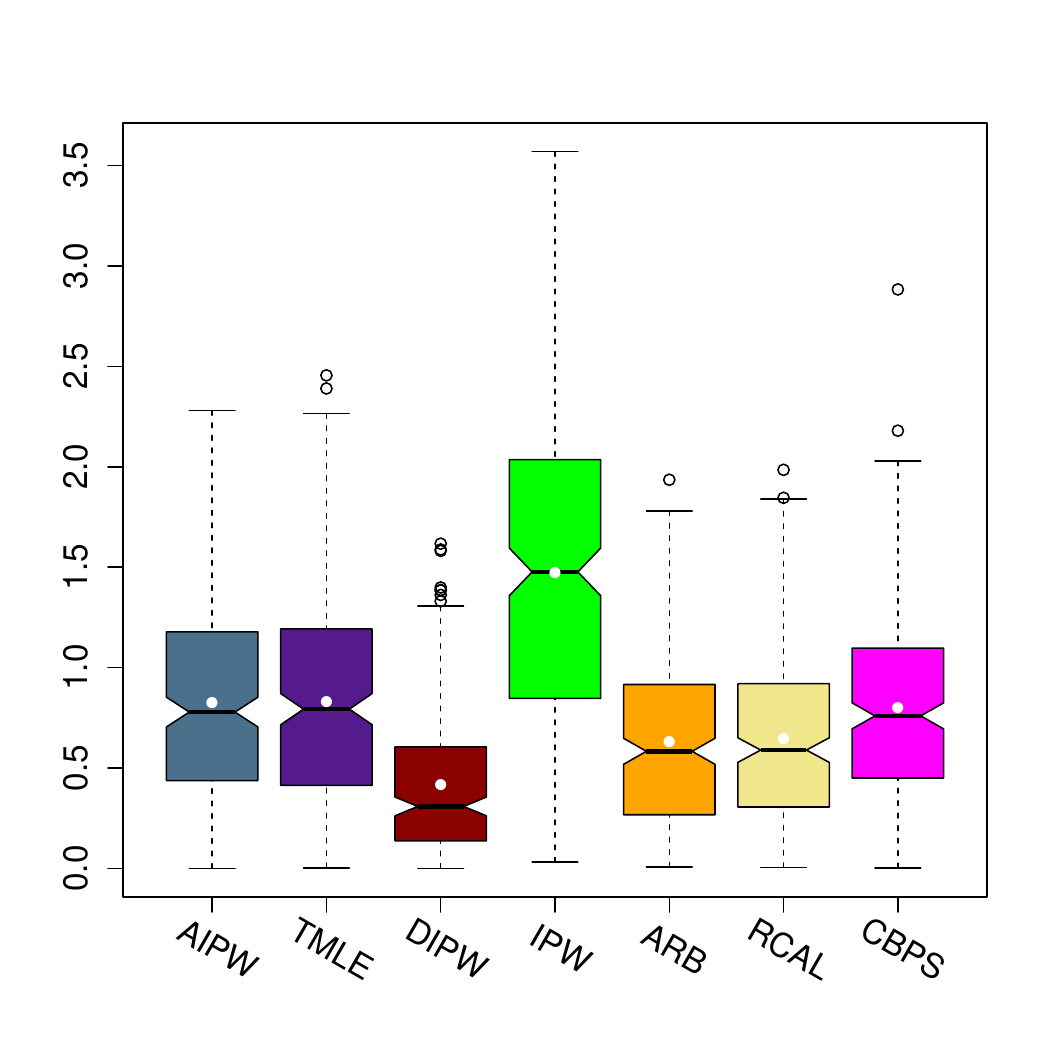}} \hfill
	\subfigure[Toeplitz design, $s = 20$]{\includegraphics[width=0.32\textwidth]{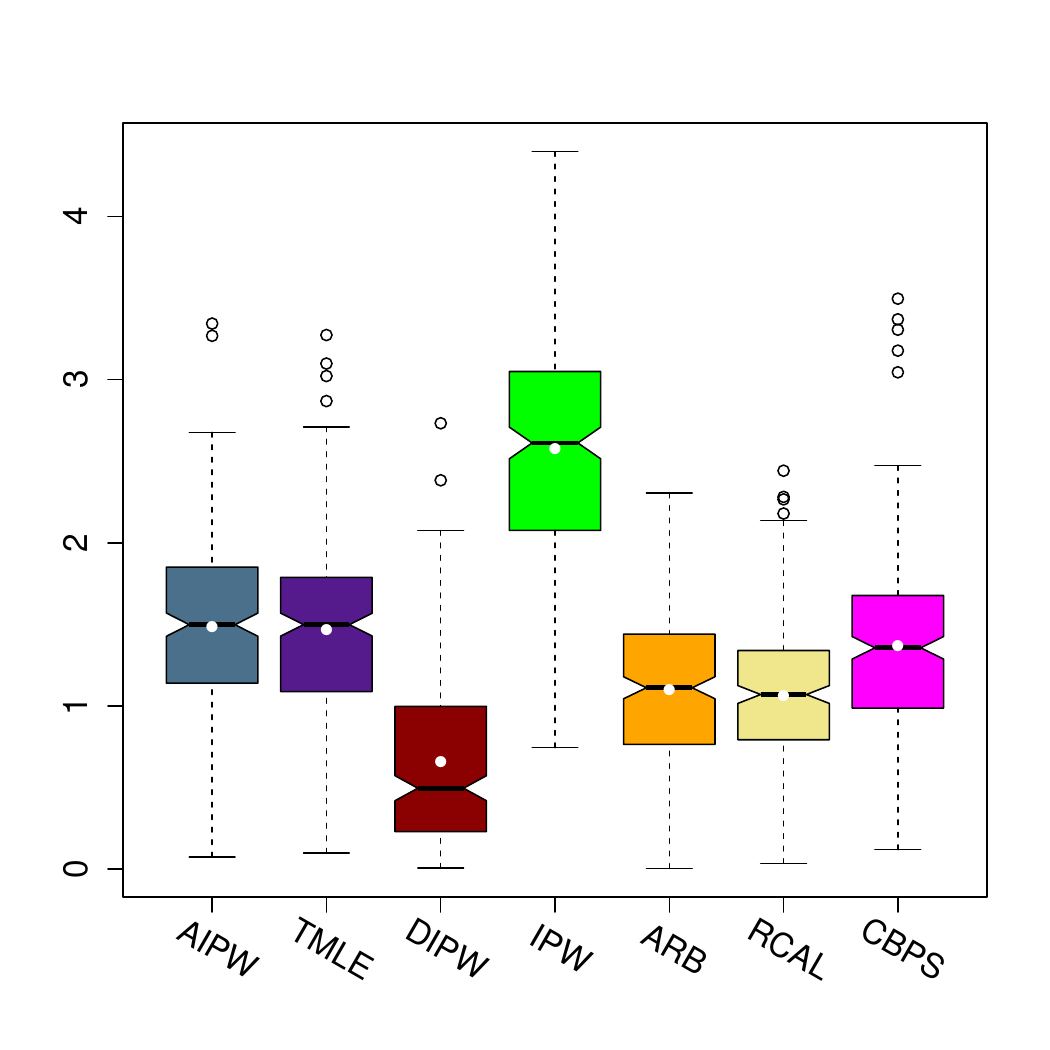}} \hfill
	\subfigure[Toeplitz design, $s = 50$]{\includegraphics[width=0.32\textwidth]{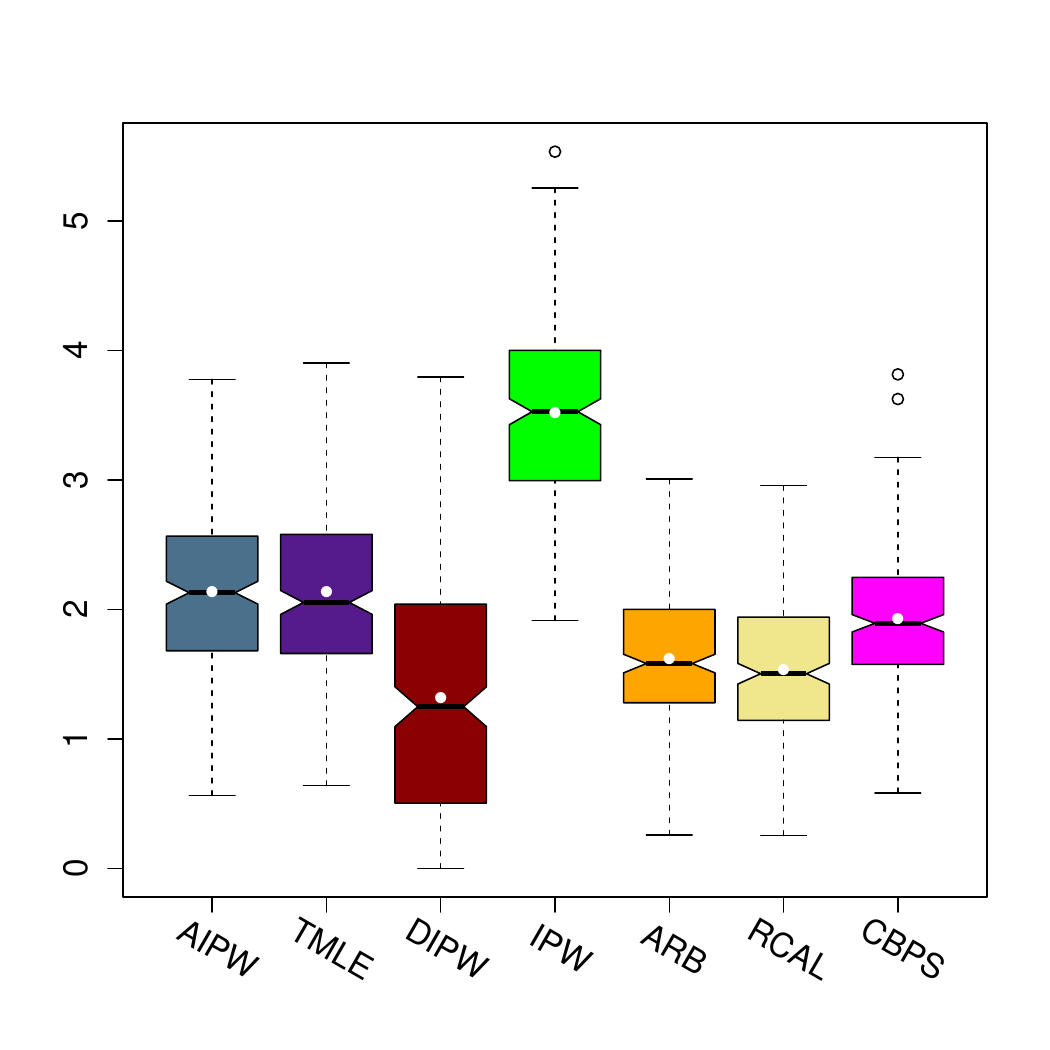}} \\
	\subfigure[Exponential design, $s = 5$]{\includegraphics[width=0.32\textwidth]{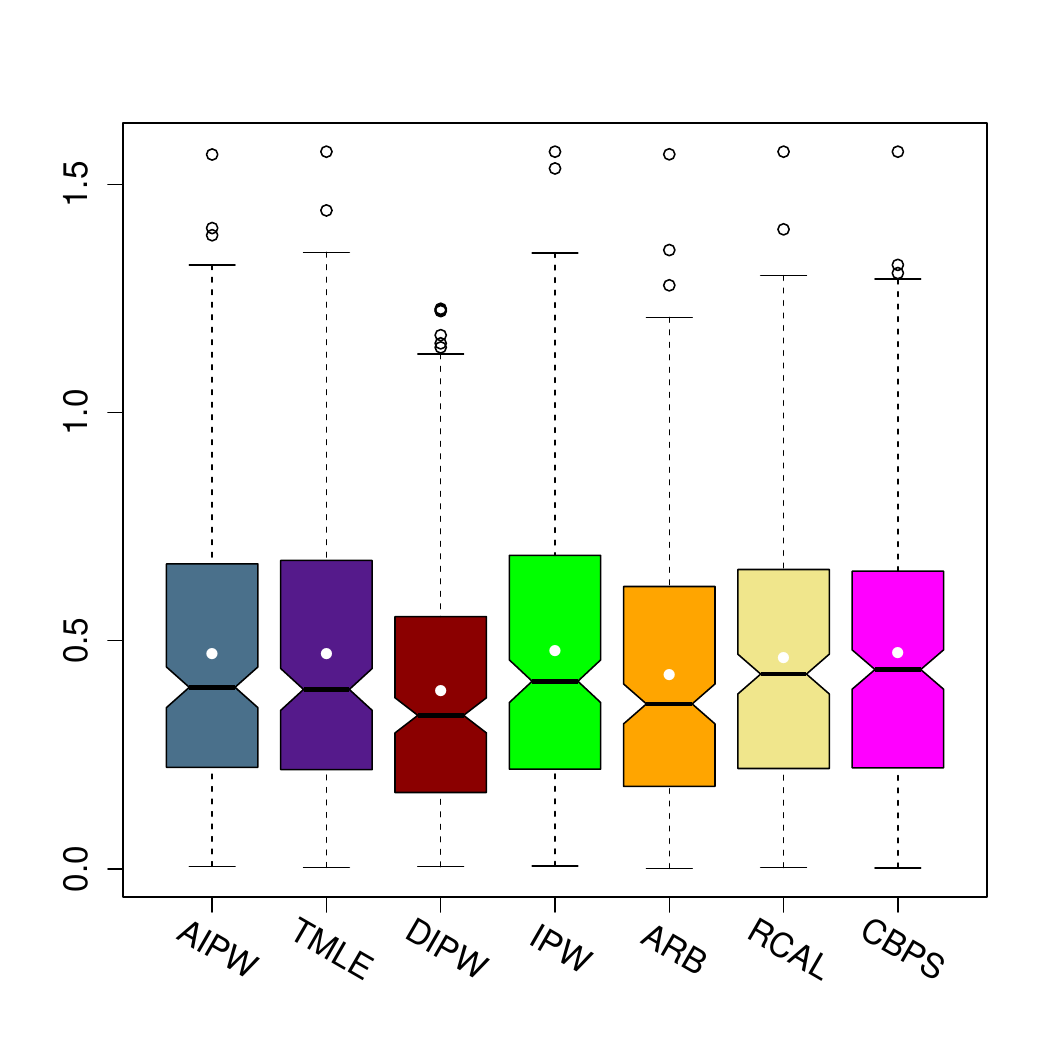}} \hfill
	\subfigure[Exponential design, $s = 20$]{\includegraphics[width=0.32\textwidth]{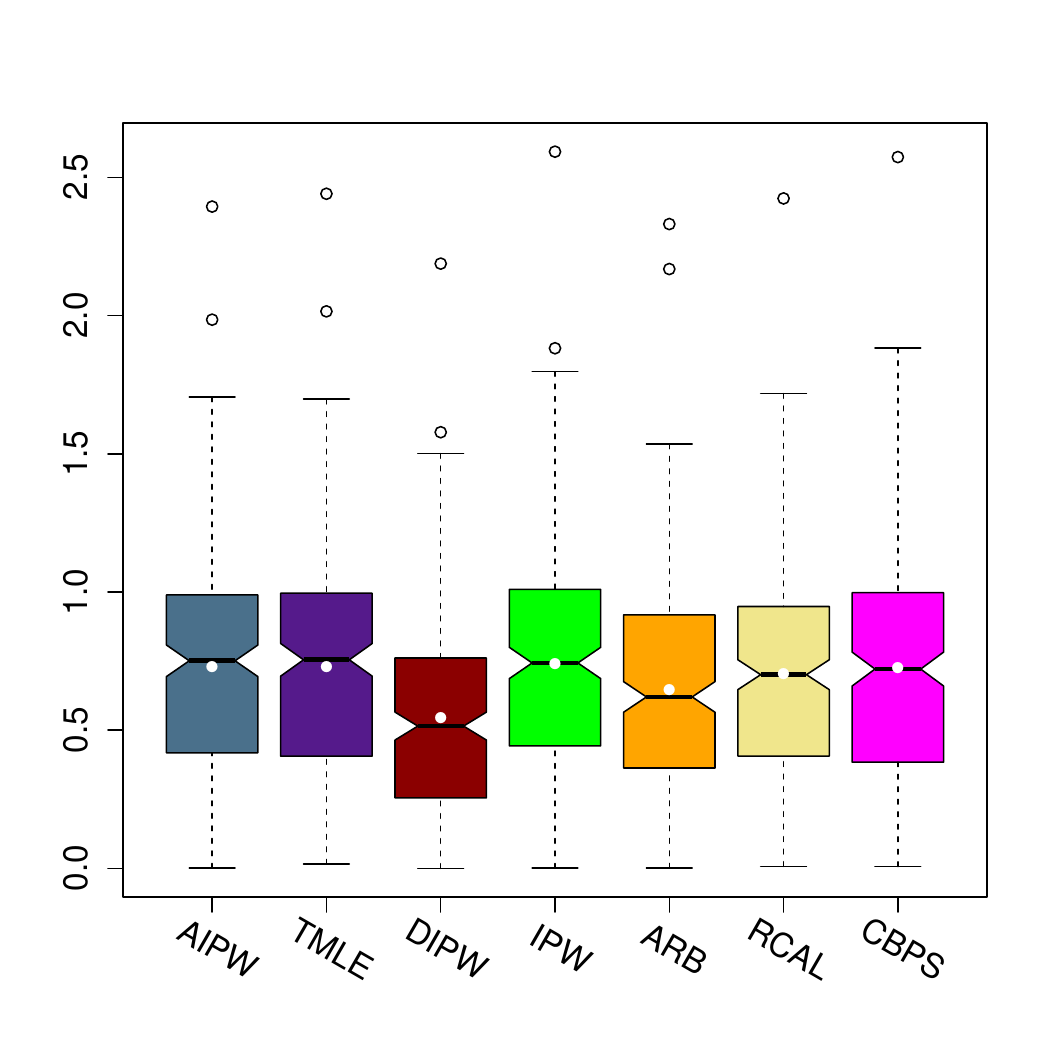}} \hfill
	\subfigure[Exponential design,  $s = 50$]{\includegraphics[width=0.32\textwidth]{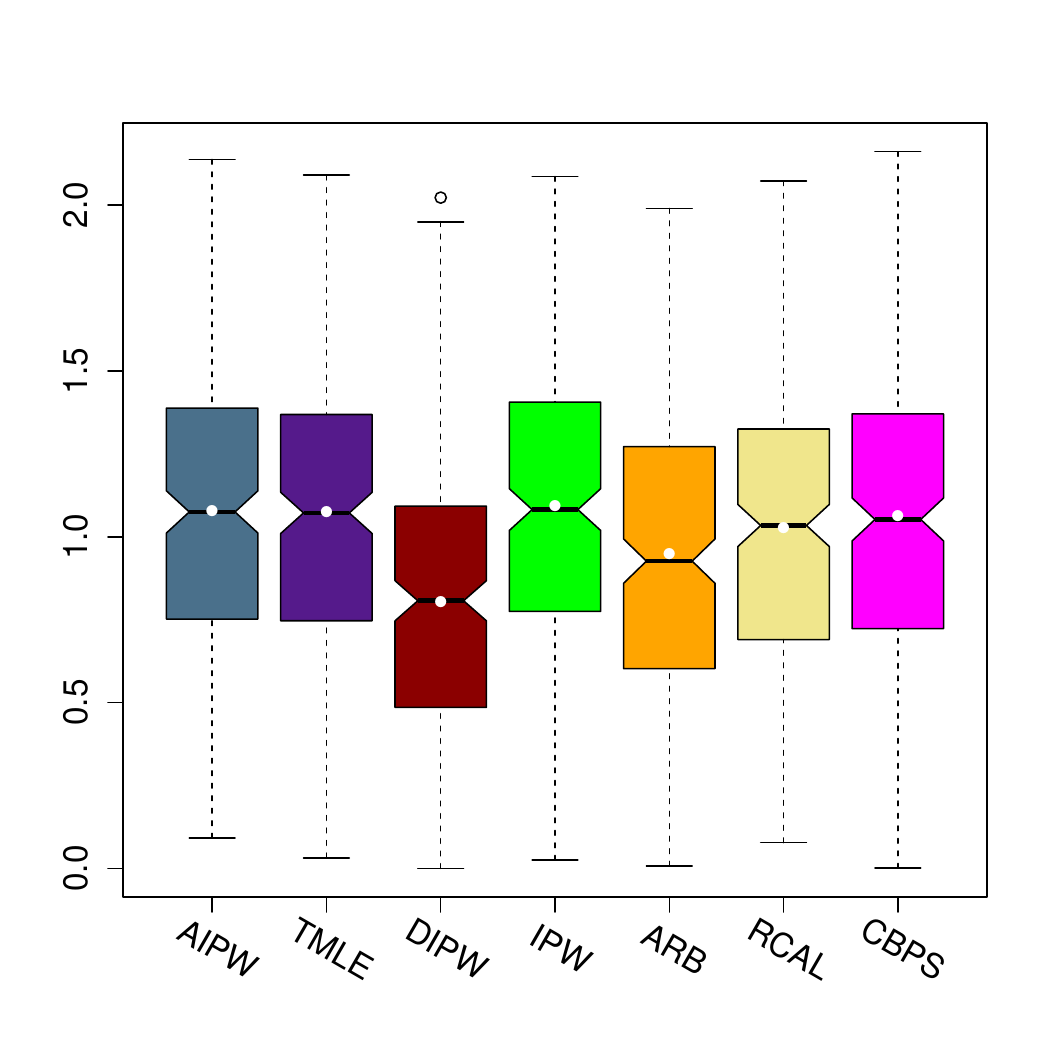}} \\
	\subfigure[Real data  design, $s = 5$]{\includegraphics[width=0.32\textwidth]{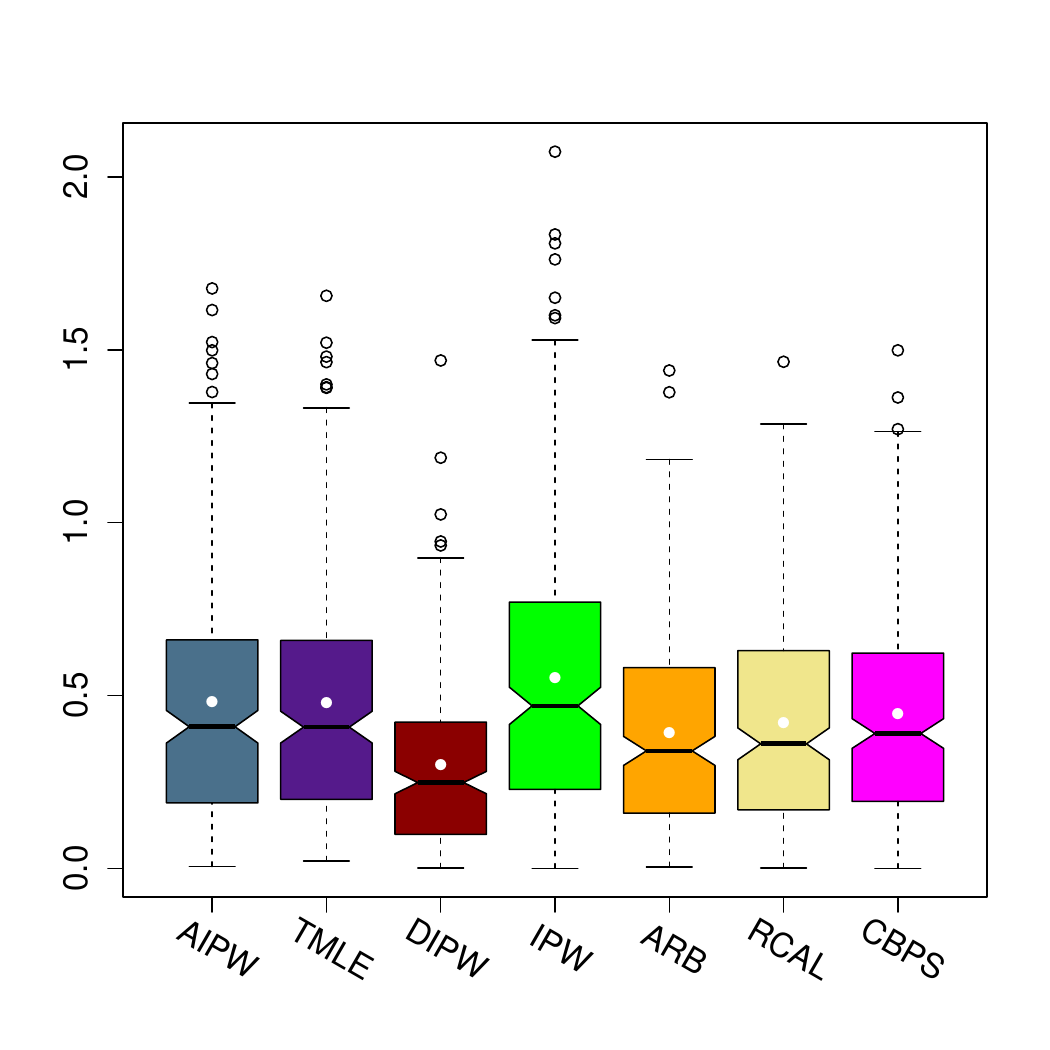}} \hfill
	\subfigure[Real data design, $s = 20$]{\includegraphics[width=0.32\textwidth]{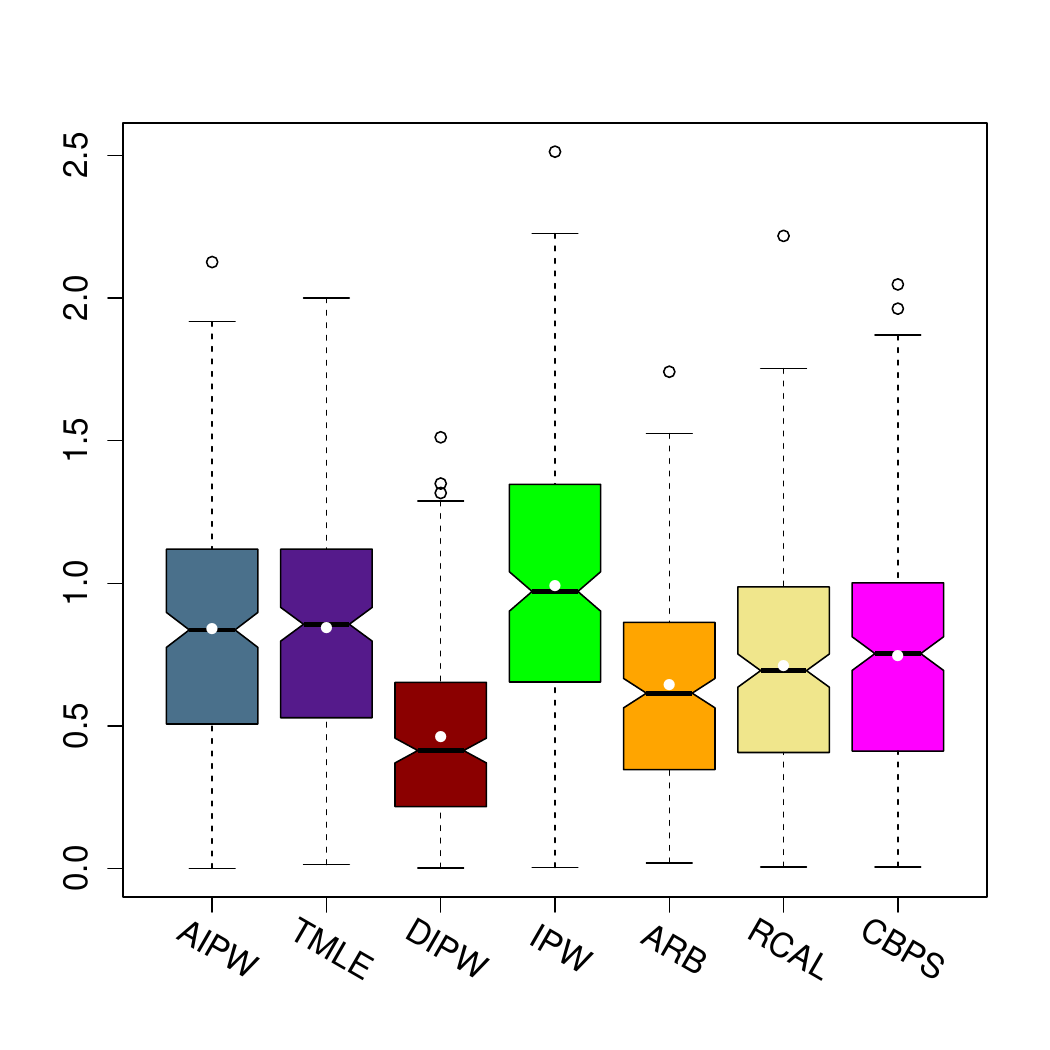}} \hfill
	\subfigure[Real data design,  $s = 50$]{\includegraphics[width=0.32\textwidth]{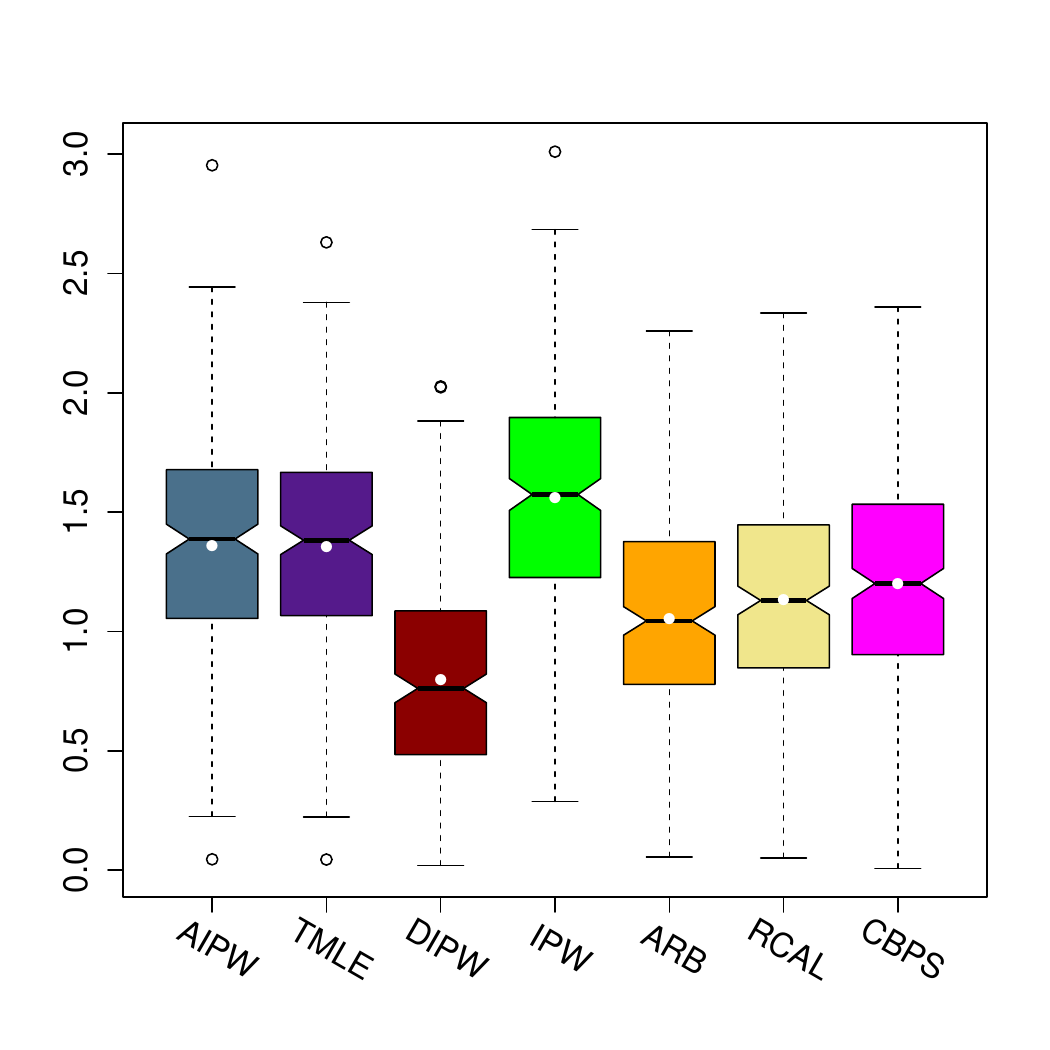}}
	\caption{Boxplots of the estimation error $|\hat{\tau} - \bar{\tau}|$ under different covariate designs and linear functions $b(\cdot)$ and $\Delta(\cdot)$ with different sparsity levels $s$ for the propensity model coefficients; the white dots correspond to means.}\label{fig:linear}
\end{figure}

\begin{figure}[t!]
	\subfigure[Toeplitz design, $s = 5$]{\includegraphics[width=0.32\textwidth]{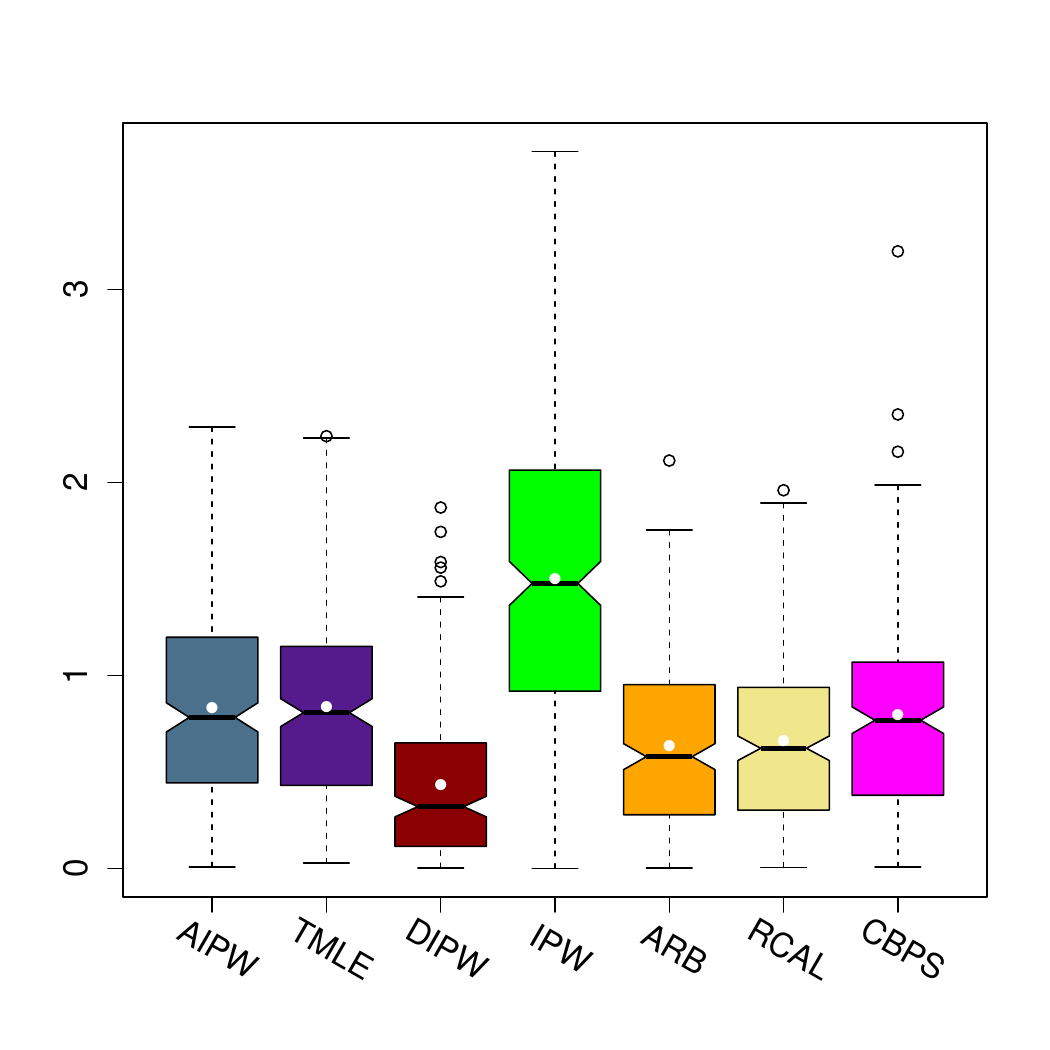}} \hfill
	\subfigure[Toeplitz design, $s = 20$]{\includegraphics[width=0.32\textwidth]{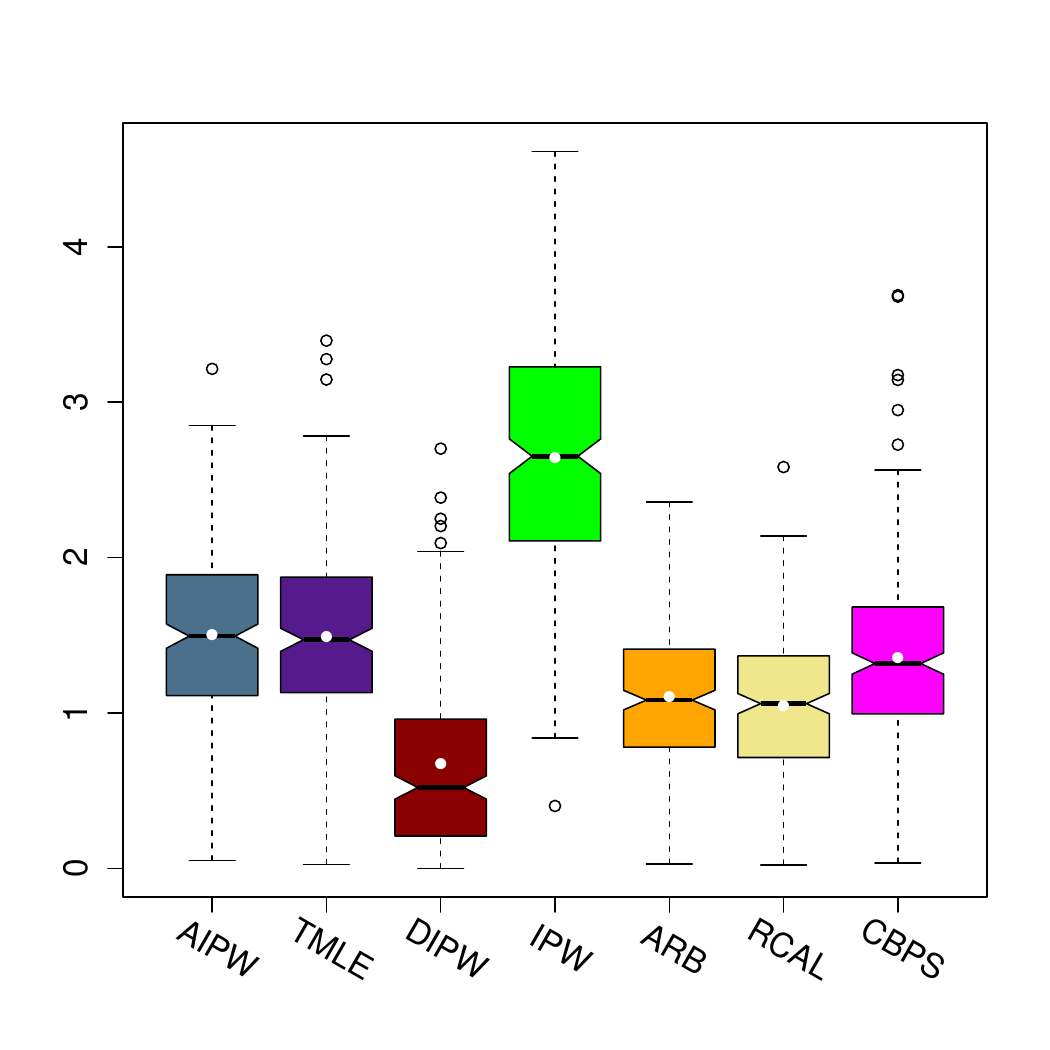}} \hfill
	\subfigure[Toeplitz design, $s = 50$]{\includegraphics[width=0.32\textwidth]{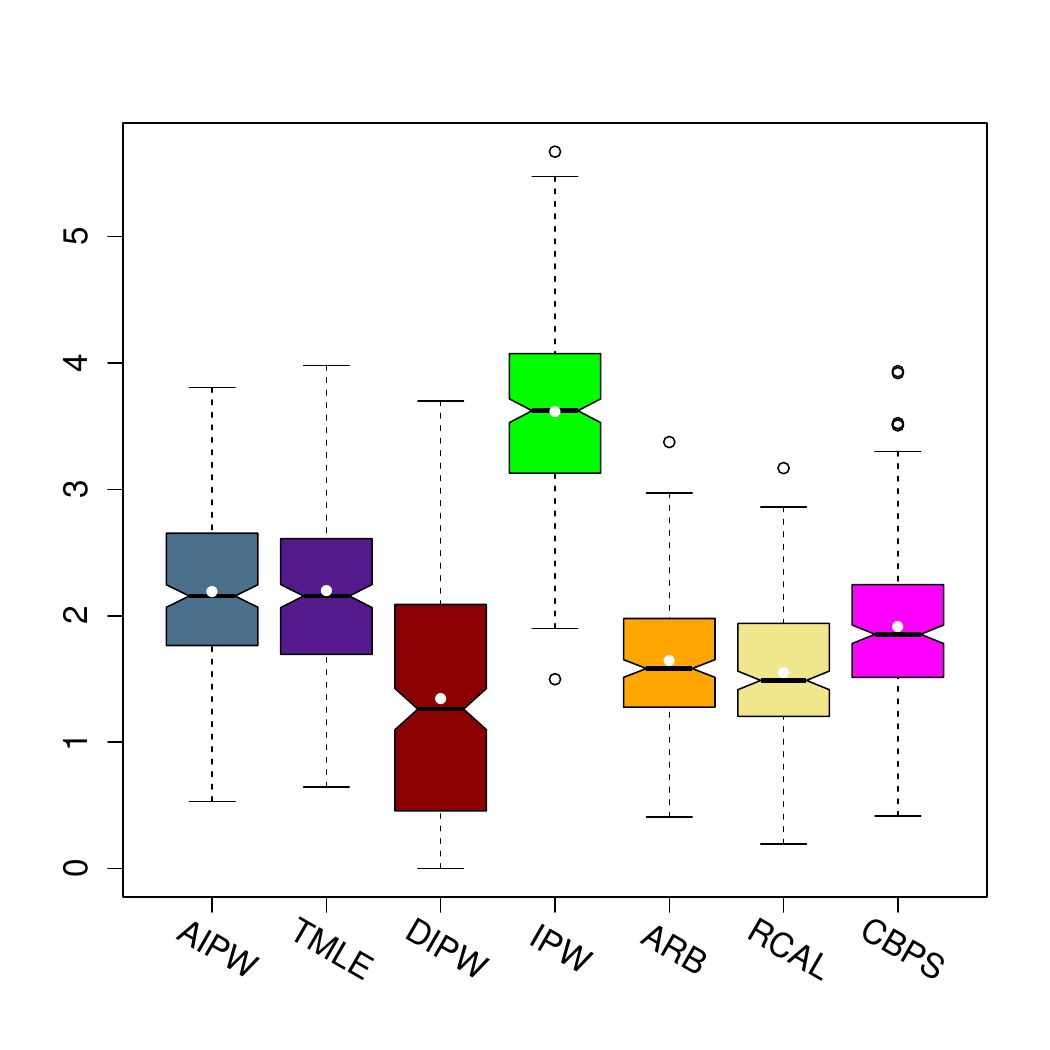}} \\
	\subfigure[Exponential design, $s = 5$]{\includegraphics[width=0.32\textwidth]{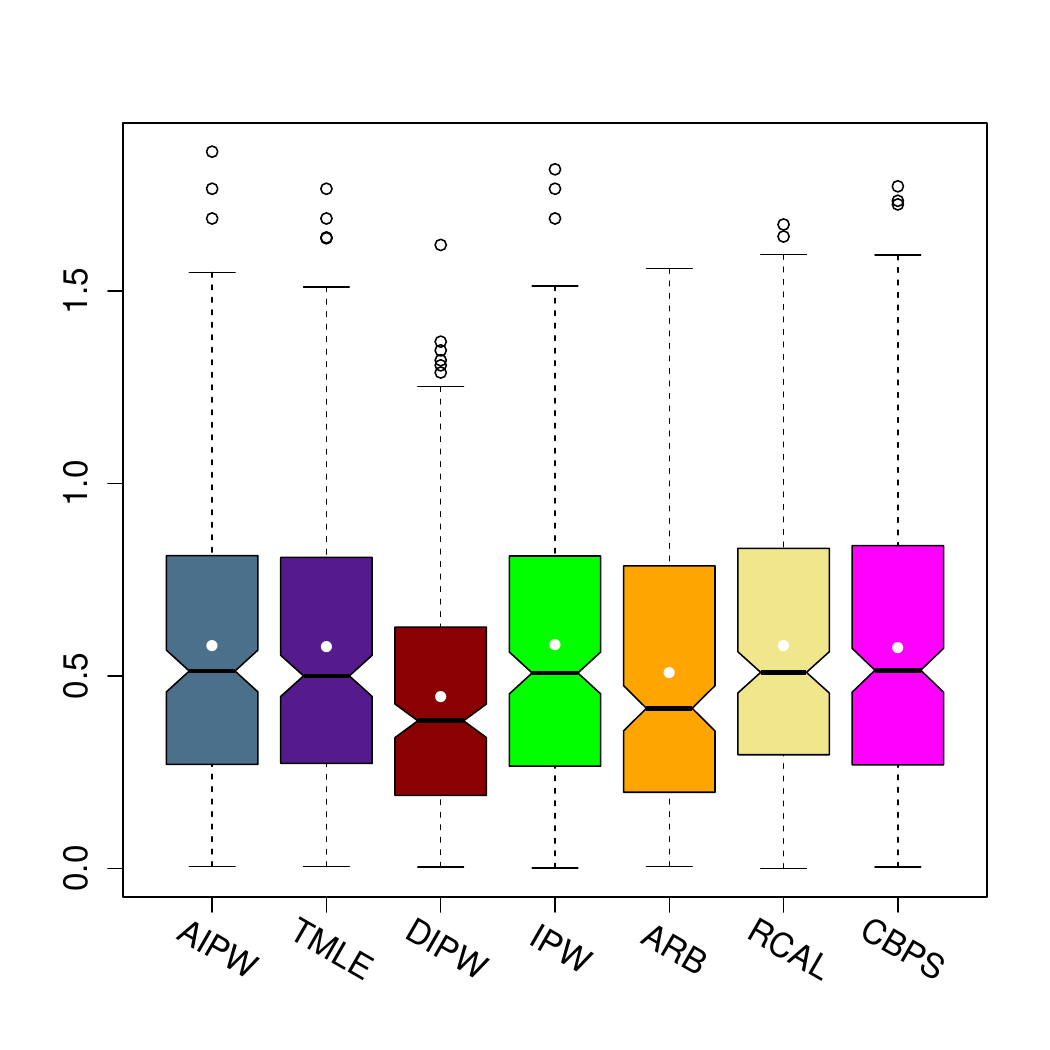}} \hfill
	\subfigure[Exponential design, $s = 20$]{\includegraphics[width=0.32\textwidth]{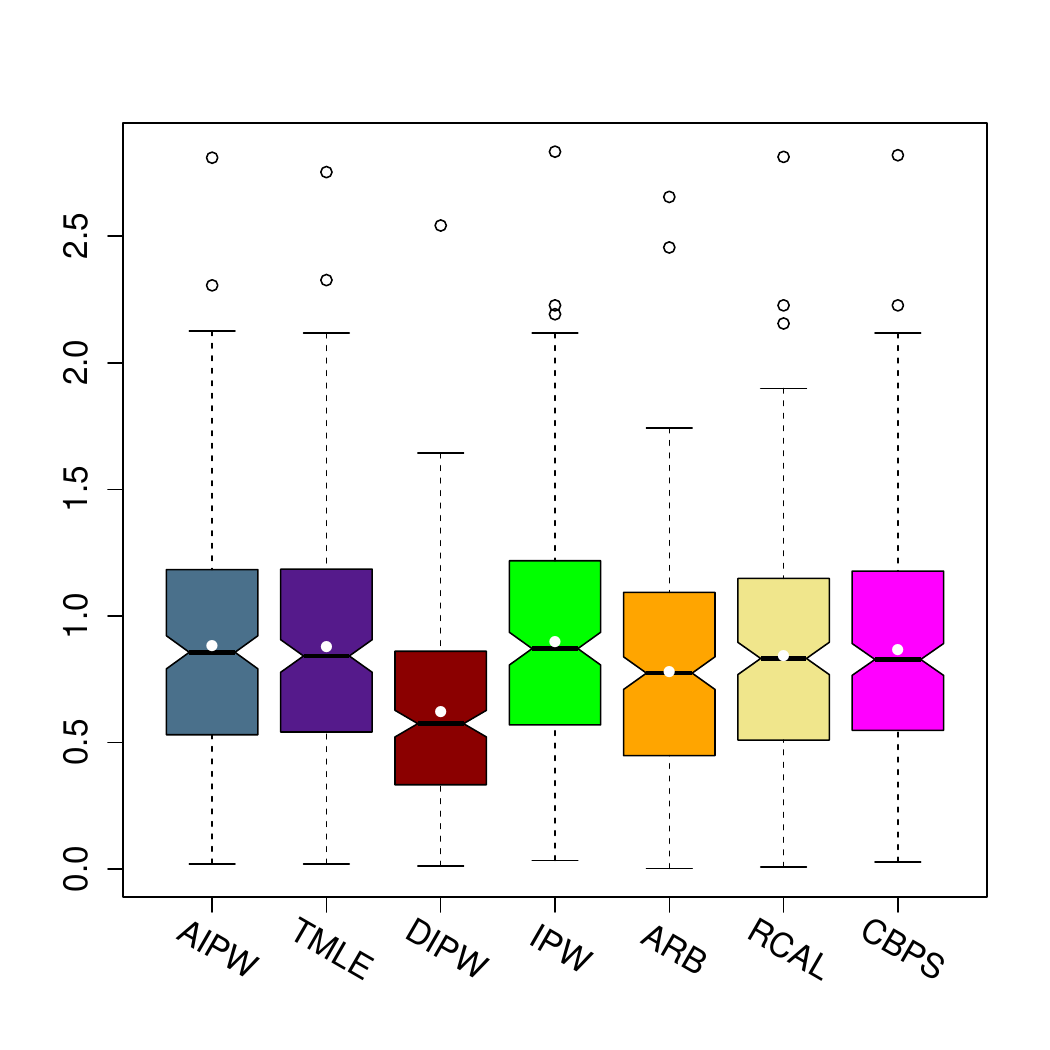}} \hfill
	\subfigure[Exponential design,  $s = 50$]{\includegraphics[width=0.32\textwidth]{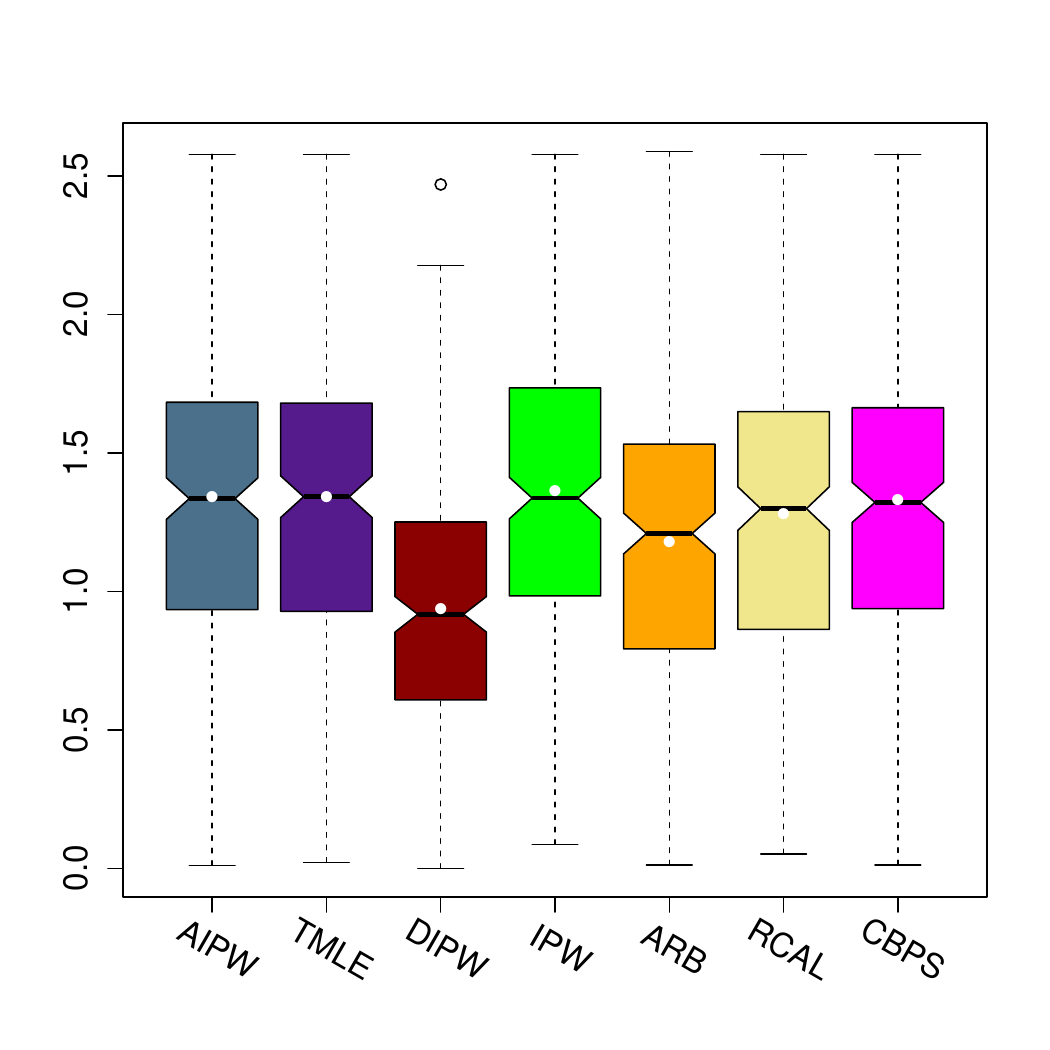}} \\
	\subfigure[Real data design, $s = 5$]{\includegraphics[width=0.32\textwidth]{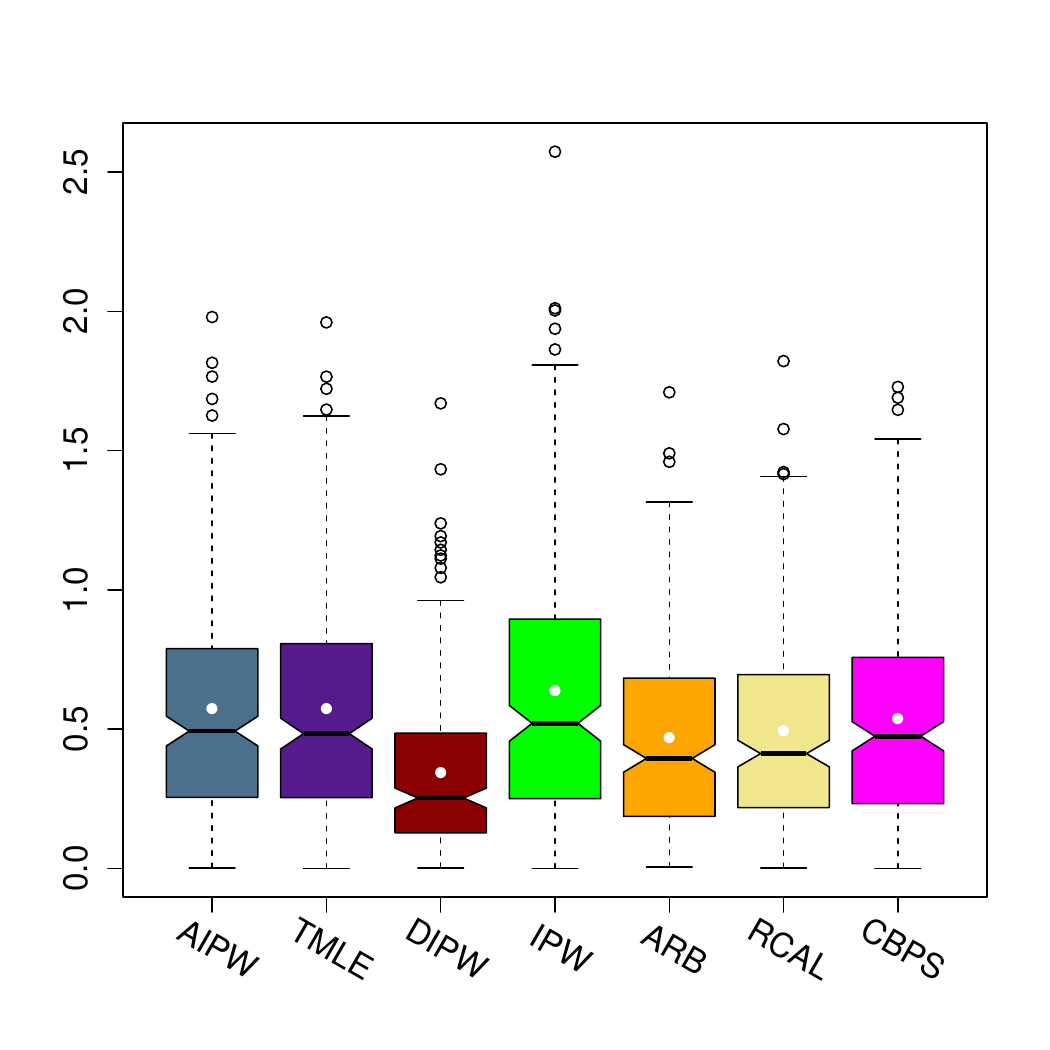}} \hfill
	\subfigure[Real data design, $s = 20$]{\includegraphics[width=0.32\textwidth]{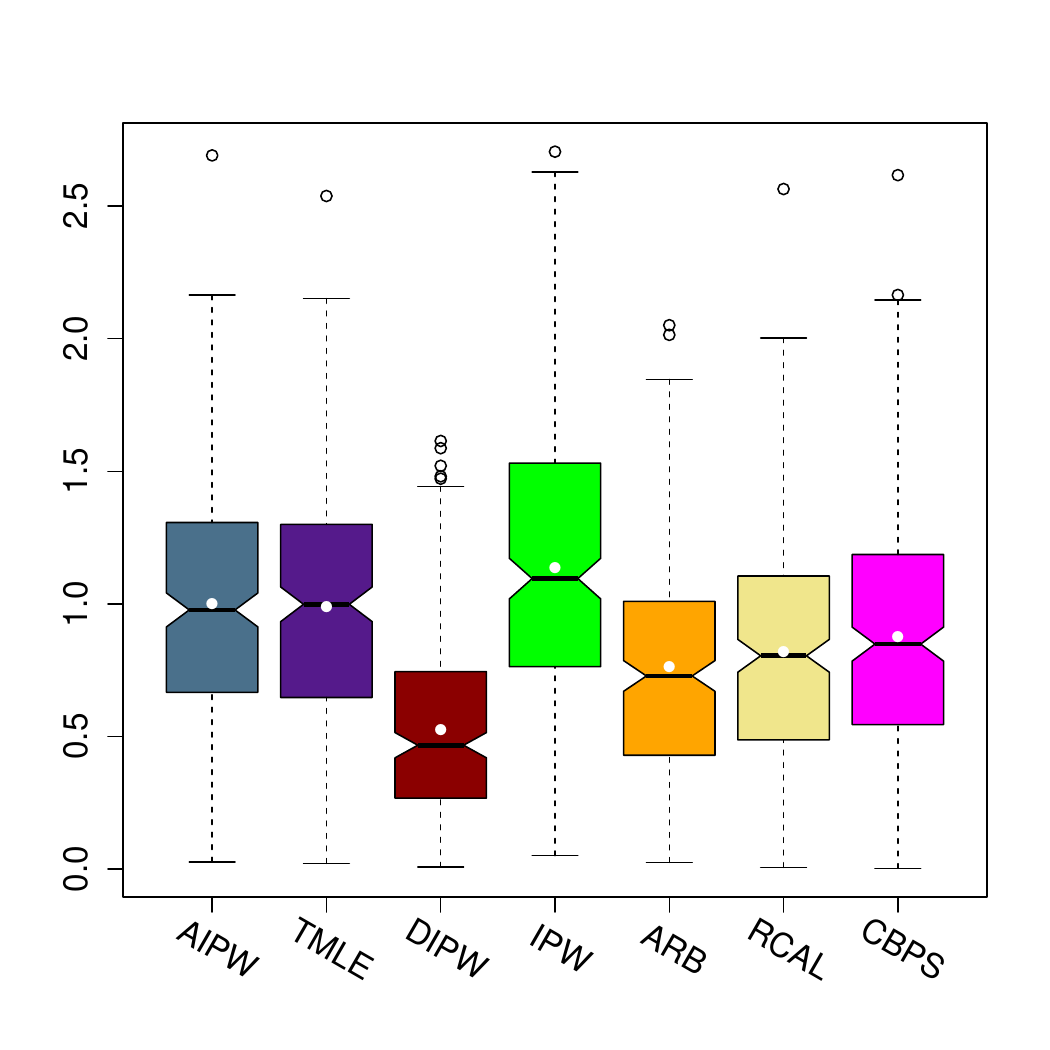}} \hfill
	\subfigure[Real data design,  $s = 50$]{\includegraphics[width=0.32\textwidth]{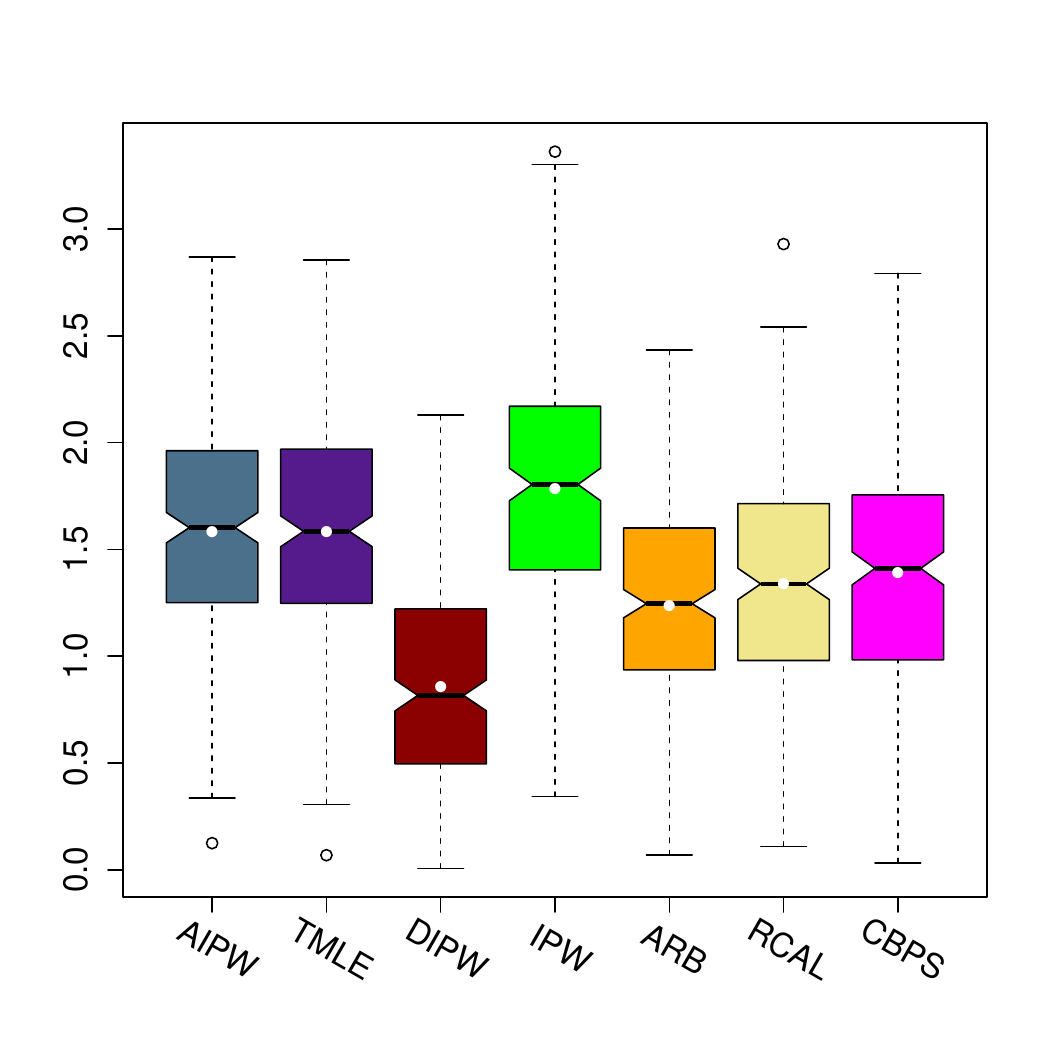}}
	\caption{As Figure~\ref{fig:linear} but with nonlinear functions $b(\cdot)$ and $\Delta(\cdot)$.}\label{fig:nonlinear}
\end{figure}

Figures~\ref{fig:linear} and~\ref{fig:nonlinear} present boxplots of the estimation error $|\hat{\tau} - \tau|$ for each of the estimators $\hat{\tau}$ and the various settings under consideration, where $\tau$ is the average treatment effect. We see that DIPW compares favourably with competing methods across all the scenarios studied. Surprisingly it performs particularly well in the near-dense propensity model settings where $s=50$, relative to other methods. This is perhaps surprising given that our theory would indicate a sparse propensity score model would be advantageous. However, the average estimation errors are higher in the $s=50$ settings, but other methods appear to struggle more with the lack of sparsity. Another interesting result is the competitive performance of ARB. Theoretical results in \citet{AIW18} suggest that the ARB estimator would only be $\sqrt{n}$-consistent in settings where the response functions are sparse linear model; however the method outperforms some of the other approaches with the exception of DIPW even in settings where the response functions are nonlinear.

\rev{\subsubsection{Results for confidence intervals} \label{sec:conf_res} We apply the confidence interval construction method as described in Section~\ref{sec:confidence}. For comparison we also consider intervals associated with each of AIPW, ARB, TMLE, RCAL and CBPS, using the options specified in Section~\ref{sec:methodsetup}. Table~\ref{table:cp} presents the empirical coverage probability and average confidence interval length for nominal $95\%$ confidence intervals among $250$ replicates. Here we consider the same simulation setups as in Figure~\ref{fig:linear}. We can see that for most of the sparse settings, $s = 5, 20$, the empirical coverage of DIPW is correct or is only slightly lower than the prespecified nominal level. For the dense setting $s = 50$, the empirical coverage is substantially lower; such low coverage is expected as the bias of $\hat{\tau}_{\mdipw}$ is large due to the difficulty in estimating propensity score. For the other approaches, their coverage probabilities are in general below the nominal level by a somewhat larger margin in both the dense and sparse cases. 
}
\begin{table}[!t]
	\centering
	\begin{tabular}{c|cc|cccccc} \hline\hline
		design & sparsity $s$ & & AIPW & DIPW & ARB & TMLE & RCAL & CBPS \\
		\hline
		\multirow{6}{*}{Toeplitz} & \multirow{2}{*}{5} & CP & 23.2\% & {\bf 92.4\%} & 66.0\% & 26.0\% & 48.0\% & 48.4\%\\
		&   & Length & 0.40 & 1.00 & 0.82 & 0.43 & 0.56 & 2.63\\
		& \multirow{2}{*}{20} & CP & 2.0\% & {\bf 75.6\%} & 27.6\% & 2.4\% & 14.8\% & 27.2\%\\
		&   & Length & 0.45 & 1.04 & 0.85 & 0.46 & 0.56 & 6.75\\
		& \multirow{2}{*}{50} & CP & 0.0\% & {\bf 46.0\%} & 10.8\% & 0.8\% & 2.8\% & 22.0\%\\
		&   & Length & 0.49 & 1.05 & 0.93 & 0.61 & 0.57 & 11.69\\ \hline
		\multirow{6}{*}{Exponential} & \multirow{2}{*}{5} & CP & 82.0\% & {\bf 99.2\%} & 92.0\% & 86.0\% & 76.4\% & 78.0\%\\
		&   & Length & 0.82 & 1.29 & 0.93 & 0.88 & 0.72 & 1.66\\
		& \multirow{2}{*}{20} & CP & 58.8\% & {\bf 99.2\%} & 73.6\% & 60.0\% & 51.2\% & 50.8\%\\
		&   & Length & 0.81 & 1.28 & 0.91 & 0.86 & 0.71 & 1.58\\
		& \multirow{2}{*}{50} & CP & 26.0\% & {\bf 86.8\%} & 44.8\% & 28.8\% & 24.8\% & 24.8\%\\
		&   & Length & 0.79 & 1.28 & 0.88 & 0.83 & 0.69 & 1.13\\ \hline
		\multirow{6}{*}{Real data} & \multirow{2}{*}{5} & CP & 77.6\% & {\bf 99.2\%} & 91.6\% & 78.8\% & 72.4\% & 74.4\%\\
		&   & Length & 0.74 & 1.15 & 0.88 & 0.78 & 0.60 & 1448.33\\
		& \multirow{2}{*}{20} & CP & 39.2\% & {\bf 96.0\%} & 72.8\% & 43.6\% & 39.2\% & 47.6\%\\
		&   & Length & 0.72 & 1.11 & 0.87 & 0.77 & 0.59 & 8.21\\
		& \multirow{2}{*}{50} & CP & 9.6\% & {\bf 80.4\%} & 33.6\% & 12.0\% & 10.8\% & 22.0\%\\
		&   & Length & 0.74 & 1.12 & 0.88 & 0.79 & 0.60 & 2.96\\
		\hline\hline
	\end{tabular}
	\caption{Coverage probability and average length of confidence intervals with $95\%$ nominal level. ``CP'' stands for the empirical coverage probability among the 250 replicates, ``Length'' stands for the average length of confidence intervals among the 250 replicates with 95\% nominal level. The largest coverage probabilities in each row are marked in bold. The settings are the same as in Figure~\ref{fig:linear}.}\label{table:cp}
\end{table}

\rev{Compared with most of the competing approaches, the lengths of confidence intervals given by our approach are slightly larger. This is consistent with our theory, as our approach uses ${\mbb \mu}$ to reduce the bias, which can at the same time potentially increase the variance of the resulting estimates by $\sigma_{\mu}^2$. Moreover, as also discussed in Section~\ref{sec:confidence}, $\hat{\sigma}_m^2$ is essentially an upper bound of the true variance, which can result in a slightly conservative confidence interval construction.}

\subsection{Variance estimation} \label{sec:simvar}
In this section, we present results concerning estimation of $\Var(Y(1))$ based on data generated using the same settings as considered in Section~\ref{sec:simate}, with the modifications described below. To do this, we first transform our data via $Y^{(1)} = Y T$. Then writing $(Y^{(1)}(0), Y^{(1)}(1)) := (0, Y(1))$, \rev{note that $Y^{(1)}(T) = Y^{(1)}$. Also, we have
	\[
	\E(Y(1)) = \E(Y^{(1)}(1) - Y^{(1)}(0)) =: \tau^{(1)};
	\]
	that is $\E(Y(1))$ is the `average treatment effect' when using transformed data $(Y, T) \mapsto (YT, T) = (Y^{(1)}, T)$}. The parameter $\tau^{(1)}$ may be estimated by first transforming the data as indicated, and then estimating this new average treatment effect. \rev{Let us call the corresponding DIPW estimator $\hat{\tau}^{(1)}$}.

Next consider $Y^{(2)} := (Y - \tau^{(1)})^2T$ and $(Y^{(2)}(0), Y^{(2)}(1)) := (0, (Y(1) - \tau^{(1)})^2)$. \rev{Then similarly, $Y^{(2)} = Y^{(2)}(T)$.} Also $\Var(Y(1)) = \E \{ Y^{(2)}(1) - Y^{(2)}(0)\}$, \rev{an `average treatment effect' involving transformed outcomes $Y \mapsto Y^{(2)}$; note however that $Y^{(2)}$ involves the unknown oracular quantity $\tau^{(1)}$}. Replacing $\tau^{(1)}$ with an estimate $\hat{\tau}^{(1)}$ and working with transformed outcomes $Y \mapsto (Y - \hat{\tau}^{(1)})^2T$, we can form an estimate $\hat{\tau}$ of $\Var(Y(1))$. We follow these steps using each of the methods outlined in Section~\ref{sec:simate} applied to the appropriately transformed outcomes to construct estimates of $\Var(Y(1))$.

\subsubsection{Experimental setup}
The settings we consider as the same as those studied in section~\ref{sec:simate} with the following modifications.
We replace \eqref{eq:sim_Y} by
\begin{align*}
	Y_i(1) = \beta^\top  X_i + T_i \delta^\top  X_i + \ind_{\{\pi(X_i) \geq 0.5\}} \varepsilon_i^{(+)} + \ind_{\{\pi(X_i) < 0.5\}} \varepsilon_i^{(-)},
\end{align*}
where $\epsilon_i^{(+)} \sim \mathcal{N}(0, 0.5)$ and $\epsilon_i^{(-)} \sim \mathcal{N}(0, 2)$; the data we observe is $Y_iT_i$. This heteroscedasticity increases the gap between $\Var(Y(1) \,|\, T=1)$ and the target $\Var(Y(1))$ so those methods that might tend towards the former do not artificially appear to perform well. \newrev{Further, we decrease the degree of overlap between the treatment and control groups by setting $\|\gamma\|_2=3$; results with $\|\gamma\|_2=1$ are given in the appendix.} 

Since overall the problem is more challenging that vanilla average treatment effect estimation, we increase the sample size $n$ to $2500$. Rather than directly using the real design which has $n=491$, we fit a Gaussian copula model to the data to give a multivariate distribution from which we can generate independent realisations. We still refer to this as the `real design' in the figures that follow, although it is simulated data that approximates the distribution of the data.

\newrev{As discussed in Section~\ref{sec:multi}, one can apply any method to produce estimate $\tilde{\mu}(\cdot)$ in Step 2 of Algorithm~\ref{alg:dipw}. Given that the outcome regression model here is highly nonlinear, it is natural to use a more flexible regression method to attempt to estimate this. We therefore additionally consider employing
	random forest \citep{breiman2001random} to construct $\tilde{\mu}$ using the approach in Step 2b of Algorithm~\ref{alg:dipw}. For comparison, in addition to all of the methods used in  Section~\ref{sec:ATE_results}, we consider versions of		
	AIPW and TMLE with random forest used in the estimation of outcome regression models. We use the \texttt{ranger} \citep{wright2015ranger} implementation of random forest with the default parameters throughout.}

%
%
\subsubsection{Results}
\begin{figure}[t!]
	\subfigure[Toeplitz design, $s = 5$]{\includegraphics[width=0.32\textwidth]{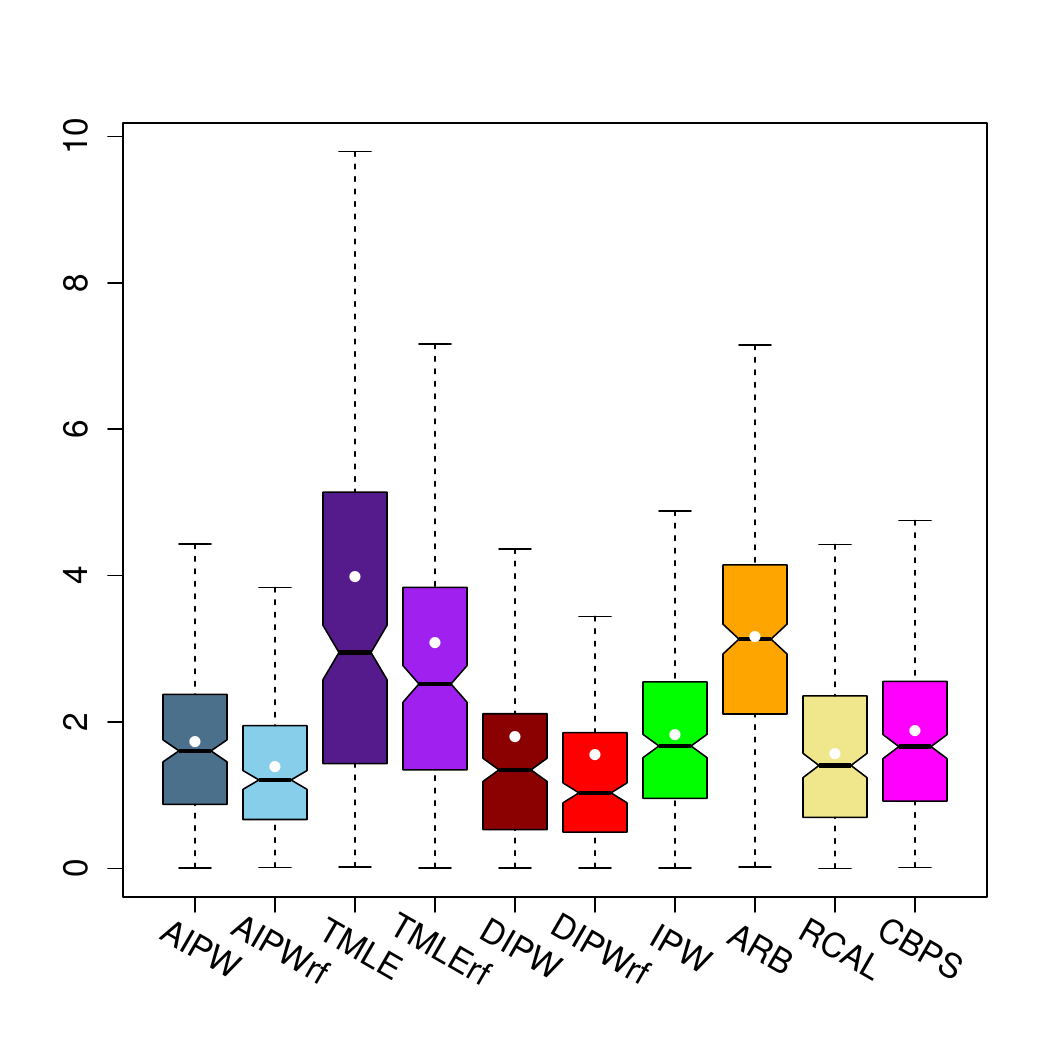}} \hfill
	\subfigure[Toeplitz design, $s = 20$]{\includegraphics[width=0.32\textwidth]{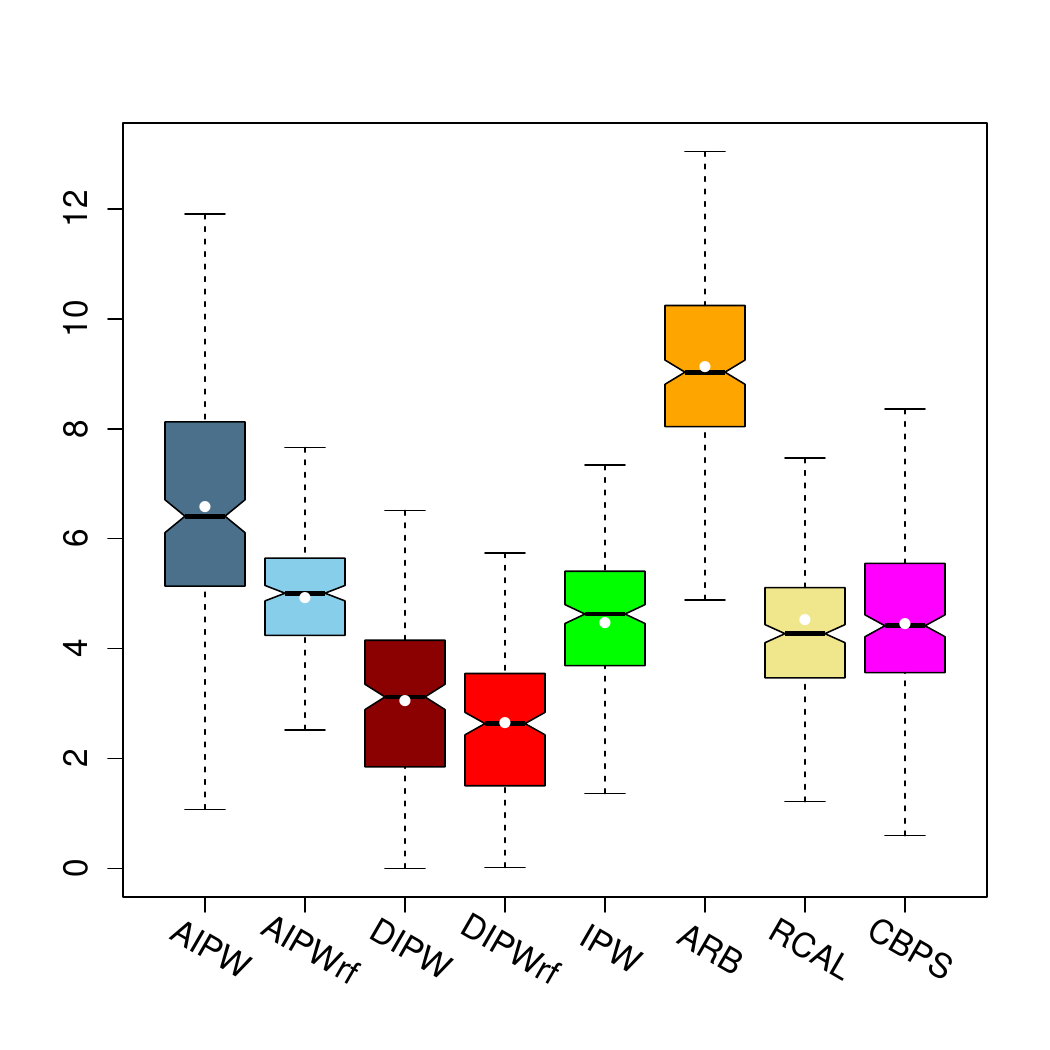}} \hfill
	\subfigure[Toeplitz design, $s = 50$]{\includegraphics[width=0.32\textwidth]{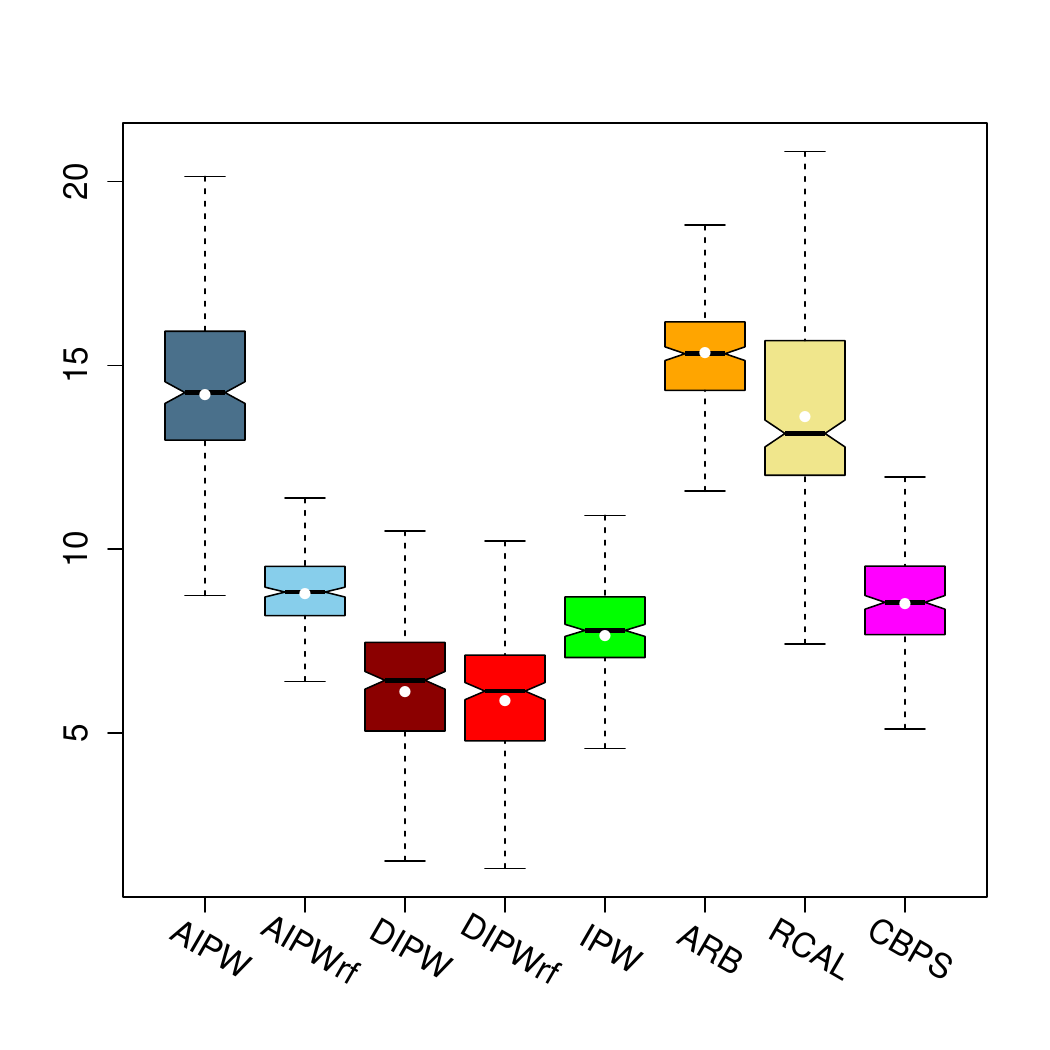}} \\
	\subfigure[Exponential design, $s = 5$]{\includegraphics[width=0.32\textwidth]{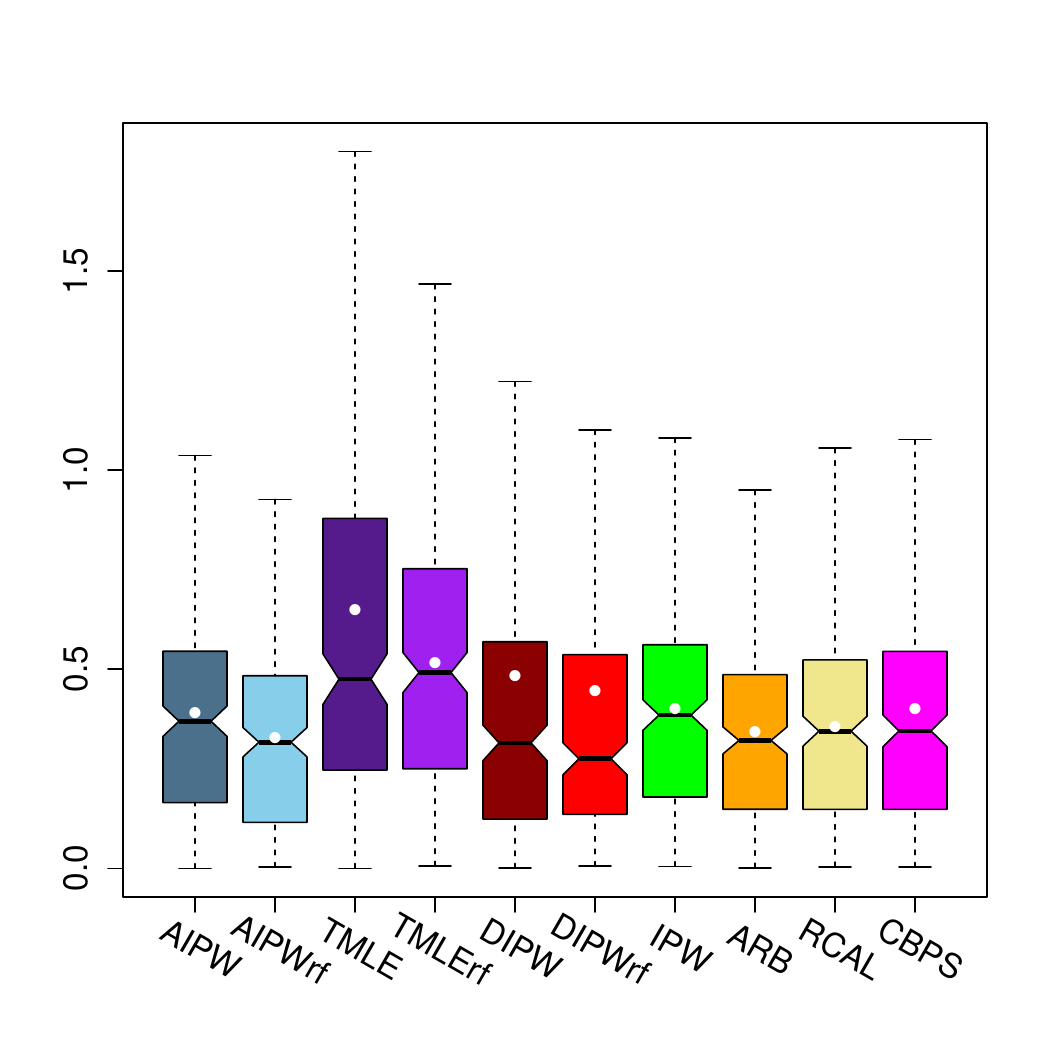}} \hfill
	\subfigure[Exponential design, $s = 20$]{\includegraphics[width=0.32\textwidth]{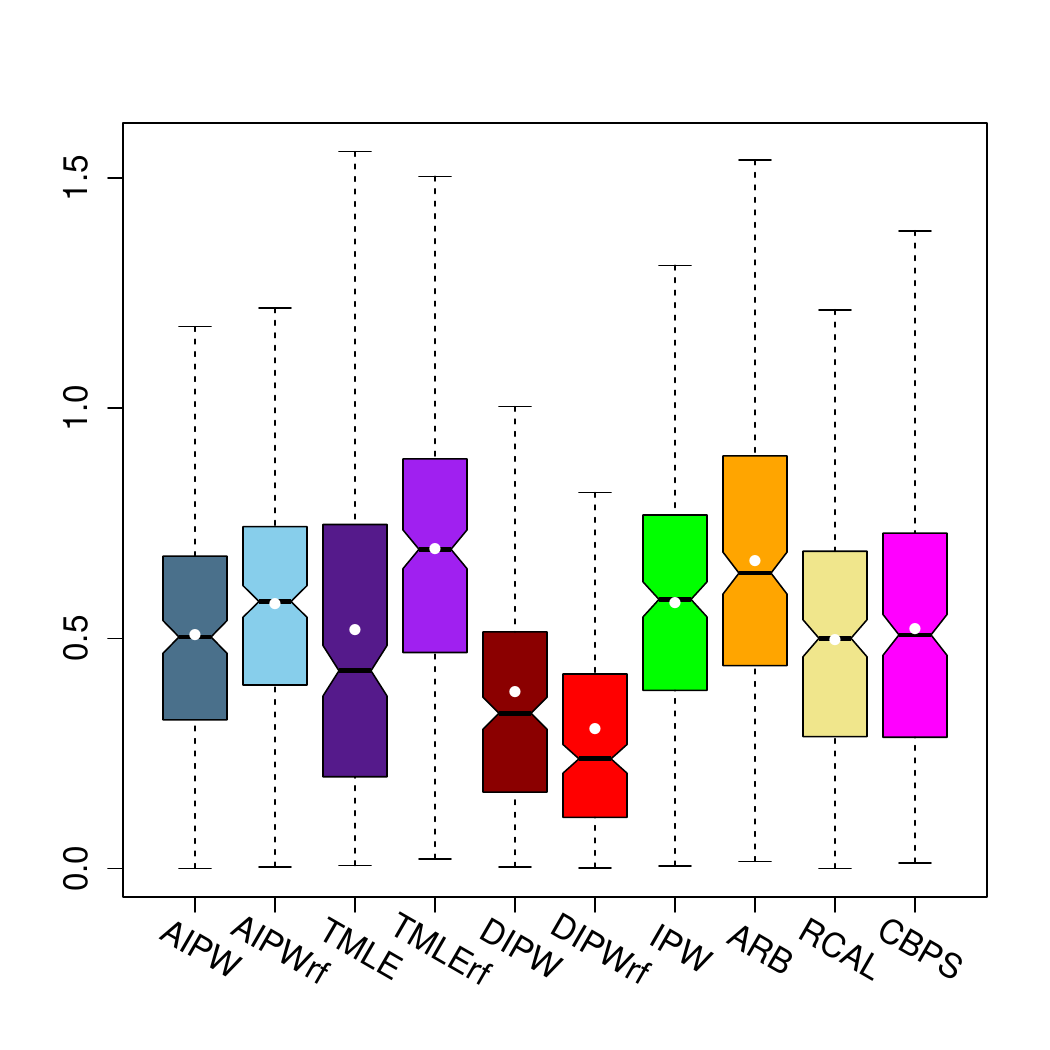}} \hfill
	\subfigure[Exponential design,  $s = 50$]{\includegraphics[width=0.32\textwidth]{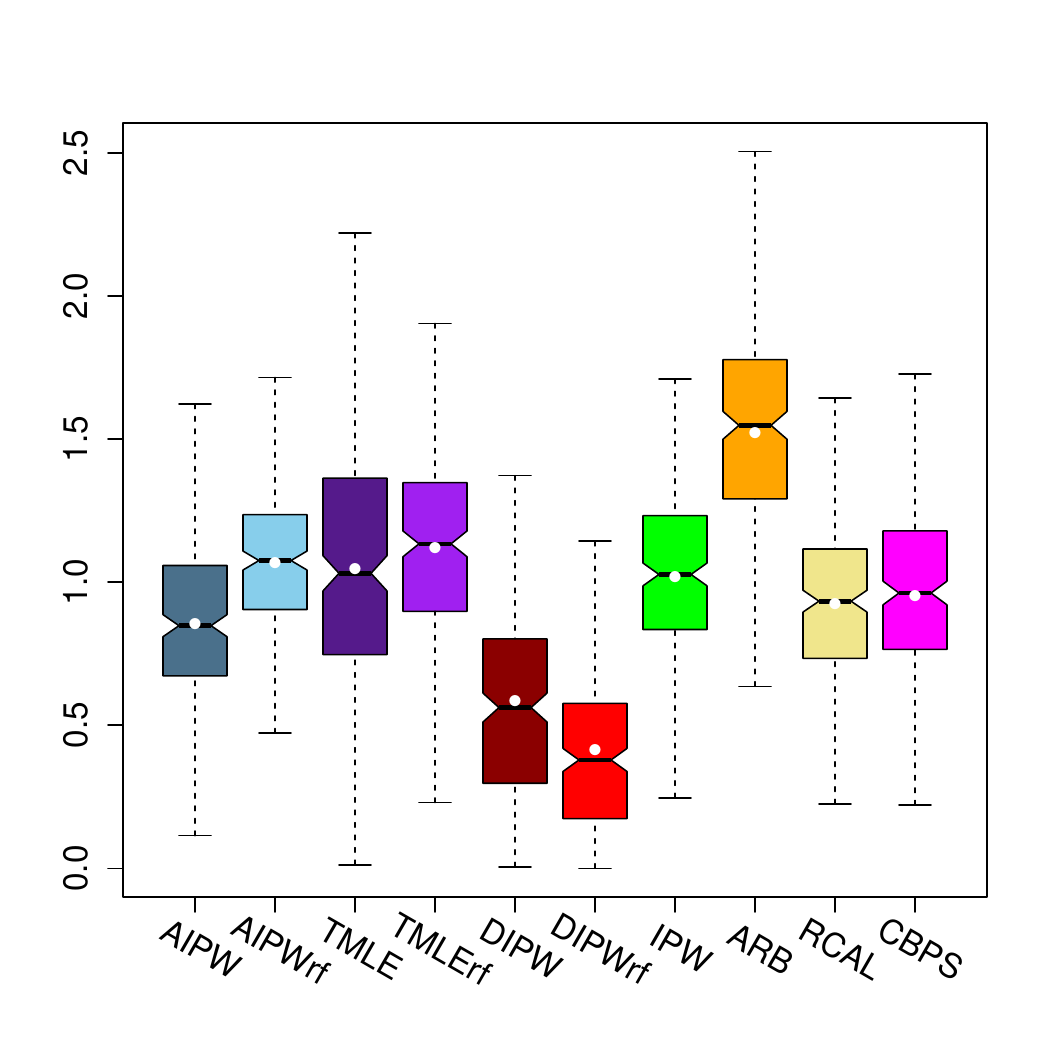}} \\
	\subfigure[Real data design, $s = 5$]{\includegraphics[width=0.32\textwidth]{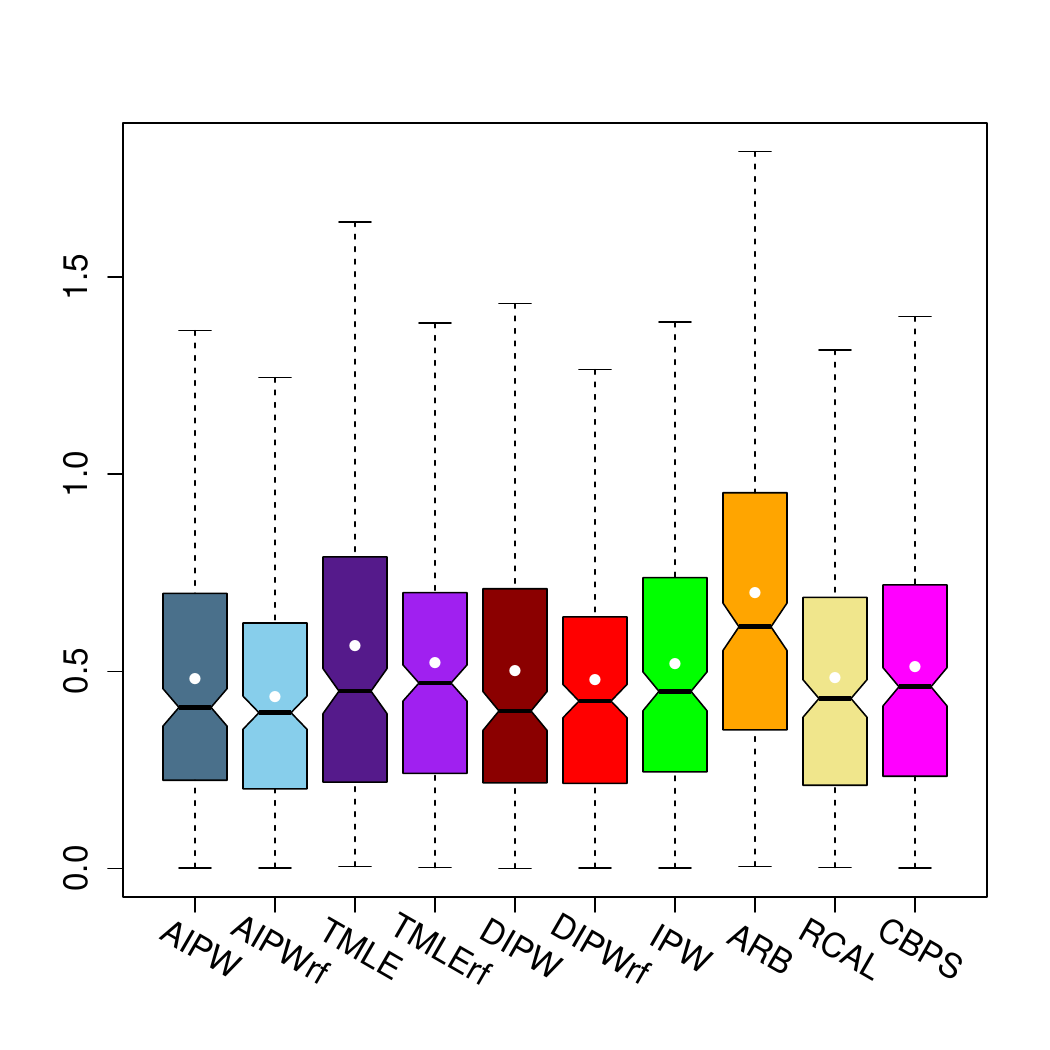}} \hfill
	\subfigure[Real data design, $s = 20$]{\includegraphics[width=0.32\textwidth]{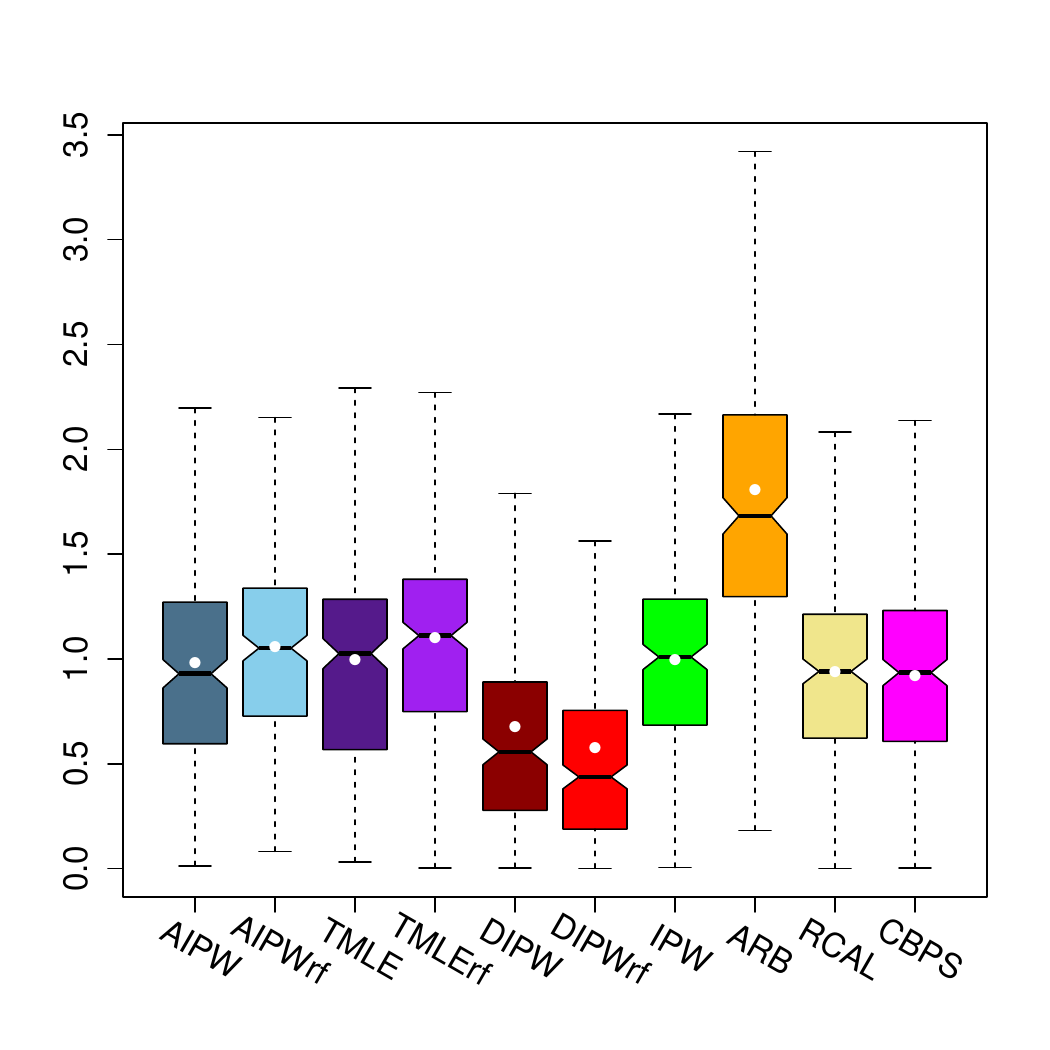}} \hfill
	\subfigure[Real data design,  $s = 50$]{\includegraphics[width=0.32\textwidth]{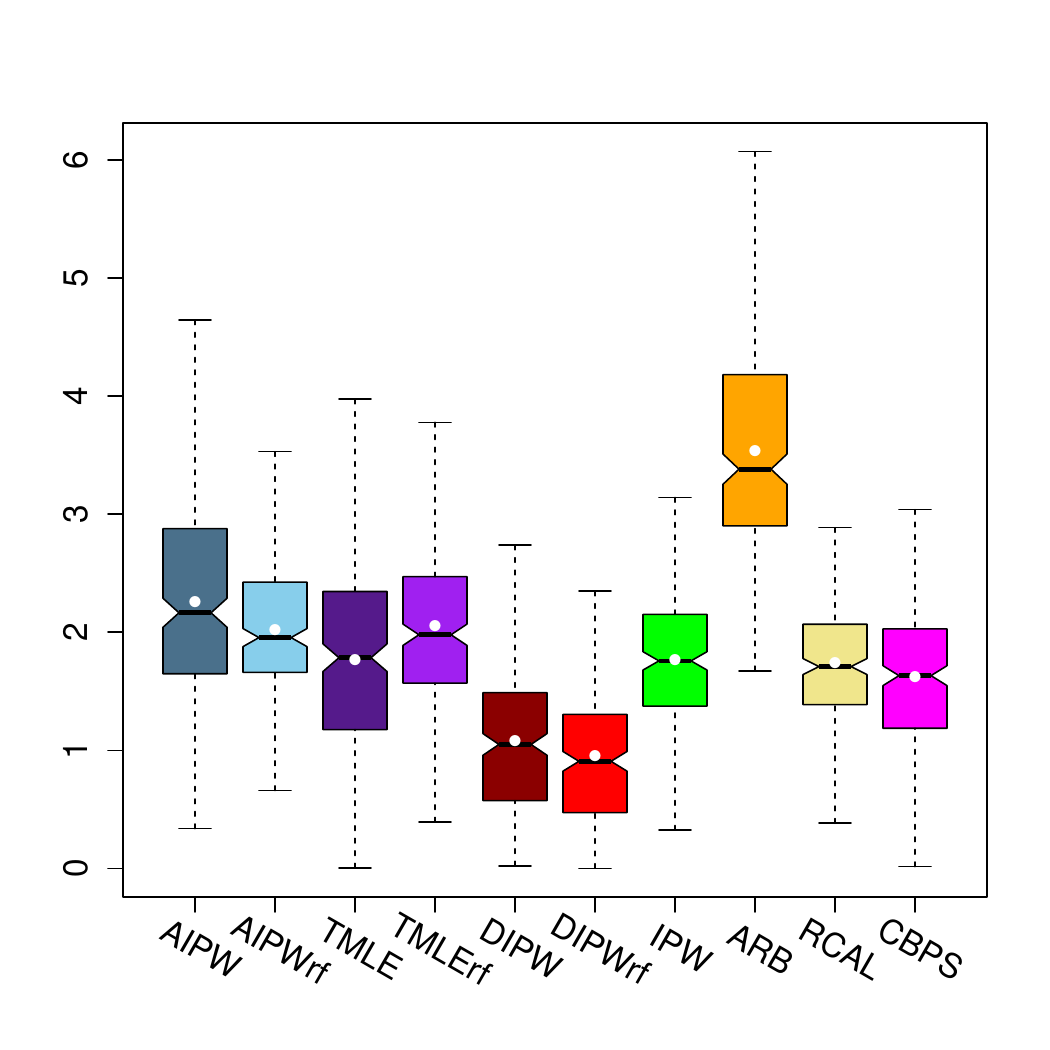}}
	\caption{Boxplots of the error in estimating $\Var(Y(1))$; the interpretation is analogous to Figure~\ref{fig:linear}. For the Toeplitz design settings with $s = 20, 50$, due to the long error bar for TMLE and TMLErf, we do not display their boxplots in the figure. The median absolute errors of TMLE and TMLErf are $6.4, 10.8$ for $s = 20$, and $81.8, 205.1$ for $s = 50$. }\label{fig:varrfall2}
\end{figure}
\newrev{Figure~\ref{fig:varrfall2} presents boxplots analogous to those in Figure~\ref{fig:linear}. We see that the DIPW methods perform well, random forect variant perhaps the best prformer overall. As with the  average treatment effect estimation experiments, the relative advantage of DIPW over other methods is however prominent in the more challenging settings where $s$ is large. Interestingly, in contrast to the results on ATE estimation presented earlier, ARB seems to struggle more here, perhaps because of its reliance on the regression model being approximately sparse, an assumption which is heavily violated in this set of examples.}

\section{Discussion} \label{sec:discuss}
In this paper we have proposed a debiased inverse propensity weighting scheme for estimating average treatment effects in high-dimensional settings. At a very basic level, our approach is similar to augmented inverse propensity score weighting, in that the latter uses regression adjustment to mitigate the bias in estimation of the propensity weights. Our method however exploits the fact that if the propensity weights are estimated based on  a well-specified sparse high-dimensional logistic regression model, we can correct for the biases without knowledge of the regression functions. Instead of relying on regression adjustment for bias correction, the regression involved in constructing $\tilde{\mu}$ has the purpose of reducing variance: if $\tilde{\mu}$ estimates $\mu_{\ora}$ sufficiently well, our estimator attains the semiparametric efficient variance bound.


We have seen how DIPW can be used to estimate expectations of other functions of the potential outcomes such as $\Var(Y(1))$ (see Section~\ref{sec:simvar}); it would be interesting to extend the approach to estimation of quantiles of the distributions of potential outcomes, for example.
In a related vein, it would be interesting to see to what extent the ideas here can be used to target other causal estimands, for example those in the broader classes considered by \citet{HW17,chernozhukov2018double}.

\section*{Acknowledgement}
The research of both authors was supported by an Engineering and Physical Sciences Research Council (EPSRC) `First' grant of the second author. Both authors thank two anonymous referees for helpful comments that improved the paper.

\bibliographystyle{plainnat}

\bibliography{reference}

\appendix
\section{Derivations and results relating to Sections~\ref{sec:intro} and \ref{sec:basic}}
\subsection{Relationship with AIPW} \label{sec:AIPW}
Consider the estimator
\[
\hat{\tau} = \frac{1}{n}\sum_{i=1}^n \frac{T_i (Y_i-\mu_i)}{\hat{\pi}(X_i)} - \frac{1}{n} \sum_{i=1}^n\frac{(1 - T_i) (Y_i-\mu_i)}{1 - \hat{\pi}(X_i)}.
\]
Taking $\mu_i = \{1-\hat{\pi}(X_i)\} \hat{r}_1(X_i) +\hat{\pi}(X_i)\hat{r}_0(X_i)$ gives
\begin{align*}
	\hat{\tau} &= \frac{1}{n}\sum_{i=1}^n \frac{T_i [Y_i-\hat{r}_1(X_i) + \hat{\pi}(X_i) \{\hat{r}_1(X_i) -  \hat{r}_0(X_i)\} ]}{\hat{\pi}(X_i)} \\
	& \qquad
	- \frac{1}{n} \sum_{i=1}^n\frac{(1 - T_i) [Y_i- \hat{r}_0(X_i) - \{1-\hat{\pi}(X_i)\} \{\hat{r}_1(X_i) -  \hat{r}_0(X_i)\} ]}{1 - \hat{\pi}(X_i)} \\
	&= \frac{1}{n}\sum_{i=1}^n  \frac{T_i \{Y_i-\hat{r}_1(X_i)\}}{\hat{\pi}(X_i)} -  \frac{1}{n} \sum_{i=1}^n  \frac{(1-T_i) \{Y_i-\hat{r}_0(X_i)\}}{1-\hat{\pi}(X_i)} \\
	& \qquad + \frac{1}{n}\sum_{i=1}^n \{\hat{r}_1(X_i) - \hat{r}_0(X_i)\},
\end{align*}
which is of the familiar form of an AIPW estimator.

\subsection{Derivation of \eqref{eq:mu_ora}} \label{sec:mu_ora}
Now $\E (T_i Y_i \,|\, X_i) = r_1(X_i)\E(T_i \,|\, X_i) = r_1(X_i) \pi_i$, and similarly $\E((1-T_i) Y_i \,|\, X_i) = r_0(X_i)(1-\pi_i)$. Thus
\begin{align*}
	\E (\tilde{Y}_{\ora,i}|X_i) &= \E\left(\frac{T_iY_i(1-\pi_i)}{\pi_i} + \frac{(1-T_i)Y_i \pi_i}{1 - \pi_i} \,|\, X_i\right) \\
	&= \{1-\pi_i\} r_1(X_i) +\pi_ir_0(X_i)
\end{align*}
as required.

\subsection{Proof Lemma~\ref{lem:var}} \label{sec:proof_lem_var}
First note that as $X_j \independent (X_i, Y_i, T_i)$ for $j \neq i$, we may take $\mu_i$ as a function of $X_i$ only. Indeed, we can always replace $\mu_i$ with $\E(\mu_i \,|\,X_i)$ without decreasing $\Var(\hat{\tau}_{\ora,i})$. Next observe that any minimiser of $ \Var(\hat{\tau}_{\ora, i})$ must minimise $\E \hat{\tau}_{\ora, i}^2$ as $\E \hat{\tau}_{\ora, i} = \tau$ for all functions $\mu_i(X_i)$.
Then from Lemma~\ref{lem:cond_exp} below we know that
\begin{equation} \label{eq:mu_i}
	\mu_i(X_i) = \frac{\E \left(\frac{(T_i - \pi_i)^2}{\pi_i^2(1-\pi_i)^2} Y_i \Big| X_i \right)}{\E \left(\frac{(T_i - \pi_i)^2}{\pi_i^2(1-\pi_i)^2} \Big| X_i \right)} = \frac{\E\{(T_i-\pi_i)^2 Y_i | X_i\}}{ \E\{(T_i-\pi_i)^2| X_i\} }
\end{equation}
minimises $\E \hat{\tau}_{\ora, i}^2$. Now
\[
\E\{(T_i-\pi_i)^2| X_i\} = \pi_i - 2\pi^2 + \pi_i^2 = \pi_i(1-\pi_i),
\]
and
\begin{align*}
	\E\{(T_i-\pi_i)^2 Y_i | X_i\} &= (1 - 2 \pi_i)\E( T_iY_i |X_i) + \pi_i^2 \E(Y_i | X_i) \\
	&= (1- \pi_i)^2\E( T_iY_i |X_i) + \pi_i^2  \E\{(1-T_i)Y_i | X_i\}.
\end{align*}
Thus \eqref{eq:mu_i} equals $\E (\tilde{Y}_{i, \ora} | X_i)$.
\begin{lemma} \label{lem:cond_exp}
	Let $V \in \R$ be a square integrable random variable and suppose random variable $W$ has support in $[a, b]$ where $a > 0$. Suppose $U \in \R^p$ is a random vector. Then $h^*(U) := \E( W^2 V|U) / \E(W^2 | U)$ minimises
	\[
	\E [W\{V - h(U)\}]^2
	\]
	over all measurable functions $h: \R^p \to \R$.
\end{lemma}
\begin{proof}
	We have
	\begin{align}
		\E [W\{V - h(U)\}]^2 &= \E [W\{V -h^*(U) + h^*(U) - h(U)\}]^2 \notag \\
		&\geq \E [W\{V - h^*(U)\}]^2 +2\E W^2\{V -h^*(U)\}\{h^*(U) - h(U)\}. \label{eq:cond_exp2}
	\end{align}
	Now
	\begin{align*}
		\E [W^2\{V -h^*(U)\}|U] &= \E (W^2 V | U) - h^*(U) \E(W^2 | U) = 0
	\end{align*}
	almost surely, so the second term is \eqref{eq:cond_exp2} is zero, proving the claim.
\end{proof}

%

\section{On the event~$\Omega$}\label{sec:omega}
\newrev{
	Here we present an explicit set of sufficient conditions under which properties (i) and (ii) involved in the event $\Omega$ occur with high probability. We consider the case where the regressions to form estimates $\hat{\pi}_i$ and $\hat{\pi}_{Ai}$ are $\ell_1$-penalised logistic regressions: specifically
	\begin{align} \label{eq:l1_logistic2}
		\hat{\gamma} &:= \argmin_{b \in \R^p} \left(\frac{1}{n_B}\sum_{i=1}^{n_B} [\log\{1+ \exp(X_{Bi}^{\top}b)\} - T_{Bi}X_{Bi}^{\top} b] + \lambda_\gamma \|b\|_1\right), \\
		\hat{\pi}_i &:= \frac{1}{1 + \exp(-X_i^{\top}\hat{\gamma})}, \notag
	\end{align}
	and similarly for $\hat{\pi}_{Ai}$. 
	\begin{lemma} \label{lem:Omega_large}
		Suppose Assumptions~\ref{as:logistic} and \ref{as:subgX} are satisfied, and additionally that there exists a constant $\alpha >0 $ such that 
		the minimimum eigenvalue of the covariance matrix $\Var(Z)$ is at least $\alpha$. Suppose further that for some sequence $d_n \to 0$, 
		$s = d_n n/(\log p \log n)$, and that there exists a constant $C > 0$ such that $\|\gamma\|_2 \leq C$.
		For simplicity, suppose $n=n_A = n_B$. Then for all $m>0$, there exist constants $c_1, c_\gamma, c_{\hat{\pi}} > 0$  such that with probability at least $1-c_1n^{-m}$, there exists tuning parameter $\lambda_{\gamma} >0$ for which the following hold:
		\begin{enumerate}[(i)]
			\item $\|\hat{\gamma} - \gamma\|_1\leq c_\gamma s \sqrt{\log (p) /n }$ and $\|\hat{\gamma} - \gamma\|_2 \leq c_\gamma \sqrt{s \log (p) / n}$;
			\item $c_{\hat{\pi}} \leq \hat{\pi}_i \leq 1 - c_{\hat{\pi}}$ and $c_{\hat{\pi}} \leq \hat{\pi}_{Ai} \leq 1 - c_{\hat{\pi}}$ for all $i=1,\ldots,n$.
		\end{enumerate}
	\end{lemma}
	\begin{proof}
		In the following, $c_{1j}$ and $c_{2j}$ for $j=1,2,\ldots$ are positive constants.
		Lemma~\ref{lem:RSC} below shows that with probability at least $1 - c_{11}e^{-c_{21}n}$, a restricted strong convexity condition holds. Then by \citet[Cor.~9.26]{wainwright_2019}, we have that with probability at least $1 - c_{12}e^{-c_{22}n}$, (i) holds for some $c_{\gamma} > 0$; let us work on this event. Finally, note that $\hat{\gamma}$ and $X_i$'s are independent, so that by conditioning on $\hat{\gamma}$, $X_i^\top (\hat{\gamma} - \gamma)$'s are independent sub-Gaussian random variables with mean $\hat{\gamma}_1 - \gamma_1$ and variance proxy bounded above by $\sigma_Z^2 \|\hat{\gamma}_{-1} - \gamma_{-1}\|_2^2$. Then we have from a union bound that (recalling (i) holds), 
		\[
		\pr\left(\max_i |X_i^\top (\gamma - \hat{\gamma})| \ge  \sqrt{ (m+1)c_\gamma s \log(p) \log (n) / n} \mid \hat{\gamma}\right) \le 2 n^{-m}.
		\]
		Thus, for $n$ sufficiently large and Som constant $c_1 > 0$, we have with probability at least $1 - c_1n^{-m}$, $c_\pi /2 <\hat{\pi}_i < 1- c_\pi / 2$ for all $i$. Arguing similarly on the auxiliary dataset, we obtain the result.
	\end{proof}
}
\newrev{
	The following result is is very similar to \citet[Thm.~9.36]{wainwright_2019}, but with the difference that the first component of the vector of predictors $X_i$ is always $1$ in our case.
	\begin{lemma} \label{lem:RSC}
		Suppose Assumptions~\ref{as:logistic} and \ref{as:subgX} are satisfied, there exists a constant $\alpha >0 $ such that 
		the minimimum eigenvalue of the covariance matrix $\Var(Z)$ is at least $\alpha$ and  $\|\gamma\|_2 \leq C$ for some constant $C>0$. Let
		\[
		\mathcal{L}(b) := \frac{1}{n}\sum_{i=1}^{n}  \left( \log\{1 + \exp(X_{i}^{\top} b)\} - Y_{i}X_{i}^{\top} b\right).
		\]
		We have that for some constants $\kappa,c, c_1, c_2 > 0$,
		the restricted strong convexity condition
		\begin{align*}
			\mathcal{E}_n(\Delta) &:= \mathcal{L}(\gamma + \Delta) - \mathcal{L}(\gamma) - \Delta^{\top}\nabla \mathcal{L}(\gamma)  \\
			&\geq \frac{\kappa}{2}\|\Delta\|_2^2 - c \frac{\log p}{n}\|\Delta\|_1^2 \quad \text{for all } \|\Delta\|_2 \leq 1
		\end{align*}
		holds with probability at least $1 - c_{1}e^{-c_{2}n}$.
	\end{lemma}
	\begin{proof}
		Let us write $\Delta := (\iota, \delta) \in \R \times \R^{p-1}$ and $\theta\in \R^{p-1}$ for a version of $\gamma$ with its first component removed.
		
		Now, following the proof of \citet[Thm.~9.36]{wainwright_2019}, we have, by a Taylor expansion, that
		\[
		\mathcal{E}_n(\Delta) = \frac{1}{n}\sum_{i=1}^n \psi''\left(\gamma^{\top}X_i + t\Delta^{\top}X_i\right) (\Delta^{\top}X_i)^2,
		\]
		for some random $t \in [0,1]$, where $\psi(\eta) = \log\{1 + e^{\eta}\}$. Now let $\tau = K \|\delta\|_2$ for a constant $K>0$ to be chosen, and let $T > 0$. Then 
		\begin{align*}
			\mathcal{E}_n(\Delta)  &= \frac{1}{n}\sum_{i=1}^n\psi''\left(\gamma_1 + \theta^{\top}Z_i + t(\iota + \delta^{\top}Z_i)  \right) (\Delta^{\top}X_i)^2 \\
			&\geq \inf_{|u| \leq 1 + T + 2K }\psi''(\gamma_1 + u) \frac{1}{n}\sum_{i=1}^n \varphi_{\tau}(\Delta^{\top}X_i) \ind_{\{|\theta^{\top}Z_i| \leq T \}},
		\end{align*}
		where $\varphi_{\tau} (u) := u^2\ind_{[-2\tau, 2\tau]}(u)$. Note that the infimum above is positive for all $K, T$.
		
		As noted in the proof of  \citet[Thm.~9.36]{wainwright_2019}, it suffices to consider the case where $\|\Delta\|_2 = 1$, and so $\tau = K$. Following the argument therein, writing
		\[
		\widetilde{\varphi}_{K}(u) = u^2 \ind_{[-K,K]}(u) + (u-2K)^2\ind_{[K,2K]}(u) + (u + 2K)^2\ind_{[-2K,-K]}(u),
		\]
		(which we note is Lipschitz with parameter $2K$) it suffices to show that
		\begin{equation} \label{eq:expec_bd}
			\E[ \widetilde{\varphi}_{K}(\Delta^{\top}X) \ind_{[-T,T]}(\theta^{\top}Z)] \geq \frac{3}{4}\alpha,
		\end{equation}
		and for i.i.d.\ Rademacher random variables $\varepsilon_1, \ldots, \varepsilon_n$, independent of 
		$(Z_1,\ldots,Z_n)$,
		\begin{equation} \label{eq:symm_bd}
			\E \left( \sup_{\|\Delta\|_1 \leq r} \frac{1}{n}\sum_{i=1}^n \varepsilon_i \widetilde{\varphi}_K(\Delta^{\top}X_i) \ind_{[-T,T]}(\theta^{\top}Z_i) \right) \leq 4Kr (\sigma_Z \vee 1) \sqrt{\frac{\log p}{n}}. 
		\end{equation}
		Consider \eqref{eq:expec_bd}. To show this, as noted in the proof of  \citet[Thm.~9.36]{wainwright_2019}, it suffices to show
		\begin{equation} \label{eq:expec_bd1}
			\E[\widetilde{\varphi}_K(\Delta^{\top}X)] \geq \frac{7}{8}\alpha 
		\end{equation}
		and
		\begin{equation} \label{eq:expec_bd2}
			\E[\widetilde{\varphi}_K(\Delta^{\top}X) \ind_{|\theta^\top Z| > T}) ] \leq \frac{1}{8}\alpha.
		\end{equation}
		Turning to \eqref{eq:expec_bd1},
		\begin{align*}
			\E[\widetilde{\varphi}_K(\Delta^{\top}X)] &\geq \E[(\Delta^{\top}X)^2] - \E [(\Delta^{\top}X)^2 \ind_{\{|\delta^{\top}Z + \iota| > K\}}]\\
			&\geq \E[(\Delta^{\top}X - \E \Delta^{\top}X)^2] - \E [(\iota + \delta^{\top}Z)^2 \ind_{\{|\delta^{\top}Z + \iota| > K\}}].
		\end{align*}
		The first term above is at least $\alpha$ by assumption. Now as $Z$ is sub-Gaussian, there exists a constant $\beta>0$ such that $\sup_{v : \|v\|_2=1} \E[ (v_1 + v_{-1}^{\top}Z)^4] < \beta$.
		But then by the Cauchy--Schwarz inequality and then Markov's inequality, we have
		\begin{align*}
			\E [(\iota + \delta^{\top}Z)^2 \ind_{\{|\delta^{\top}Z + \iota| > K\}}] \leq \sqrt{\E[(\iota + \delta^{\top}Z)^4] }\sqrt{\pr\{|\delta^{\top}Z + \iota| > K\}} \leq \frac{\beta}{K^2},
		\end{align*}
		so setting $K^2 = 8\beta/\alpha$ gives \eqref{eq:expec_bd1}. Next, by the Cauchy--Schwarz inequality and Markov's inequality,
		\[
		\{\E[\widetilde{\varphi}_K(\Delta^{\top}X) \ind_{[-T,T]}(\theta^\top Z) ]\}^2  \leq \E[(\Delta^{\top}X)^4] \pr (|\theta^\top Z| > T) \leq \beta \frac{ C^4 \beta}{T^4}
		\]
		Thus taking $T^2 = 8 \beta C^2 / \alpha$ gives \eqref{eq:expec_bd2}.
		
		Turning to \eqref{eq:symm_bd}, we have by the Ledoux--Talagrand contraction principle that
		\begin{align*}
			& \E \left( \sup_{\|\Delta\|_1 \leq r} \frac{1}{n}\sum_{i=1}^n \varepsilon_i \widetilde{\varphi}_K(\Delta^{\top}X_i) \ind_{[-T,T]}(\theta^{\top}Z_i) \right) \\
			\leq & 2K \E \left( \sup_{\|\Delta\|_1 \leq r} \Delta^{\top} \frac{1}{n}\sum_{i=1}^n \varepsilon_iX_i \ind_{[-T,T]}(\theta^{\top}Z_i) \right) \leq 2Kr \E \left( \norm{\frac{1}{n}\sum_{i=1}^n \varepsilon_iX_i \ind_{[-T,T]}(\theta^{\top}Z_i)}_{\infty}\right),
		\end{align*}
		using H\"older's inequality 
		in the final line. Since each coordinate of the random vector
		\[
		\frac{1}{n}\sum_{i=1}^n \varepsilon_iX_i \ind_{[-T,T]}(\theta^{\top}Z_i)
		\]
		is a sub-Gaussian random variable with variancy proxy bounded above by $\frac{\sigma_Z^2 \vee 1}{n}$, we have from~\citet[Exercise~2.12(b)]{wainwright_2019} (see also~\citet[Theorem~1.14]{rigollet2023high}) that
		\[
		\E \left( \norm{\frac{1}{n}\sum_{i=1}^n \varepsilon_iX_i \ind_{[-T,T]}(\theta^{\top}Z_i)}_{\infty}\right) \le 2 (\sigma_Z \vee 1) \sqrt{\frac{\log p}{n}}.
		\]
	\end{proof}
}

\section{Proof of Theorem~\ref{thm:dipw}}\label{sec:thmdipw}

\newrev{Throughout this section, we prove a stronger version of Theorem~\ref{thm:dipw}, where we make the dependencies of the constants in Theorem~\ref{thm:dipw} on $(c_{\hat{\pi}}, c_\pi)$ explicit:
	\begin{theorem} \label{thm:dipwpf}
		Let $\hat{\tau}_{\dipw}$ be the estimator \eqref{eq:dipw} where $\hat{\mu} \in \R^n$ is constructed via \eqref{eq:mu_hat} with tuning parameter $\eta = c_\eta c_{\hat{\pi}}^{-1} \sqrt{\log(p) / n}$ and constant $c_\eta>0$ is sufficiently large. Suppose Assumptions \ref{as:unconfounded}--\ref{as:p_large} hold.
		Then we have the decomposition
		\[
		\sqrt{n}(\hat{\tau}_{\dipw} - \bar{\tau}) = \delta + \sqrt{ \sigma_{\mu}^2 + \bar{\sigma}^2} \zeta
		\]
		in which given constants $c_\gamma, c_{\tilde{\mu}} > 0$, $c_{\hat{\pi}} \in (0, \frac{1}{2}]$ and $m \in \mathbb{N}$, we have that $\delta$ and $\zeta$ satisfy the following properties:
		\begin{enumerate}[(i)]
			\item there exist constants $c_\delta, c > 0$ that do not depend on $(c_{\hat{\pi}}, c_\pi)$,
			such that
			\begin{equation*}
				\pr\left(|\delta|  > c_{\hat{\pi}}^{-2} c_\pi^{-1} c_\delta (s + \sqrt{s \log n}) \frac{\log p}{\sqrt{n}}\right) \leq \pr(\Omega^c(c_\gamma,c_{\tilde{\mu}},c_{\hat{\pi}})) + c(p^{-m} + n^{-m});
			\end{equation*}
			\item there exists a universal constant $c_\zeta >0$ such that
			\begin{align*}
				\sup_{t \in \R} |\pr(\zeta \leq t | \mathcal{D}) - \Phi(t)| \leq \frac{c_\pi^{-2} c_\zeta}{\sqrt{n}} \frac{\sigma_\mu^2\|\hat{\mbb\mu} - \mbb\mu_{\ora}\|_\infty + \bar{\rho}^3}{(\sigma_{\mu}^2 + \bar{\sigma}^2)^{3/2}}.
			\end{align*}
		\end{enumerate}
	\end{theorem}
}
We begin by deriving the decomposition \eqref{eq:hat_tau_decomp}. We consider here a version of $\tau_{\ora}$ with $\mu_i = \hat{\mu}_i$.

Observe that for all $i=1,\ldots,n$ we have
\begin{align*}
	\frac{T_i(Y_i - \hat{\mu}_i)}{\hat{\pi}_i} - \frac{T_i(Y_i - \hat{\mu}_i)}{\pi_i} &= - \frac{T_i(Y_i - \hat{\mu}_i)}{\pi_i\hat{\pi}_i} (\hat{\pi}_i - \pi_i)  \\
	\frac{(1 - T_i)(Y_i - \hat{\mu}_i)}{1 - \hat{\pi}_i} - \frac{(1 - T_i)(Y_i - \hat{\mu}_i)}{1 - \pi_i} &= \frac{(1 - T_i)(Y_i - \hat{\mu}_i)}{(1 - \pi_i)(1 - \hat{\pi}_i)} (\hat{\pi}_i - \pi_i).
\end{align*}
Thus
\begin{align*}
	\hat{\tau}_{\dipw} - \tau_{\ora} &= \frac{1}{n} \sum_{i=1}^n \bigg\{\bigg(\frac{T_i(Y_i - \hat{\mu}_i)}{\hat{\pi}_i} - \frac{(1-T_i)(Y_i - \hat{\mu}_i)}{1-\hat{\pi}_i}\bigg) - \bigg(\frac{T_i(Y_i - \hat{\mu}_i)}{\pi_i} - \frac{(1-T_i)(Y_i- \hat{\mu}_i)}{1-\pi_i}\bigg)\bigg\} \\
	&=- \frac{1}{n} \sum_{i=1}^n \bigg(\frac{T_i(Y_i - \hat{\mu}_i)}{\hat{\pi}_i \pi_i} + \frac{(1-T_i)(Y_i - \hat{\mu}_i)}{(1-\hat{\pi}_i) (1-\pi_i)}\bigg) (\hat{\pi}_i - \pi_i) \\
	&= -\frac{1}{n}\sum_{i=1}^n \left(\frac{T_iY_i}{\hat{\pi}_i\pi_i} + \frac{(1-T_i)Y_i}{(1-\hat{\pi}_i)(1-\pi_i)} - \frac{\hat{\mu}_i}{\hat{\pi}_i(1-\hat{\pi}_i)}\right)(\hat{\pi}_i - \pi_i) \notag \\
	&\qquad +  \frac{1}{n} \sum_{i=1}^n (T_i - \pi_i)\left(\frac{\hat{\mu}_i}{\hat{\pi}_i \pi_i} -
	\frac{\hat{\mu}_i}{(1-\hat{\pi}_i) (1-\pi_i)}\right)(\hat{\pi}_i - \pi_i)\\
	&=: Q_A + Q_B.
\end{align*}
Now recall that $Y_i = T_i(r_1(X_i) + \varepsilon_i(1)) + (1 - T_i)(r_0(X_i) + \varepsilon_i(0))$. Writing $r_t(X_i) = r_{t,i}$ and $\varepsilon_i(t) = \varepsilon_{t,i}$, $t=0,1$, for notational simplicity, we have
\begin{align*}
	Q_C &:= \tau_{\ora} -\bar{\tau} \\
	&= \frac{1}{n} \sum_{i=1}^n  \bigg( \frac{T_i(r_{1,i} + \varepsilon_{1,i} - \hat{\mu}_i)}{\pi_i} - \frac{(1-T_i)(r_{0, i} + \varepsilon_{0, i} - \hat{\mu}_i)}{1 - \pi_i} -(r_{1,i} -r_{0,i}) \bigg)\\
	&= \frac{1}{n} \sum_{i=1}^n \bigg(\frac{r_{1,i}}{\pi_i} + \frac{r_{0,i}}{1 - \pi_i} - \frac{\hat{\mu}_i}{\pi_i (1 - \pi_i)}\bigg) (T_i - \pi_i) + \frac{1}{n} \sum_{i=1}^n \bigg(\frac{T_i \varepsilon_{1,i}}{\pi_i} - \frac{(1 - T_i) \varepsilon_{0,i}}{1 - \pi_i}\bigg).
\end{align*}
Putting things together, we have the decomposition
\begin{align} \label{eq:decomp}
	\hat{\tau}_{\dipw} & - \bar{\tau} \nonumber\\
	= & \underset{Q_A}{\underbrace{- \frac{1}{n} \sum_{i=1}^n \bigg(\frac{T_i Y_i}{\hat{\pi}_i \pi_i} + \frac{(1-T_i)Y_i}{(1-\hat{\pi}_i) (1-\pi_i)} - \frac{\hat{\mu}_i}{\hat{\pi}_i(1 - \hat{\pi}_i)}\bigg) (\hat{\pi}_i - \pi_i)}} \nonumber\\
	& +\underset{Q_B}{\underbrace{\frac{1}{n} \sum_{i=1}^n (T_i - \pi_i)\left(\frac{\hat{\mu}_i}{\hat{\pi}_i \pi_i} +
			\frac{\hat{\mu}_i}{(1-\hat{\pi}_i) (1-\pi_i)}\right)(\hat{\pi}_i - \pi_i)}} \nonumber\\
	& +\underset{Q_C}{\underbrace{\frac{1}{n} \sum_{i=1}^n \bigg(\frac{r_{1,i}}{\pi_i} + \frac{r_{0,i}}{1 - \pi_i} - \frac{\hat{\mu}_i}{\pi_i (1 - \pi_i)}\bigg) (T_i - \pi_i) + \frac{1}{n} \sum_{i=1}^n \bigg(\frac{T_i \varepsilon_{1,i}}{\pi_i} - \frac{(1 - T_i) \varepsilon_{0,i}}{1 - \pi_i}\bigg)}}.
\end{align}

Armed with this decomposition, we approach the proof by assigning $\delta = \sqrt{n}(Q_A + Q_B)$
and $\zeta = \sqrt{n}Q_C / \sqrt{\sigma_\mu^2 + \bar{\sigma}^2}$.
The rest of the proof is organised as follows. In Lemma~\ref{lem:qa} and Lemma~\ref{lem:qb}, we show that $\delta$ is of order $(s + \sqrt{s \log n}) \log p / \sqrt{n}$ on a high probability event. This event is characterised by the intersection of four events $\T_1, \ldots,\T_4$ defined below, of which the study is deferred to Section~\ref{sec:events}. In Lemma~\ref{lem:qc}, we prove the Berry-Esseen type bound for $\zeta$.

Define
\begin{align*}
	\T_1 &:= \mathcal{T}_1(c_\gamma, c_{\hat{\pi}}) \\
	&:= \Big\{\| \hat{\gamma} - \gamma \|_1 \leq c_\gamma s \sqrt{\log p / n} \;,\; \| \hat{\gamma} - \gamma \|_2 \leq c_\gamma \sqrt{s \log p / n} \;\text{ and }\; c_{\hat{\pi}} < \min_i \hat{\pi}_i \leq \max_i \hat{\pi}_i < 1 - c_{\hat{\pi}} \Big\}.
\end{align*}
Also define for $c_1, c_2, c_{\hat{\mu}}, c_Y > 0$ the following events:
\begin{align*}
	\T_2 &:= \T_2(c_1) :=\bigg\{\bigg(\frac{1}{n} \sum_{i=1}^n \{X_i^\top (\hat{\gamma} - \gamma)\}^4\bigg)^{1/2} \leq c_1 s \frac{\log p}{n}\bigg\} \\
	\T_3 &:= \T_3(c_{\hat{\mu}}, c_2, \newrev{c_{\hat{\pi}}}) :=\bigg\{ \; \frac{1}{n} \| \hat{\mbb\mu} \|_2^2 \leq \newrev{c_{\hat{\pi}}^{-2}} c_{\hat{\mu}}, \;\text{  and }\; \| \mb X^\top \tilde{\mb Y} - \mb X^\top  \hat{\mbb\mu}\|_\infty / n \leq \newrev{c_{\hat{\pi}}^{-1}} c_2 \sqrt{\frac{\log p}{n}}\bigg\}, \\
	\T_4 &:= \T_4(c_Y) := \Big\{\frac{1}{n} \|{\mb Y}\|_2^2 \leq c_Y\Big\}.
\end{align*}
We will at times suppress the dependence of the events above on the constants $c_1,c_2, c_\gamma, ,c_{\hat{\pi}}, c_{\hat{\mu}}, c_Y$ in order to ease notation.
In what follows, we will denote by $c$ and $c'$ positive constants that can vary from line to line. We also assume throughout the conditions of Theorem~\ref{thm:dipw} (namely Assumptions~\ref{as:unconfounded}--\ref{as:p_large}).

It will also be helpful to introduce the notation $Z_i \in \R^{p-1}$ for $X_i$ with its first component removed. Furthermore, let $\theta$ and $\hat{\theta}$ denote $\gamma$ and $\hat{\gamma}$ respectively with their first components removed, and let $\theta_0$ and $\hat{\theta}_0$ denote $\gamma_1$ and $\hat{\gamma}_1$ respectively. Note that as the first component of $X_i$ is $1$, $\gamma_1 = \theta_0$ is the intercept coefficient.

\begin{lemma}\label{lem:qa}
	There exists a constant $c_{\delta1} > 0$ \newrev{depending only on $c_1, c_2, c_{\hat{\mu}}, c_\gamma, c_Y$}, such that
	on the event $\mathcal{T} := \T_1 \cap \T_2 \cap \T_3 \cap \T_4$, we have
	\begin{align*}
		|Q_A| \leq \newrev{c_{\hat{\pi}}^{-2}c_\pi^{-1}} c_{\delta1} s \frac{\log p}{n}.
	\end{align*}
\end{lemma}
\begin{proof}
	\newrev{Unless explicitly specified, all the constants defined in this proof depend only on the constants $c_1, c_2, c_{\hat{\mu}}, c_\gamma, c_Y$, and not on other constants specified in the events or the assumptions. Further, in the below, $c$ is a constant which may change from line to line.}
	We may decompose $Q_A$ as
	\begin{align} \label{eq:qa1}
		\begin{split}
			Q_A = & - \frac{1}{n} \sum_{i=1}^n \Big(\frac{T_i Y_i}{\hat{\pi}_i^2} + \frac{(1-T_i)Y_i}{(1-\hat{\pi}_i)^2} - \frac{\hat{\mu}_i}{\hat{\pi}_i(1 - \hat{\pi}_i)}\Big)(\hat{\pi}_i - \pi_i) \\
			& - \frac{1}{n} \sum_{i=1}^n \Big(\frac{T_i Y_i}{\hat{\pi}_i^2 \pi_i} - \frac{(1-T_i)Y_i}{(1-\hat{\pi}_i)^2(1-\pi_i)} \Big)(\hat{\pi}_i - \pi_i)^2 \\
			=: & Q_{A1} + Q_{A2}.
		\end{split}
	\end{align}
	Now we prove the claim by controlling the two terms $Q_{A1}$ and $Q_{A2}$ separately. 
	
	To control $Q_{A1}$, observe that by performing Taylor expansion of the logistic function $\psi$ about  $X_i^\top \hat{\gamma}$, we have that for each $i$, there exists some $\xi_i \in [0,1]$ such that
	\begin{align*}
		\pi_i = \hat{\pi}_i + \psi'(X_i^\top  \hat{\gamma}) X_i^\top  (\hat{\gamma} - \gamma) + \frac{1}{2} \psi''\big(X_i^\top  \hat{\gamma} \xi_i + X_i^\top  \gamma (1 - \xi_i)\big) \{X_i^\top  (\hat{\gamma} - \gamma)\}^2.
	\end{align*}
	This allows us to further decompose $Q_{A1}$ as follows:
	\begin{align*}
		Q_{A1} =\; & - \frac{1}{n} \sum_{i=1}^n \bigg(\frac{T_i Y_i}{\hat{\pi}_i^2} + \frac{(1-T_i)Y_i}{(1-\hat{\pi}_i)^2} - \frac{\hat{\mu}_i}{\hat{\pi}_i (1 - \hat{\pi}_i)}\bigg) \psi'(X_i^\top  \hat{\gamma}) X_i^\top  (\hat{\gamma} - \gamma) \\
		\; & - \frac{1}{n} \sum_{i=1}^n \frac{1}{2} \bigg(\frac{T_i Y_i}{\hat{\pi}_i^2} + \frac{(1-T_i)Y_i}{(1-\hat{\pi}_i)^2} - \frac{\hat{\mu}_i}{\hat{\pi}_i (1 - \hat{\pi}_i)}\bigg) \psi''\big(X_i^\top  \hat{\gamma} \xi_i + X_i^\top  \gamma (1 - \xi_i)\big) \{X_i^\top  (\hat{\gamma} - \gamma)\}^2 \\
		=: & Q_{A11} + Q_{A12}.
	\end{align*}
	For $Q_{A11}$, using the same analysis as in~\eqref{eq:approx_bias}--\eqref{eq:holder} yields
	\begin{align*}
		|Q_{A11}| \leq \Big\|\frac{1}{n} \sum_{i = 1}^n (\tilde{Y}_i - \hat{\mu}_i) X_i\Big\|_\infty \|\hat{\gamma} - \gamma\|_1.
	\end{align*}
	It then follows from the fact we are on events $\mathcal{T}_1$ and $\mathcal{T}_3$ that $|Q_{A11}| \leq \newrev{c_{\hat{\pi}}^{-1}} c_2 c_\gamma s \frac{\log p}{n}$.
	
	For $Q_{A12}$, observe that by using Cauchy--Schwarz inequality and the fact that $|\psi''(u)|\leq 1$ for all $u \in \R$,
	\begin{align*}
		Q_{A12}^2 \leq & \frac{1}{4n} \sum_{i=1}^n \bigg(\frac{T_i Y_i}{\hat{\pi}_i^2} + \frac{(1-T_i)Y_i}{(1-\hat{\pi}_i)^2} - \frac{\hat{\mu}_i}{\hat{\pi}_i (1 - \hat{\pi}_i)}\bigg)^2   \times \frac{1}{n}\sum_{i=1}^n (X_i^\top  (\hat{\gamma} - \gamma))^4.
	\end{align*}
	Then by using the fact that we are on $\T$, we have that
	\begin{align*}
		Q_{A12}^2 & \overset{(a)}{\leq} c^2 \bigg(\newrev{c_{\hat{\pi}}^{-4} \cdot} \frac{1}{n} \|\mb Y\|_2^2 + \newrev{c_{\hat{\pi}}^{-2} \cdot} \frac{1}{n}\|\hat{\mbb\mu}\|_2^2\bigg)\frac{1}{n}\sum_{i=1}^n \{X_i^\top  (\hat{\gamma} - \gamma)\}^4 \\
		& \overset{(b)}{\leq} c^2 \bigg(\newrev{c_{\hat{\pi}}^{-4} \cdot} \frac{1}{n} \|\mb Y\|_2^2 + \newrev{c_{\hat{\pi}}^{-2} \cdot} \frac{1}{n}\|\hat{\mbb\mu}\|_2^2\bigg) \bigg(s \frac{\log p}{n}\bigg)^2 \\
		& \overset{(c)}{\leq} \newrev{c_{\hat{\pi}}^{-4}} c^2 s^2 \frac{(\log p)^2}{n^2}.
	\end{align*}
	Here for inequality $(a)$, we use the fact that we are on $\T_1$; for $(b)$, we use that we are on $\T_2$; for $(c)$, we use that we are on $\T_3$ and $\T_4$.
	Hence, 
	\begin{equation}\label{eq:qa2}
		|Q_{A1}|  \leq |Q_{A11}| + |Q_{A12}| \leq \newrev{c_{\hat{\pi}}^{-2}} (c_2 c_\gamma + c) s \frac{\log p}{n}.
	\end{equation}
	Following a similar argument as that used to bound $Q_{A12}$ and using the fact that $|\psi'(u)| \leq 1$, we obtain
	\begin{equation}\label{eq:qa3}
		|Q_{A2}| \leq \newrev{c_{\hat{\pi}}^{-2} c_\pi^{-1}} c' s \frac{\log p}{n}.
	\end{equation}
	In light of the decomposition~\eqref{eq:qa1}, the desired result follows from~\eqref{eq:qa2} and~\eqref{eq:qa3}.
\end{proof}

\begin{lemma}\label{lem:qb}
	Let $\T := \T_1 \cap \T_3$. Given $m \in \mathbb{N}$, there exists constants $c, c_{\delta2} > 0$ \newrev{depending only on $c_\gamma, c_{\hat{\mu}}, c_2, \sigma_Z$} such that
	\begin{align*}
		\pr\Big(\Big\{|Q_B| \geq \newrev{c_{\hat{\pi}}^{-2} c_\pi^{-1}}c_{\delta2} \sqrt{s \log n} \frac{\log p}{n} \Big\} \cap \T \Big) \leq c(p^{-m} + n^{-m}).
	\end{align*}
\end{lemma}
\begin{proof}
	\newrev{Unless explicitly specified, all the constants defined in this proof depend only on the constants $c_\gamma, c_{\hat{\mu}}, c_2, \sigma_Z$, and not on other constants specified in the events or the assumptions.}
	Let $\mathcal{A}$ denote the event 
	\begin{align*}
		\mathcal{A} := \big\{ \|\mb X (\hat{\gamma} - \gamma)\|_\infty \leq c_A \sqrt{\log n} \|\hat{\gamma} - \gamma\|_2\big\},
	\end{align*}
	where $c_A>0$ is a constant to be chosen later. We first consider controlling $Q_B$ on the events $\T$ and $\mathcal{A}$. First observe that $\E(Q_B \mid \mathcal{D}_B, \mathcal{D}_A, \mb X) = 0$. Now $T_i - \pi_i$ is sub-Gaussian with variance proxy $\sigma^2=1/4$ conditional on $\mathcal{D}_B, \mathcal{D}_A, \mb X$ by Hoeffding's lemma. Thus applying Hoeffding's inequality, we have that for all $t > 0$,
	\begin{align*}
		\pr\big(|Q_B| \geq t \mid \mathcal{D}_B, \mathcal{D}_A, \mb X \big) & \leq 2 \exp(-2n^2t^2 / \sigma_Q^2)
	\end{align*}
	where
	\begin{align*}
		\sigma_Q^2 &= \sum_{i=1}^n \bigg(\frac{\hat{\mu}_i}{\hat{\pi}_i \pi_i} - \frac{\hat{\mu}_i}{(1-\hat{\pi}_i) (1-\pi_i)}\bigg)^2 (\hat{\pi}_i - \pi_i)^2 \\
		&\leq \max_i (\hat{\pi}_i - \pi_i)^2\sum_{i=1}^n \bigg(\frac{\hat{\mu}_i}{\hat{\pi}_i \pi_i} - \frac{\hat{\mu}_i}{(1-\hat{\pi}_i) (1-\pi_i)}\bigg)^2\\
		&\leq \|\mb X (\hat{\gamma} - \gamma)\|_\infty^2 \sum_{i=1}^n \bigg(\frac{\hat{\mu}_i}{\hat{\pi}_i \pi_i} - \frac{\hat{\mu}_i}{(1-\hat{\pi}_i) (1-\pi_i)}\bigg)^2.
	\end{align*}
	Note the final inequality uses the fact that 
	\begin{equation}\label{eq:qb0}
		|\hat{\pi}_i - \pi_i| \leq \sup_{u \in \R} |\psi'(u)| |X_i^\top (\hat{\gamma} - \gamma)| \leq |X_i^\top (\hat{\gamma} - \gamma)|.
	\end{equation}
	In the following we work on $\mathcal{T} \cap \mathcal{A}$, which we note is $(\mathcal{D}_B, \mathcal{D}_A, \mb X)$-measurable.
	Note that on $\T$ we have $\|\hat{\mbb\mu}\|_2^2 \leq \newrev{c_{\hat{\pi}}^{-2}}c_{\hat{\mu}} n$. Also, by Assumption~\ref{as:logistic} and the fact that we are on $\mathcal{T}_1$ we have that for all $i$,
	$c_{\hat{\pi}} \leq \hat{\pi}_i \leq 1-c_{\hat{\pi}}$ and $c_{\pi} \leq \pi_i \leq 1-c_{\pi}$. Thus
	\[
	\ind_{\T}\sum_{i=1}^n \bigg(\frac{\hat{\mu}_i}{\hat{\pi}_i \pi_i} - \frac{\hat{\mu}_i}{(1-\hat{\pi}_i) (1-\pi_i)}\bigg)^2 \leq \newrev{c_{\hat{\pi}}^{-4}c_\pi^{-2}} c n
	\]
	for constant $c\geq 0$.
	We may therefore conclude that
	\[
	\sigma_Q^2 \ind_{\T \cap \A} \leq \newrev{c_{\hat{\pi}}^{-4}c_\pi^{-2}} c c_\gamma c_A s \log (n) \log (p),
	\]
	and so for some $c > 0$ depending on $c_A$, 
	\begin{align*}
		2 \exp\bigg(- \frac{n^2 t^2}{\newrev{c_{\hat{\pi}}^{-4}c_\pi^{-2}} c s \log (n) \log (p)}\bigg) &\geq \E  \{\pr(|Q_B| \geq t \mid \mathcal{D}_B, \mathcal{D}_A, \mb X) \ind_{\mathcal{A}\cap \mathcal{T}}\} \\
		&= \pr\big(\{|Q_B| \geq t\} \cap \T \cap \mathcal{A}\big).
	\end{align*}
	By choosing $t = m \newrev{c_{\hat{\pi}}^{-2}c_\pi^{-1}} \sqrt{c s \log n}\log p / n$, we obtain
	\begin{align}\label{eq:qb1}
		\pr\bigg(\Big\{|Q_B| \geq m \newrev{c_{\hat{\pi}}^{-2}c_\pi^{-1}} \sqrt{cs \log n} \frac{\log p}{n}\Big\} \cap \T \cap \mathcal{A}\bigg) \leq 2 p^{- m}.
	\end{align}
	It remains to choose $c_A$ and bound
	$\pr(\mathcal{A}^c)$.
	Appealing to Assumption~\ref{as:subgX}, a standard sub-Gaussian tail bound yields
	\[
	\pr \bigg( \frac{|Z_i^\top (\hat{\theta} - \theta)|}{\|\hat{\gamma} - \gamma\|_2} \geq t \bigg) \leq 2e^{-t^2 /(2\sigma_Z^2)}.
	\]
	Observe that $|\theta_0 - \hat{\theta}_0| \leq \|\hat{\gamma} - \gamma\|_2$.
	By choosing $c_A \geq \sigma_Z\sqrt{2(m+1)}$ and applying a union bound, we have that for all $n$ sufficiently large,
	\begin{align}\label{eq:qb2}
		\pr (\mathcal{A}^c) & = \pr \bigg(\bigcup_{i=1}^n \bigg\{\frac{|X_i^\top (\hat{\gamma} - \gamma)|}{\|\hat{\gamma} - \gamma\|_2} \geq c_A\sqrt{\log n} \bigg\} \bigg) \nonumber\\
		&\leq n \pr \bigg(\frac{|Z_i^\top  (\hat{\theta} - \theta)|}{\|\hat{\gamma} - \gamma\|_2} \geq \sigma_Z \sqrt{2(m+1) \log n} \bigg)\\
		& \leq 2n^{-m} \nonumber.
	\end{align}
	The desired result then follows from noting that
	\begin{align*}
		& \pr\bigg(\Big\{|Q_B| \geq \newrev{c_{\hat{\pi}}^{-2}c_\pi^{-1}} c \sqrt{s \log n} \frac{\log p}{n}\Big\} \cap \T \bigg) \\
		&\quad \leq \pr\bigg(\Big\{|Q_B| \geq  \newrev{c_{\hat{\pi}}^{-2}c_\pi^{-1}} c\sqrt{s \log n} \frac{\log p}{n}\Big\} \cap \T \cap \mathcal{A} \bigg) + \pr( \A_1^c).
	\end{align*}
\end{proof}

\begin{lemma}\label{lem:qc}
	We have
	\[
	\sqrt{n}Q_C = \sqrt{\sigma_\mu^2 + \bar{\sigma}^2}\zeta
	\]
	where $\zeta$ satisfies (ii) of Theorem~\ref{thm:dipwpf}.
\end{lemma}
\begin{proof}
	Let us introduce, for $i=1,\ldots,n$
	\[
	u_i := \bigg(\frac{r_1(X_i)}{\pi_i} + \frac{r_0(X_i)}{1 - \pi_i} - \frac{\hat{\mu}_i}{\pi_i (1 - \pi_i)}\bigg) (T_i - \pi_i) + \frac{T_i \varepsilon_i(1)}{\pi_i} - \frac{(1 - T_i) \varepsilon_i(0)}{1 - \pi_i}.
	\]
	Note that $Q_C = \sum_{i=1}^n u_i / n$. Next observe that
	\begin{align*}
		\mu_{\ora,i} &= \E \left(\frac{T_iY_i(1-\pi_i)}{\pi_i} + \frac{(1-T_i)Y_i\pi_i}{1-\pi_i} \Big| X_i \right) \\
		&= \E \left(\frac{T_i(r_1(X_i) + \varepsilon_i(1))(1-\pi_i)}{\pi_i} + \frac{(1-T_i)(r_0(X_i) + \varepsilon_i(0))\pi_i}{1-\pi_i} \Big| X_i \right) \\
		&= r_1(X_i)(1-\pi_i) + r_0(X_i)\pi_i,
	\end{align*}
	so we may write $u_i$ as
	\begin{equation} \label{eq:u_i}
		u_i = \frac{\mu_{\ora,i} - \hat{\mu}_i}{\pi_i(1-\pi_i)}(T_i-\pi_i) + \frac{T_i \varepsilon_i(1)}{\pi_i} - \frac{(1 - T_i) \varepsilon_i(0)}{1 - \pi_i}.
	\end{equation}
	We see that $\E(u_i | \mathcal{D}) = 0$, as for example $\E (\varepsilon_i(1) | \D, T_i) = 0$. 
	Furthermore, the $u_i$ are independent conditional on $\D$ and
	\begin{equation} \label{eq:sigma_i}
		\sigma_i^2 := \Var(u_i \mid \mathcal{D}) = \frac{(\mu_{\ora, i} - \hat{\mu}_i)^2}{\pi_i (1 - \pi_i)} + \frac{\E(\varepsilon_i(1)^2 \mid X_i)}{\pi_i} + \frac{\E(\varepsilon_i(0)^2 \mid X_i)}{1 - \pi_i}
	\end{equation}
	and
	\[
	\frac{1}{n}\sum_{i=1}^n \sigma_i^2 = \frac{1}{n}\sum_{i=1}^n \frac{(\mu_{\ora,i}-\hat{\mu}_i)^2}{\pi_i(1-\pi_i)} + \bar{\sigma}^2=\sigma_\mu^2 + \bar{\sigma}^2.
	\]
	Next set
	\[
	\zeta := \frac{\frac{1}{\sqrt{n}}\sum_{i=1}^n u_i}{\sqrt{\sigma_\mu^2 + \bar{\sigma}^2}}.
	\]
	Observe that 
	\[
	\rho_i^3 := \E(|u_i|^3 | \mathcal{D}) \leq c \newrev{c_\pi^{-2}}\left(\frac{|\mu_{\ora, i} - \hat{\mu}_i|^3}{\pi_i (1 - \pi_i)} + \E(|\varepsilon_i(1)|^3 | X_i) + \E(|\varepsilon_i(0)|^3 | X_i)\right),
	\]
	for some \newrev{universal constant $c > 0$}.
	The Berry-Esseen theorem \citep{Esseen42} gives us that for a universal constant $C > 0$,
	\begin{equation}\label{eq:berry1}
		\sup_{t \in \R} |\pr(\zeta \geq t \mid \mathcal{D}) - \Phi(t)| \leq \frac{C}{\sqrt{n}} \bigg(\frac{1}{n}\sum_{i=1}^n \sigma_i^2\bigg)^{-3/2} \frac{1}{n}\sum_{i=1}^n \rho_i^3.
	\end{equation}
	Now using
	\[
	\frac{|\mu_{\ora, i} - \hat{\mu}_i|^3}{\pi_i (1 - \pi_i)} \leq  \|\mbb\mu_{\ora} - \hat{\mbb\mu}\|_\infty \frac{(\mu_{\ora, i} - \hat{\mu}_i)^2}{\pi_i(1 - \pi_i)},
	\]
	we see that
	\[
	\frac{1}{n}\sum_{i=1}^n \rho_i^3 \leq \newrev{c_\pi^{-2}} c(\|\mbb\mu_{\ora} - \hat{\mbb\mu}\|_\infty \sigma_{\mu}^2 + \bar{\rho}^3)
	\]
	from which the result easily follows.
\end{proof}
\subsection{Analysis of the events}\label{sec:events}

In this section, we prove that given $c_\gamma, c_{\hat{\pi}}, c_{\tilde{\mu}}$ and $m \in \mathbb{N}$, there exist positive constants $c, c_1, c_2, c_{\hat{\mu}}$ and $c_Y$ \newrev{that do not depend on $c_\pi$ and $c_{\hat{\pi}}$} such that the event
\[
\T:=\T(c_\gamma, c_{\hat{\pi}}, c_1, c_{\hat{\mu}}, c_2, c_Y) := \T_1 (c_\gamma, c_{\hat{\pi}}) \cap \T_2(c_1) \cap \T_3(c_{\hat{\mu}}, c_2\newrev{, c_{\hat{\pi}}}) \cap \T_4(c_Y)
\]
satisfies 
\begin{align*}
	\pr(\mathcal{T}^c) \leq \pr(\Omega^c(c_\gamma, c_{\tilde{\mu}}, c_{\hat{\pi}})) + c'(p^{-m} + n^{-m}).
\end{align*}
To do this, we first observe that
\begin{align}
	\pr(\T^c) &\leq \pr(\T^c \cap \Omega) + \pr(\Omega^c) \notag \\
	&= \pr((\T_2^c \cup \T_3^c \cup \T_4^c) \cap \Omega) + \pr(\Omega^c) \notag\\
	&= \pr(\{\T_2^c \cup (\T_3^c \cap \T_4) \cup \T_4^c\} \cap \Omega) + \pr(\Omega^c)  \notag\\
	&\leq \pr(\T_2^c \cap \Omega) + \pr(\T_3^c \cap \T_4 \cap \Omega) + \pr(\T_4^c) + \pr(\Omega^c), \label{eq:t_bds}
\end{align}
and then individually bound the first three terms in \eqref{eq:t_bds}. \newrev{Unless explicitly specified, all the constants defined in this section do not depend on $c_{\hat{\pi}}$ and $c_\pi$.}


%
\subsubsection{Bounds relating to $\T_2$}
\begin{lemma}\label{lem:prop3}
	For any $m \in \mathbb{N}$, there exists constants $c,c'>0$ such that for all $v \in \R^p$ and $n \in \mathbb{N}$,
	\begin{align*}
		\pr\bigg(\frac{1}{n} \sum_{i =1}^n (X_i^\top v)^4 > c \|v\|_2^4\bigg) \leq \frac{c'}{n^m}.
	\end{align*}
\end{lemma}
\begin{proof}
	First note that it suffices to show the result when $\|v\|_2=1$.
	Now denote by $v_{-1}\in \R^p$ the vector $v$ with the first component removed. We have
	\begin{align} \label{eq:X_ibd}
		(X_i^\top v)^4 \leq 2^4\{(Z_i^\top v_{-1})^4 + v_1^4\} \leq 2^4 \{(Z_i^\top v_{-1})^4 + 1\}.
	\end{align}
	Note also that as $Z_i^\top v_{-1}$ is mean-zero and sub-Gaussian with variance proxy $\sigma_Z^2 \|v_{-1}\|_2^2 \leq \sigma_Z^2$, 
	\begin{equation} \label{eq:mom_bd}
		\E \{(Z_i^\top v_{-1})^4\} \leq c,
	\end{equation}
	where $c>0$ is a constant depending on $\sigma_Z^2$.
	Now let $u_i := Z_i^\top  v_{-1}$. From \eqref{eq:X_ibd} and \eqref{eq:mom_bd}, we see that it suffices to show that there exist constants $c, c'>0$ such that 
	\[
	\pr\bigg(\frac{1}{n}\sum_{i=1}^n (u_i^4 - \E u_i^4)  > c \bigg) \leq \frac{c'}{n^m}.
	\]
	A tail bound for mean-zero random variables under moment constraints \citep[page~55]{wainwright_2019} yields  that
	\begin{equation} \label{eq:mom3_bd}
		\pr\bigg(\frac{1}{n}\sum_{i=1}^n (u_i^4 - \E u_i^4)  > c \bigg) \leq \frac{c'}{c^{2m}n^m}.
	\end{equation}
	where $c'$ depends only on $\{\E (u_i^4 - \E u_i^4)^{2m}\}^{\frac{1}{2m}}$ (and in particular is finite if $c'$ is finite). But $\E (u_i^4 - \E u_i^4)^{2m} \leq 2^{2m} \E u_i^{8m} < c''$ where $c'' < \infty$ and $c''$ depends only on $m$ and $\sigma_Z^2$.
\end{proof}

\begin{lemma}\label{lem:t2}
	Given $c_{\gamma}, c_{\hat{\pi}}, c_{\tilde{\mu}} >0$ and $m \in \mathbb{N}$, there exists $c_1, c' >0$ such that
	\[
	\pr(\T_2^c(c_1) \cap \Omega(c_\gamma, c_{\tilde{\mu}}, c_{\hat{\pi}})) \leq c'n^{-m}.
	\]
\end{lemma}
\begin{proof}
	Applying Lemma~\ref{lem:prop3} with $v=\gamma - \hat{\gamma}$ we have that there exists $c, c'$ such that
	\[
	\pr \bigg( \frac{1}{n} \sum_{i=1}^n \{X_i^\top (\hat{\gamma} - \gamma)\}^4 > c \|\hat{\gamma} - \gamma\|_2^4 \bigg) \leq \frac{c'}{n^m}.
	\]
	Thus setting $c_1 = c c_\gamma$, we have
	\begin{align*}
		\T_2^c(c_1) \cap \Omega(c_\gamma, c_{\tilde{\mu}}, c_{\hat{\pi}}) &=\bigg\{ \frac{1}{n}\sum_{i=1}^n \{X_i^\top (\hat{\gamma} - \gamma)\}^4 > c_\gamma^2 c^2 s^2\frac{(\log p)^2}{n^2}\bigg\} \cap \Omega(c_\gamma, c_{\tilde{\mu}}, c_{\hat{\pi}}) \\
		&\subseteq  \bigg\{ \frac{1}{n} \sum_{i=1}^n \{X_i^\top (\hat{\gamma} - \gamma)\}^4 > c \|\hat{\gamma} - \gamma\|_2^4  \bigg\},
	\end{align*}
	which implies the result.
\end{proof}

\subsubsection{Bounds relating to $\T_3$}
\begin{lemma} \label{lem:feasibility}
	Given $c_{\gamma}, c_{\tilde{\mu}}, c_{\hat{\pi}} > 0$ and $m \in \mathbb{N}$, there exist $c > 0$ and 
	$c_\eta > 0$ 
	such that, 
	\[
	\pr(\{ \tilde{\mb Y} \text{ is not a feasible solution to } \eqref{eq:mu_hat}\} \cap \Omega) \leq c p^{-m}.
	\]
\end{lemma}

\begin{proof}
	We define the events $\tilde{\Omega}_1 (c_\gamma, c_{\tilde{\mu}}) := \{\text{(i) and (iii) of } \Omega((c_\gamma, c_{\tilde{\mu}}, c_{\hat{\pi}})\}$ and $\tilde{\Omega}_2 (c_{\hat{\pi}}) := \{\text{(ii) of } \Omega(c_\gamma, c_{\tilde{\mu}}, c_{\hat{\pi}})\}$. Then we can prove the desired statement via proving that there exists $c, c_{\eta} > 0$ such that on $\tilde{\Omega}_1$, writing 
	$\mathcal{A} := \{\tilde{\mb Y} \text{ is not a feasible solution to } \eqref{eq:mu_hat}\}$,
	\[
	\pr(\mathcal{A} \cap \tilde{\Omega}_2 \mid \mathcal{D}_B) \leq c p^{-m}.
	\]
	Indeed, taking expectations and noting that $\tilde{\Omega}_1$ is $\mathcal{D}_B$-measurable, we would have
	\begin{align*}
		cp^{-m} &\geq \E [\pr(\mathcal{A} \cap \tilde{\Omega}_2 \mid \mathcal{D}_B) \ind_{\tilde{\Omega}_1}] \\
		&= \E [\E\{\E(\ind_{\{\mathcal{A} \cap \tilde{\Omega}_2\}} \mid \mathcal{D}_B) \mid \ind_{\tilde{\Omega}_1}\} \ind_{\tilde{\Omega}_1}] \\
		&= \E [\pr(\mathcal{A} \cap \tilde{\Omega}_2 \mid \tilde{\Omega}_1) \ind_{\tilde{\Omega}_1}] \\
		&= \pr(\mathcal{A} \cap \Omega).
	\end{align*}
	Now let $\check{Y}_i := \tilde{Y}_i \mathbbm{1}_{\{c_{\hat{\pi}}\le \hat{\pi}_i \le 1 - c_{\hat{\pi}}\}}, \check{\mb Y} := (\check{Y}_i)_{i=1}^n$ and define $\check{\mb Y}_A$ analogously. To show the above, we prove the above claim via proving a stronger claim that on $\tilde{\Omega}_1$,
	\[
	\pr\left(\frac{1}{n} \|\mb X_A^\top \{\check{\mb Y}_A - \tilde{\mu}(\mb X_A)\} - \mb X^\top \{\check{\mb Y} - \tilde{\mu}(\mb X)\} \|_\infty \ge c_\eta \newrev{c_{\hat{\pi}}^{-1}} \sqrt{\frac{\log p}{n}} \mid \mathcal{D}_B \right) \le c p^{-m}.
	\]	
	Note first that for each fixed $j$, conditional on $\mathcal{D}_B$ the
	\begin{equation} \label{eq:bern_summand}
		X_{A,ij}(\check{Y}_{A,i} - \tilde{\mu}(X_{A,i})) - X_{ij}(\check{Y}_i - \tilde{\mu}(X_i))
	\end{equation}
	are i.i.d.\ for $i=1,\ldots,n$ and mean-zero (i.e.\ their expectation conditional on $\mathcal{D}_B$ is zero). On $\tilde{\Omega}_1$ and conditional on $\mathcal{D}_B$, each $X_{ij}$, $\check{Y}_i$, $\tilde{\mu}(X_i)$ and their counterparts on $\mathcal{D}_A$ are sub-Gaussian. Moreover, on $\tilde{\Omega}_1$ and  conditional on $\mathcal{D}_B$, each of their means are bounded. 
	\newrev{Hence by \citet[Lem.~2.7.7 and Prop. 2.7.1]{vershynin2018high}, on $\tilde{\Omega}_1$ and  conditional on $\mathcal{D}_B$, the quantity in \eqref{eq:bern_summand} is a sub-exponential random variable with sub-exponential norm $c_{\hat{\pi}}^{-1} K$, where $K$ is some constant. Thus applying a union bound, \citet[Proposition 2.7.1]{vershynin2018high} and Bernstein's inequality of the form given by \citet[Equation (2.18)]{wainwright_2019} yields that on $\tilde{\Omega}_1$,
		\begin{align*} 
			& \pr \bigg( \frac{1}{n} \|\mb X_A^\top \{\check{\mb Y}_A - \tilde{\mu}(\mb X_A)\} - \mb X^\top \{\check{\mb Y} - \tilde{\mu}(\mb X)\} \|_\infty \geq t  \,|\, \mathcal{D}_B \bigg) \\
			& \leq 2p \exp \left( - \min\left(\frac{n t^2}{2 c_{\hat{\pi}}^{-2}K^2}, \frac{n t}{2 c_{\hat{\pi}}^{-1}K}\right)\right).
		\end{align*}
		Then substituting $t = \eta = c_\eta c_{\hat{\pi}}^{-1} \sqrt{ \log p / n}$ shows that on $\tilde{\Omega}_1$ we have
		\begin{align*}
			&\; \pr \bigg( \frac{1}{n} \|\mb X_A^\top \{\check{\mb Y}_A - \tilde{\mu}(\mb X_A)\} - \mb X^\top \{\check{\mb Y} - \tilde{\mu}(\mb X)\} \|_\infty < c_\eta c_{\hat{\pi}}^{-1} \sqrt{ \log p / n}  \mid \mathcal{D}_B \bigg) \\
			& \; \geq 1 - 2\exp\left( - \min\left(\frac{c_\eta^2 \log p}{2 K^2}, \frac{c_\eta \sqrt{n \log p}}{2 K}\right) + \log p \right).
		\end{align*}
		Thus provided $c_\eta^2 / (2 K^2) \geq m+1$, we have from Assumption~\ref{as:p_large} that for $n$ sufficiently large, the last display above is at least $1-2p^{-m}$, as required.}
\end{proof}

\begin{lemma} \label{lem:program}
	Given $c_{\gamma}, c_{\tilde{\mu}}, c_{\hat{\pi}}, c_Y > 0$ and $m \in \mathbb{N}$, there exist constants $c, c_{\hat{\mu}}, c_2 > 0$ such that
	\begin{align}
		\pr\bigg(\Big\{\frac{1}{n} \|\hat{\mbb\mu}\|_2^2 > \newrev{c_{\hat{\pi}}^{-2}} c_{\hat{\mu}}\Big\} \cap \T_4 \cap \Omega \bigg) &\leq e^{-cn} \label{eq:mu_bd} \\
		\pr\bigg( \bigg\{ \frac{1}{n}\|\mb X^\top (\hat{\mbb\mu} - \tilde{\mb Y})\|_\infty >  \newrev{c_{\hat{\pi}}^{-1}} c_2 \sqrt{\frac{\log p}{n}} \bigg\} \cap \Omega \bigg) &\leq c p^{-m}. \label{eq:tildeY_bd}
	\end{align}
\end{lemma}
\begin{proof}
	We show \eqref{eq:tildeY_bd} first. Let $\A_1 := \{\tilde{\mb Y} \text{ is feasible for } \eqref{eq:mu_hat_opt}\}$. On the event $\A_1$,
	\begin{align*}
		\frac{1}{n} \|\mb X_A^\top \{\tilde{\mb Y}_A - \tilde{\mu}(\mb X_A)\} - \mb X^\top \{\tilde{\mb Y} - \tilde{\mu}(\mb X)\} \|_\infty  &\leq c_\eta \newrev{c_{\hat{\pi}}^{-1}} \sqrt{\frac{\log p}{n}} \\
		\frac{1}{n} \|\mb X_A^\top \{\tilde{\mb Y}_A - \tilde{\mu}(\mb X_A)\} - \mb X^\top \{\hat{\mbb \mu} - \tilde{\mu}(\mb X)\} \|_\infty  &\leq c_\eta \newrev{c_{\hat{\pi}}^{-1}} \sqrt{\frac{\log p}{n}} .
	\end{align*}
	Thus by the triangle inequality,
	\[
	\frac{1}{n} \|\mb X^\top (\hat{\mbb\mu} - \tilde{\mb Y})\|_\infty \leq 2 c_{\eta} \newrev{c_{\hat{\pi}}^{-1}} \sqrt{\frac{\log p}{n}} .
	\]
	From Lemma~\ref{lem:feasibility}, we have that there exists constant $c >0$ and a choice of $c_\eta$ such that $\pr(\A_1^c \cap \Omega) \le cp^{-m}$. Thus
	\[
	\pr\bigg( \Big\{ \frac{1}{n} \|\mb X^\top (\hat{\mbb\mu} - \tilde{\mb Y})\|_\infty > 2 c_{\eta} \newrev{c_{\hat{\pi}}^{-1}}\sqrt{\frac{\log p}{n}} \Big\} \cap \Omega\bigg) \leq c p^{-m},
	\]
	which shows \eqref{eq:tildeY_bd}. To show \eqref{eq:mu_bd} we first observe that if $\tilde{\mb Y}$ is feasible, then
	\begin{align*}
		\|\hat{\mbb\mu}\|_2^2 & \leq 
		2 \|\hat{\mbb\mu} - \tilde{\mu}(\mb X)\|_2^2 + 2 \|\tilde{\mu}(\mb X)\|_2^2 \leq 2 \|\tilde{\mb Y} - \tilde{\mu}(\mb X)\|_2^2 + 2 \|\tilde{\mu}(\mb X)\|_2^2 \\
		&\leq 4 \|\tilde{\mb Y}\|_2^2 + 12 \|\tilde{\mu}(\mb X) - \E (\tilde{\mu}(\mb X) |\tilde{\mu}) \|_2^2 + 12\|\E (\tilde{\mu}(\mb X) |\tilde{\mu})\|_2^2.
	\end{align*}
	\newrev{Now on $\Omega$ we have $|\tilde{Y}_i| \leq c_{\hat{\pi}}^{-1} |Y_i|$ for all $i$.}
	Thus on the event $\Omega \cap \A_1 \cap \T_4$
	\[
	\frac{4}{n} \|\tilde{\mb Y}\|_2^2 + \frac{12}{n}\|\E (\tilde{\mu}(\mb X) |\tilde{\mu})\|_2^2 \leq 4\newrev{c_{\hat{\pi}}^{-2}}c_Y + 12 c_{\tilde{\mu}}^2.
	\]
	Conditional on $\Omega$, $\tilde{\mu}(X_i) - \E(\tilde{\mu}(X_i)| \tilde{\mu})$ is mean-zero and sub-Gaussian with variance proxy $c_{\tilde{\mu}}^2$. Thus its square satisfies Bernstein's condition with parameters $(8c_{\tilde{\mu}}^2, 4 c_{\tilde{\mu}}^2)$ \citep[Lem.~2.7.7]{vershynin2018high}, and hence Bernstein's inequality yields
	\[
	\pr\bigg( \frac{1}{n} \|\tilde{\mu}(\mb X) - \E (\tilde{\mu}(\mb X) |\tilde{\mu}) \|_2^2  \geq m_{\tilde{\mu}} + t  \mid \Omega \bigg) \leq \exp \left( - \frac{nt^2}{2(8 c_{\tilde{\mu}}^2 + 4 c_{\tilde{\mu}}^2 t)} \right)
	\]
	for $t \geq 0$, 
	where $m_{\tilde{\mu}} := \E \{\tilde{\mu}(X_i) - \E(\tilde{\mu}(X_i)| \tilde{\mu})\}^2$.  Note that on $\Omega$, we have $m_{\tilde{\mu}} \leq c_{\tilde{\mu}}^2$. Thus putting things together, we see that for $c_{\hat{\mu}} > 4c_Y + 24c_{\tilde{\mu}}^2$, there exists $c'>0$ such that
	\[
	\pr\bigg(\Big\{\frac{1}{n} \|\hat{\mbb\mu}\|_2^2 > \newrev{c_{\hat{\pi}}^{-2}} c_{\hat{\mu}} \Big\} \cap \T_4 \cap \Omega \bigg) \leq e^{-c'n}
	\]
	as required.
\end{proof}

\subsubsection{Bounds relating to $\T_4$}
\begin{lemma}\label{lem:t4}
	There exist constants $c, c_Y > 0$ such that
	\[
	\pr(\|{\mb Y}\|_2^2 /n\geq c_Y) \leq e^{-cn}.
	\]
\end{lemma}
\begin{proof}
	We have that
	\[
	Y_i^2 \leq \max_{t=0,1} Y_i^2(t) \leq 2\{Y_i(0) - \E Y_i(0)\}^2 + 2\{Y_i(1) - \E Y_i(1)\}^2 + 2m_Y^2.
	\]
	Thus
	\[
	\frac{1}{2 n}\|\mb Y\|_2^2 \leq m_Y^2 + \frac{1}{n}\sum_{i=1}^n \{Y_i(0) - \E Y_i(0)\}^2 + \frac{1}{n}\sum_{i=1}^n  \{Y_i(1) - \E Y_i(1)\}^2.
	\]
	As each $Y_i(t) - \E Y_i(t)$ is mean-zero and sub-Gaussian (Assumption~\ref{as:subgY}) their square is is sub-expoential \citep[Lem.~2.7.7]{vershynin2018high}. The result then follows from Bernstein's inequality.
\end{proof}

\section{Proof of Theorem~\ref{thm:split}}
First note that
\[
\bar{\tau} - \tau = \frac{1}{n} \sum_{i=1}^n \bigg(r_1(X_i) - r_0(X_i) - \tau\bigg).
\]
By combining this with the decomposition of $\hat{\tau}_\dipw - \bar{\tau}$ in~\eqref{eq:decomp}, we obtain
\begin{align}\label{eq:splitdecomp}
	\hat{\tau}_\dipw - \tau & = \underbrace{Q_A + Q_B}_{=n^{-1/2}\delta} + \underset{=: Q_{C1}}{\underbrace{\frac{1}{n} \sum_{i=1}^n \bigg(\frac{r_{1,i}}{\pi_i} + \frac{r_{0,i}}{1 - \pi_i} - \frac{\hat{\mu}_i}{\pi_i (1 - \pi_i)}\bigg) (T_i - \pi_i)}}  \\
	& \quad + \underset{=: Q_{C2}}{\underbrace{\frac{1}{n} \sum_{i=1}^n \bigg(\frac{T_i \varepsilon_{1,i}}{\pi_i} - \frac{(1 - T_i) \varepsilon_{0,i}}{1 - \pi_i} + r_1(X_i) - r_0(X_i) - \tau\bigg)}},\notag
\end{align}
where $Q_A$ and $Q_B$ are defined  as in~\eqref{eq:decomp}.
We may thus take $\delta$ in Theorem~\ref{thm:split} as the same as that from Theorem~\ref{thm:dipw}, so (i) of Theorem~\ref{thm:split} follows directly from (i) of Theorem~\ref{thm:dipw}.
For (ii), we assign $\zeta_1 = \sqrt{n}Q_{C1} / \sigma_\mu$ and $\zeta_2 = \sqrt{n}Q_{C2}/\sigma$.
To prove the required property of $\zeta_1$ thus defined, we follow the argument of Lemma~\ref{lem:qc} with $u_i$ redefined as
\[
u_i := \bigg(\frac{r_{1,i}}{\pi_i} + \frac{r_{0,i}}{1 - \pi_i} - \frac{\hat{\mu}_i}{\pi_i (1 - \pi_i)}\bigg) (T_i - \pi_i).
\]
To prove the required property of $\zeta_2$, we apply the Berry--Esseen theorem~\citep{Esseen42} to the mean zero i.i.d.\ random variables
\[
\bigg(\frac{T_i \varepsilon_{1,i}}{\pi_i} - \frac{(1 - T_i) \varepsilon_{0,i}}{1 - \pi_i} + r_1(X_i) - r_0(X_i) - \tau \bigg).
\]
Turning to property (iii), we have
\begin{align*}
	\E (r_t(X_i)(T_i-\pi_i)r_{t'}(X_i) \,|\,\mathcal{D}) = r_t(X_i)r_{t'}(X_i) \E ((T_i-\pi_i) \,|\,\mathcal{D}) = 0
\end{align*}
for all $t, t' \in \{0,1\}$. Also 
\begin{align*}
	\E (r_t(X_i)(T_i-\pi_i)\varepsilon_{i}(t') \,|\,\mathcal{D}) &= \E \{\E (r_t(X_i)(T_i-\pi_i)\varepsilon_{i}(t') \,|\,\mathcal{D}, Y_i) \, |\, \mathcal{D}\}  \\
	&=  \E \{\varepsilon_{i}(t') r_t(X_i) \E (T_i-\pi_i \,|\,\mathcal{D}) \, |\, \mathcal{D}\} =0,
\end{align*}
using Assumption~\ref{as:unconfounded} for the final equality. \qed

\section{Confidence intervals for $\tau$} \label{sec:CI_tau}
\newrev{
	Here we provide a confidence interval for $\tau$. Recall that without making assumptions on the outcome regrssion model which would allow $\hat{\mbb \mu} - \mbb\mu_{\ora}$ to be small, $\hat{\tau}_{\dipw}$ does not have a Gaussian distribution, but rather its distribution is expected to be a mixture of Gaussians (see Theorem~\ref{thm:split}). This makes the construction of a confidence interval challenging, but we can provide a conservative confidence interval, that contracts at the optimal $1/\sqrt{n}$ rate.
	
	To see how this may work, our goal is to first provide a valid confidence interval for $\bar{\tau}$; then using that $\sqrt{n}(\bar{\tau} - \tau)$ is also asymptotically normal, we can construct a valid confidence interval for $\tau$ using a union bound. To see how to construct an interval for $\bar{\tau}$, since by Theorem~\ref{thm:dipw}, $\sqrt{n}(\hat{\tau}_\dipw - \bar{\tau})$ is approximately normal conditional on $\{{\mbb X}, \mathcal{D}_A, \mathcal{D}_B\}$, provided in addition that Assumptions~\ref{as:sparsity} and~\ref{as:subeps} hold. Then we may construct an interval for $\bar{\tau}$ via estimating an upper bound of our estimator's conditional variance:
	\[
	\hat{\sigma} := \frac{1}{n}\sum_{i=1}^n \left(\frac{T_i(Y_i - \hat{\mu}_i)}{\hat{\pi}_i} - \frac{(1 - T_i)(Y_i - \hat{\mu}_i)}{1 - \hat{\pi}_i} - \hat{\tau}_{\dipw}\right)^2.
	\]
	At the same time, we also need to estimate an upper bound of the variance of $\sqrt{n}(\bar{\tau} - \tau)$, which is based on an estimate of the variance of the IPW estimator:
	\[
	\hat{\sigma}^2_\ipw := \frac{1}{n} \sum_{i=1}^n \left(\frac{T_iY_i}{\hat{\pi}_i} - \frac{(1 - T_i)Y_i}{1 - \hat{\pi}_i} - \hat{\tau}_{\ipw}\right)^2.
	\]
	Combining them together, we can construct a conservative confidence interval via the following approach:
	\[
	\hat{C}_{\alpha} := \Big[\hat{\tau}_{\dipw} - \frac{\hat{\sigma} + \hat{\sigma}_\ipw}{\sqrt{n}}z_{\alpha/2}, \, \hat{\tau}_{\dipw}+ \frac{\hat{\sigma} + \hat{\sigma}_\ipw}{\sqrt{n}}z_{\alpha/2} \Big].
	\] 
	where recall that $z_\alpha$ is the upper $\alpha/2$ point of a standard Gaussian distribution. 
	\begin{theorem}\label{thm:conf_tau}
		Consider the setup of Theorem~\ref{thm:dipw} but additionally suppose Assumptions~\ref{as:sparsity} and \ref{as:subeps} hold. Then there exists a constant $c_{\zeta}>0$ such that for all $\alpha \in (0, 1]$,
		\begin{align*}
			&\pr(\tau \in \hat{C}_{\alpha}) \geq 1 - \alpha  - \pr(\Omega^c(c_{\gamma}, c_{\tilde{\mu}}, c_{\hat{\pi}}))\\
			&\qquad - c_{\zeta} \bigg(\E\{\|\hat{\mbb\mu} - \mbb\mu_{\ora}\|_\infty\ind_{\Omega(c_{\gamma}, c_{\tilde{\mu}}, c_{\hat{\pi}}) }\} \sqrt{\frac{\log n}{n}}  + \frac{\sqrt{ b_n \log p \log n}  + 1}{n^{1/4}} + b_n +  p^{-m} \bigg).
		\end{align*}
	\end{theorem}
	As discussed  following Theorem~\ref{thm:ci_new} and also Theorem~\ref{thm:dipw}, we expect $\E\{\|\hat{\mbb\mu} - \mbb\mu_{\ora}\|_\infty\ind_{\Omega(c_{\gamma}, c_{\tilde{\mu}}, c_{\hat{\pi}}) }\}$ to only grow relatively slowly with $n$, and under rather weak conditions we should have that this is $o(\sqrt{n / \log (n)})$.
}

\section{Proofs of Theorems~\ref{thm:ci_new} and~\ref{thm:conf_tau}} \label{sec:CI_proofs}

\subsection{An intermediate result}

To prove the results on confidence intervals, we first need an intermediate result, which is given in Theorem~\ref{thm:ci}. To prove this theorem, we need some preliminary lemmas. It is helpful to define $\sigma_*^2 := \sigma_{\mu}^2 + \bar{\sigma}^2$,  $\sigma_2^2 := \Var(r_1(X) - r_0(X))$ and $\sigma_+^2 := \sigma_*^2 + \sigma_2^2$.

\begin{lemma}\label{lem:d}
	Consider the setup of Theorem~\ref{thm:conf_tau}. We have that there exists constants $c_1, \ldots, c_4 > 0$
	such that on an event $\Lambda_1 \subseteq \Omega(c_\gamma,c_{\tilde{\mu}},c_{\hat{\pi}})$ with probability at least $1 - \pr(\Omega^c(c_\gamma,c_{\tilde{\mu}},c_{\hat{\pi}})) - c_1 (p^{-m} + n^{-m})$, the random element $\mathcal{D}$ satisfies following inequalities:
	\begin{gather}
		\max_{i=1,\ldots,n}|\hat{\pi}_i - \pi_i| \leq c_2 \frac{\sqrt{b_n \log n}}{n^{\frac{1}{4}}}; \label{eq:d1}\\
		\max_{t=0,1} \frac{1}{n} \sum_{i=1}^n r_t^2 (X_i) \leq c_2; \label{eq:d2}\\ 
		\pr(\|{\mb Y}\|_2^2 / n > c_2 \mid \mathcal{D}) \leq c_3 n^{-m},   \qquad \frac{1}{n} \|\hat{\mbb\mu}\|_2^2 \leq c_2; \label{eq:d3}\\
		\pr\bigg(|\delta| > c_2 (s + \sqrt{s \log n}) \frac{\log p}{\sqrt{n}} \,\Big|\, \mathcal{D}\bigg) \leq c_2 (n^{-m} + p^{-m}); \label{eq:d4} \\
		c_4 \leq \sigma^2_+ \leq c_2, \qquad \bar{\rho}^3 \leq c_2. \label{eq:d5}
	\end{gather}
\end{lemma}
\begin{proof}
	\eqref{eq:d1}: This follows from~\eqref{eq:qb2}, \eqref{eq:qb0}, and  Assumption~\ref{as:sparsity}.
	
	\noindent\eqref{eq:d2}: Fix $t \in \{0,1\}$. By Jensen's inequality $\E(|r_t(X)- \E Y(t) |^l) \leq \E |Y(t) - \E Y(t)|^l$  for $l \geq 1$. Thus, by the characterisation of sub-Gaussian random variables in terms of moments bounds \citep[Thm.~2.6]{wainwright_2019} and Assumption~\ref{as:subgY}, $r_t(X)- \E Y(t)$ is sub-Gaussian and so the desired inequality follows similarly to Lemma~\ref{lem:t4}.
	
	
	\noindent\eqref{eq:d3}: For the first result, note that by Lemma~\ref{lem:t4}, there exists $c$ such that $\pr(\|\mb Y\|_2^2/n > c) \leq n^{-2m}$ for all $n$ sufficiently large.
	Thus by Markov's inequality, for such $n$,
	\[
	\pr\{\pr(\|\mb Y\|_2^2/n > c \,|\, \mathcal{D}) > n^{-m}\} \leq n^m n^{-2m} = n^{-m}.
	\]
	The second result follows from \eqref{eq:mu_bd}.


	\noindent\eqref{eq:d4}: From Theorem~\ref{thm:dipw} and its proof (in particular Lemmas~\ref{lem:qa} and \ref{lem:qb}), we know that there exist constants $c_\delta, c'>0$ such that
	\[
	\pr( \A \cap \Omega ) \leq c' \{\max(n,p)\}^{-2m}
	\]
	where
	\[
	\A:=\bigg\{ |\delta| > c_\delta (s + \sqrt{s \log n}) \frac{\log p}{\sqrt{n}} \bigg\}.
	\]
	Now as $\Omega$ is $\mathcal{D}$-measurable, we have
	\[
	c' \{\max(n,p)\}^{-2m} \geq \E (\ind_{\A} \ind_\Omega  ) = \E \E (\ind_{\A} \ind_\Omega \, | \, \mathcal{D} ) = \E \{\ind_\Omega \E (\ind_{\A}  \, | \, \mathcal{D} )\} = \E \{\ind_\Omega \pr( \A \,|\,\mathcal{D})\},
	\]
	whence by Markov's inequality,
	\[
	\pr(\{\pr( \A \,|\,\mathcal{D}) > \{\max(n,p)\}^{-m}\} \cap \Omega)  = \pr(\ind_\Omega \pr( \A \,|\,\mathcal{D}) > \{\max(n,p)\}^{-m})  \leq c'\{\max(n,p)\}^{-m}.
	\]
	Thus $\pr\{ \pr(\A \,|\,\mathcal{D}) > \{\max(n,p)\}^{-m}\} \leq c'\{\max(n,p)\}^{-m} + \pr(\Omega^c)$.
	
	\noindent\eqref{eq:d5}: Fixing $t \in \{0,1\}$ and writing $a$ for the mean of $Y(t)$, we have that for $l \in \mathbb{N}$
	\begin{align}
		\E \varepsilon(t)^{2l} &= \E \{Y(t) - a + a - \E (Y(t) | X)\}^{2l} \notag\\
		&\leq 2^{2l} \E \{Y(t) - a\}^{2l} + 2^{2l}\E\{\E (Y(t) -a| X)\}^{2l} \notag\\
		&\leq 2 \cdot 2^{2l}\E \{Y(t) - a\}^{2l}, \notag
	\end{align}
	applying Jensen's inequality in the final line. As $Y(t) - a$ is mean-zero and sub-Gaussian we see that $\varepsilon(t)$ is also sub-Gaussian. Thus, we know there exists $c>0$ such that 
	\begin{align} \label{eq:mom2_bd}
		\frac{c^r (2r)!}{2^r r!} \geq \E \varepsilon(t)^{2r} \geq \E[\E\{\varepsilon(t)^2|X\}]^r,
	\end{align}
	using Jensen's inequality for the final inequality.
	A tail bound for averages of independent mean-zero random variables under moment constraints \citep[pg.~55]{wainwright_2019} along with Assumptions~\ref{as:subeps}~and~\ref{as:logistic}
	then show that $\bar{\sigma}^2$ concentrates around its (positive) mean (see also \eqref{eq:mom3_bd}), thereby giving the required high probability lower bound, and a corresponding upper bound. In view of this, for the upper bound, it suffices to argue that $\sigma_\mu^2$ is bounded above with high probability, but this follows from Lemma~\ref{lem:program}. The bound for $\bar{\rho}^3$ follows similarly to those for $\bar{\sigma}^2$. 
\end{proof}

\begin{lemma}\label{lem:sigma}
	Consider the setup of Theorem~\ref{thm:conf_tau} and Lemma~\ref{lem:d}. Given $m \in \mathbb{N}$, there exists constants 
	Given $m \in \mathbb{N}$, there exist constants $c_1, c_2, c_3 > 0$ such that on an event
	$\Lambda_2 \subseteq \Lambda_1 $ with probability at least $1 - \pr(\Lambda_1^c) - c_1 n^{-m}$, we have
	\[
	\pr\bigg(|\hat{\sigma}^2 - \sigma_+^2| \leq c_2 \big(n^{1/4}\sqrt{b_n} + \|\mbb\mu_{\ora} - \hat{\mbb\mu}\|_\infty \big)\sqrt{\frac{ \log n}{ n}}  \,\Big|\, \mathcal{D}\bigg) \geq 1- c_3 n^{-m}.
	\]
\end{lemma}
\begin{proof}
	Let us introduce for $i=1,\ldots,n$,
	\begin{align*}
		q_{i} := - \bigg(\frac{T_i (Y_i - \hat{\mu}_i)}{\hat{\pi}_i \pi_i} + \frac{(1 - T_i) (Y_i - \hat{\mu}_i)}{(1 - \hat{\pi}_i) (1 - \pi_i)}\bigg) (\hat{\pi}_i - \pi_i),
	\end{align*}
	and recall that
	\begin{align*}
		u_i &:= \bigg(\frac{r_1(X_i)}{\pi_i} + \frac{r_0(X_i)}{1 - \pi_i} - \frac{\hat{\mu}_i}{\pi_i (1 - \pi_i)}\bigg) (T_i - \pi_i) + \bigg(\frac{T_i \varepsilon_i (1)}{\pi_i} - \frac{(1 - T_i) \varepsilon_i (0)}{1 - \pi_i}\bigg) \\
		&= \frac{\mu_{\ora,i} - \hat{\mu}_i}{\pi_i(1-\pi_i)}(T_i-\pi_i) + \frac{T_i \varepsilon_i(1)}{\pi_i} - \frac{(1 - T_i) \varepsilon_i(0)}{1 - \pi_i};
	\end{align*}
	see \eqref{eq:u_i}.
	Similarly to~\eqref{eq:decomp}, we can show that for all $i=1,\ldots,n$,
	\[
	\frac{T_i (Y_i - \hat{\mu}_i)}{\hat{\pi}_i} + \frac{(1 - T_i) (Y_i - \hat{\mu}_i)}{1 - \hat{\pi}_i} = q_{i} + u_{i} + r_1(X_i) - r_0(X_i),
	\]
	so
	\begin{align*}
		\hat{\sigma}^2 = \frac{1}{n}\sum_{i=1}^n \{q_{i} + (u_{i} + r_1(X_i) - r_0(X_i) - \tau ) +  \tau - \hat{\tau}_{\dipw}\}^2
	\end{align*}
	Let us write
	\begin{align*}
		A&:= \bigg|\frac{1}{n} \sum_{i=1}^n (u_{i} + r_1(X_i) - r_0(X_i) -\tau )^2 - \sigma_+^2\bigg|\\
		B &:= \frac{1}{n} \sum_{i=1}^n (q_{i} + \tau - \hat{\tau}_\dipw)^2.
	\end{align*}
	Then we have that
	\begin{align}
		|\hat{\sigma}^2 - \sigma_+^2| &\leq A + B + \frac{2}{n}\abs{\sum_{i=1}^n (u_i + r_1(X_i) - r_0(X_i) -\tau)(q_i + \tau - \hat{\tau}_{\dipw})} \\
		&\leq A + B + 2\underbrace{\sqrt{(A + \sigma_+^2) B}}_{=:C},\label{eq:sigma1} 
	\end{align}
	using the Cauchy--Schwarz inequality in the final line.
	
	To control $A$, first recall that
	\begin{align*}
		\sigma_+^2 &= \frac{1}{n}\sum_{i=1}^n \bigg(\frac{(\mu_{\ora,i}-\hat{\mu}_i)^2}{\pi_i(1-\pi_i)} + \underbrace{\frac{\E(\varepsilon_i(1)^2 \mid X_i)}{\pi_i} + \frac{\E(\varepsilon_i(0)^2 \mid X_i)}{1 - \pi_i}}_{=:\bar{\sigma}_i^2}\bigg) + \sigma_2^2;
	\end{align*}
	see \eqref{eq:sigma_i}. Thus
	\begin{align*}
		A &\leq	\bigg| \frac{1}{n}\sum_{i=1}^n \bigg\{ \frac{(\mu_{\ora,i} - \hat{\mu}_i)^2}{\pi_i(1-\pi_i)}\bigg(\frac{(T_i-\pi_i)^2}{\pi_i(1-\pi_i)}-1\bigg) \\
		&\qquad +\bigg(\frac{T_i\varepsilon_i(1)}{\pi_i} - \frac{(1-T_i)\varepsilon_i(0)}{1-\pi_i} + r_1(X_i) - r_0(X_i) - \tau \bigg)^2 - \sigma_i^2 - \sigma_2^2   \\
		&\qquad + 2(T_i-\pi_i)\frac{\mu_{\ora,i} - \hat{\mu}_i}{\pi_i(1-\pi_i)}\bigg(\frac{T_i\varepsilon_i(1)}{\pi_i} - \frac{(1-T_i)\varepsilon_i(0)}{1-\pi_i} + r_1(X_i) - r_0(X_i) - \tau \bigg) \bigg\} \bigg| \\
		&=: |A_1 + A_2 + 2A_3|.
	\end{align*}
	
	Now conditional on $\mathcal{D}$, $A_1$ is an average of independent sub-Gaussian variables, so is itself sub-Gaussian with variance proxy
	\[
	\frac{c}{n^2}\sum_{i=1}^n \frac{(\mu_{\ora,i} - \hat{\mu}_i)^4}{\pi_i(1-\pi_i)} \leq \frac{c}{n} \|\mbb\mu_{\ora} - \hat{\mbb\mu}\|_\infty^2 \sigma_{\mu}^2,
	\]
	for a constant $c>0$. Thus for a constant $c'>0$ sufficiently large,
	\begin{equation} \label{eq:A_1}
		\pr\big(|A_1| \geq c' \|\mbb\mu_{\ora} - \hat{\mbb\mu}\|_\infty \sigma_{\mu} \sqrt{\log(n) / n)}\,|\,\mathcal{D} \big) \leq 2n^{-m}.
	\end{equation}
	From \eqref{eq:mom2_bd}, we see that, $A_2$ is an average of mean-zero independent sub-exponential random variables, so Bernstein's inequality gives us that
	\begin{equation*}
		\pr(|A_2| \geq c \sqrt{\log(n)/n} ) \leq 2n^{-2m}
	\end{equation*}
	for a constant $c>0$ sufficiently large. But then by Markov's inequality,
	\begin{equation} \label{eq:A_2}
		\pr\{\pr(|A_2| \geq c \sqrt{\log(n)/n} \,|\, \mathcal{D}) > n^{-m}\} \leq 2n^{-m}.
	\end{equation}
	
	Finally note that $|A_3| \leq \|\mbb\mu_{\ora} - \hat{\mbb\mu}\|_\infty |A_{31}|$ where
	\[
	A_{31}:=\frac{1}{n}\sum_{i=1}^n \frac{T_i - \pi_i}{\pi_i(1-\pi_i)}\bigg(\frac{T_i \varepsilon_i(1)}{\pi_i} - \frac{(1-T_i)\varepsilon_i(0)}{1-\pi_i} + r_1(X_i) - r_0(X_i) - \tau\bigg).
	\]
	By \eqref{eq:mom2_bd} the $\varepsilon_i(0), \varepsilon_i(1), r_1(X_i), r_0(X_i)$ are all sub-Gaussian, so $A_{31}$ is an averages of mean-zero sub-Gaussian random variables. 
	
	Hoeffding's inequality then gives us that there exists constant $c>0$ such that
	\[
	\pr(|A_{31}| \geq c \sqrt{\log(n)/n}) \leq 2n^{-2m},
	\]
	and similarly for $|A'_{31}|$.
	Thus by Markov's inequality
	\begin{equation} \label{eq:A_3}
		\pr\big\{\pr\big(|A_3| \geq c \|\mbb\mu_{\ora} - \hat{\mbb\mu}\|_\infty \sqrt{\log(n) / n)}\,|\,\mathcal{D} \big)>n^{-m} \big\} \leq 2n^{-m},
	\end{equation}
	and similarly for $|A'_{3}|$.
	Note that on $\Lambda_1$, $\sigma_\mu$ is bounded from above due to \eqref{eq:d2} and \eqref{eq:d3}. We thus see from \eqref{eq:A_1}, \eqref{eq:A_2}, \eqref{eq:A_3}, that on an event $\Lambda_{2A} \subseteq \Lambda_1$ with probability at least $1 - \pr(\Lambda_1) - c_1 n^{-m}$ for some constant $c_1>0$,
	\begin{equation} \label{eq:A}
		\pr\big(|A| \geq c_2 (1+\|\mbb\mu_{\ora} - \hat{\mbb\mu}\|_\infty) \sqrt{\log(n)/n} \,|\,\mathcal{D}\big) \leq c_3 n^{-m}
	\end{equation}
	for constants $c_2, c_3 >0$, and similarly for $A'$.
	
	We now turn to $B$. Working on $\Lambda_1$, we see that applying~\eqref{eq:d1} and using the fact that $\hat{\pi}_i$ and $\pi_i$'s are bounded within $[c_{\hat{\pi}}, 1 - c_{\hat{\pi}}]$ and $[c_\pi, 1 - c_\pi]$ respectively, we have that there exists a constant $c > 0$ such that
	\[
	B \leq 2 \frac{1}{n} \sum_{i=1}^n q_{i}^2 + 2 (\hat{\tau}_\dipw - \tau)^2 
	\leq c(\|{\mb Y}\|_2^2 / n + \|\hat{\mbb\mu}\|_2^2 / n) \frac{b_n \log n}{\sqrt{n}} + 2 (\hat{\tau}_\dipw - \tau)^2.
	\]
	Then using \eqref{eq:d3}, we see that on $\Lambda_1$,
	\begin{equation}\label{eq:sigma3}
		\pr\bigg( B \leq c  \frac{b_n \log n}{\sqrt{n}} + 2 (\hat{\tau}_\dipw - \tau)^2 \, \Big| \, \mathcal{D}\bigg) \leq n^{-m}.
	\end{equation}
	We are left with the task of dealing with $(\hat{\tau}_\dipw - \tau)^2 \leq 2 \delta^2 / n + 2 Q_C^2$ where recall that
	\begin{align*}
		Q_C &= \frac{1}{n} \sum_{i=1}^n \frac{\mu_{\ora,i} - \hat{\mu}_i}{\pi_i(1-\pi_i)} (T_i - \pi_i) + \frac{1}{n} \sum_{i=1}^n \bigg(\frac{T_i \varepsilon_{1,i}}{\pi_i} - \frac{(1 - T_i) \varepsilon_{0,i}}{1 - \pi_i} + r_1(X_i) - r_0(X_i) - \tau\bigg) \\
		&=: Q_{C1} + Q_{C2}.
	\end{align*}
	Now conditional on $\mathcal{D}$, $Q_{C1}$ is an average of mean zero sub-Gaussian random variables, so there exists constant $c>0$ such that 
	\begin{equation} \label{eq:QC1}
		\pr(2Q_{C1}^2 >c \sigma_\mu^2\log(n) / n  \,|\, \mathcal{D}) \leq 2n^{-m}.
	\end{equation}
	Also, $Q_{C2}$ is an average of mean zero sub-Gaussian random variables, so there exists constant $c>0$ such that $\pr( 2Q_{C2}^2 > c \log(n)/n) \leq 2n^{-2m}$. Markov's inequality then gives that
	\begin{equation} \label{eq:QC2}
		\pr\{ \pr(2Q_{C2}^2 \geq c \log(n)/n \,|\,\mathcal{D}) >n^{-m}\} \leq 2 n^{-m}. 
	\end{equation}
	Note that 
	the quantity $\delta^2/n$ is straightforward to control using~\eqref{eq:d4}, with a high probability upper bound
	\[
	(s^2 + s \log n) \frac{(\log p)^2}{n^2} \leq 2\bigg(\frac{b_n^2}{n} + \frac{a_n b_n}{\sqrt{n}}\bigg)
	\]
	using Assumptions~\ref{as:p_large} and \ref{as:sparsity}.
	Thus putting together \eqref{eq:sigma3}, \eqref{eq:QC1}, \eqref{eq:QC2} and the above, we see that on an event $\Lambda_{2B} \subseteq \Lambda_{2A}$ with probability at least $1 - \pr(\Lambda_1) - c_1 n^{-m}$ for some $c_1>0$,
	\begin{equation} \label{eq:B}
		\pr\big(|B| \geq c_2 b_n \log n /\sqrt{n} \,|\,\mathcal{D}\big) \leq c_3 n^{-m}
	\end{equation}
	for constants $c_2, c_3 >0$, where we have used that $b_n \geq \log(2)/\sqrt{n}$ (see Assumption~\ref{as:p_large}).
	
	Turning now the final term in \eqref{eq:sigma1}, recall that $\sigma_+$ is bounded from above on $\Lambda_1$. Thus from 
	\eqref{eq:A}, \eqref{eq:d5} and \eqref{eq:B}, we see that on $\Lambda_{2B}$, there exists constants $c, c'>0$ such that
	\[
	\pr\big(|C| \geq c \{\sqrt{b_n \log n} / n^{1/4}  + (1+\|\mbb\mu_{\ora} - \hat{\mbb\mu}\|_\infty) \sqrt{\log(n)/n}\}\,|\,\mathcal{D}\big) \leq c' n^{-m}.
	\]
	Putting this together with \eqref{eq:A} and  \eqref{eq:B} we have the result.
	%
	%
\end{proof}

To prove both Theorems~\ref{thm:conf_tau} and \ref{thm:ci_new}, it is helpful to first prove an intermediate result, Theorem~\ref{thm:ci}, which gives guarantees for a confidence interval for $\bar{\tau}$ in the setting of Theorem~\ref{thm:dipw} and Section~\ref{sec:basic}, that is, when auxiliary data are used.

\begin{theorem}\label{thm:ci}
	Consider the setup of Theorem~\ref{thm:dipw} and  in addition suppose that Assumptions~\ref{as:sparsity} and \ref{as:subeps} hold.
	Let
	\begin{equation} \label{eq:C_check_def}
		\check{C}_\alpha := \Big[\hat{\tau}_{\dipw} - \frac{\hat{\sigma}}{\sqrt{n}} z_\alpha, \hat{\tau}_{\dipw} + \frac{\hat{\sigma}}{\sqrt{n}} z_\alpha \Big].
	\end{equation}
	Given constants $c_\gamma, c_{\tilde{\mu}}, c_\varepsilon > 0, c_{\hat{\pi}} \in (0, \frac{1}{2}]$ and $m \in \mathbb{N}$, there exist constants $c, c_\zeta >0$ such that with probability at least
	\[
	1 - \pr(\Omega^c(c_\gamma,c_{\tilde{\mu}},c_{\hat{\pi}})) - c(n^{-m} + p^{-m}),
	\]
	we have for all $\alpha \in (0, 1]$ the coverage guarantee
	\begin{align} \label{eq:coverage}
		\pr\Big( \bar{\tau} \in \check{C}_{\alpha} \mid \mathcal{D}\Big) \geq 1 -\alpha - c_\zeta \bigg(\|\hat{\mbb\mu} - \mbb\mu_{\ora}\|_\infty  \sqrt{\frac{\log n}{n}}  + \frac{\sqrt{ b_n \log p \log n}}{n^{1/4}} + b_n +  p^{-m}\bigg).
	\end{align}
\end{theorem}
%
\begin{proof}[Proof of Theorem~\ref{thm:ci}]
	Throughout the proof we work on the event $\Lambda_2$ defined in Lemma~\ref{lem:sigma}.
	Now
	\begin{align*}
		\pr\Big( & |\hat{\tau}_\dipw - \bar{\tau}| \leq \frac{\hat{\sigma}}{\sqrt{n}} \Phi^{-1}(1 - \alpha / 2) \mid \mathcal{D}\Big) \\
		& \geq \underset{=: A}{\underbrace{\pr\bigg(|\zeta| \leq \frac{1}{\sigma_*}\Big(\hat{\sigma} \Phi^{-1}(1 - \alpha / 2) - c_\delta (s + \sqrt{s \log n}) \frac{\log p}{\sqrt{n}}\Big) \mid \mathcal{D}\bigg)}} \\
		&\qquad - \underset{=: B}{\underbrace{\pr\bigg(|\delta| \geq c_\delta (s + \sqrt{s \log n}) \frac{\log p}{\sqrt{n}} \mid \mathcal{D}\bigg)}},
	\end{align*}
	where $c_\delta$ is inherited from~\eqref{eq:d4} and $\zeta$ is as in Theorem~\ref{thm:dipw}. From~\eqref{eq:d4}, we have that
	\begin{equation}\label{eq:ci1}
		B \leq c (p^{-m} + n^{-m}).
	\end{equation}
	
	To control $A$, recalling Assumption~\ref{as:sparsity}, we have that,
	\[
	A \geq \pr\Big(|\zeta| \leq \frac{\hat{\sigma}}{\sigma_*} \Phi^{-1}(1 - \alpha / 2) - \frac{b_n}{\sigma_*} - \frac{\sqrt{ b_n \log p \log n}}{\sigma_* n^{1/4}} \mid \mathcal{D}\Big).
	\]
	Using the elementary inequality $(a - b)^2 \leq (a^2 - b^2)^2 / a^2$ for $a>0$, $b\geq 0$, we find 
	\[
	\bigg| \frac{\hat{\sigma}}{\sigma_+} - 1\bigg| \leq \frac{|\hat{\sigma}^2 - \sigma_+^2|}{\sigma_+^2}.
	\]
	Hence
	\[
	\frac{\hat{\sigma}}{\sigma_*} \geq 1 - \frac{|\hat{\sigma}^2 - \sigma_+^2|}{\sigma_+^2}.
	\]
	Then it follows from Lemmas~\ref{lem:d} and \ref{lem:sigma} that
	\begin{align*}
		A \geq & \pr\bigg(|\zeta| \leq \Phi^{-1}(1 - \alpha / 2) - c (\Phi^{-1}(1 - \alpha / 2) + 1) \bigg(\|\hat{\mbb\mu} - \mbb\mu_{\ora}\|_\infty \sqrt{\frac{\log n}{n}}
		+  b_n  + \frac{\sqrt{ b_n \log p \log n}}{n^{1/4}}\bigg) \mid \mathcal{D}\bigg) \\
		& - c' (p^{-m} + n^{-m})
	\end{align*}
	for some constants $c, c' > 0$ that do not depend on $\alpha$. Then by combining the above inequality with Lemma~\ref{lem:qc}, we have that
	\begin{align*}
		A & \geq \pr\bigg(|Z| \leq \Phi^{-1}(1 - \alpha / 2) - c (\Phi^{-1}(1 - \alpha / 2) + 1) \Big(\|\hat{\mbb\mu} - \mbb\mu_{\ora}\|_\infty \sqrt{\frac{\log n}{n}}  + b_n + \frac{\sqrt{ b_n \log p \log n}}{n^{1/4}} \Big)  \bigg) \nonumber\\
		& \quad\quad - c' (p^{-m} + n^{-m}) - \frac{c_\zeta' (\sigma_\mu^2\|\hat{\mbb\mu} - \mbb\mu_{\ora}\|_\infty + \bar{\rho}^3)}{\sigma_*^3} \frac{1}{\sqrt{n}},
	\end{align*}
	where $Z$ is a standard Gaussian random variable and $c_\zeta'$ is the same as the ``$c_\zeta$'' in Theorem~\ref{thm:dipw}(ii). Using a Taylor expansion for $\Phi$, noting that it has bounded derivative,  and using~\eqref{eq:ci1} and Lemma~\ref{lem:d}, we obtain that with $c_n := \|\hat{\mbb\mu} - \mbb\mu_{\ora}\|_\infty \sqrt{\frac{\log n}{n}}  + b_n + \frac{\sqrt{ b_n \log p \log n}}{n^{1/4}}$,
	\begin{align*}
		A & \geq \pr\{|Z| \leq \Phi^{-1}(1 - \alpha / 2) - c \Phi^{-1}(1 - \alpha / 2) c_n  \} - c_{\zeta_1} (c_n + n^{-m} + p^{-m}).
	\end{align*}
	We now consider two cases, in the first case, $c \cdot c_n \ge \frac{1}{2}$, then~\eqref{eq:coverage} holds trivially with $c_\zeta := 2 c$. In the second case, we naturally have
	\[
	\frac{1}{2}  \Phi^{-1}(1 - \alpha / 2) \le \Phi^{-1}(1 - \alpha / 2) - c \Phi^{-1}(1 - \alpha / 2) c_n \le \Phi^{-1}(1 - \alpha / 2).
	\]
	In light of this and the mean value theorem, we have that for some $\iota \in [0, 1]$,
	\begin{align*}
		A & \geq 1 - \alpha -  \frac{2}{\sqrt{2 \pi}} e^{- ((1 - \iota)z_\alpha + z_\alpha \iota / 2  )^2 / 2} z_\alpha \cdot c c_n - c_{\zeta_1} (c_n + n^{-m} +  p^{-m}) \\
		&  \ge 1 - \alpha -  \frac{2}{\sqrt{2 \pi}} e^{- z_\alpha^2 / 8} z_\alpha \cdot c c_n - c_{\zeta_1} (c_n + n^{-m} + p^{-m}),
	\end{align*}
	where recall that $z_\alpha \equiv \Phi^{-1}(1 - \alpha / 2)$. Putting things together, we prove~\eqref{eq:coverage} with
	\[
	c_\zeta := \max\left\{c_{\zeta_1} + \frac{2c}{\sqrt{2 \pi}} \cdot \max_t t e^{-t^2 / 8}, 2c\right\}.
	\]
\end{proof}

\subsection{Proof of Theorem~\ref{thm:conf_tau}}
We first introduce Lemmas~\ref{lem:popsigma} and~\ref{lem:ci_ate}, which are crucial for the proof of Theorem~\ref{thm:conf_tau}.

\begin{lemma}\label{lem:popsigma}
	Consider the set up of Theorem~\ref{thm:conf_tau} and let $m \in \mathbb{N}$ be given, there exist constants $c, c' > 0$ such that
	\[
	\pr\left(\left\{\left|\hat{\sigma}_\ipw^2 - \sigma_{\ipw}^2\right| > c \frac{\sqrt{b_n \log n}}{n^{1/4}}\right\} \cap  \Omega(c_\gamma, c_{\tilde{\mu}}, c_{\hat{\pi}}) \right) \le c' n^{-m},
	\]
	where $\sigma_{\ipw}^2 := \var\left(\frac{TY}{\pi(X)} - \frac{(1 - T)Y}{1 - \pi(X)}\right)$.
\end{lemma}

\begin{proof}
	Following analogous analysis before~\eqref{eq:sigma1}, we have the decomposition
	\begin{align*}
		|\hat{\sigma}_\ipw^2 - \sigma_\ipw^2| &\leq \underset{=: A}{\underbrace{\bigg|\frac{1}{n} \sum_{i=1}^n u_{i}^2 - \sigma_\ipw^2\bigg|}} + \underset{=: B}{\underbrace{\frac{1}{n} \sum_{i=1}^n (q_{i} + \tau - \hat{\tau}_\ipw)^2}} + 2\bigg|\frac{1}{n} \sum_{i=1}^n u_{i} (q_{i} + \tau - \hat{\tau}_\ipw)\bigg| \\
		& \le A + B + 2 \underset{=: C}{\underbrace{\sqrt{(A + \sigma_\ipw^2) B}}}
	\end{align*}
	where with a slight abuse of notation, we redefine
	\begin{align*}
		q_i := - \bigg(\frac{T_i Y_i}{\hat{\pi}_i \pi_i} + \frac{(1 - T_i) Y_i}{(1 - \hat{\pi}_i) (1 - \pi_i)}\bigg) (\hat{\pi}_i - \pi_i), \quad \text{and} \quad u_i := \frac{T_i Y_i}{\pi_i} - \frac{(1 - T_i) Y_i}{1 - \pi_i} - \tau.
	\end{align*}
	Applying Bernstein's inequality to the term $A$ yields that given any constant $m \in \mathbb{N}$, there exists a constant $c > 0$ depending on $m$ such that
	\[
	\pr(|A| \ge c \sqrt{\log n / n}) \le 2 n^{-m}.
	\]
	For $B$, we see that
	\[
	B \le \frac{2}{n} \sum_{i=1}^n (q_i - \bar{q})^2 + 2\left(\bar{u} - \tau\right)^2 \le \frac{2}{n} \sum_{i=1}^n q_i^2 + 2\left(\bar{u} - \tau\right)^2,
	\]
	where $\bar{q}, \bar{u}$ are empirical averages of $q_i$ and $u_i$, respectively. Following exactly the same analysis as the term ``$\sum_{i=1}^n q_i^2 / n$'' in the proof of Lemma~\ref{lem:sigma} and applying a Chernoff bound to control the second term, we see that 
	\[
	\pr\left(\left\{|B| \ge c \frac{b_n \log n}{\sqrt{n}}\right\} \cap \Omega(c_\gamma, c_{\tilde{\mu}}, c_{\hat{\pi}}) \right) \le 2 n^{-m}.
	\]
	In light of our control of both $A, B$, it follows from exactly the same analysis as the term ``$C$'' in the proof of Lemma~\ref{lem:sigma} that
	\[
	\pr\left(\left\{|C| \ge c \frac{\sqrt{b_n \log n}}{n^{1/4}}\right\} \cap \Omega(c_\gamma, c_{\tilde{\mu}}, c_{\hat{\pi}}) \right) \le 2 n^{-m}.
	\]
	Putting together, we obtain the desired result.
\end{proof}

\begin{lemma}\label{lem:ci_ate}
	Consider the set up of Theorem~\ref{thm:conf_tau}, we have there exists a constant $c_\zeta > 0$ such that for all $\alpha \in (0, 1]$,
	\[
	\pr\left(\left|\bar\tau - \tau \right| > \frac{\hat{\sigma}_\ipw}{\sqrt{n}} z_{\alpha / 2}, \Omega(c_\gamma, c_{\tilde{\mu}}, c_{\hat{\pi}}) \right) \le \alpha / 2 + c_\zeta \frac{\sqrt{ b_n \log n} + 1}{n^{1/4}}.
	\]
\end{lemma}

\begin{proof}
	Write $\rho_2^3 := \E[|r_1(X) - r_0(X) - \tau|^3]$. 
	We prove this result by considering two cases.
	
	Case I: $\sigma_2^2 \le \sigma_\ipw^2 n^{-1/4}$. Then we have for all $\alpha \in (0, 1]$, by Chebyshev's inequality,
	\[
	\pr\left(\left|\bar\tau - \tau \right| \ge \frac{\sigma_\ipw}{\sqrt{n}} z_{\alpha / 2} \right) \le \frac{\sigma_2^2}{\sigma_\ipw^2 z_{0.5}} \le \frac{1}{ z_{0.5} n^{1/4}},
	\]
	where recall that $z_{0.5} = \Phi(0.75) \approx 0.674$ is a universal constant.
	
	Case II: $\sigma_2^2 > \sigma_\ipw^2 n^{-1/4}$. We have
	\begin{align*}
		\rho_2^3 \; & =\E[|r_1(X) - r_0(X) - \tau|^3 \kron(|r_1(X) - r_0(X) - \tau| < 1)]  \\
		&\quad + \E[|r_1(X) - r_0(X) - \tau|^3 \kron(|r_1(X) - r_0(X) - \tau| \ge 1)] \\
		& \le \sigma_2^2 + \sqrt{\E[|r_1(X) - r_0(X) - \tau|^6]} \sqrt{\pr(|r_1(X) - r_0(X) - \tau| \ge 1)},
	\end{align*}
	where for the last inequality we apply the Cauchy--Schwarz inequality. 
	
	By Jensen's inequality we have $\E[|r_t(X) - \E[r_t(X)]|^6] \le \E[|Y(t) - \E[Y(t)]|^6] \le c'$ for $t \in \{0, 1\}$, where for the last inequality we apply Condition~\ref{as:subgY}. In light of this and using Chebyshev's inequality to control $\pr(|r_1(X) - r_0(X) - \tau| \ge 1)$, we further have that for some constant $c > 0$, $
	\rho_2^3 \le \sigma_2^2 + c \sigma_2$.
	
	Armed with the above, as a direct consequence of the Berry--Esseen theorem~\citep{Esseen42}, we have that for some universal constant $c_\zeta$,
	\begin{equation}\label{eq:berryessen}
		\sup_{t \in \R} \left|
		\pr\left(\frac{\sqrt{n}(\bar\tau - \tau)}{\sigma_2} \le t\right) - \Phi(t) \right| \le c_\zeta \frac{ \rho_2^3}{\sqrt{n}\sigma_2^3} \le c_\zeta\frac{\sigma_2^2 + c \sigma_2}{\sqrt{n} \sigma_2^3} \le \frac{c_\zeta'}{n^{1/4}},
	\end{equation}
	where for the last inequality we apply that $\sigma_\ipw^2$ is bounded below by some constant since $\sigma_\ipw^2 \ge \var\left(\frac{T \varepsilon(1)}{\pi(X)} + \frac{(1 - T) \varepsilon(0)}{1 - \pi(X)}\right)$ and we are under Condition~\ref{as:subeps}.
	
	In light of~\eqref{eq:berryessen}, Lemma~\ref{lem:popsigma} and $\sigma_\ipw^2 \ge \sigma_2^2$, it follows from exactly the same analysis as the term ``$A$'' defined in the proof of Theorem~\ref{thm:ci} that in Case II, 
	\begin{align*}
		& \pr\left(\left|\bar\tau - \tau \right| > \frac{\hat{\sigma}_\ipw}{\sqrt{n}} z_{\alpha / 2}, \Omega(c_\gamma, c_{\tilde{\mu}}, c_{\hat{\pi}}) \right) = \pr\left(\left|\frac{\sqrt{n}(\bar\tau - \tau)}{\sigma_2} \right| > \frac{\hat{\sigma}_\ipw}{\sigma_2} z_{\alpha / 2}, \Omega(c_\gamma, c_{\tilde{\mu}}, c_{\hat{\pi}}) \right) \\
		& \le \pr\left(\left|\frac{\sqrt{n}(\bar\tau - \tau)}{\sigma_2} \right| > \left(1 - c\frac{\sqrt{b_n \log n}}{n^{1/4}}\right) z_{\alpha / 2}\right) \le \alpha / 2 + c_\zeta \left(\frac{\sqrt{b_n \log n}}{n^{1/4}} + \frac{1}{n^{1/4}}\right).
	\end{align*}
	\[
	\]
	
	From our analysis of both Case I and II, we prove the desired result.
\end{proof}

\begin{proof}[Proof of Theorem~\ref{thm:conf_tau}]
	Let $\Lambda_2$ be the $\mathcal{D}$-measurable event on which the result \eqref{eq:coverage} of Theorem~\ref{thm:ci} holds and recall that $\Lambda_2 \subseteq \Omega(c_{\gamma}, c_{\tilde{\mu}}, c_{\hat{\pi}})$ and $\pr(\Lambda_2) \geq \pr(\Omega(c_{\gamma}, c_{\tilde{\mu}}, c_{\hat{\pi}})) - c(n^{-m}+p^{-m})$ for a constant $c>0$. 
	Recalling the definition of $\check{C}_{\alpha}$ \eqref{eq:C_check_def}, by a union bound,
	\begin{align*}
		&\pr\left( |\hat{\tau}_{\dipw} - \tau| \leq \frac{(\hat{\sigma} + \hat{\sigma}_\ipw)z_{\alpha/2}}{\sqrt{n}} \right) \geq \pr\left( |\hat{\tau}_{\dipw} - \bar{\tau}| \leq \frac{\hat{\sigma}z_{\alpha/2}}{\sqrt{n}}, \, \, |\bar{\tau} - \tau| \leq  \frac{\hat{\sigma}_\ipw z_{\alpha/2}}{\sqrt{n}}, \,\,\Lambda_2 \right)\\
		&\qquad \geq \pr\left( |\hat{\tau}_{\dipw} - \bar{\tau}| \leq \frac{\hat{\sigma}z_{\alpha/2}}{\sqrt{n}}, \,\, \Lambda_2 \right) - \pr\left( |\bar{\tau} - \tau| > \frac{\hat{\sigma}_\ipw z_{\alpha/2}}{\sqrt{n}}, \,\, \Lambda_2 \right) \\
		&\qquad = \E \{\pr(\bar{\tau} \in \check{C}_{\alpha/2} \,|\, \mathcal{D}) \ind_{\Lambda_2}\} - \pr\left( |\bar{\tau} - \tau| > \frac{\hat{\sigma}_\ipw z_{\alpha/2}}{\sqrt{n}}, \, \, \Lambda_2 \right)\\
		&\qquad \geq \E \{\pr(\bar{\tau} \in \check{C}_{\alpha/2} \,|\, \mathcal{D}) \ind_{\Omega(c_{\gamma}, c_{\tilde{\mu}}, c_{\hat{\pi}}) }\} - \pr\left( |\bar{\tau} - \tau| > \frac{\hat{\sigma}_\ipw z_{\alpha/2}}{\sqrt{n}}, \, \, \Lambda_2 \right) \\
		&\qquad\qquad - c(n^{-m}+p^{-m}).
	\end{align*}
	Theorem~\ref{thm:ci} provides the following lower bound on the first term:
	\begin{align*}
		& 1 -\frac{\alpha}{2} - c_\zeta \bigg(\E\{\|\hat{\mbb\mu} - \mbb\mu_{\ora}\|_\infty\ind_{\Omega(c_{\gamma}, c_{\tilde{\mu}}, c_{\hat{\pi}}) }\}   \sqrt{\frac{\log n}{n}}  + \frac{\sqrt{ b_n \log p \log n}}{n^{1/4}} + b_n +  p^{-m}\bigg)\\
		&\quad - \pr(\Omega^c(c_\gamma, c_{\tilde{\mu}}, c_{\hat{\pi}})).
	\end{align*}
	For the second term, we apply Lemma~\ref{lem:ci_ate}. Putting things together, we obtain the desired result.
\end{proof}

\subsection{Proof of Theorem~\ref{thm:ci_new}}

We are now in a position to prove Theorem~\ref{thm:ci_new} giving guarantees for the confidence interval \eqref{eq:ci_new} centred on the cross-fit estimate $\hat{\tau}_{\ave}$ \eqref{eq:cross-fit}. Recall that in contrast to Theorem~\ref{thm:ci}, here we do not assume the existence of auxiliary data.
\begin{proof}[Proof of Theorem~\ref{thm:ci_new}]
	From Theorem~\ref{thm:ci} and a union bound, we have that there exist a constants $c, c_\zeta > 0$ such that on an event $\Lambda$ with probability at least
	\[
	1 - \sum_{j=1}^3 \pr(\Omega_j^c(c_\gamma,c_{\tilde{\mu}},c_{\hat{\pi}})) - c(n^{-m} + p^{-m}),
	\]
	we have
	\begin{equation*}
		\begin{split} 
			& \pr\Big( |\hat{\tau}_{\dipw, j} -  \bar{\tau}_j | \leq \frac{\hat{\sigma}_j}{\sqrt{n/3}} z_{\alpha/3} \,|\, \mb X \Big) \geq 1 -\alpha/3  \\
			& \qquad - \frac{c_\zeta}{3} \bigg(\E(\|\hat{\mbb\mu} - \mbb\mu_{\ora}\|_\infty \,|\, \mb X) \sqrt{\frac{\log n}{n}}  + \frac{\sqrt{ b_n \log p \log n}}{n^{1/4}} + b_n +  p^{-m}\bigg)
		\end{split}
	\end{equation*}
	for all $j=1,2,3$. Thus on $\Lambda$, by a union bound, we have
	\begin{equation*}
		\begin{split} 
			& \pr\Big( |\hat{\tau}_{\dipw, j} -  \bar{\tau}_j | \leq \frac{\hat{\sigma}_j}{\sqrt{n/3}} z_{\alpha / 3} \; \forall \; j \; \,|\, \mb X \Big) \geq 1 -\alpha  \\
			& \qquad - c_\zeta\bigg(\E(\|\hat{\mbb\mu} - \mbb\mu_{\ora}\|_\infty \,|\, \mb X) \sqrt{\frac{\log n}{n}}  + \frac{\sqrt{ b_n \log p \log n}}{n^{1/4}} + b_n +  p^{-m}\bigg).
		\end{split}
	\end{equation*}
	Averaging the inequalities
	\[
	-\frac{\hat{\sigma}_j}{\sqrt{n/3}} z_{\alpha / 3} \leq \hat{\tau}_{\dipw, j} -  \bar{\tau}_j \leq \frac{\hat{\sigma}_j}{\sqrt{n/3}} z_{\alpha / 3}
	\]
	for $j=1,2,3$, we obtain
	\[
	-\frac{\tilde{\sigma}}{\sqrt{n}} z_{\alpha / 3} \leq \hat{\tau}_{\ave} -  \bar{\tau} \leq \frac{\tilde{\sigma}}{\sqrt{n}} z_{\alpha / 3}.
	\]
	Hence on $\Lambda$, we have
	\begin{equation*}
		\begin{split} 
			& \pr\Big( |\hat{\tau}_{\ave} -  \bar{\tau} | \leq \frac{\tilde{\sigma}}{\sqrt{n}} z_{\alpha / 3} \,|\, \mb X \Big) \geq 1 -\alpha  \\
			& \qquad - c_\zeta\bigg(\E(\|\hat{\mbb\mu} - \mbb\mu_{\ora}\|_\infty \,|\, \mb X) \sqrt{\frac{\log n}{n}}  + \frac{\sqrt{ b_n \log p \log n}}{n^{1/4}} + b_n +  p^{-m}\bigg).
		\end{split}
	\end{equation*}
	as required.
\end{proof}

\section{Efficiency in estimating $\bar{\tau}$}\label{sec:efficiency_appendix}

\newrev{
	In Section~\ref{sec:efficiency}, we considered the efficiency of $\hat{\tau}_{\ave}$ for estimating $\tau$;  in this section, we show that a similar conclusion applies to estimation of $\bar{\tau}$ with $\hat{\tau}_{\ave}$.
	
	\begin{corollary}\label{cor:sampleeff}
		Consider the setup of Theorem~\ref{thm:efficiency}. We have the decomposition $\sqrt{n} (\hat{\tau}_\ave - \bar\tau) = \delta + \bar\sigma \zeta$, where $\delta$ is as in Theorem~\ref{thm:efficiency} (i), and 
		\[
		\sup_{t \in \R} |\pr(\zeta \leq t \mid \mb X) - \Phi(t)| \leq \frac{c_\zeta}{\sqrt{n}} \frac{\bar\rho^3}{\bar\sigma^3}.
		\]
	\end{corollary}
	Under the assumptions of Corollary~\ref{cor:efficiency}, we may also argue
	that provided there exists some $e_n \to \infty$ such that $\|\hat{r}_t(\mbb X) - r_t(\mbb X)\|_2 / \sqrt{n} = o_\pr(e_n^{-1})$ for $t = 0, 1$, then with probability converging to $1$, $\sqrt{n}(\hat{\tau}_\ave - \bar{\tau}) \mid {\mbb X} \approx \mathcal{N}(0, \bar{\sigma}^2)$.
	The quantity $\bar{\sigma}^2$ is the same as the conditional variance that the AIPW estimator would achieve, when supplied with ground truth nuisance functions.
}

\section{Proofs of Theorem~\ref{thm:efficiency} and Corollaries~\ref{cor:efficiency} and~\ref{cor:sampleeff}}

Theorem~\ref{thm:efficiency} and Corollary~\ref{cor:sampleeff} follow from the following variant of Theorem~\ref{thm:dipw} followed by applications of the Berry--Eseen theorem (as in the argument for  Theorem~\ref{thm:split}).

\begin{theorem}\label{thm:dipweff}
	\newrev{	Consider the setup of Theorem \ref{thm:dipw}. We have the following decompositions: 
		\begin{align*}
			\sqrt{n}(\hat{\tau}_\dipw - \tau) &= \delta + \frac{1}{\sqrt{n}} \sum_{i=1}^n \left(r_1(X_i) - r_0(X_i) - \tau + \frac{T_i \varepsilon_i(1)}{\pi_i} - \frac{(1 - T_i) \varepsilon_i(0)}{1 - \pi_i}\right),\\
			\sqrt{n}(\hat{\tau}_\dipw - \bar\tau) &= \delta + \frac{1}{\sqrt{n}} \sum_{i=1}^n \left(\frac{T_i \varepsilon_i(1)}{\pi_i} - \frac{(1 - T_i) \varepsilon_i(0)}{1 - \pi_i}\right).
		\end{align*}
		Here $\delta$ is such that given constants $c_\gamma, c_{\tilde{\mu}} > 0, c_{\hat{\pi}} \in (0, 1/2]$ and $m \in \mathbb{N}$, there exists a constant $c_\delta > 0$ such that given any sequence $(e_n)_{n = 1}^\infty$, with probability at least $1 - \pr(\Omega^c(c_\gamma, c_{\tilde{\mu}}, c_{\hat{\pi}})) - c(n^{-m} + p^{-m}) - 2 e^{-e_n^2}$, 
		\[
		|\delta| \le c_\delta(s + \sqrt{s \log n}) \frac{\log p}{\sqrt{n}} + c_\delta \frac{e_n}{\sqrt{n}} \|\tilde{\mu}({\mbb X}) - {\mbb \mu_\ora}\|_2.
		\]}
\end{theorem}

To prove this, first let $\hat{\mbb \mu}_\ora$ be an $n$-dimensional vector such that
\[
\hat{\mu}_{\ora, i} := \frac{\pi_i(1 - \hat{\pi}_i)}{\hat{\pi}_i} r_1(X_i) + \frac{(1 - \pi_i) \hat{\pi}_i}{1 - \hat{\pi}_i} r_0(X_i).
\]
\begin{lemma}\label{lem:muora}
	Consider the set up of Theorem~\ref{thm:dipweff}, we have that there exist constants $c_1, c_2 > 0$ such that on an event $\Lambda$ contained in $\Omega := \Omega(c_\gamma, c_{\tilde{\mu}}, c_{\hat{\pi}})$ with $\pr(\Lambda^c \cap \Omega) \le c_1n^{-m}$,
	\[
	\frac{1}{\sqrt{n}} \|{\mbb \mu}_\ora - \hat{\mbb \mu}_\ora\|_2 \le c_2 \sqrt{s \frac{\log p}{n}}.
	\]
\end{lemma}

\begin{proof}
	We have that
	\[
	\frac{1}{\sqrt{n}}\|\hat{\mbb \mu}_\ora - {\mbb \mu}_\ora\|_2 \le \sqrt{ \frac{1}{n} \sum_{i = 1}^n \left(\frac{\pi_i(1 - \hat{\pi}_i)}{\hat{\pi}_i} - (1 - \pi_i)\right)^2 r_{1, i}^2} +  \sqrt{\frac{1}{n} \sum_{i = 1}^n \left(\frac{(1 - \pi_i)\hat{\pi}_i}{1 - \hat{\pi}_i} - \pi_i\right)^2 r_{0, i}^2}.
	\]
	Using that we are on $\Omega$ and following~\eqref{eq:qb0}, we have that for some constant $c$,
	\[
	\frac{1}{\sqrt{n}}\|\hat{\mbb \mu}_\ora - {\mbb \mu}_\ora\|_2 \le c \sqrt{ \frac{1}{n} \sum_{i = 1}^n (X_i^\top(\hat{\gamma} - \gamma))^2 r_{1, i}^2} + c \sqrt{ \frac{1}{n} \sum_{i = 1}^n (X_i^\top(\hat{\gamma} - \gamma))^2 r_{0, i}^2}.
	\]
	Applying the Cauchy--Schwarz inequality, we have
	\begin{align*}
		& \frac{1}{\sqrt{n}}\|\hat{\mbb \mu}_\ora - {\mbb \mu}_\ora\|_2 \le c \left(\frac{1}{n} \sum_{i = 1}^n (X_i^\top(\hat{\gamma} - \gamma))^4\right)^{\frac{1}{4}} \left\{\left(\frac{1}{n} \sum_{i = 1}^n r_{1, i}^4\right)^{\frac{1}{4}} + \left(\frac{1}{n} \sum_{i = 1}^n r_{0, i}^4\right)^{\frac{1}{4}}\right\} \\
		& \le 2c \left(\frac{1}{n} \sum_{i = 1}^n (X_i^\top(\hat{\gamma} - \gamma))^4\right)^{\frac{1}{4}} \cdot \max_{t \in \{0, 1\}}\left(\frac{1}{n} \sum_{i = 1}^n (r_{t, i} - \E[Y(t)])^4 + \E[Y(t)]^4\right)^{\frac{1}{4}}.
	\end{align*}
	In the final line, we applied Jensen's inequality which gives that for any $m' \in \mathbb{N}$, $\E[|r_t(X) - \E[Y(t)]|^{m'}] \le \E[|Y(t) - \E[Y(t)]|^{m'}]$.
	By the argument of Lemma~\ref{lem:d},
	we have that there exist constants $c', c'' > 0$ such that with probability at least $1 - c''n^{-m}$,
	\[
	\max_{t \in \{0, 1\}} \frac{1}{n} \sum_{i = 1}^n (r_{t, i} - \E[Y(t)])^4 \le c'.
	\]
	Thus taking $\Lambda$ to be the intersection of the event that the above holds, $\mathcal{T}_2$ defined in Lemma~\ref{lem:t2} and $\Omega$, we obtain the desired result.
\end{proof}


\begin{proof}[Proof of Theorem~\ref{thm:dipweff}]
	\newrev{First note that without loss of generality, we may assume that $e_n \leq \sqrt{\log n}$ by absorbing the contribution of $2e^{-e_n^2}$ in the probability into the term $cn^{-m}$.}
	Recalling the decomposition in~\eqref{eq:splitdecomp}, it remains to prove that there exist constants $c, c' > 0$ such that on an event $\mathcal{A}_1$ with $\pr(\mathcal{A}_1^c \cap \Omega(c_\gamma, c_{\tilde{\mu}}, c_{\hat{\pi}})) \le c' (n^{-m} + p^{-m}) + 2 e^{-e_n^2}$, the term
	\[
	\delta_2 := \frac{1}{\sqrt{n}} \sum_{i=1}^n \bigg(\frac{r_{1,i}}{\pi_i} + \frac{r_{0,i}}{1 - \pi_i} - \frac{\hat{\mu}_i}{\pi_i (1 - \pi_i)}\bigg) (T_i - \pi_i) = \frac{1}{\sqrt{n}} \sum_{i=1}^n \bigg(\frac{\mu_{\ora, i} - \hat{\mu}_i}{\pi_i (1 - \pi_i)}\bigg) (T_i - \pi_i)
	\]
	satisfies
	\newrev{
		\begin{equation}\label{eq:biaseff}
			|\delta_2| \le \frac{c e_n}{\sqrt{n}} \left(\sqrt{s \log p } + \|\tilde{\mu}({\mbb X}) - {\mbb \mu}_\ora\|_2\right).
		\end{equation}
		(Recall that $e_n \leq \sqrt{\log n}$.)} 
	
	To achieve this goal, we first control $\|\hat{\mbb \mu} - \tilde{\mu}({\mbb X})\|_2$. To this end, let $\check{\mu}_{\ora, i} := \hat{\mu}_{\ora, i} \mathbbm{1}_{\{c_{\hat{\pi}} \le \hat{\pi}_i \le 1 - c_{\hat{\pi}}\}}$. We now argue similarly to the proof of Lemma~\ref{lem:feasibility} and furthermore use the definitions of $\tilde{\Omega}_1$ and $\tilde{Y}_{A,i}$ therein.
	
	Then conditional on $\mathcal{D}_B$ and working on $\tilde{\Omega}_1$, the
	\begin{equation} \label{eq:mu_check_optim}
		X_{A,ij}\{\check{Y}_{A, i} - \tilde{\mu}(X_{A, i})\} - X_{ij}\{\check{\mu}_{\ora, i} - \tilde{\mu}(X_i)\}
	\end{equation}
	are i.i.d. mean-zero random variables; indeed, note that for example
	\begin{align*}
		\E\left(\frac{T_iY_i(1-\hat{\pi}_i)}{\hat{\pi}_i} \mathbbm{1}_{\{c_{\hat{\pi}} \le \hat{\pi}_i \le 1 - c_{\hat{\pi}}\}} \, \Big| \, \mathcal{D}_B, X_i\right) = \frac{(1-\hat{\pi}_i) \mathbbm{1}_{\{c_{\hat{\pi}} \le \hat{\pi}_i \le 1 - c_{\hat{\pi}}\}}}{\hat{\pi}_i} \pi_i r_1(X_i).
	\end{align*}
	Furthermore, conditional $\mathcal{D}_B$ and working on $\tilde{\Omega}_1$, \eqref{eq:mu_check_optim} is a linear combination of products of sub-Gaussian random variables with finite conditional means as, for example,  $r_t(X_i)$ is conditionally (and unconditionally) sub-Gaussian (see e.g. proof of Lemma~\ref{lem:d})).
	Thus arguing similarly as in the proof of Lemma~\ref{lem:feasibility}, we have that by choosing the constant $c_\eta > 0$ sufficiently large, there exists a constant $c > 0$ such that on an event $\mathcal{A}_2$ satisfying  $\pr(\mathcal{A}_2^c \cap \Omega(c_\gamma, c_{\tilde{\mu}}, c_{\hat{\pi}})) \le c p^{-m}$, $\hat{\mbb \mu}_\ora$ is a feasible solution to~\eqref{eq:mu_hat}, so that
	\[
	\|\hat{\mbb \mu} - \tilde{\mu}({\mbb X})\|_2 \le \|\hat{\mbb \mu}_\ora - \tilde{\mu}({\mbb X})\|_2 \le 
	\|{\mbb \mu}_\ora - \tilde{\mu}({\mbb X})\|_2 + \|{\mbb \mu}_\ora - \hat{\mbb \mu}_\ora\|_2.
	\]
	Now
	\begin{align}
		\|\hat{\mbb \mu} - {\mbb \mu}_\ora\|_2 & \le \|\hat{\mbb \mu} - \tilde{\mu}({\mbb X})\|_2 +  
		\|{\mbb \mu}_\ora - \tilde{\mu}({\mbb X})\|_2 \label{eq:horamu}\\
		& \le \|{\mbb \mu}_\ora - \hat{\mbb \mu}_\ora\|_2 +  
		2 \|{\mbb \mu}_\ora - \tilde{\mu}({\mbb X})\|_2, \notag
	\end{align}
	using the above.
	
	\newrev{On $\mathcal{A}_2$ then, we have
		\[
		|\delta_2| \leq \frac{1}{\sqrt{n}} \sum_{i=1}^n a_i (T_i - \pi_i)
		\]
		for some $\mathcal{D}$-measurable random vector $\mbb a$ such that $c_\pi(1-c_{\pi}) \|\mbb a\|_2 \leq  \|{\mbb \mu}_\ora - \hat{\mbb \mu}_\ora\|_2 +  
		2 \|{\mbb \mu}_\ora - \tilde{\mu}({\mbb X})\|_2$. Thus by Hoeffding's inequality, for some constant $c>0$,
		\begin{align*}
			\pr\left(|\delta_2| \le \frac{c e_n}{\sqrt{n}} \left( \|{\mbb \mu}_\ora - \hat{\mbb \mu}_\ora\|_2 +  \|{\mbb \mu}_\ora - \tilde{\mu}({\mbb X})\|_2\right) \;\Big|\; \mathcal{D} \right) \leq 2 e^{-e_n^2}.
		\end{align*}
		Combining the above with Lemma~\ref{lem:muora}, we obtain~\eqref{eq:biaseff}, thereby proving the desired result.}
\end{proof}

\begin{proof}[Proof of Corollary~\ref{cor:efficiency}]
	First note that
	\[
	(s + \sqrt{s \log n}) \frac{\log p}{\sqrt{n}} \leq b_n + \sqrt{s}\frac{\sqrt{\log p}}{n^{1/4}} \frac{\sqrt{\log n \log p}}{n^{1/4}} \leq (1+\sqrt{c})\sqrt{b_n}.
	\]
	
	Using the fact that for $a,b, \hat{a}, \hat{b} \in \R$, $ab - \hat{a}\hat{b} = (a-\hat{a})b + \hat{a} (b-\hat{b})$, we have 
	\begin{align*}
		|\tilde{\mu}(X_i) - \mu_{\ora}(X_i)| \leq |\pi_i - \hat{\pi}_i| \big( |r_{1}(X_i)| + |r_{0}(X_i)|\big) + | \hat{r}_{1}(X_i) - r_{1}(X_i)| + | \hat{r}_{0}(X_i) - r_{0}(X_i)|.
	\end{align*}
	Now from Assumption~\ref{as:subgY} and Jensen's inequality, each $r_{t}(X_i) - \E r_t(X_i) $ is sub-Gaussian for each $t=0,1$. Thus by a similar argument to the proof of Lemma~\ref{lem:t4}, there exists a constant $c_1$ such that on an event with probability at least $2n^{-m}$, we have
	\begin{equation} \label{eq:r_bd}
		\frac{1}{n} \sum_{i=1}^n(r_{1}^2(X_i) + r_{0}^2(X_i)) \leq c_1.
	\end{equation}
	Now from \eqref{eq:d1}, we have that on an event with probability at least
	\[
	1 - \sum_{j=1}^3 \pr(\Omega_j^c(c_\gamma, c_{\tilde{\mu}}, c_{\hat{\pi}})) - c_1(n^{-m} + p^{-m}),
	\]
	for some constant $c_1$, both \eqref{eq:r_bd} holds and
	\[
	\max_{i=1,\ldots,n} |\pi_i - \hat{\pi}_i| \leq c_1\frac{\sqrt{b_n \log n}}{n^{1/4}}.
	\]
	Thus on this event, there exists a constant $c_2$ such that
	\begin{equation}\label{eq:pirbnd}
		\sqrt{\sum_{i = 1}^n (\hat{\pi}_i - \pi_i)^2 (|r_1(X_i)| + |r_0(X_i)|)^2} \le c_2  n^{1/4} \sqrt{b_n \log n}.
	\end{equation}
	Putting things together, we have the required bound on $|\delta|$.
\end{proof}

\section{A lower bound for treatment effect estimation}\label{sec:lower}

\newrev{
	In this section, we study the minimax optimatility of average treatment effect estimation when the outcome regression model is a generalised linear model with link function $\psi$ and the propensity model follows a generalised linear model with link function
	\[
	\phi(u) := \frac{1 + \exp(u)}{1 + 2 \exp(u)}.
	\]
	We consider a setting where $Y$ is binary with $Y(0)$ conditional on $(X, T)$ a Bernoulli random variable with parameter $0.5$; $Y(1)$ and $T$ follow
	\[
	\pr(Y(1) = 1 \mid X) = \psi(\beta^\top X) \quad\text{and}\quad \pr(T = 1 \mid X) = \phi(\gamma^\top X);
	\]
	and $X$ follows a continuous distribution with density $f(x) := C_\phi\frac{p(x; \Sigma)}{\phi(\gamma^\top x)}$, where $p(x; \Sigma)$ represents a density function of mean zero multivariate normal random vector with covariance $\Sigma \in \R^{p \times p}$ with $C_\phi$ is a normalising constant. Moreover, one can derive that $C_\phi = \frac{2}{3}$, since
	\[
	\int C_\phi\frac{p(x; \Sigma)}{\phi(\gamma^\top x)} dx = C_\phi \int p(x; \Sigma) \left(1 + \psi(\gamma^\top x)\right) dx \overset{(a)}{=} \frac{3}{2} C_\phi = 1,
	\] 
	where for $(a)$ we use that $\gamma^\top X$ is a mean zero Gaussian random variable, and so the density of $\psi(\gamma^{\top} X)$ can be shown to be symmetric around $1/2$. With this construction, it is straightforward to see that the distribution of $X$ follows $p(x; \Sigma)$ conditional on $T = 1$. Moreover
	\[
	\frac{C_\phi p(x; \Sigma) (1 - \phi(\gamma^\top x))}{(1 - C_\phi)\phi(\gamma^\top x)} = \frac{C_\phi}{1 - C_\phi} f(x) \psi(\gamma^\top x) = 2 f(x) \psi(\gamma^\top x)
	\]
	is the conditional distribution of $X$ given $T=0$.
	
	We see that then the joint distribution of $(X, T, Y)$ can be parameterised by $\theta := (\beta, \gamma, \Sigma) \in \R^p \times \R^p \times \R^{p \times p}$. Given a parameter space $\Theta$ for $\theta$, we define $\mathcal{I}_\alpha(\Theta)$ as the set of $(1 - \alpha)$-level confidence intervals for $\tau$ over $\Theta$:
	\[
	\mathcal{I}_\alpha(\Theta) := \{\ci_\alpha(\D): \forall \theta \in \Theta, \pr_\theta^{\otimes n}(\tau_\theta \in \ci_\alpha(\D)) \ge 1 - \alpha\},
	\]
	where $\pr_\theta^{\otimes n}$ means that the data in $\D := (X_i, Y_i, T_i)_{i = 1}^n$ is i.i.d. generated by the distribution parameterised by $\theta$, and $\tau_\theta := \E_\theta[Y(1) - Y(0)]$. Given $s > 0, M > 1$, let
	\[
	\Theta(s, M) := \{(\beta, \gamma, \Sigma): \|\beta\|_0, \|\gamma\|_0 \le s; \|\beta\|_2, \|\gamma\|_2 \le M; M^{-1} \le \lambda_{\min}(\Sigma) \le \lambda_{\max}(\Sigma) \le M\}.
	\]
	We have the following proposition.
	
	\begin{prop}\label{prop:lower}
		Let $M > 1$ be given. Suppose that $s \le \min\{p^c, \frac{n}{\log p}\}$ for some constant $0 \le c < 1 / 2$. Then for any $0 < \alpha < \frac{1}{2}$ there exists a constant $c_1 > 0$ depending on $M,c$ and $\alpha$ such that
		\[
		\inf_{\ci_\alpha(\D) \in \mathcal{I}_\alpha(\Theta(s, M))} \sup_{\theta \in \Theta(s, M)} \E_\theta^{\otimes n} L(\ci_\alpha(\D)) \ge c_1 \left(s \frac{\log p}{n} + \frac{1}{\sqrt{n}}\right),
		\]
		where $L(\cdot)$ is the length of the confidence interval, and $\E_\theta^{\otimes n}$ represents the expectation over the distribution $\pr_\theta^{\otimes n}$.
	\end{prop}
	
	Informally then, Proposition~\ref{prop:lower} claims that the $n \gtrsim s^2 / \log p$ is essentially a minimal sample complexity for an ATE estimator to be $\sqrt{n}$-consistent. It also shows that the upper bound presented in Theorem~\ref{thm:dipw} is minimax optimal even for $s \gg \sqrt{n} / \log p$. Proposition~\ref{prop:lower} is a direct consequence of the following two propositions.
	
	\begin{prop}\label{prop:lower1}
		Consider the set up of Proposition~\ref{prop:lower}, we have that for any $0 < \alpha < \frac{1}{2}$ there exists some constant $c_1 > 0$ depending on $M, c$ and $\alpha$ such that
		\[
		\inf_{\ci_\alpha(\D) \in \mathcal{I}_\alpha(\Theta(s, M))} \sup_{\theta \in \Theta(s, M)} \E_\theta^{\otimes n} L(\ci_\alpha(\D)) \ge c_1 s \frac{\log p}{n}.
		\]
	\end{prop}
	
	\begin{prop}\label{prop:lower2}
		Same as Proposition~\ref{prop:lower1}, but with $s \frac{\log p}{n}$ replaced by $\frac{1}{\sqrt{n}}$.
	\end{prop}
}
\newrev{
	
	\subsection{Proof of Proposition~\ref{prop:lower1}}
	
	\begin{proof}
		Let $\lambda^2 := \frac{1}{c_{\lambda_1} n} \log\left(1 + \frac{p - 1}{c_{\lambda_2} (s - 1)^2}\right)$ for some constants $c_{\lambda_1}, c_{\lambda_2} > 0$ to be chosen. Given any $\gamma \in \R^p$, define $\Omega_0, \Omega_{1,\gamma} \in \R^{p \times p}$ by
		\[
		\Omega_0 := \left(\begin{matrix}
			1 - (s - 1) \lambda^2 & 0\\
			0 & I
		\end{matrix}\right) \quad \text{ and }\quad \Omega_{1,\gamma} := \left(\begin{matrix}
			1 & \gamma_{-1}^\top\\
			\gamma_{-1} & I
		\end{matrix}\right).
		\]
		Let $[p] := \{1, \ldots, p\}$ and let $\mathcal{A}_{s - 1}$ be the set of $(s-1)$-subsets of $[p] \setminus \{1\}$.
		Let 
		\[
		\Gamma_{s - 1} := \left\{\gamma \in \R^p: \gamma = (s - 1)\lambda^2 e_1 + \sum_{j \in A_{s - 1}} \lambda e_{j},\,  A_{s - 1} \in \mathcal{A}_{s - 1}\right\}
		\]
		where $e_{j}, j = 1, \ldots, p$ are the standard basis vectors. We define two sets of hypotheses $\mathcal{H}_0 := \{H_0\}$, $\mathcal{H}_1 := \{H_{1, \gamma}: \gamma \in \Gamma_{s - 1}\}$ where $H_0$ and $H_{1, \gamma}$ are defined by:
		\begin{equation}\label{eq:hypothesis}
			\begin{split}
				H_{0}: &\; \theta = ((0, \ldots, 0)^\top, (0, \ldots, 0)^\top, \Omega_0^{-1}) \in \R^p \times \R^p \times\R^{p \times p} \\
				H_{1, \gamma}: &\; \theta = (\gamma, \gamma, \Omega_{1, \gamma}^{-1}) \in \R^p \times \R^p \times \R^{p \times p}.
			\end{split}
		\end{equation} 
		We shall take the constants $c_{\lambda_1}, c_{\lambda_2}$ are chosen large enough such that all above hypotheses are within $\Theta(s, M)$, and that  uniformly for all $\gamma, \gamma' \in \Gamma_{s - 1}$, $\|\gamma\|_2 \le 1 / 2$ and $\Omega_{1, \gamma} + \Omega_{1, \gamma'} - \Omega_0$ is positive definite. 
		To see that this is possible, note that $s\lambda^2$ can be chosen to be arbitrarily small for all $n, p$ and $s $ with $s \leq n / \log p$.
		Finally, for notational simplicity we also let $f_{1, \gamma}(x \mid T = 1)$ denote the density of $X$ conditional on $T = 1$.
		
		Then it is straightforward to show that the ATE under $H_0$ satisfies $\tau(H_0) = 0$, and all the $\gamma \in \Gamma_{s - 1}$ share the same value of $\tau(H_{1, \gamma})$, which satisfies that for some constant $c' > 0$ depending only on $M$,
		\begin{equation}
			\label{diff}
			\begin{split}
				\tau ({H}_{1, \gamma}) & = \frac{2}{3}  \int \left(\psi(x^\top \gamma) - \frac{1}{2}\right) \frac{1 + 2 \exp(x^\top \gamma)}{1 + \exp(x^\top \gamma)} f_{1,\gamma}(x \mid T = 1) d x \\
				& = \frac{2}{3}  \int \left(\psi(x^\top \gamma) - \frac{1}{2}\right)^2 f_{1,\gamma}(x \mid T = 1) d x \ge c' \gamma^\top \Omega_{1, \gamma} \gamma \ge c' M^{-1} (s - 1) \lambda^2,
			\end{split}
		\end{equation}
		where for the penultimate inequality we use Lemma~\ref{lem:psisq}. Let $\pr_{1, \gamma}$ and $\pr_0$ denote the distribution of $(X, Y, T)$ under $H_{1, \gamma}$ and $H_0$, respectively. Let $\pr_{1, \gamma}^{\otimes n}$ and $\pr_0^{\otimes n}$ be the distributions of the observed samples $\{(X_i, Y_i, T_i)\}_{i=1}^n$ where each $(X_i, Y_i, T_i)$ is an i.i.d. realization from $\pr_{1, \gamma}$ and $\pr_0$, respectively. Finally, let $\pr_1^{\otimes n} := \frac{1}{|\Gamma_{s - 1}|} \sum_{\gamma \in \Gamma_{s - 1}} \pr_{1, \gamma}^{\otimes n}$ denote the distribution of $\{(X_i, Y_i, T_i)\}_{i=1}^n$ taking a uniform prior on $\gamma$.
		
		If we can show that 
		\begin{equation}\label{eq:singlechisq}
			\int \frac{\diff \pr_{1, \gamma}\diff \pr_{1, \gamma'}}{\diff \pr_0} \le 1 + c'' \gamma_{-1}^\top \gamma_{-1}',
		\end{equation}
		then following exactly the same proof as~\citet[Supplementary Material: Section 1.4.1]{CGM23}, we have by choosing constant $c_{\lambda_1}$ sufficiently large, the $\chi^2$-divergence between $\pr_1^{\otimes n}$ and $\pr_0^{\otimes n}$ satisfies
		\[
		\chi^2(\pr_1^{\otimes n}, \pr_0^{\otimes n}) \le \left(1 + \frac{1}{c_{\lambda_2} (s - 1)}\right)^{s - 1} - 1.
		\]
		In light of this, we have that given any $\alpha \in (0, \frac{1}{2})$, by choosing the constant $c_{\lambda_2}$ large enough, we can always have that
		\[
		\mathrm{TV}(\pr_1^{\otimes n}, \pr_0^{\otimes n}) \le \sqrt{\chi^2(\pr_1^{\otimes n}, \pr_0^{\otimes n})} < 1 - 2 \alpha - c'
		\]
		for some constant $c' \in (0, 1 - 2\alpha)$. Armed with the above and~\eqref{diff}, the desired result follows from Lemma~\ref{lem:cg17}. Finally, to prove~\eqref{eq:singlechisq} we apply Lemma~\ref{lem:singlechisq} below.
	\end{proof}
}

\newrev{
	
	\begin{lemma}\label{lem:cg17}
		Consider two sets of hypotheses $\mathcal{H}_0, \mathcal{H}_1 \subseteq \Theta(s, M)$ where $\tau_{\theta_0} \equiv \tau_0$ for all $\theta_0 \in \mathcal{H}_0$ and $\tau_{\theta_1} \equiv \tau_1$ for all $\theta_1 \in \mathcal{H}_1$. Given any $\theta \in \mathcal{H}_0 \cup \mathcal{H}_1$, let $\pr_\theta^{\otimes n}$ be the distribution of $\{(X_i, Y_i, T_i)\}_{i = 1}^n$ with each $(X_i, Y_i, T_i)$ i.i.d. generated from the distribution parametrised by $\theta$, and let $\pr_1^{\otimes n} := \frac{1}{|\mathcal{H}_1|} \sum_{\theta_1 \in \mathcal{H}_1} \pr_{\theta_1}^{\otimes n}$, $\pr_0^{\otimes n} := \frac{1}{|\mathcal{H}_0|} \sum_{\theta_1 \in \mathcal{H}_0} \pr_{\theta_1}^{\otimes n}$. Then we have that for any confidence level $\alpha \in (0, \frac{1}{2})$, 
		\[
		\inf_{\ci_\alpha(\D) \in \mathcal{I}_\alpha(\Theta(s, M))} \sup_{\theta \in \Theta(s, M)} \E_\theta^{\otimes n} L(\ci_\alpha(\D)) \ge |\tau_1 - \tau_0| (1 - 2 \alpha - \mathrm{TV}(\pr_1^{\otimes n}, \pr_0^{\otimes n}))_+,
		\]
		where $\mathrm{TV}(\cdot, \cdot)$ represents the total variation distance, and for any $a \in \R$, $a_+ := a  \kron_{\{a \ge 0\}}$.
	\end{lemma}
	
	
	Lemma~\ref{lem:cg17} follows easily from \citet[Lemma 1]{CG17}; its proof is provided in~\citet[Step 3 in Section 1.4.1 of the Supplementary Material]{CGM23}. 
	
	
	\begin{lemma}\label{lem:singlechisq}
		Consider the setup of Proposition~\ref{prop:lower1} and let $\pr_{1, \gamma}, \pr_{1, \gamma'}$ and $\pr_0$ be as in the proof of Proposition~\ref{prop:lower1}. There exists a constant $c' > 0$ such that 
		\[
		\int \frac{\diff \pr_{1, \gamma}\diff \pr_{1, \gamma'}}{\diff \pr_0} \le 1 + c' \gamma_{-1}^\top \gamma_{-1}'.
		\]
	\end{lemma}
	
	\begin{proof}
		Let $f_{1, \gamma} (x\mid T = t)$ be the density of $X$ conditional on $T=t$ under $H_{1, \gamma}$, let $f_0 (x)$ be the density of $X$ under $H_0$. Then some simple calculation yields that
		\begin{align*}
			\int \frac{\diff \bbP_{1, \gamma} \diff \bbP_{1, \gamma'}}{\diff \bbP_{0}} 
			=\; & C_\phi \int \frac{2 f_{1,\gamma}(x \mid T = 1) f_{1,\gamma'}(x \mid T = 1)}{f_0(x)} \psi(x^\top \gamma) \psi(x^\top \gamma') \diff x \\
			& + C_\phi \int \frac{2 f_{1,\gamma}(x \mid T = 1) f_{1,\gamma'}(x \mid T = 1)}{f_0(x)} (1 - \psi(x^\top \gamma)) (1 - \psi(x^\top \gamma')) \diff x
			\\
			& + (1 - C_\phi)  \int \frac{f_{1, \gamma}(x \mid T = 0) f_{1, \gamma'}(x \mid T = 0)}{f_0(x)} d x.
		\end{align*}
		Recall that $p(x; \Sigma)$ is the density of a multivariate normal random variable with mean zero and covariance matrix $\Sigma$; using the fact that $f_{1,\gamma}(x \mid T = 1) \equiv p(x; \Omega_{1, \gamma}^{-1})$ and $f_{1,\gamma}(x \mid T = 0) \equiv 2 p(x; \Omega_{1, \gamma}^{-1}) \psi(x^\top \gamma)$, we have
		\begin{align*}
			\int \frac{\diff \bbP_{1, \gamma} \diff \bbP_{1, \gamma'}}{\diff \bbP_{0}} = & C_\phi \int \frac{2 p(x; \Omega_{1, \gamma}^{-1}) p(x; \Omega_{1, \gamma'}^{-1})}{p(x; \Omega_0^{-1})} \psi(x^\top \gamma) \psi(x^\top \gamma') \diff x \\
			& + C_\phi \int \frac{2 p(x; \Omega_{1, \gamma}^{-1}) p(x; \Omega_{1, \gamma'}^{-1})}{p(x; \Omega_0^{-1})} (1 - \psi(x^\top \gamma)) (1 - \psi(x^\top \gamma')) \diff x
			\\
			& + (1 - C_\phi)  \int \frac{4 p(x; \Omega_{1, \gamma}^{-1}) p(x; \Omega_{1, \gamma'}^{-1})}{p(x; \Omega_0^{-1})} \psi(x^\top \gamma) \psi(x^\top \gamma') d x \\
			= & \int \frac{4 p(x; \Omega_{1, \gamma}^{-1}) p(x; \Omega_{1, \gamma'}^{-1})}{p(x; \Omega_0^{-1})} \psi(x^\top \gamma) \psi(x^\top \gamma') d x,
		\end{align*}
		where in the last line we use the symmetry of the logistic function and the density of mulvariate normal random variable.
		
		Expanding $p(x; \Omega_0^{-1})$ etc., we have
		\begin{align*}
			& \int \frac{\diff \bbP_{1, \gamma} \diff \bbP_{1, \gamma'}}{\diff \bbP_{0}} \\
			& = \sqrt{\frac{\det(\Omega_{1, \gamma}) \det(\Omega_{1, \gamma'})}{\det(\Omega_0) (2\pi)^p}} 4 \int \exp\left\{-\frac{1}{2}x^\top (\Omega_{1, \gamma} + \Omega_{1, \gamma'} - \Omega_0)x\right\} \psi(x^\top \gamma) \psi(x^\top \gamma') d x \\
			& = \sqrt{\frac{\det(\Omega_{1, \gamma}) \det(\Omega_{1, \gamma'})}{\det(\Omega_0) \det(\Omega_{1, \gamma} + \Omega_{1, \gamma'} - \Omega_0)}} 4 \E[\psi(Z^\top \gamma) \psi(Z^\top \gamma')],
		\end{align*}
		where $Z \sim \mathcal{N}(0, (\Omega_{1, \gamma} + \Omega_{1, \gamma'} - \Omega_0)^{-1})$. Some simple calculation yields 
		\[
		\det(\Omega_{1, \gamma}) = \det(\Omega_{1, \gamma'}) = 1 - \|\gamma_{-1}\|_2^2
		\]
		and
		\[
		\det(\Omega_{1, \gamma} + \Omega_{1, \gamma'} - I) = 1 + (s - 1)\lambda^2 - \|\gamma_{-1} + \gamma_{-1}'\|_2^2 = 1 - \|\gamma_{-1}\|_2^2 - 2 \gamma_{-1}^\top \gamma_{-1}',
		\]
		which further gives us that
		\begin{equation}\label{eq:detdiff}
			\sqrt{\frac{\det(\Omega_{1,\gamma}) \det(\Omega_{1,\gamma'})}{\det(\Omega_0) \det(\Omega_{1,\gamma} + \Omega_{1,\gamma'} - \Omega_0)}} = \sqrt{\frac{(1 - \|\gamma_{-1}\|_2^2)^2}{(1 - \|\gamma_{-1}\|_2^2 - 2 \gamma_{-1}^\top \gamma_{-1}')(1 - \|\gamma_{-1}\|_2^2)}} \le 1 + 4 \gamma_{-1}^\top \gamma_{-1}',
		\end{equation}
		where for the last equality we use that $\|\gamma_{-1}\|_2 \le \frac{1}{2}$.
		
		We now focus on $4\E[\psi(Z^\top \gamma) \psi(Z^\top \gamma')]$. Using the blockwise inversion formula, we get
		\[
		(\Omega_{1, \gamma} + \Omega_{1, \gamma'} - \Omega_0)^{-1} = \left(
		\begin{matrix}
			\frac{1}{1 - \|\gamma_{-1}\|_2^2 - 2 \gamma_{-1}^\top \gamma_{-1}'} & - \frac{(\gamma_{-1} + \gamma_{-1}')^\top}{1 - \|\gamma_{-1}\|_2^2 - 2 \gamma_{-1}^\top \gamma_{-1}'} \\
			- \frac{\gamma_{-1} + \gamma_{-1}'}{1 - \|\gamma_{-1}\|_2^2 - 2 \gamma_{-1}^\top \gamma_{-1}'} & I + \frac{(\gamma_{-1} + \gamma_{-1}') (\gamma_{-1} + \gamma_{-1}')^\top}{1 - \|\gamma_{-1}\|_2^2 - 2 \gamma_{-1}^\top \gamma_{-1}'}
		\end{matrix}
		\right),
		\]
		which means that
		\begin{equation}\label{eq:samegamma}
			\begin{aligned}
				& \gamma^\top (\Omega_{1, \gamma} + \Omega_{1, \gamma'} - \Omega_0)^{-1} \gamma = \gamma'^\top (\Omega_{1, \gamma} + \Omega_{1, \gamma'} - \Omega_0)^{-1} \gamma' \\
				& = \|\gamma_{-1}\|_2^2 + \frac{((s - 1) \lambda^2 - (\|\gamma_{-1}\|_2^2 + \gamma_{-1}^\top \gamma_{-1}'))^2}{1 - \|\gamma_{-1}\|_2^2 - 2 \gamma_{-1}^\top \gamma_{-1}'} \\
				& = \|\gamma_{-1}\|_2^2 + \frac{(\gamma_{-1}^\top \gamma_{-1}')^2}{1 - \|\gamma_{-1}\|_2^2 - 2 \gamma_{-1}^\top \gamma_{-1}'} \le \|\gamma_{-1}\|_2^2 + 4 (\gamma_{-1}^\top \gamma_{-1}')^2 \le 1.
			\end{aligned}
		\end{equation}
		In light of the above and that $\psi(u) \le \Phi(u)$ for all $u \ge 0$ (which we prove via Lemma~\ref{lem:phi}), it follows from exactly the same proof as~\citet[Lemma~2]{CGM23} that for some universal constant $C > 0$,
		\[
		4 \E[\psi(Z^\top \gamma) \psi(Z^\top \gamma')] \le 1 + C \gamma^\top (\Omega_{1, \gamma} + \Omega_{1, \gamma'} - \Omega_0)^{-1} \gamma'.
		\]
		Following an analysis analogous to~\eqref{eq:samegamma} we further have
		\begin{equation}\label{eq:gamma}
			\gamma^\top (\Omega_{1, \gamma} + \Omega_{1, \gamma'} - \Omega_0)^{-1} \gamma' = \gamma_{-1}^\top \gamma'_{-1} + \frac{(\gamma_{-1}^\top \gamma'_{-1})^2}{1 - \|\gamma_{-1}\|_2^2 - 2 \gamma_{-1}^\top \gamma_{-1}'} \le \gamma_{-1}^\top \gamma'_{-1} + 4 (\gamma_{-1}^\top \gamma'_{-1})^2.
		\end{equation}
		In light of both~\eqref{eq:detdiff} and~\eqref{eq:gamma} and that $\gamma_{-1}^\top \gamma'_{-1} \le 1$, we prove the desired result.
	\end{proof}
	
	\begin{lemma}\label{lem:phi}
		For any $u \ge 0$, we have $\psi(u) \le \Phi(u)$.
	\end{lemma}
	
	\begin{proof}
		Since $\Phi(0) = \psi(0) = \frac{1}{2}$, $\lim_{u \to \infty} \Phi(u) = \lim_{u \to \infty} \psi(u) = 1$ and $\Phi'(0) = \frac{1}{\sqrt{2 \pi}} > \frac{1}{4} = \psi'(0)$, we only need to show that the equation 
		\begin{equation}\label{eq:phi}
			\psi'(u) = \Phi'(u)
		\end{equation}
		has at most one solution within the interval $[0, \infty)$.
		
		First, for any $u \in [0, 0.95)$, $\Phi'(u) \ge \Phi'(0.95) > \frac{1}{4} \ge \phi'(u)$. Therefore, we only need to prove that~\eqref{eq:phi} has at most one solution within $[0.95, \infty)$. Now we rewrite~\eqref{eq:phi} as
		\begin{equation}\label{eq:phi1}
			\sqrt{2 \pi} \exp\left(\frac{u^2}{2} + u\right) - (1 + \exp(u))^2 = 0
		\end{equation}
		and let $g(u)$ denote the left hand side of the above equation. We have for any $u \in [0.95, \infty)$,
		\begin{align*}
			g'(u) & = \sqrt{2 \pi} \exp\left(\frac{u^2}{2} + u\right) (u + 1) - 2 (1 + \exp(u)) \exp(u) \\
			& \ge \exp(u) \underset{ =: h(u)}{\underbrace{\left(\sqrt{2 \pi} \times 1.95 \times \exp\left(\frac{u^2}{2}\right) - 2 (1 + \exp(u))\right)}} > 0,
		\end{align*}
		where for the last inequality we use that $h(0.95) > 0$ and that
		\begin{align*}
			h'(u) & = \sqrt{2\pi} \times 1.95 \times \exp\left(\frac{u^2}{2}\right) u - 2 \exp(u) \\
			& \ge \sqrt{2\pi} \times 1.95 \times 0.95 \times \exp\left(\frac{u^2}{2}\right) - 2 \exp(u) > 0.
		\end{align*}
		Notice that the last inequality holds because uniformly for all $u \in [0.95, \infty)$,
		\begin{align*}
			& \frac{\sqrt{2\pi} \times 1.95 \times 0.95 \times \exp\left(\frac{u^2}{2}\right)}{2 \exp(u)} = \frac{\sqrt{2\pi} \times 1.95 \times 0.95}{2} \exp\left(\frac{u^2}{2} - u\right) \\
			& \ge \frac{\sqrt{2\pi} \times 1.95 \times 0.95}{2} \exp\left(- \frac{1}{2}\right) > 1.
		\end{align*}
		In light of that $g'(u) > 0$ for all $u \in [0.95, \infty)$ and $g(0.95) < 0$, we prove that~\eqref{eq:phi1}, or equivalently~\eqref{eq:phi} has at most one solution in $[0.95, \infty)$, and the desired result follows.
	\end{proof}
	
	\begin{lemma}\label{lem:psisq}
		Let $M > 0$ and $Z \sim \mathcal{N}(0, 1)$. Then there eixsts a constant $c >0$ depending only on $M$ such that for any $0 < \sigma_1 \le \sigma_2 \le M$,
		\[
		\E\left[\left(\psi(\sigma_1 Z) - \frac{1}{2}\right) \left(\psi(\sigma_2 Z) - \frac{1}{2}\right)\right] \ge c \sigma_1 \sigma_2.
		\]
	\end{lemma}
	
	\begin{proof}
		By the mean value theorem, using that $\psi'(u)$ is monotonically decreasing for $u \ge 0$, we have
		\begin{align*}
			& \E\left[\left(\psi(\sigma_1 Z) - \frac{1}{2}\right) \left(\psi(\sigma_2 Z) - \frac{1}{2}\right)\right] \ge \E\left[\left(\psi(\sigma_1 Z) - \frac{1}{2}\right) \left(\psi(\sigma_2 Z) - \frac{1}{2}\right) \kron_{\{|\sigma_2 Z| \le M\}}\right] \\
			&\qquad \ge \sigma_1 \sigma_2 \psi'(M)^2 \E[Z^2 \kron_{\{|\sigma_2 Z| \le M\}}]  \ge \sigma_1 \sigma_2 \psi'(M)^2 \E[Z^2 \kron_{\{|Z| \le 1}\}],
		\end{align*}
		as desired.
	\end{proof}
}

\newrev{
	\subsection{Proof of Proposition~\ref{prop:lower2}}
	
	\begin{proof}
		We instead consider $\mathcal{H}_0 := \{H_0\}, \mathcal{H}_1 := \{H_1\}$, where 
		\begin{equation*}
			\begin{split}
				H_{0}: &\; \theta = ((1, 0, \ldots, 0)^\top, (0, \ldots, 0)^\top, I) \in \R^p \times \R^p \times\R^{p \times p} \\
				H_1: &\; \theta = \left((1, 0, \ldots, 0)^\top, \left(\frac{1}{c_\lambda\sqrt{n}}, 0, \ldots, 0\right)^{\top}, I\right) \in \R^p \times \R^p \times \R^{p \times p}.
			\end{split}
		\end{equation*} 
		for some constant $c_\lambda \ge 1$ to be chosen. Note that apparently, both $H_0$ and $H_1$ are in $\Theta(s, M)$. Apparently $\tau(H_0) = 0$, and following analogous analysis as~\eqref{diff}, for some $c'$ depending only on $M$,
		\[
		\tau(H_1) \ge c' \beta^\top \gamma \ge \frac{c'}{c_\lambda \sqrt{n}}.
		\]
		In light of the above, letting $\pr_1, \pr_0$ denote the distributions of $(X, T, Y)$ under $H_0$, $H_1$, respectively, and letting $Z \sim \mathcal{N}(0, 1 / (c_\lambda^2 n))$, we can have from exactly the same analysis in Lemma~\ref{lem:singlechisq} that
		\[
		\int \frac{\diff \pr_1^2}{\diff \pr_0} = \frac{2}{3} + \frac{4}{3} \E[\psi(Z)^2] = 1 + \frac{4}{3} \cdot \E\left[\left(\psi(Z) - \frac{1}{2}\right)^2\right] \le 1 + \frac{1}{12} \var(Z) = 1 + \frac{1}{12} \frac{1}{c_\lambda^2 n},
		\]
		and moreover that
		\[
		\chi^2(\pr_1^{\otimes n}, \pr_0^{\otimes n}) = \left(\int \frac{\diff \pr_1^2}{\diff \pr_0}\right)^n - 1 \le \left(1 + \frac{1}{12} \frac{1}{c_\lambda^2 n}\right)^n - 1.
		\]
		This means that given any $\alpha \in (0, \frac{1}{2})$, by choosing the constant $c_\lambda$ large enough, 
		\[
		\mathrm{TV}(\pr_1^{\otimes n}, \pr_0^{\otimes n}) \le \sqrt{\chi^2 (\pr_1^{\otimes n}, \pr_0^{\otimes n})} \le 1 - 2 \alpha - c'
		\]
		for some constant $c' \in (0, 1 - 2 \alpha)$. Then the desired result follows from Lemma~\ref{lem:cg17}.
	\end{proof}
}

\section{ATE results with heteroscedastic noise} \label{sec:simate_hetero}

In Figures~\ref{fig:hlinear} and~\ref{fig:hnonlinear}, we present the results from repeating the experiments in Section~\ref{sec:simate} but replacing the homoscedastic noise $\varepsilon_i$ in \eqref{eq:sim_Y} with heteroscedastic errors
\[
\ind_{\{\pi(X_i) \geq 0.5\}} \varepsilon_i^{(+)} + \ind_{\{\pi(X_i) < 0.5\}} \varepsilon_i^{(-)},
\]
where $\epsilon_i^{(+)} \sim \mathcal{N}(0, 0.5)$ and $\epsilon_i^{(-)} \sim \mathcal{N}(0, 2)$.
The results are similar to those presented in Figures~\ref{fig:linear} and \ref{fig:nonlinear} with DIPW performing well across the scenarios considered.

\begin{figure}[t!]
	\subfigure[Toeplitz design, $s = 5$]{\includegraphics[width=0.32\textwidth]{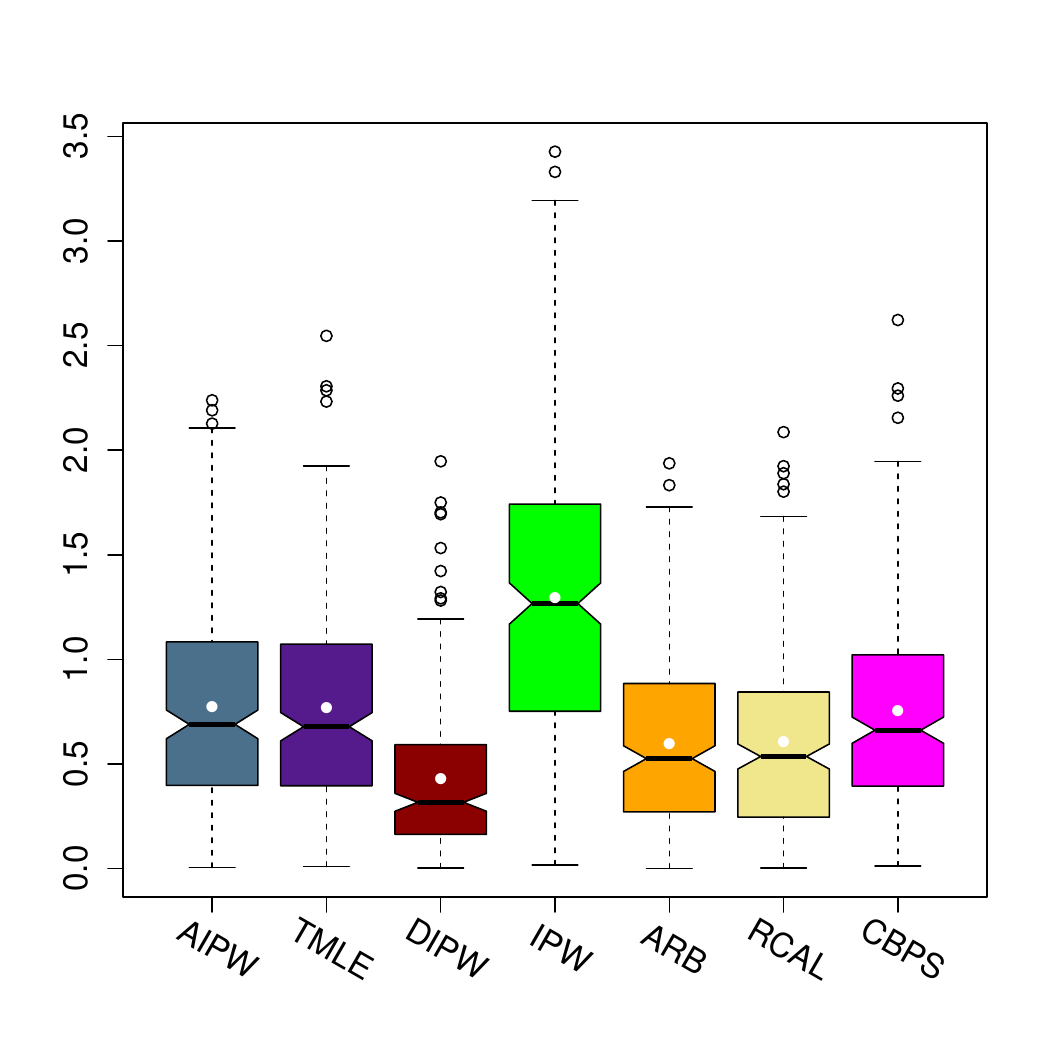}} \hfill
	\subfigure[Toeplitz design, $s = 20$]{\includegraphics[width=0.32\textwidth]{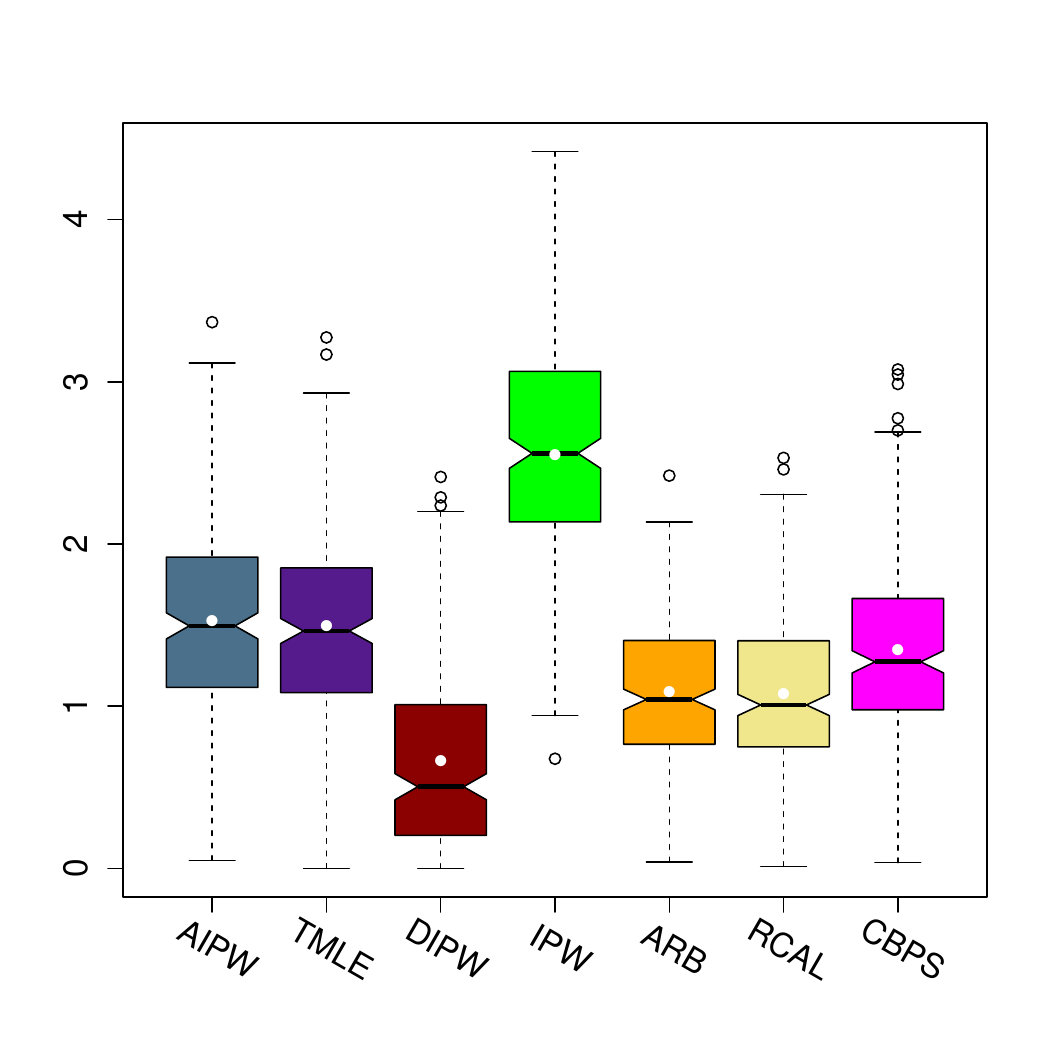}} \hfill
	\subfigure[Toeplitz design, $s = 50$]{\includegraphics[width=0.32\textwidth]{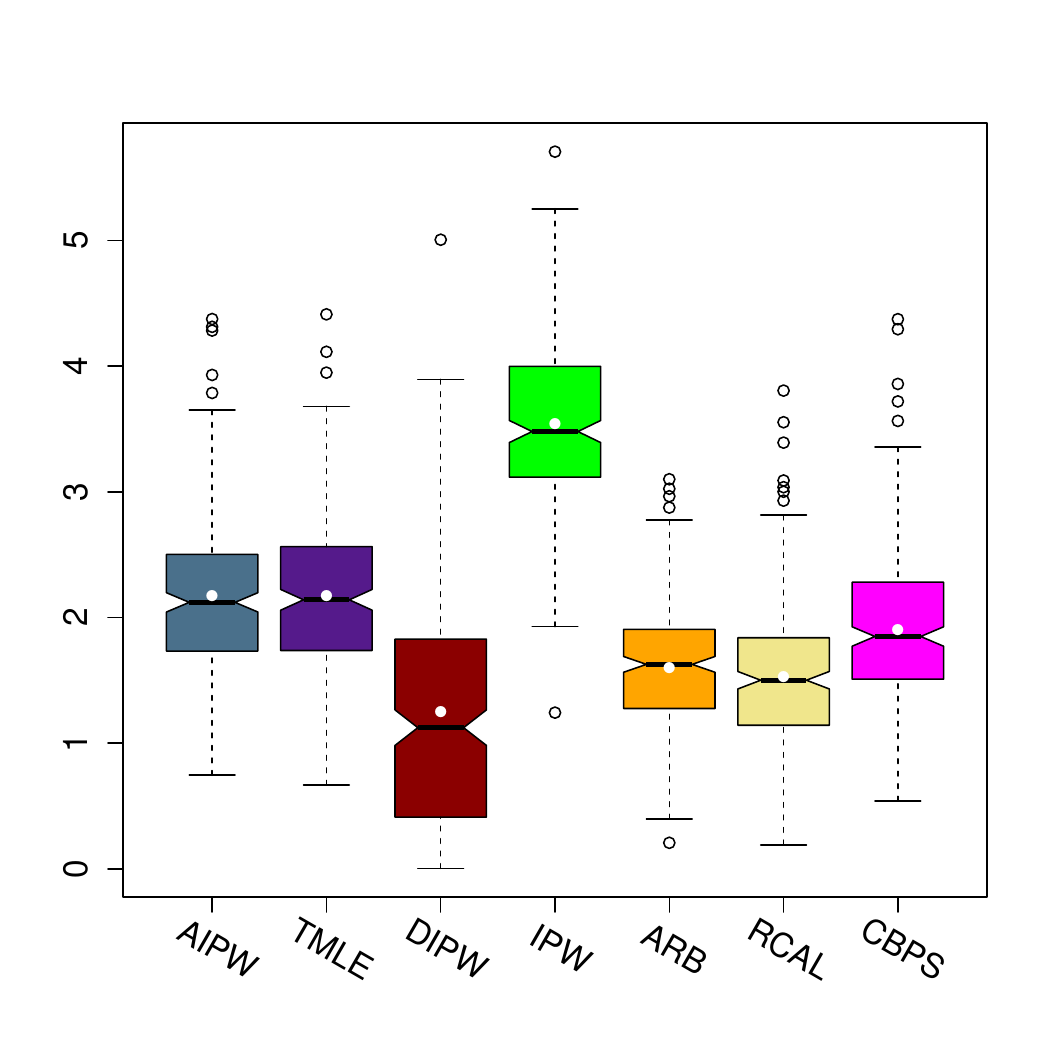}} \\
	\subfigure[Exponential design, $s = 5$]{\includegraphics[width=0.32\textwidth]{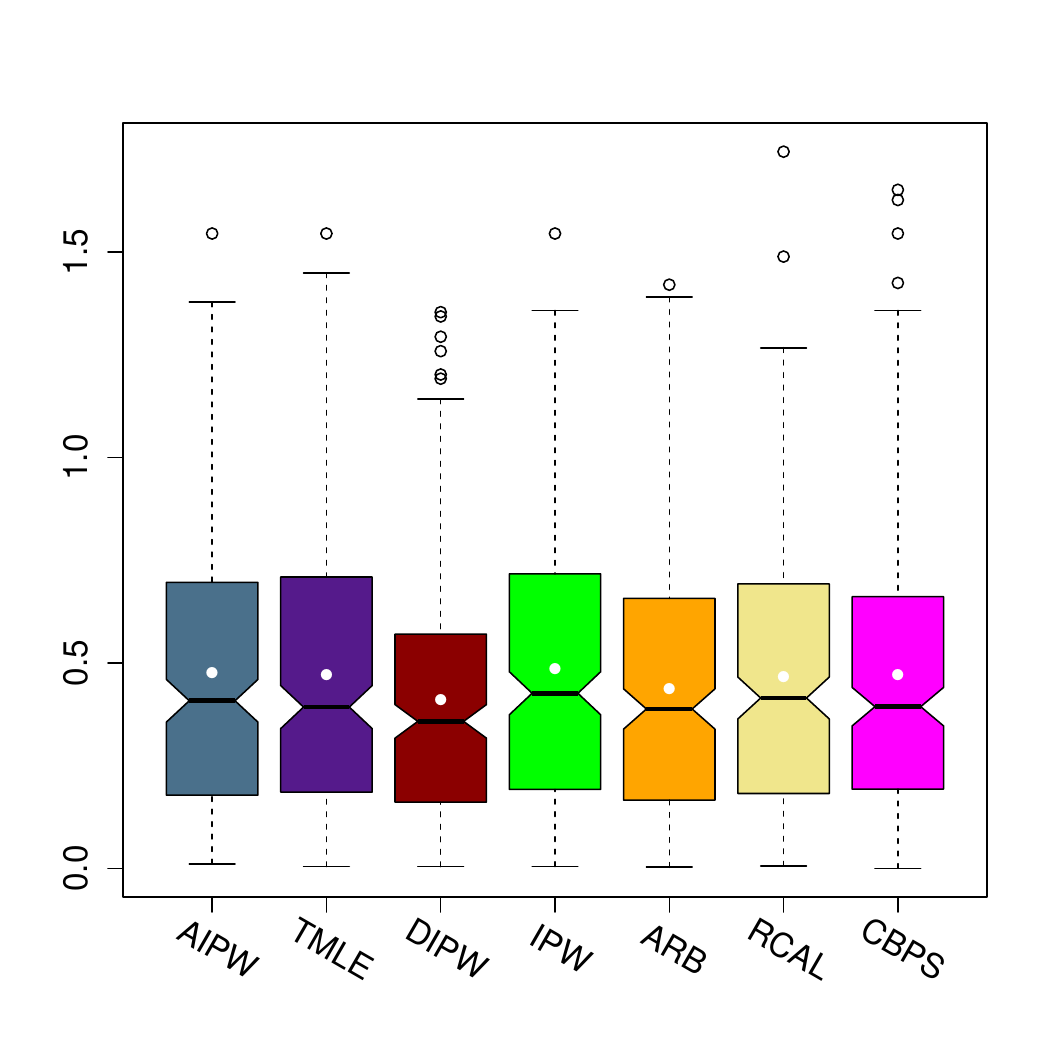}} \hfill
	\subfigure[Exponential design, $s = 20$]{\includegraphics[width=0.32\textwidth]{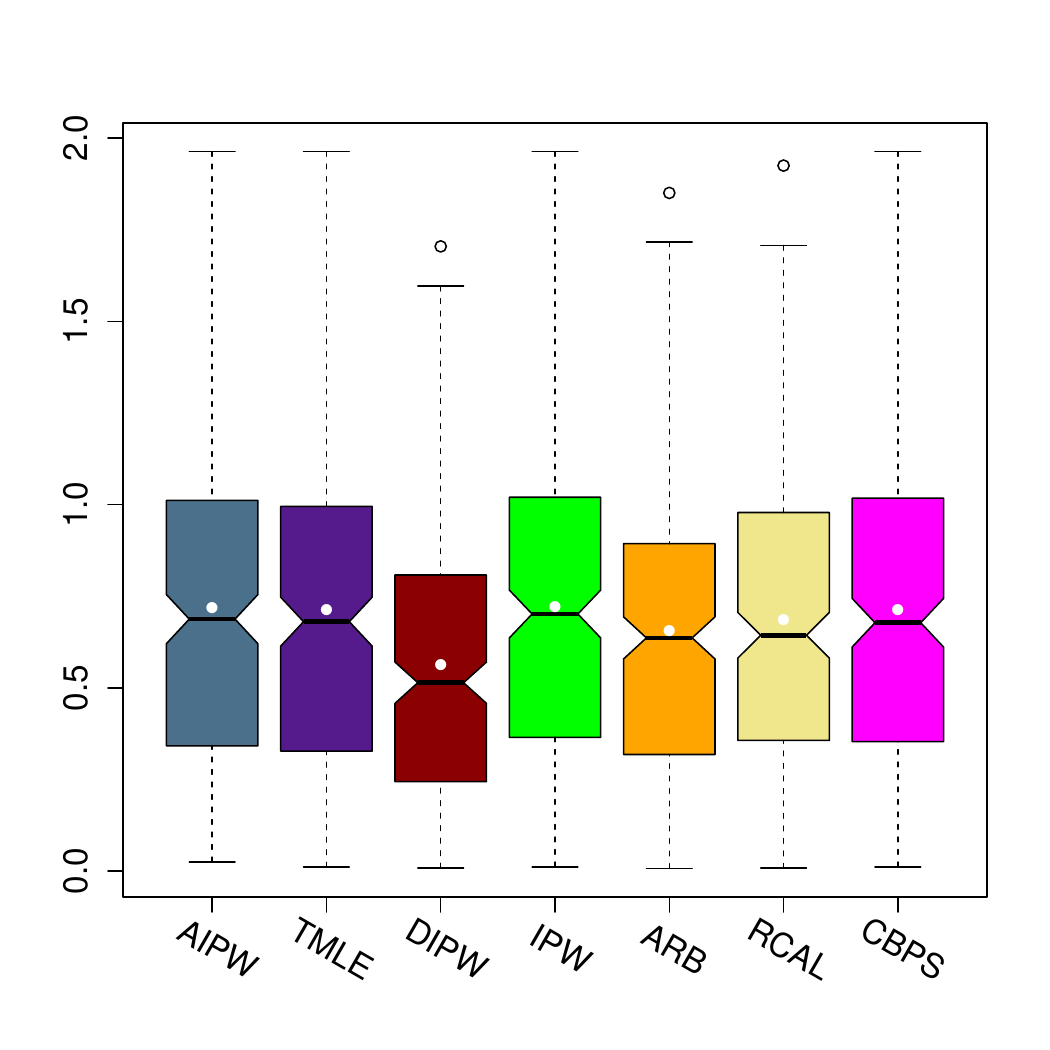}} \hfill
	\subfigure[Exponential design,  $s = 50$]{\includegraphics[width=0.32\textwidth]{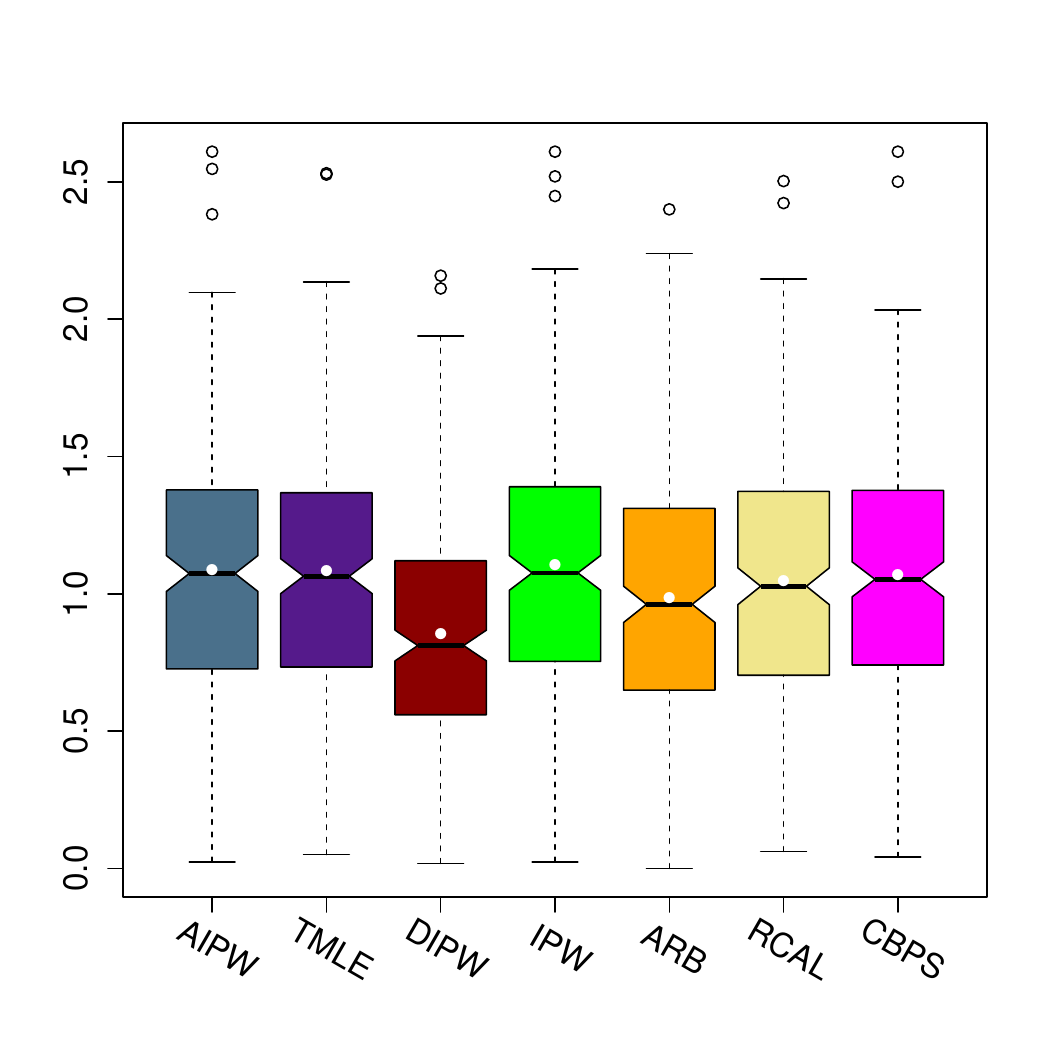}} \\
	\subfigure[Real data  design, $s = 5$]{\includegraphics[width=0.32\textwidth]{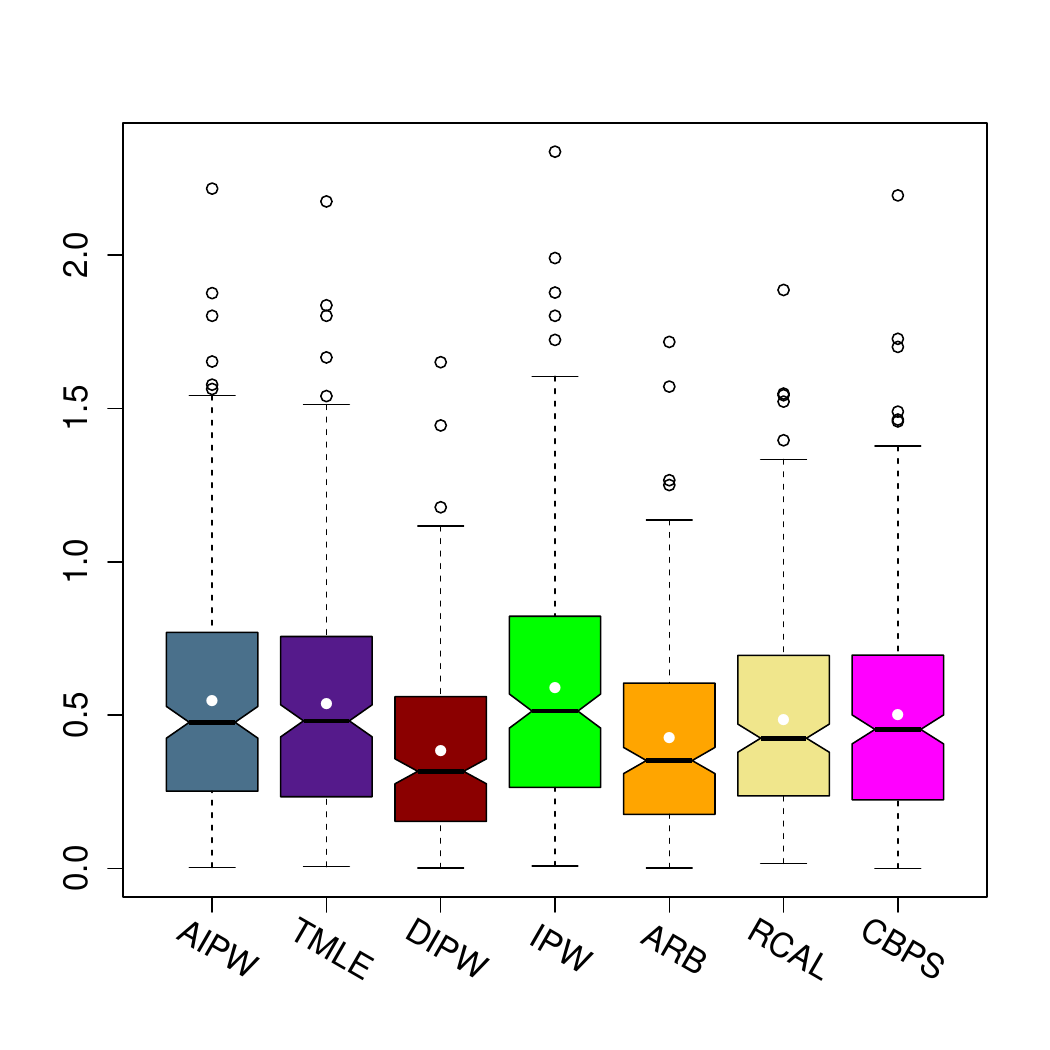}} \hfill
	\subfigure[Real data design, $s = 20$]{\includegraphics[width=0.32\textwidth]{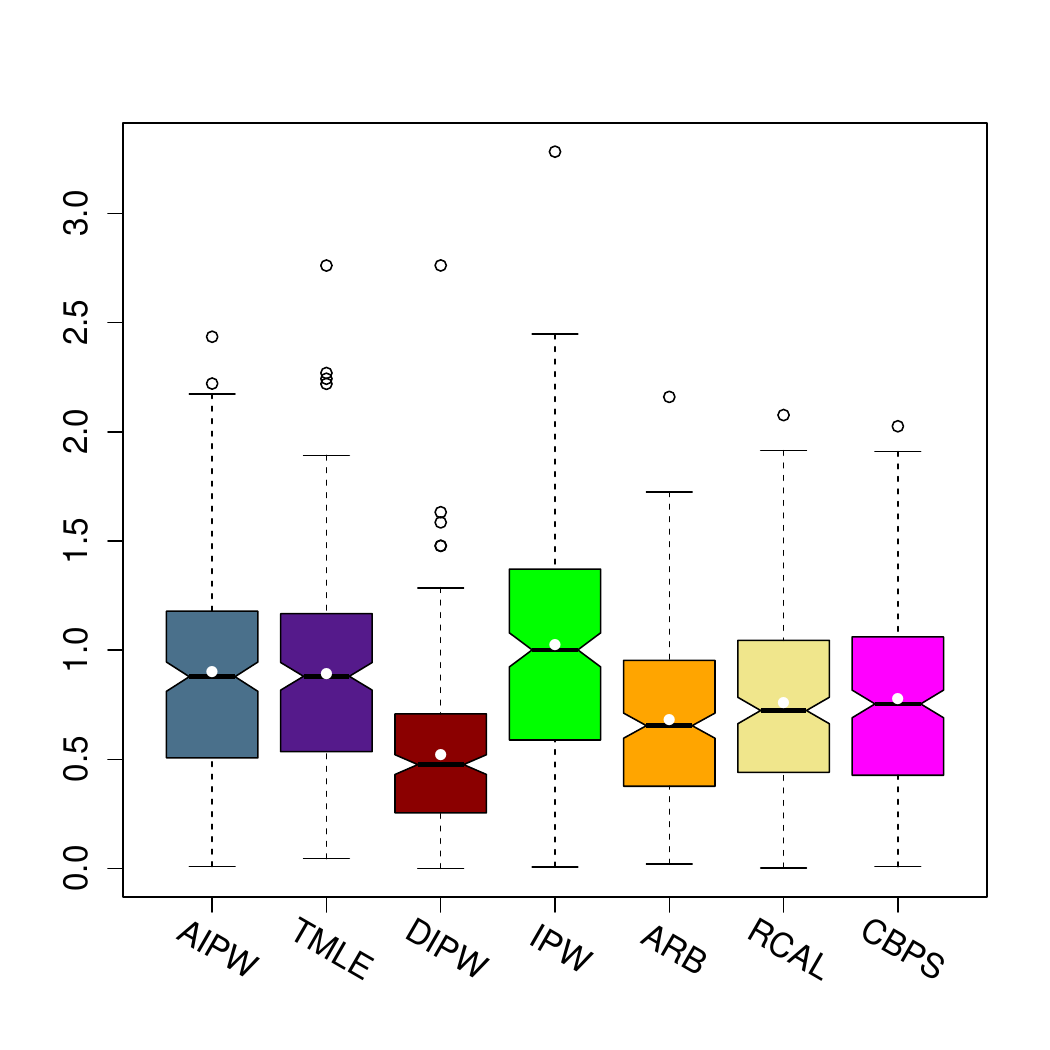}} \hfill
	\subfigure[Real data design,  $s = 50$]{\includegraphics[width=0.32\textwidth]{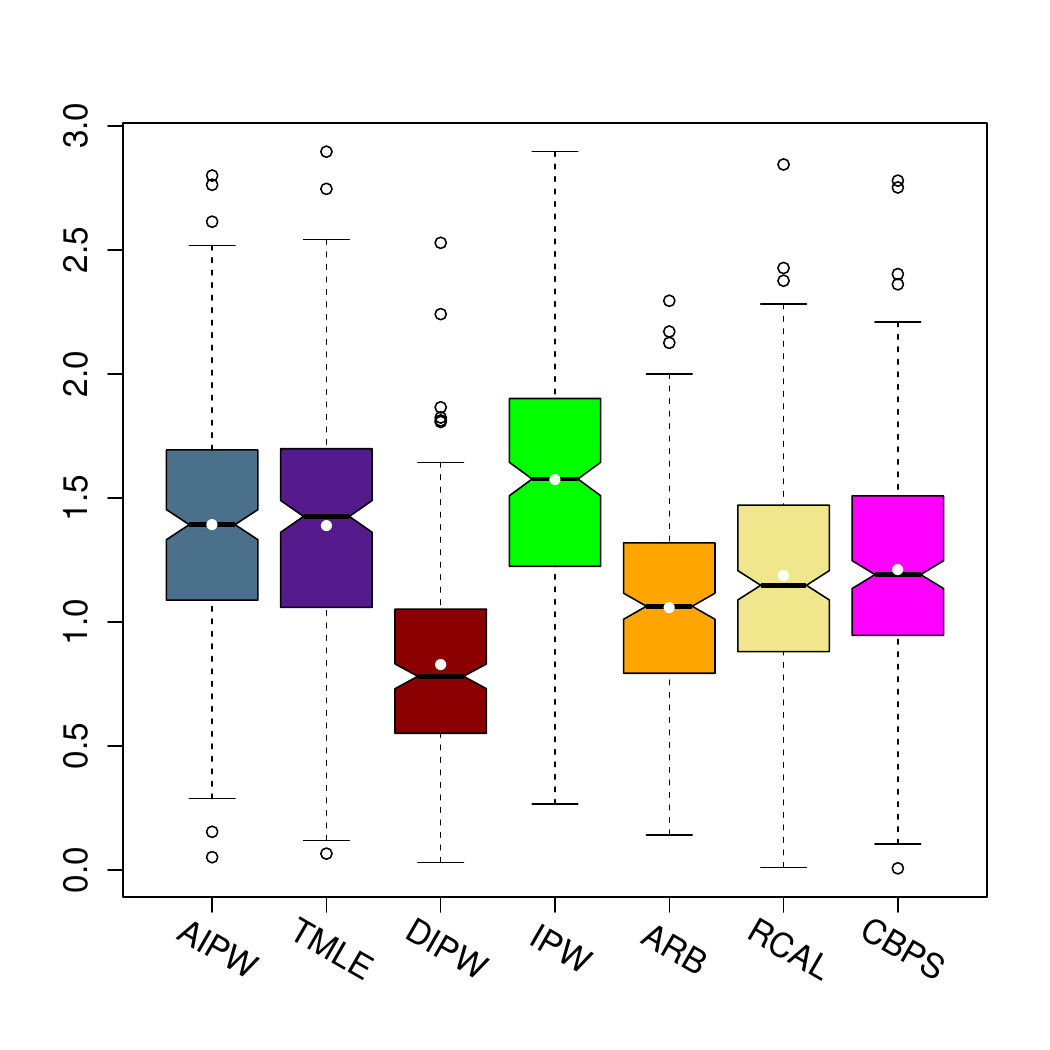}}
	\caption{Boxplot of estimation error with heteroscedastic noise and linear responses.}\label{fig:hlinear}
\end{figure}

\begin{figure}[t!]
	\subfigure[Toeplitz design, $s = 5$]{\includegraphics[width=0.32\textwidth]{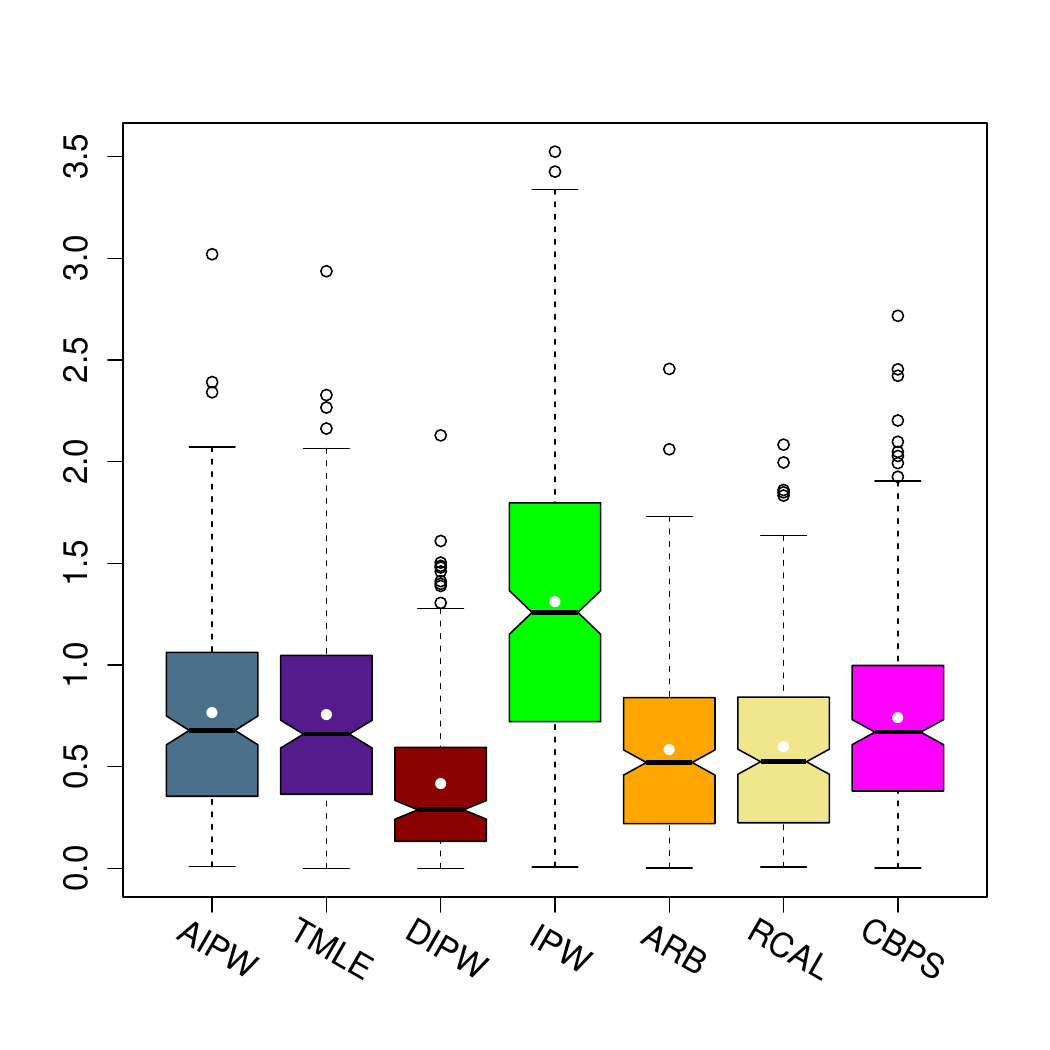}} \hfill
	\subfigure[Toeplitz design, $s = 20$]{\includegraphics[width=0.32\textwidth]{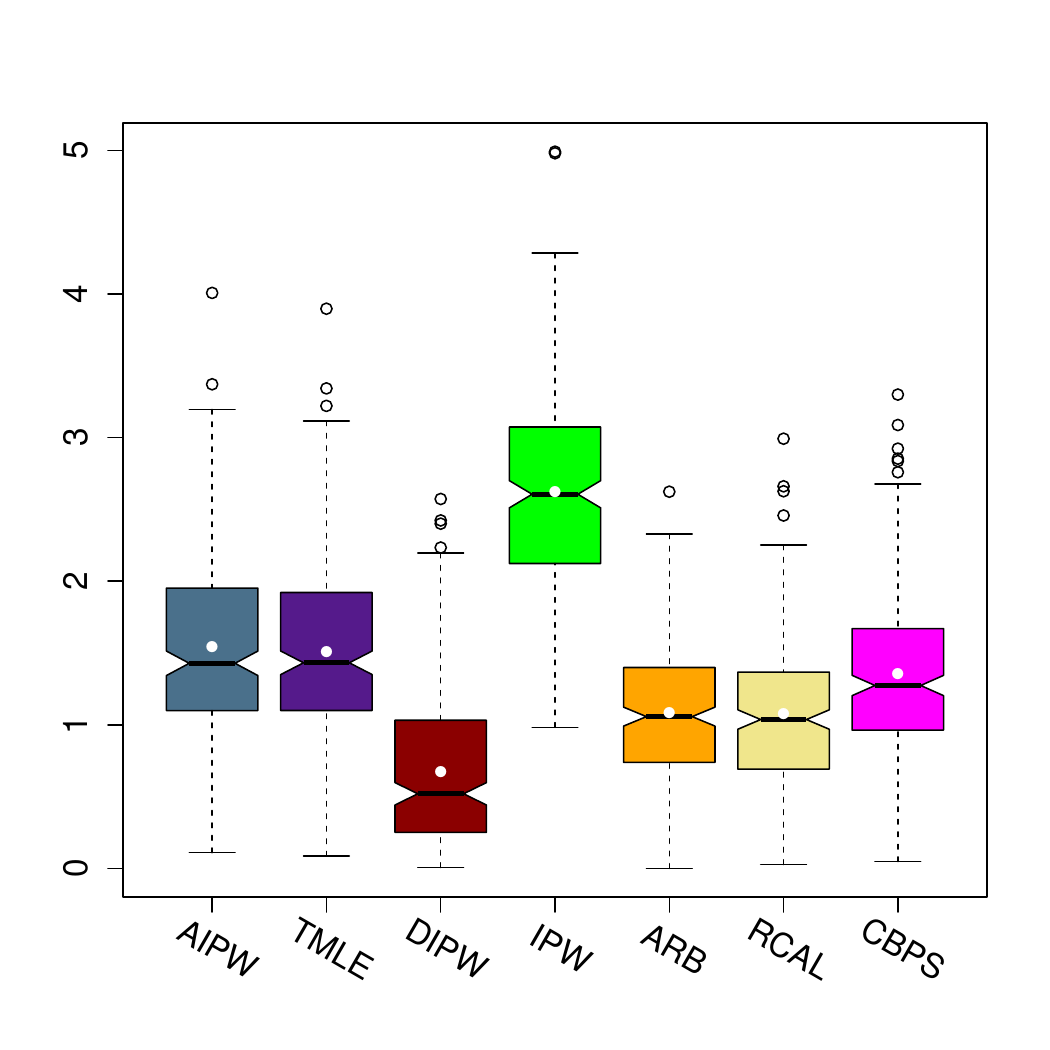}} \hfill
	\subfigure[Toeplitz design, $s = 50$]{\includegraphics[width=0.32\textwidth]{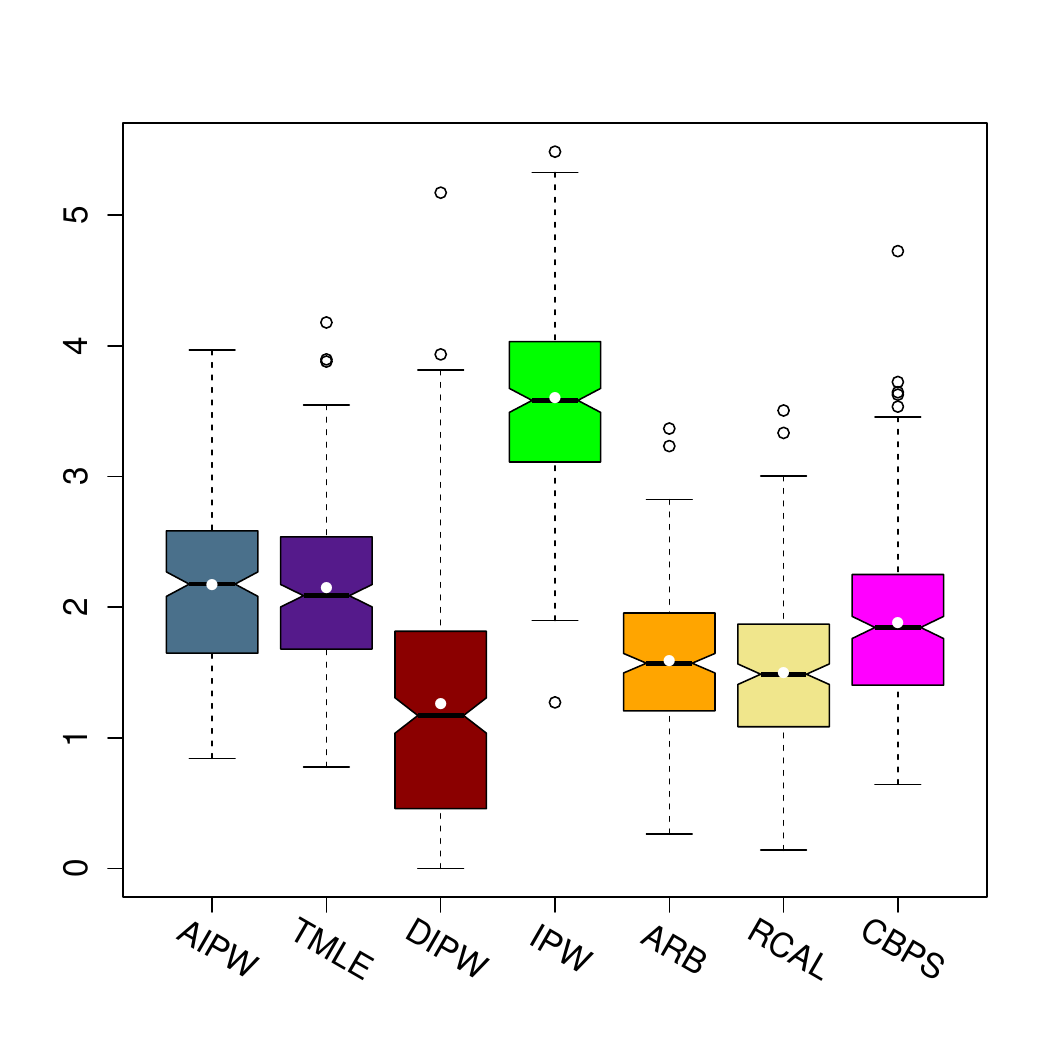}} \\
	\subfigure[Exponential design, $s = 5$]{\includegraphics[width=0.32\textwidth]{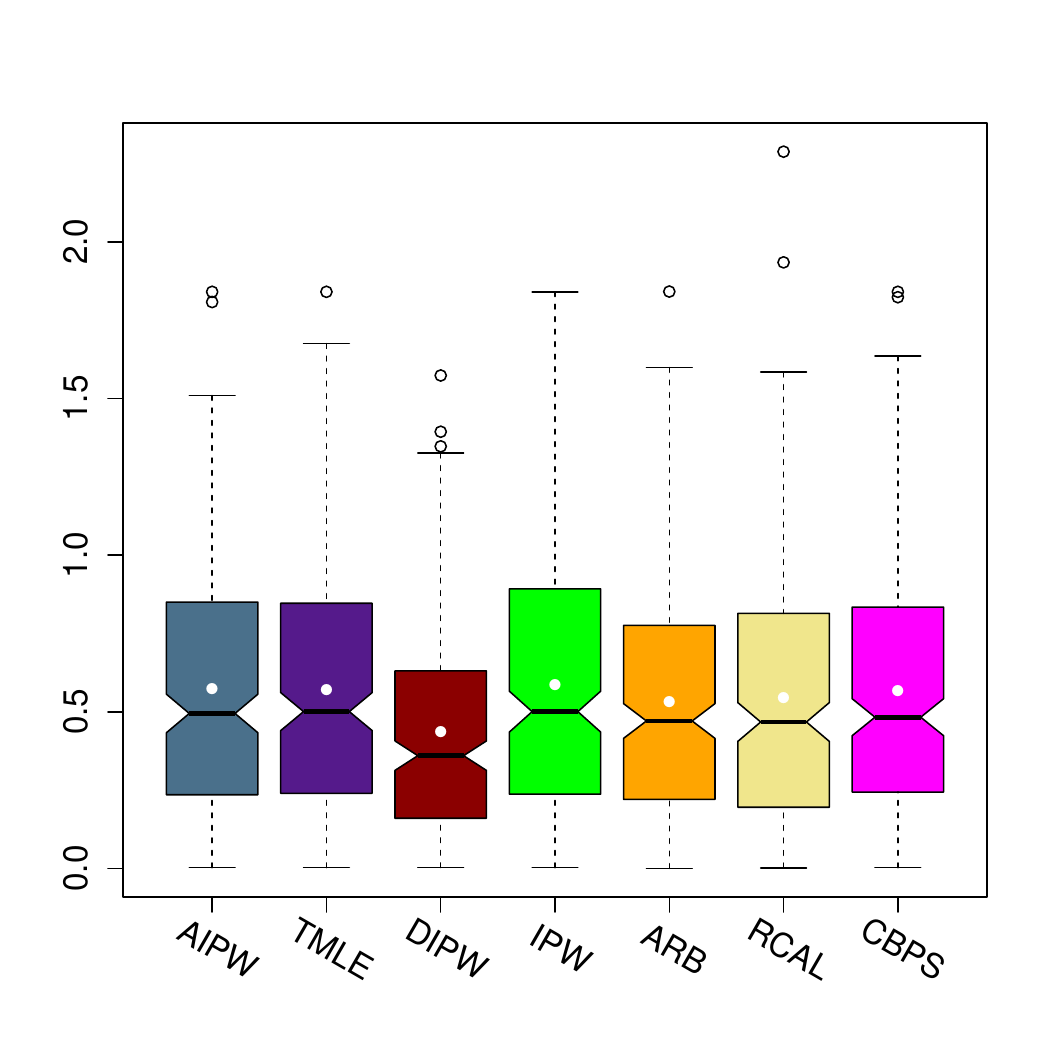}} \hfill
	\subfigure[Exponential design, $s = 20$]{\includegraphics[width=0.32\textwidth]{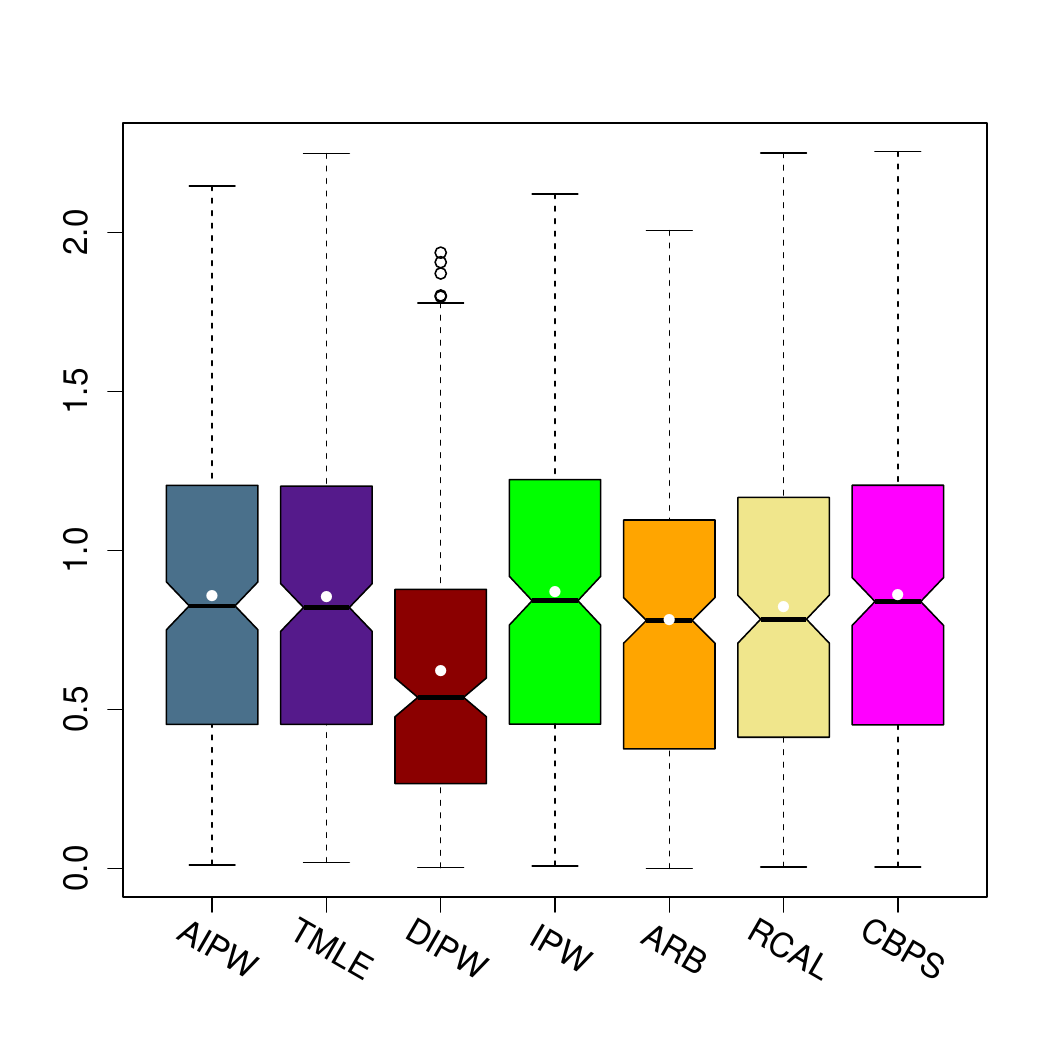}} \hfill
	\subfigure[Exponential design,  $s = 50$]{\includegraphics[width=0.32\textwidth]{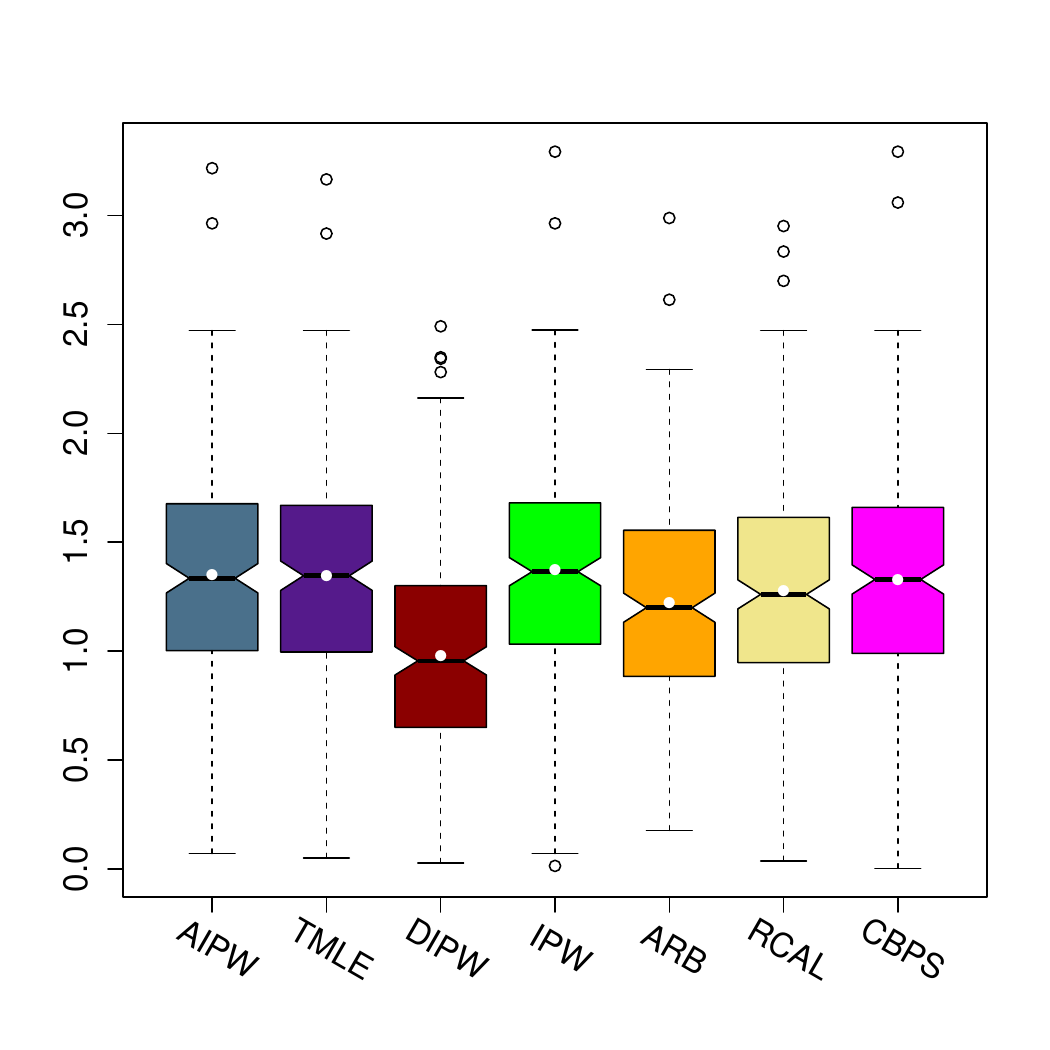}} \\
	\subfigure[Real data design, $s = 5$]{\includegraphics[width=0.32\textwidth]{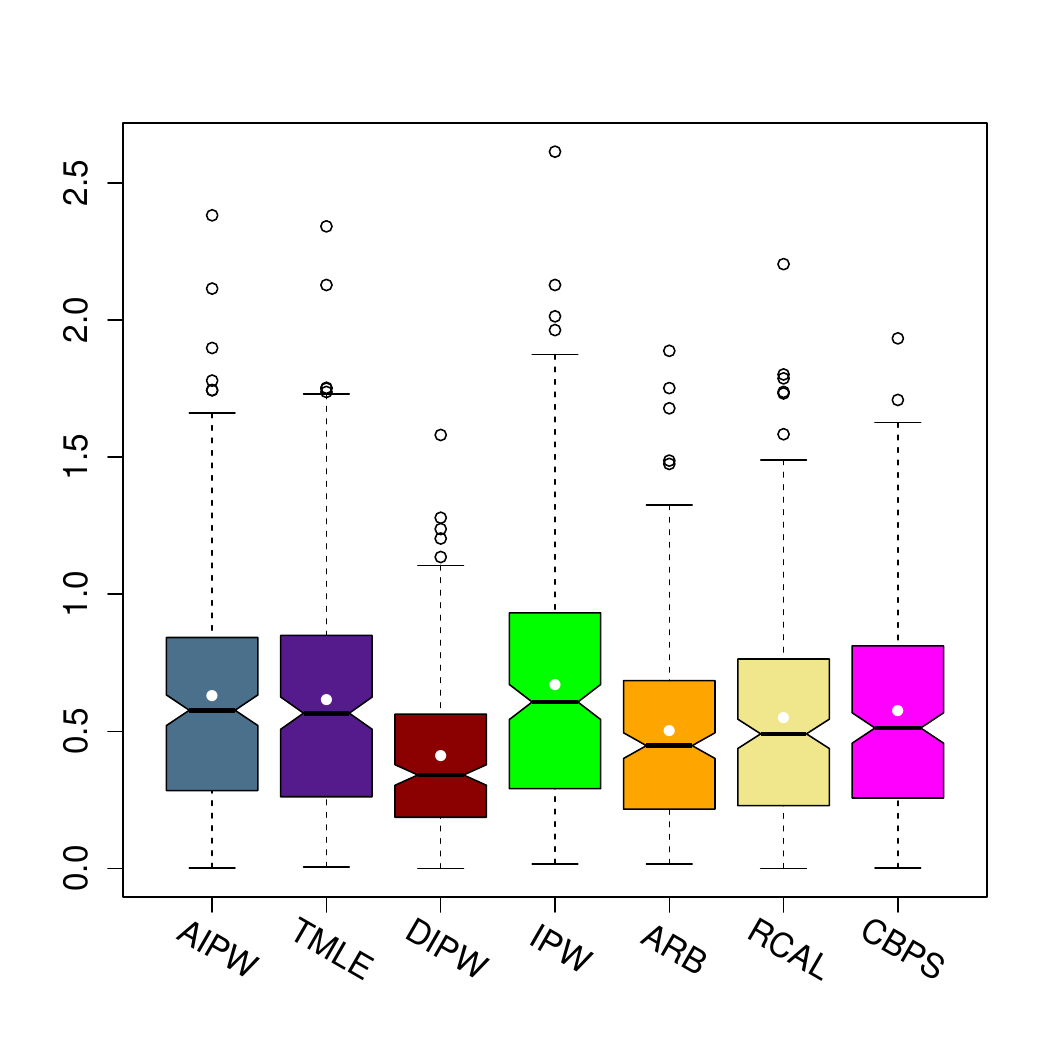}} \hfill
	\subfigure[Real data design, $s = 20$]{\includegraphics[width=0.32\textwidth]{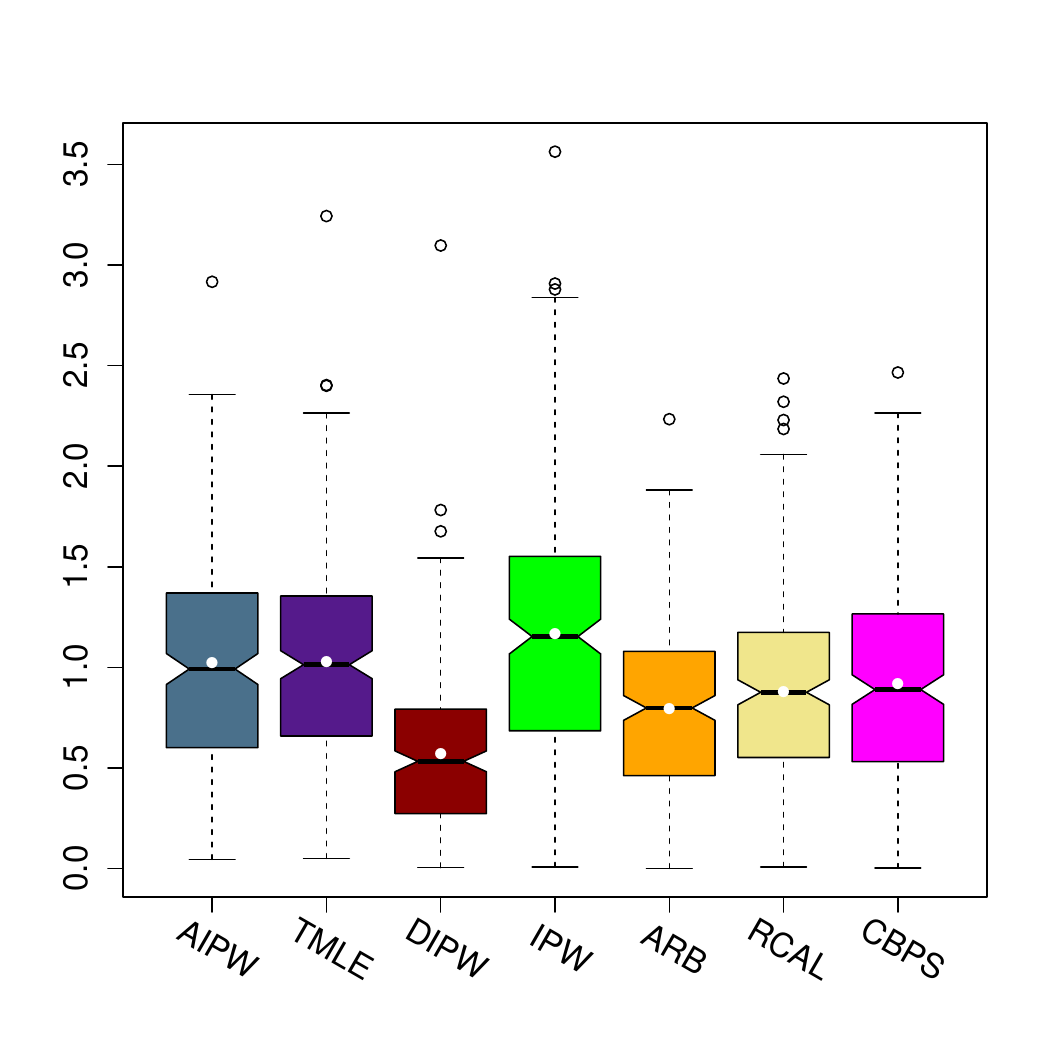}} \hfill
	\subfigure[Real data design,  $s = 50$]{\includegraphics[width=0.32\textwidth]{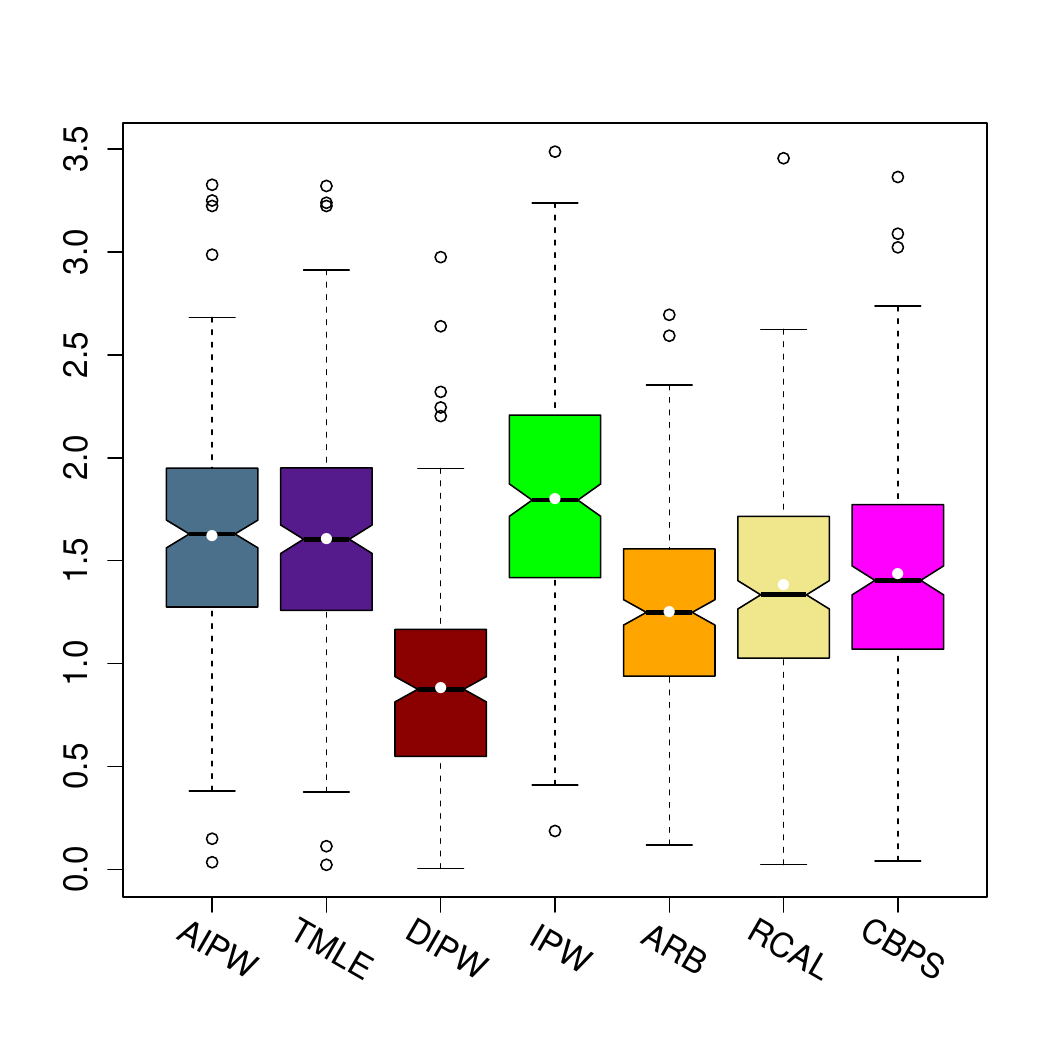}}
	\caption{Boxplot of estimation error with heteroscedastic noise and nonlinear responses.}\label{fig:hnonlinear}
\end{figure}


\section{Additional variance estimation results} \label{sec:simvar_addition}



\newrev{In Figure~\ref{fig:varrfall} we present the results of the experiments in Section~\ref{sec:simvar} but with $\|\gamma\|_2=1$, so there is a larger degree of overlap among the classes. We see that as before, the DIPW variants perform relatively well, though the difference is less pronounced overall as the esitmation problems are somewhat simpler.}

\begin{figure}[t!]
	\subfigure[Toeplitz design, $s = 5$]{\includegraphics[width=0.32\textwidth]{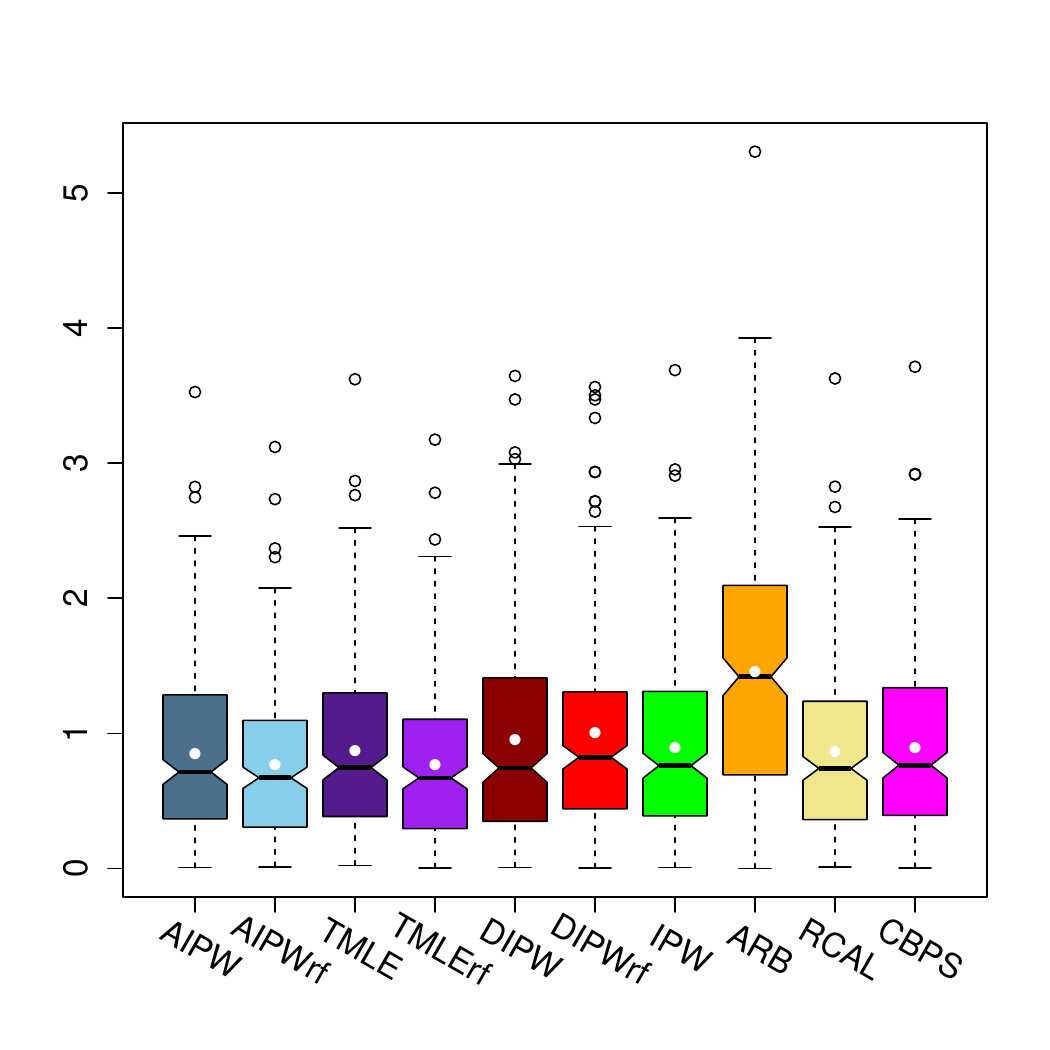}} \hfill
	\subfigure[Toeplitz design, $s = 20$]{\includegraphics[width=0.32\textwidth]{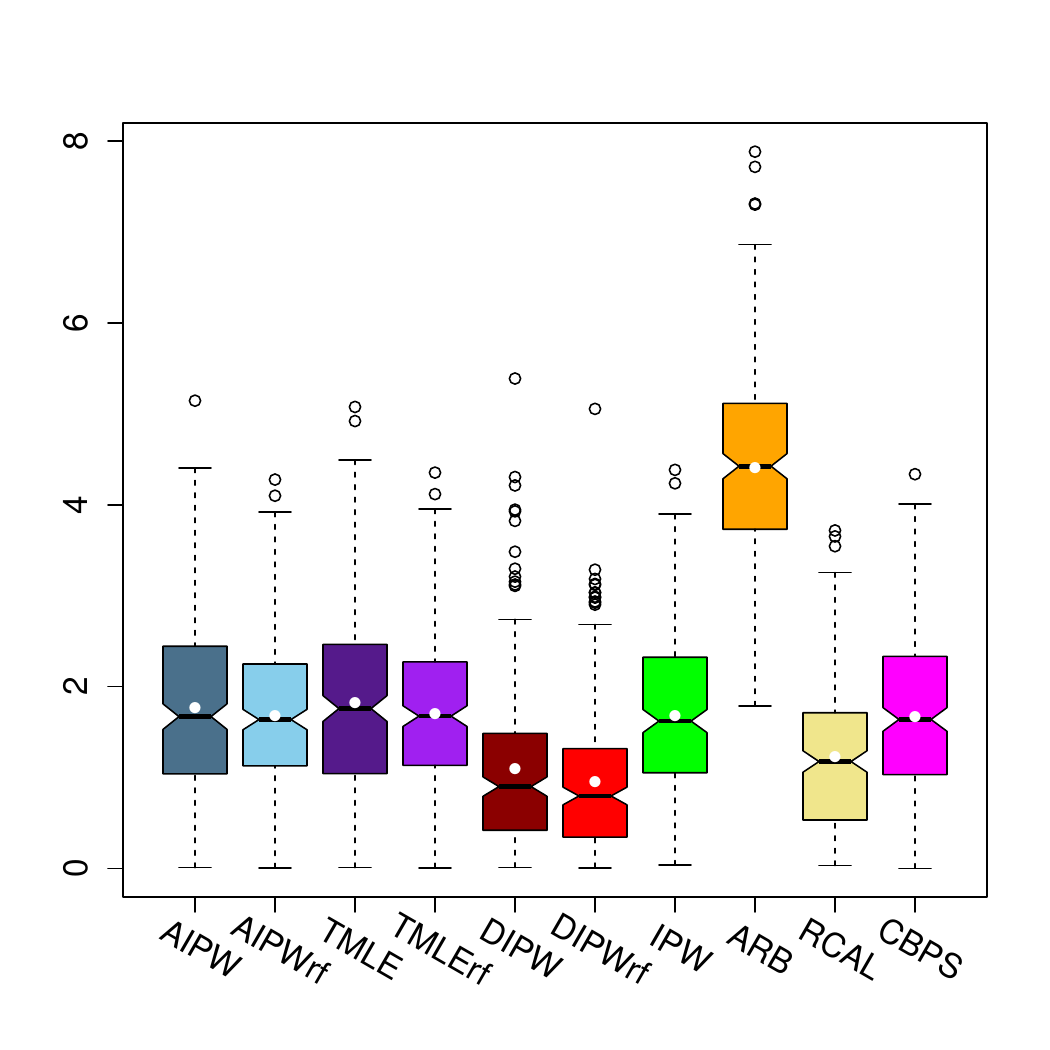}} \hfill
	\subfigure[Toeplitz design, $s = 50$]{\includegraphics[width=0.32\textwidth]{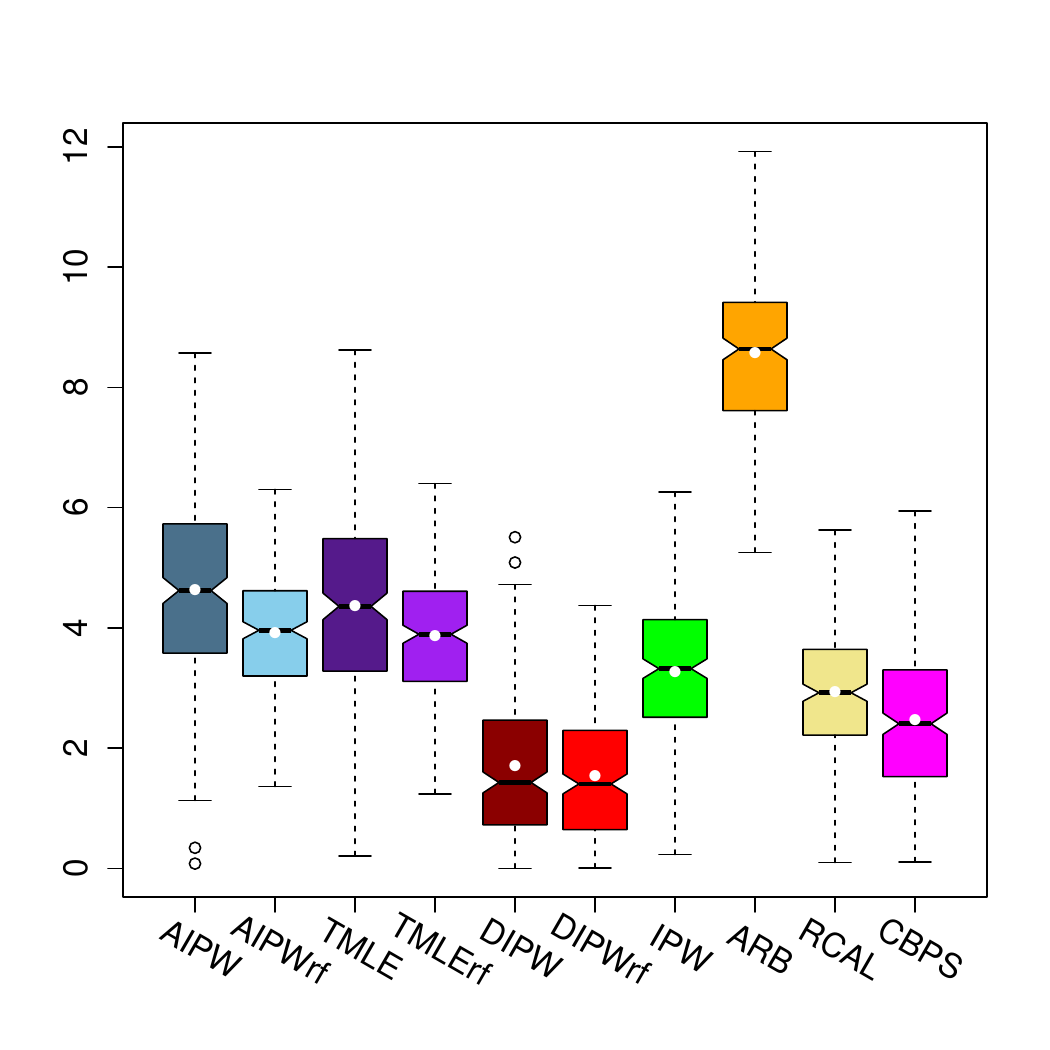}} \\
	\subfigure[Exponential design, $s = 5$]{\includegraphics[width=0.32\textwidth]{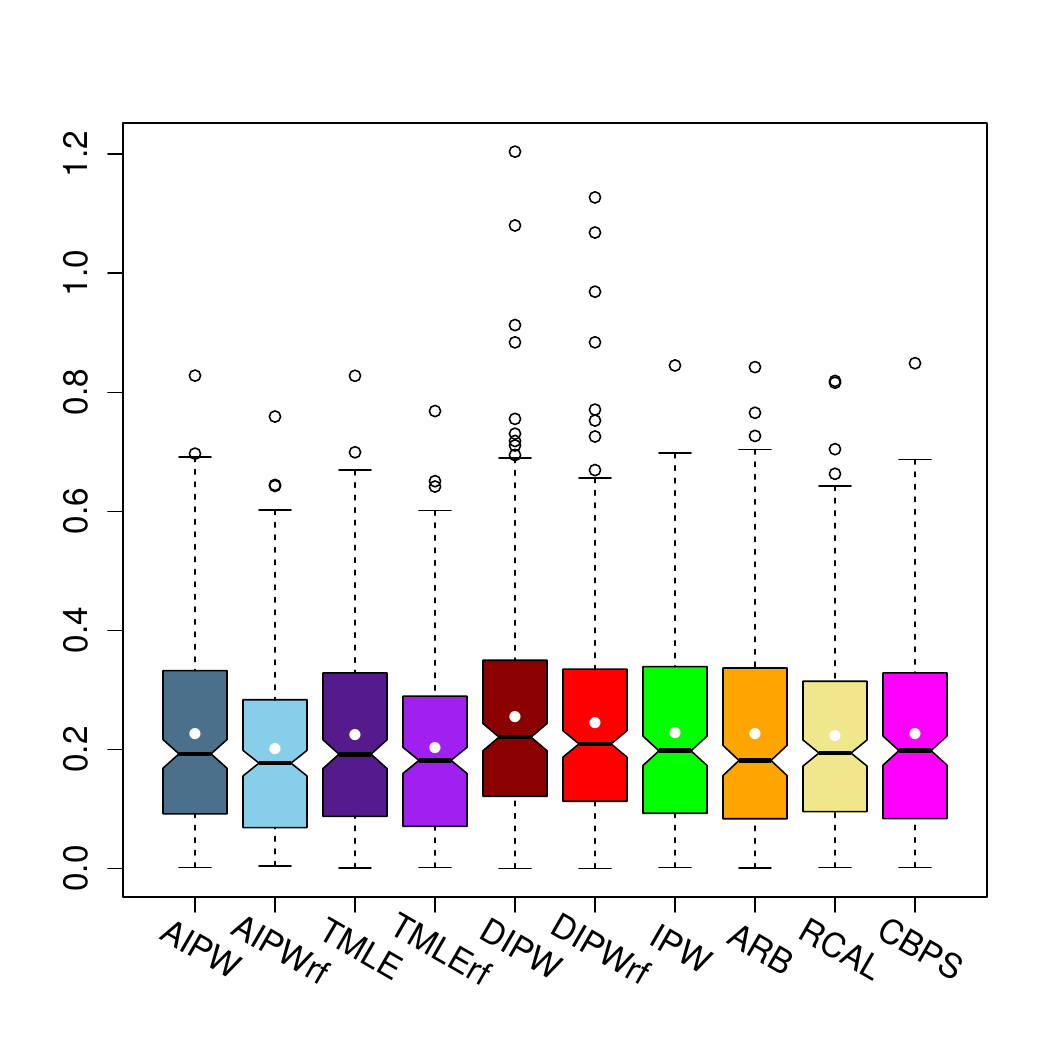}} \hfill
	\subfigure[Exponential design, $s = 20$]{\includegraphics[width=0.32\textwidth]{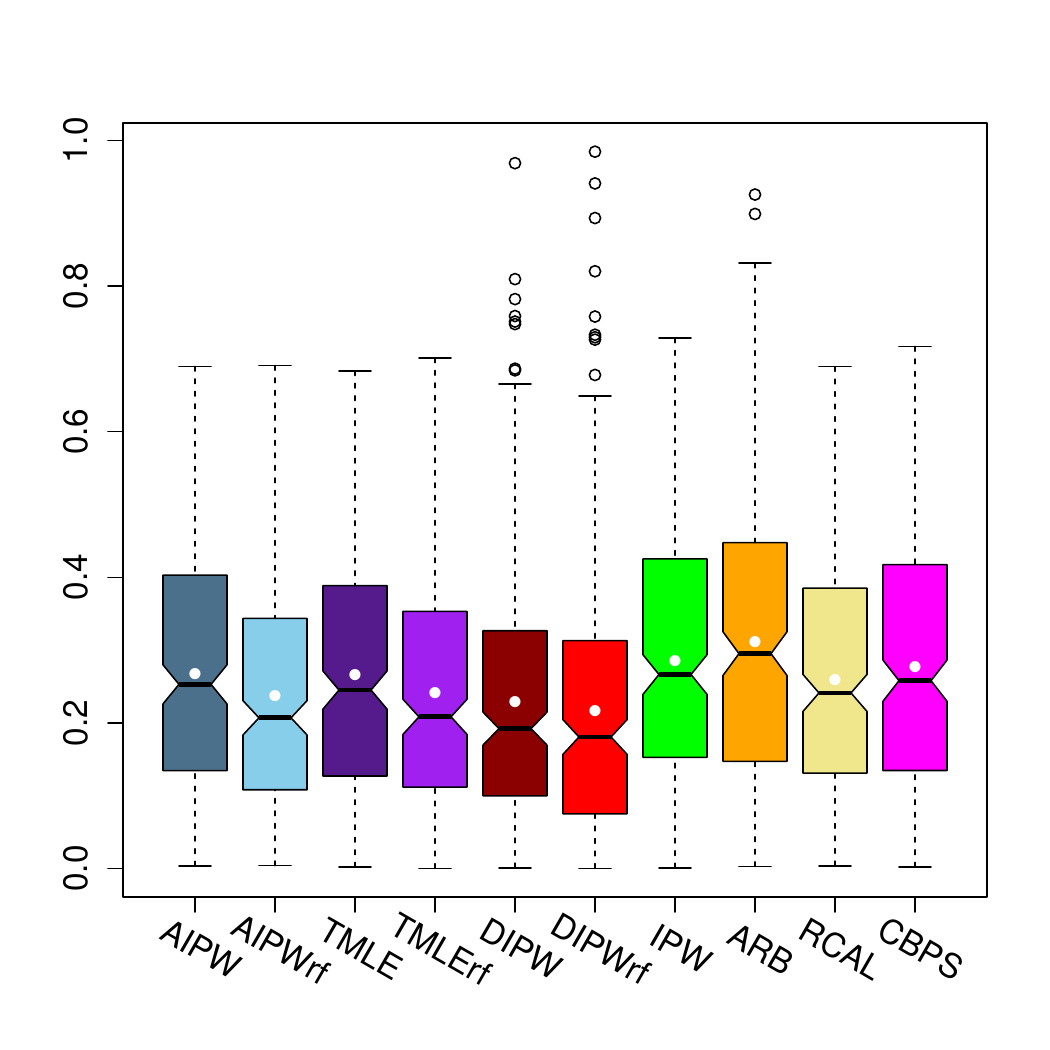}} \hfill
	\subfigure[Exponential design,  $s = 50$]{\includegraphics[width=0.32\textwidth]{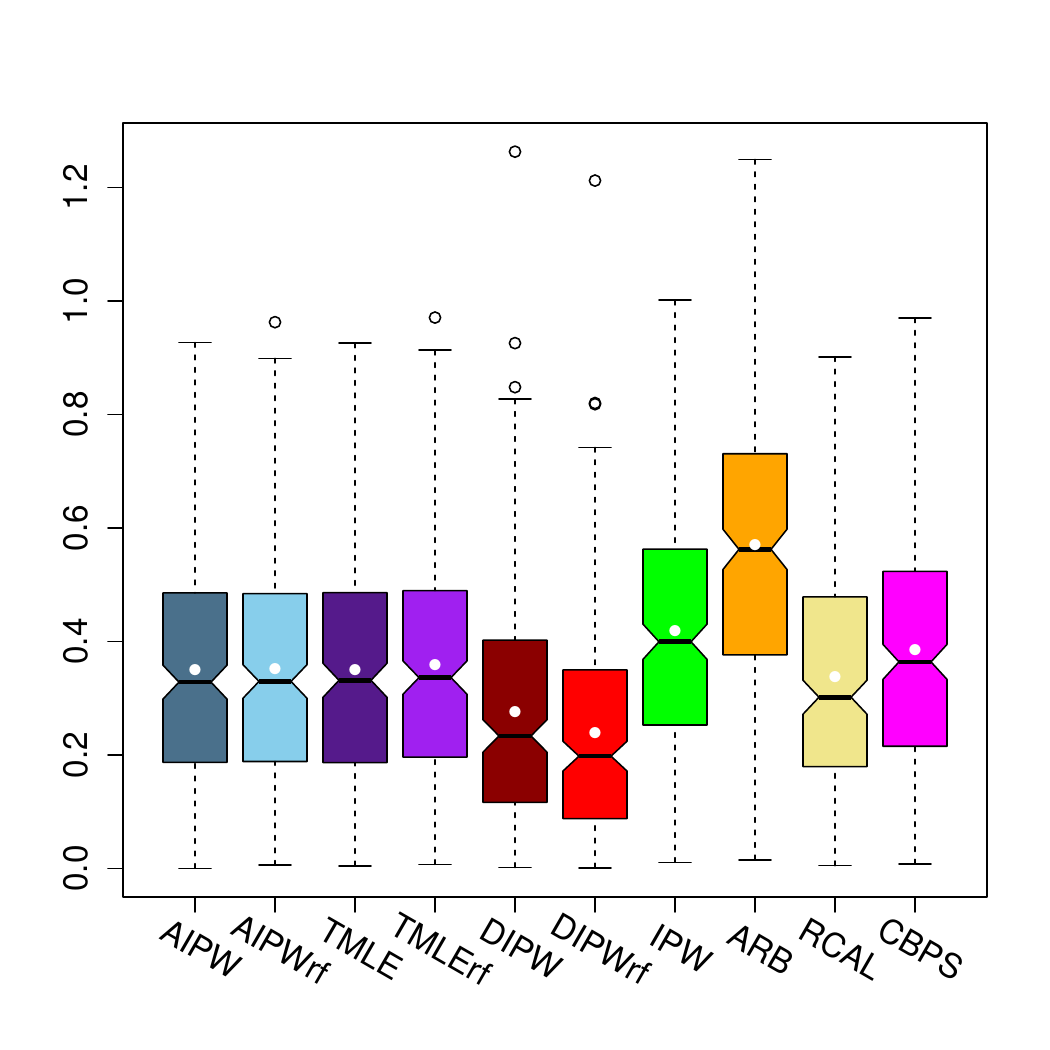}} \\
	\subfigure[Real data design, $s = 5$]{\includegraphics[width=0.32\textwidth]{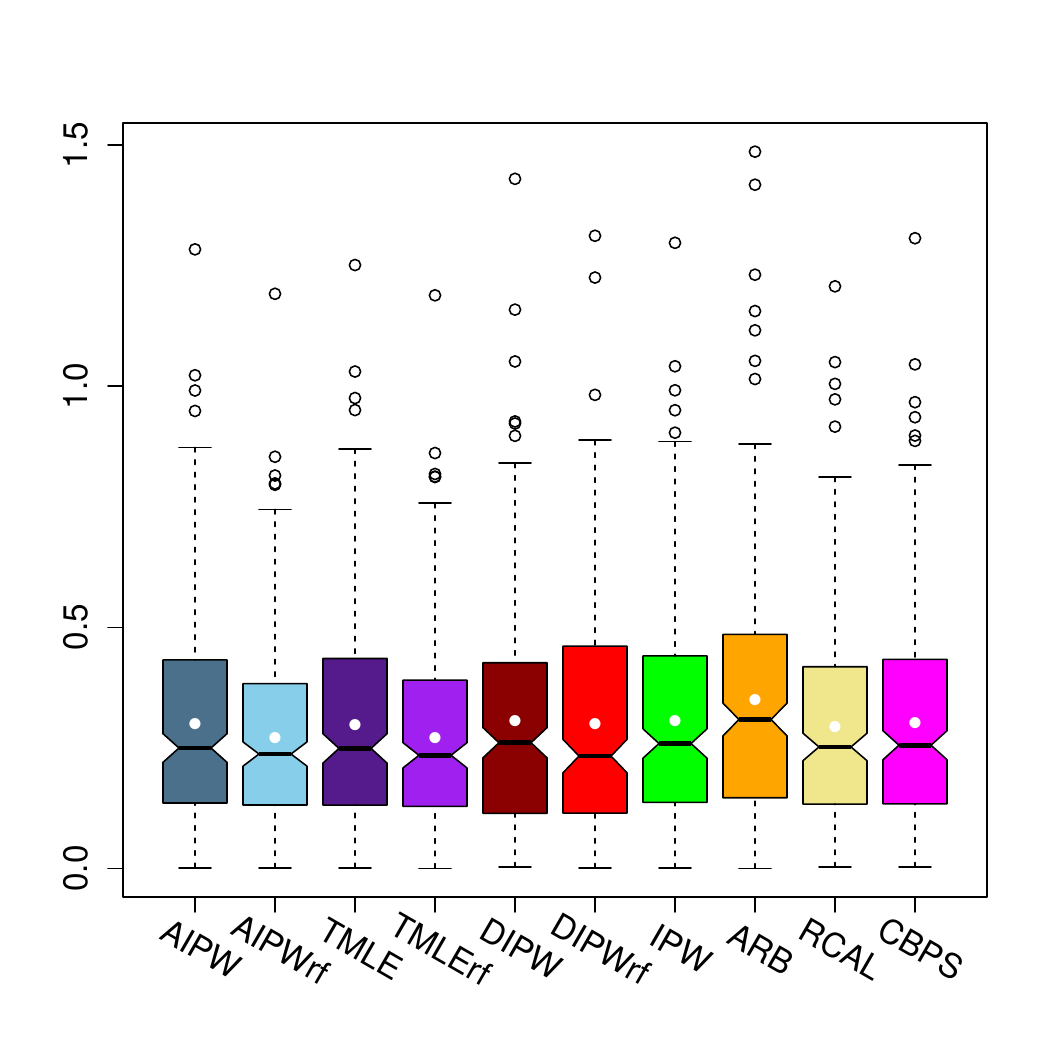}} \hfill
	\subfigure[Real data design, $s = 20$]{\includegraphics[width=0.32\textwidth]{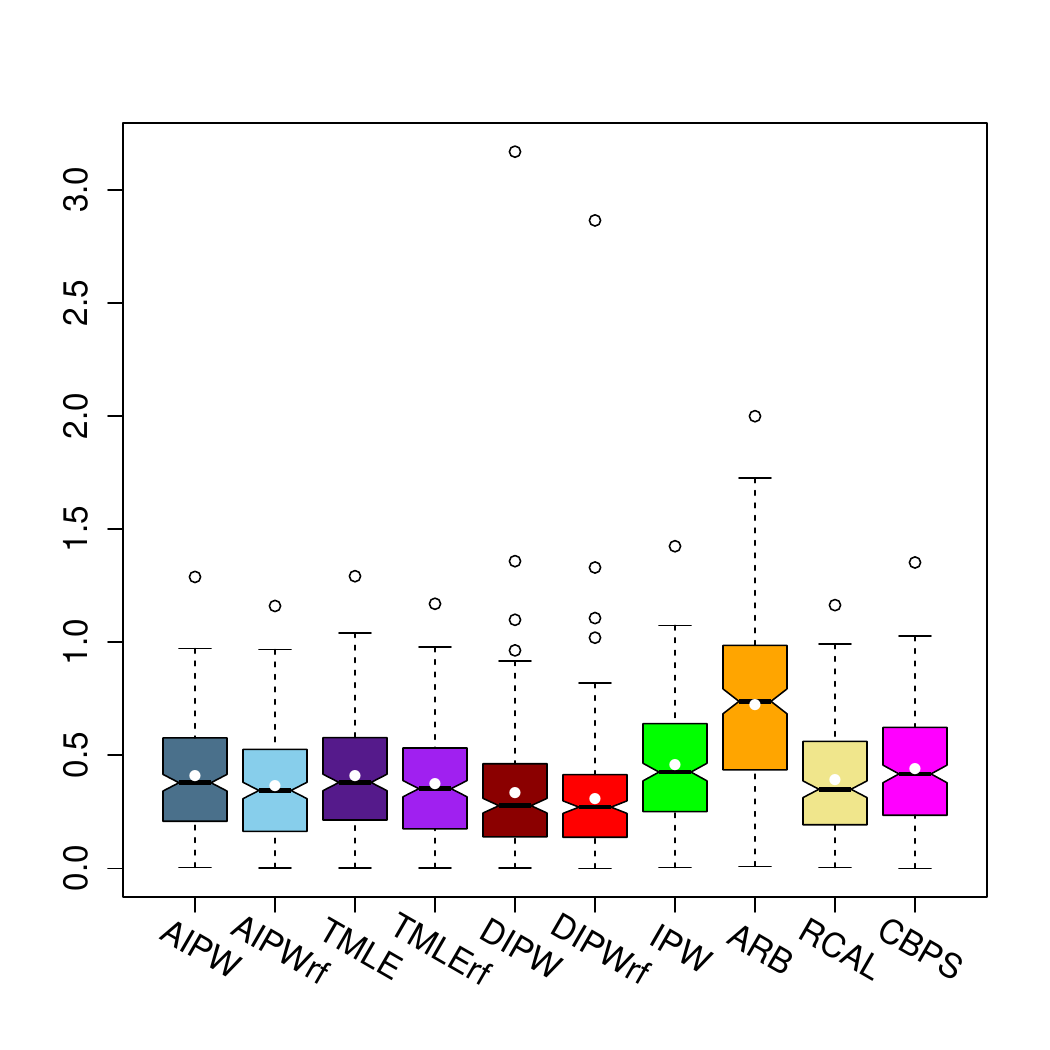}} \hfill
	\subfigure[Real data design,  $s = 50$]{\includegraphics[width=0.32\textwidth]{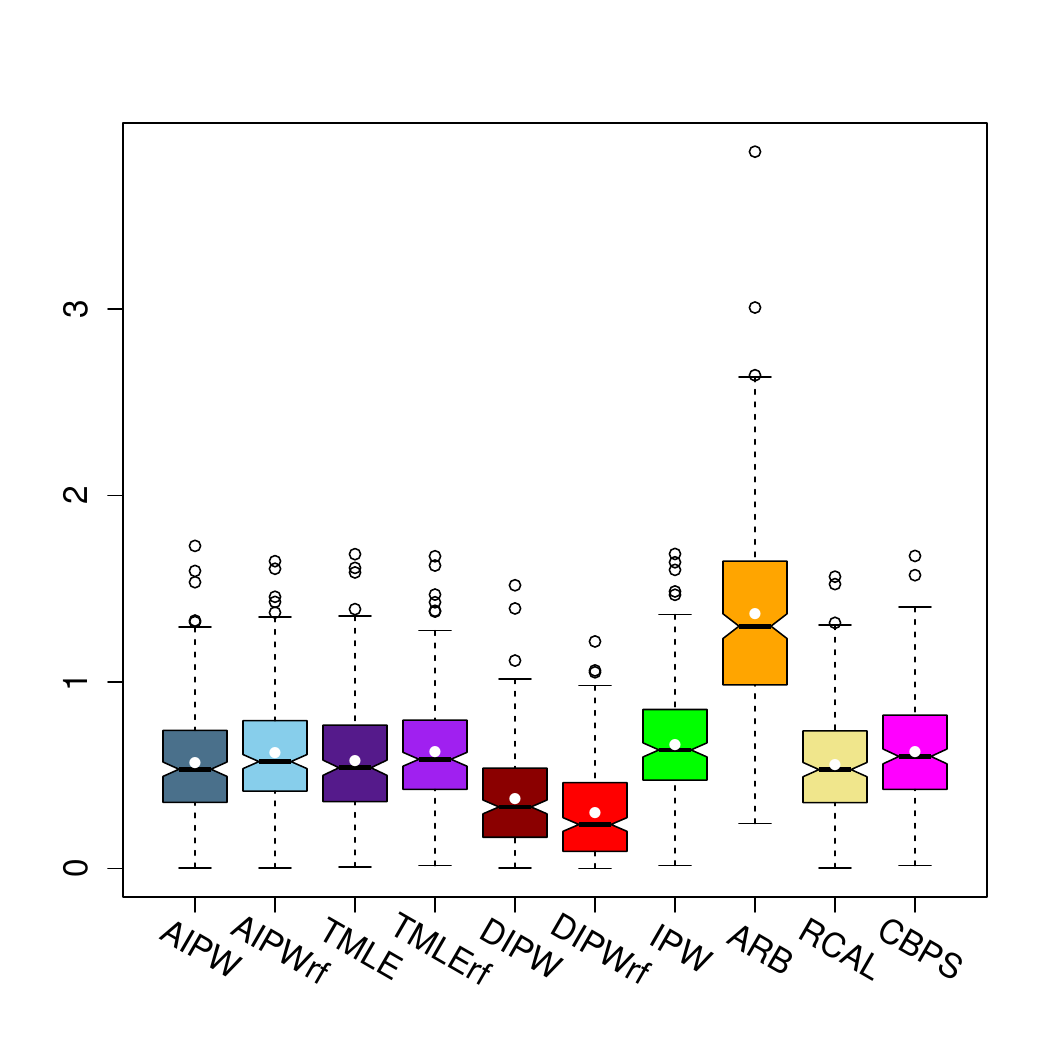}}
	\caption{Boxplots of errors as in Figure~\ref{fig:varrfall2}, but with $\|\gamma\|_2=1$ so there is greater overlap among the treatment and control groups.}\label{fig:varrfall}
\end{figure}


\end{document}